\documentclass[b5paper, fleqn, 10pt]{book}

\usepackage[utf8]{inputenc}
\usepackage[english]{babel}
\usepackage{amsmath,amssymb, amsthm, mathabx, mathtools}
\usepackage{url, pdfsetup}
\usepackage{diagrams}
\usepackage{imakeidx}
\usepackage[nottoc]{tocbibind}

\newenvironment{dedication}
{
   \cleardoublepage
   \thispagestyle{empty}
   \vspace*{\stretch{1}}
   \hfill\begin{minipage}[t]{0.7\textwidth}
   \raggedright
}%
{
   \end{minipage}
   \vspace*{\stretch{3}}
   \clearpage
}

\newenvironment{abstract}%
    {\cleardoublepage\thispagestyle{empty}\null\vfill\begin{center}%
    \bfseries\abstractname\end{center}}%
    {\vfill\null}

\newenvironment{acknowledgements}%
    {\cleardoublepage\thispagestyle{empty}\null\vfill\begin{center}%
    \bfseries Acknowledgements\end{center}}%
    {\vfill\null}

\usepackage{fancyhdr}
\makeatletter
\def\cleardoublepage{\clearpage\if@twoside \ifodd\c@page\else
  \hbox{}
  \thispagestyle{empty}
  \newpage
  \if@twocolumn\hbox{}\newpage\fi\fi\fi}
\makeatother

\DeclareMathAlphabet{\mathsfit}{T1}{\sfdefault}{\mddefault}{\sldefault}
\SetMathAlphabet{\mathsfit}{bold}{T1}{\sfdefault}{\bfdefault}{\sldefault}

\newtheorem{lemma}{Proposition}[section] 
\newtheorem{corollary}[lemma]{Corollary} 

\theoremstyle{definition}
\newtheorem{definition}{Definition}[section]

\def\eg{{\em e.g.}\:}
\def\cf{{\em cf.}\:}
\def\ie{{\em i.e.}\:}

\newcommand{\naturals}{\mathbb{N}}

\newcommand{\mystrut}{\rule[0pt]{0pt}{7pt}}
\newcommand{\stnsp}{\hspace{-0.03cm}}
\newcommand{\stsp}{\hspace{0.03cm}}
\newcommand{\stdsp}{\hspace{0.05cm}}
    
\newcommand{\scomp}{\bowtie}
\newcommand{\kstar}{*}
\newcommand{\pr}{^{\scriptscriptstyle\prime}}
\newcommand{\prr}{^{\scriptscriptstyle\prime\scriptscriptstyle\prime}}
\newcommand{\langA}[1]{\mathtt{\mathcal{L}}_{#1}}
\newcommand{\ret}[1]{\rho\:#1}
\newcommand{\op}[1]{\mathsfit{#1}}
\newcommand{\ty}[1]{\mathtt{#1}}
\newcommand{\com}[1]{\mathtt{#1}}
\newcommand{\mv}[1]{\mathit{#1}}
\newcommand{\sv}[1]{\mathbf{#1}}

\newcommand{\ljumps}{\mathcal{J}\!\mathit{umps}_{\mathit{local}}}
\newcommand{\jumps}{\mathcal{J}\!\mathit{umps}}
\newcommand{\skipp}{\com{skip}}
\newcommand{\basic}{\com{basic}}
\newcommand{\jump}[1]{\com{jump}\:#1}
\newcommand{\cjump}[3]{\com{cjump}\:#1\;\com{to}\:#2\;\com{otherwise}\;#3\;\com{end}}
\newcommand{\ite}[3]{\com{if}\:#1\;\com{then}\;#2\;\com{else}\;#3\;\com{fi}}

\newcommand{\while}[3]{\com{while}\:#1\;\com{do}\;#2\;\com{subsequently}\;#3\;\com{od}}
\newcommand{\whileS}[2]{\com{while}\:#1\;\com{do}\;#2\;\com{od}}
\newcommand{\await}[2]{\com{await}\:#1\;\com{do}\;#2\;\com{od}}
\newcommand{\wait}[1]{\com{wait}\:#1}
\newcommand{\Parallel}[1]{\parallel\!#1}

\newcommand{\sq}{\mathit{sq}}
\newcommand{\envC}[1]{\mathcal{E}\hspace{-1.1pt}\mathit{nv}\:#1}
\newcommand{\progC}[1]{\mathcal{P}\!\mathit{rog}\:#1}
\newcommand{\inC}[1]{\mathcal{I}\hspace{-1.1pt}\mathit{n}\:#1}
\newcommand{\outC}[1]{\mathcal{O}\!\mathit{ut}\:#1}
\newcommand{\envCi}[1]{\mathcal{E}\!\mathit{nv}^\omega\:#1}
\newcommand{\defeq}{\mathrel{\stackrel{\scriptscriptstyle{\mathit{def}}}{=}}}
\newcommand{\pcsi}[1]{\lsemantic#1\rsemantic^\omega}
\newcommand{\pcs}[1]{\lsemantic#1\rsemantic}
\newcommand{\rpcsi}[2]{\lsemantic#1\rsemantic^\omega_{#2}}
\newcommand{\rpcs}[2]{\lsemantic#1\rsemantic_{#2}}
\newcommand{\pstep}{\mathop{\rightarrow}\nolimits_{\scriptscriptstyle{\mathcal{P}}}\:}
\newcommand{\rpstep}[3]{#1 \vdash #2 \pstep #3}
\newcommand{\psteps}{\stackrel{\scriptscriptstyle\kstar}{\mathop{\rightarrow}}_{\scriptscriptstyle{\mathcal{P}}}}
\newcommand{\pstepsN}[1]{\stackrel{\scriptscriptstyle #1}{\mathop{\rightarrow}}_{\scriptscriptstyle{\mathcal{P}}}}
\newcommand{\rpsteps}[3]{#1 \vdash #2 \psteps #3}
\newcommand{\rpstepsN}[4]{#1 \vdash #2 \pstepsN{#3} #4}
\newcommand{\estep}{\mathop{\rightarrow}\nolimits_{\scriptscriptstyle{\mathcal{E}}}\:}
\newcommand{\rgvalid}[5]{\models\hspace{-0.5pt}\{#1,\:#2\}\:#3\:\{#4,\:#5\}}
\newcommand{\rgvalidr}[6]{#1\models\hspace{-0.5pt}\{#2,\:#3\}\:#4\:\{#5,\:#6\}}
\newcommand{\rgvalidi}[5]{\models^\omega\hspace{-1.5pt}\{#1,\:#2\}\:#3\:\{#4,\:#5\}}
\newcommand{\rgvalidri}[6]{#1 \models^\omega\hspace{-1pt}\{#2,\:#3\}\:#4\:\{#5,\:#6\}}
\newcommand{\rgvalidrext}[6]{#1\models_{\scriptscriptstyle2}\hspace{-1pt}\{#2,\:#3\}\:#4\:\{#5,\:#6\}}
\newcommand{\rgvalidext}[5]{\models_{\scriptscriptstyle2}\hspace{-1pt}\{#1,\:#2\}\:#3\:\{#4,\:#5\}}
\newcommand{\pcorr}[3]{\rho, \rho\pr \models #1\sqsupseteq_{#2}#3}
\newcommand{\pcorrG}[5]{#1, #2 \models #3\sqsupseteq_{#4}#5}
\newcommand{\peqv}[3]{\rho, \rho\pr \models #1\approx_{#2}#3}

\newcommand{\pcorrS}[3]{\rho\hspace{0.2pt} \models #1\sqsupseteq_{#2}#3}
\newcommand{\peqvS}[3]{\rho \models #1\approx_{#2}#3}
\newcommand{\peqvSG}[4]{#1 \models #2\approx_{#3}#4}
\newcommand{\pcorrC}[3]{\models #1\sqsupseteq_{#2}#3}
\newcommand{\peqvC}[3]{\models #1\approx_{#2}#3}
\newcommand{\progOf}[1]{\mathcal{P}(#1)}
\newcommand{\stateOf}[1]{\mathcal{S}(#1)}
\newcommand{\rimg}[2]{#1\cdot #2}
\newcommand{\rcomp}[2]{#1 \diamond #2}
\newcommand{\rconv}[1]{#1^{\circ}}
\newcommand{\negate}[1]{\neg #1}
\newcommand{\prefix}[2]{_{\scriptscriptstyle #1}|#2}
\newcommand{\suffix}[2]{^{\scriptscriptstyle #1}|#2}
\newcommand{\Rule}[2]{\dfrac{#1}{#2}}

\newcommand{\rgvalida}[5]{\begin{array}{l l}
                          \models\!\!\! & \{#1,\:#2\} \\
                                  & #3 \\
                                  & \{#4,\:#5\}
                           \end{array}}
\newcommand{\rgvalidap}[5]{\begin{array}{l l}
                          \models\!\!\! & \{#1,\:#2\} \\
                                  & #3 \\
                                  & \{#4,\:#5\}
\end{array}}
\newcommand{\rgvalidapext}[5]{\begin{array}{l l}
                          \models_{\scriptscriptstyle2}\hspace{-7.5pt} & \{#1,\:#2\} \\
                                  & #3 \\
                                  & \{#4,\:#5\}
                           \end{array}}

\newcommand{\myrhd}{\smalltriangleright}
\newcommand{\pos}{\mathit{{\mathcal P}\!os}}
\newcommand{\Next}{{\mathcal N}}
\newcommand{\rpcsf}[2]{\lsemantic#1\rsemantic^{\scriptscriptstyle{\mathit{fair}}}_{#2}}
\newcommand{\pcsf}[1]{\lsemantic#1\rsemantic^{\scriptscriptstyle{\mathit{fair}}}}
\newcommand{\lcond}[4]{#1 \Vdash #2 \myrhd #3 \myrhd #4}

\newcommand{\lcondN}[4]{#1 \Vdash_{\scriptscriptstyle{\mathcal N}} #2 \myrhd #3 \myrhd #4}
\newcommand{\NlcondN}[4]{#1 \not\Vdash_{\scriptscriptstyle{\mathcal N}} #2 \myrhd #3 \myrhd #4}
\newcommand{\lcondT}[5]{#1 \Vdash_{\scriptscriptstyle{\mathcal T}} #2 \myrhd #3 \myrhd #4 \myrhd #5}
\newcommand{\steqv}[2]{#1 \sim #2}
\newcommand{\plook}[2]{#1 |_{#2}}
\newcommand{\psubst}[3]{#1[#2]_{#3}}
\newcommand{\sqt}{\widehat{\sq}}
\newcommand{\sqts}{\widehat{\sigma}}

\makeindex[intoc]

\begin{document}

\frontmatter

\pagenumbering{arabic}

\begin{titlepage}
\thispagestyle{empty}
\begin{tabular}{l}
\\[3cm]
 \begin{tabular}{c} 
 \large{\textbf{\textsf{A Framework for Modelling, Verification and}}} \\ 
 \large{\textbf{\textsf{Transformation of Concurrent Imperative Programs}}}
 \end{tabular} \\[4cm]
\hspace{3.9cm}{\large{Maksym Bortin}}
\end{tabular}
\end{titlepage}

\begin{dedication}
  {\it dedicated to the people of Ukraine}
\end{dedication}

\begin{abstract}
This report gives a detailed presentation of a framework, embedded into the simply typed higher-order logic
and aimed to support sound and structured reasoning
upon various properties of models of imperative programs with interleaved computations. 
As a case study, a model of the Peterson's mutual exclusion algorithm will be scrutinised in the course of the report
illustrating applicability of the framework.
\end{abstract}

       \begin{acknowledgements}
         A perception that some research in this direction could be useful emerged when
         the author was part of the Trustworthy System group in Sydney,
         where 
         a software system component utilising multiple cores had to be formally modelled and the intended behaviour verified.
         In short, an exceptional opportunity to apply theoretical results fostering computer aided verification of practically relevant software.  

         The bulk of this work (especially concerning liveness reasoning)
         was subsequently done at the Tallinn University of Technology with support
         by the Estonian IT Academy under grant 2014-2020.4.05.19-0001.
        \end{acknowledgements}

\tableofcontents
\mainmatter

\chapter*{Introduction}\label{S:intro}
The behaviours of programs running in parallel can become entangled to such an extent that
crucial interferences may easily escape all of thorough code reviews and tests, 
leaving although rather rare yet potentially weighty faults undetected.
In situations when for example safe functioning of some nearly invaluable autonomous devices is at stake,
providing dependable assurances about the absence of faults in actions and interactions of the controlling processes can thus hardly be regarded as dispensable.
Such assurances however demand rigorous verification of the deployed program code
which basically cannot pass by formal modelling with subsequent derivation of the required properties.

From the specific perspective of safety- and mission-critical software, more recent standards encourage application of formal modelling and reasoning
for certification of systems with high safety integrity levels.
On the other hand, they also demand to explicitly outline the limits of the chosen method regarding compliance 
with the actual requirements specification, as the formal methods supplement~\cite{do2011333} to the DO-178C standard~\cite{standard:dod178c} particularly stresses.
Indeed, like with outcomes of any other measure,
applicability of formally derived conclusions to requirements on eventual implementations
is bound to a series of assumptions.

In comparison to the sequential case, development of concurrent software has significantly more sources
for that.
Considering for example some conveniently
structured code such as
\vspace{-7.5pt}
\begin{verbatim}
     ...
     b = f(a);
     ...
\end{verbatim}
\vspace{-5pt}
the input/output properties of the invoked function \verb|f|
generally depend on
the properties of concurrently running processes:
in the sequential case one can simply assert that no other process can interfere when \verb|f(a)| is evaluated drawing thus further conclusions involving the value of \verb|b|,
but if the assertion is not met then the variable \verb|b| may assume quite different values.
Besides that, it is a common practice to rely on a compiler
that translates the above fragment to an assembly level code which shall behave accordingly. 
By doing so one actually makes an implicit assumption that pushing a return address as well as the argument on a stack in preparation to the function call
also proceeds without any interference. 
In presence of interleaving it is however \emph{per se} not evident
why another process could not push some different values on the same stack in between
and hence rather hazy to what extent the conclusions concerning behaviour of the structured code 
shall be applicable to the outcome of such compilation.

To approach these kinds of problems in a cohesive manner, 
the presented framework
encompasses state abstractions as well as high- and low-level language features
facilitating structured verification of input/output, invariant and liveness properties of interleaved computations
at various levels of abstraction. Moreover, as a conservative extension to the simply typed higher-order logic extended with Hindley-Milner polymorphism,
the framework places great emphasis on the soundness of the deduced verification methods.

A sound framework is surely a must as a theoretical foundation but from the practical point of view,
rigorous
reasoning with larger models typically requires an impressive
amount of successive applications of various logical rules
which need to be instantiated at each particular step. 
And this is the point where proof assistants ultimately become relevant, being capable to accomplish
such complex, yet largely mechanisable tasks in a highly reliable and efficient way. 
Based on Robin Milner's~\cite{Milner72, Milner411} influential work, 
there are tools 
such as \emph{Isabelle}\footnote{A version of the presented framework for the Isabelle/HOL prover
is available at \url{https://github.com/maksym-bortin/a_framework}}~\cite{Nipkow_PW:Isabelle}, \emph{HOL-Light}~\cite{DBLP:conf/tphol/Harrison09a} and \emph{HOL}~\cite{DBLP:conf/tphol/Gordon91}
that provide the assistance, which is particularly valuable concerning compliance with the mentioned software standards~\cite{do-case-studies}. 

Regarding related work, 
the development essentially leans on
the work by Stirling~\cite{STIRLING1988347}:
starting with the concept of potential computations controlled by
the environment, program, input and output conditions,
it systematically extends the Owicki and Gries method~\cite{Owicki_Gries_76} to
a Hoare-style rely/guarantee program logic for a concurrent while-language.
With a few adaptations the notions from~\cite{STIRLING1988347} become ubiquitous in what follows. 
Moreover, the technique combining deep and shallow embedding with
abstraction over the type of program states, 
which has been notably applied by Nieto~\cite{DBLP:phd/dnb/Nieto02} and Schirmer~\cite{Schirmer:PhD},
is also at the core of the presented encoding of the framework's language in the higher-order logic,
whereas the elegant approach to a light-weight extension
of the Hoare-style program logic enabling
state relations as postconditions has been adopted
from the presentation by van Staden~\cite{DBLP:conf/mpc/Staden15} and slightly tuned.
More of the related work will be discussed at the beginning of Chapter~\ref{S:pcorr} and in Chapter~\ref{S:live}. 

This report is organised as follows. Next chapter gives brief explanations for some basic decisions
behind the framework's design. Chapter~\ref{S:lang} describes the framework's language
and the underlying computational model featuring interleaving and jumps.
Chapter~\ref{S:pcorr} is devoted to the notion of program correspondence and its properties.
Chapter~\ref{S:pcs-props} introduces the relevant conditions on potential computations,
places correspondences in this context and pays moreover due attention to the special case of program equivalence.
Chapter~\ref{S:prog-log} presents 
a Hoare-style logic and
Chapter~\ref{S:PM1} applies it to derive a rule for establishing input/output properties of a
generic model of Peterson's mutual exclusion algorithm.
Chapter~\ref{S:gener} shows how the program logic from Chapter~\ref{S:prog-log}
can be lifted to enable reasoning with state relations as postconditions,
whereas Chapter~\ref{S:PM2} utilises this extension to strengthen the results of Chapter~\ref{S:PM1}. 
Based on a concise notion of fair computations that is elaborated first,
Chapter~\ref{S:live} describes an approach to verification of liveness properties.
Essentially all of the developed techniques will subsequently be applied in Chapter~\ref{S:PM3} which concludes the case study by proving termination of
the mutual exclusion model.
Finally, Chapter~\ref{S:concl} concludes the entire presentation outlining most important enhancements. 
\chapter{Outlining Some Basic Design Decisions}\label{S:concept}
  This small chapter offers a superficial introduction to the framework's
  `low-level instructions' and program transformations involving these.
  Although reaching an assembly level representation of a concurrent program does not need to be an objective of 
  each development process, enabling this option
  shapes the framework's design as outlined below. 
  This will be accomplished without delving into much detail, appealing rather to the intuition
  that each program can be transformed to another program that provably does exactly the same but uses only jumps instead of conditional and while-statements to this end. 
  Based on this, the claim in what follows will be that any appropriately shaped program can be brought into a form which is evaluated
  by the computational model (\cf\stsp Section~\ref{Sb:sem}) without resorting to instructions 
  unknown to an ordinary assembly interpreter.

  One of the most striking differences between a structured and an assembly-level language is
  that in the former case 
  we can conveniently follow the syntax tree in order to evaluate a program.
  For instance, the sequential composition of some $p$ and $q$ (as usually denoted by $p;\hspace{-0.5pt}q$)
  is evaluated by descending to the subtree $p$, processing this to some $p\pr$ and 
  `pasting' $p\pr$ back which ultimately results in $p\pr;q$ or just $q$ whenever $p\pr\! =\! \skipp$ \ie the left branch has been completely evaluated.
  Placing jumps in this context, the effect of the
  command $\cjump{\hspace{0.7pt}C\hspace{-0.5pt}}{\hspace{0.5pt}i\hspace{0.5pt}}{\hspace{0.4pt} p}$
  shall be the following:
  if $C$ holds in the current state then we proceed to the code $\ret{\hspace{1.2pt} i}\stsp$ associated to the label $i$ by means of a `retrieving' function $\rho$
  and continue with $p$ otherwise.
  Although this in principle mimics
  processing of conditional jumps by assembly interpreters, evaluations of the sequential composition
  $\cjump{\hspace{-0.2pt} C\hspace{-1.9pt}}{\hspace{0.7pt} i\stnsp}{\hspace{-0.5pt} p\hspace{-0.5pt}};q\stsp$
  do not:
they would simply advance to
the same $q$ once $\ret{\hspace{2.1pt} i}$ or $p$ is done with its computations. 
Depending on where the code block $\ret{\hspace{1pt} i}$ is actually placed in the program text,
such behaviour can generally be achieved at the assembly level only when either $p$ or $\ret{\stsp i}$   
concludes with an explicit jump to the entry label of $q$.
These considerations 
underline that deploying jumps in tree-structured programs demands certain preparations
to make sense from the assembly-level perspective and the transformations, sketched below, take care of that. 
\subsubsection{(1a) expressing conditional statements by jumps}
A statement $\ite{\hspace{0.9pt}\mv{C}}{p_1}{p_2}$ can equivalently be evaluated by evaluating $\cjump{\negate{C}}{\stsp i\stsp}{p_1}$ instead
if we  additionally ensure $\ret{\stsp i} = p_2$ with a fresh label $i$
($\negate{C}$ will denote the complement of the condition $C$).
Applying this rule straight to a program such as $\ite{\mv{C}}{p_1}{p_2};q$ yields $\cjump{\negate{C}\hspace{-0.5pt}}{\stsp i}{p_1\hspace{-2pt}};q$ with $\ret{\stsp i} = p_2$
and hence the issues pointed out above which are however addressed by the next rule.
\subsubsection{(1b) sequential normalisation of conditional statements}
Any sequential composition $\ite{\hspace{-0.5pt}\mv{C}\hspace{-1.5pt}}{\hspace{-1pt}p_1\hspace{-2pt}}{\hspace{-0.5pt}p_2\hspace{-1pt}};q$ can be transformed to 
the equivalent $\ite{\hspace{-0.7pt}\mv{C}\hspace{-0.7pt}}{\hspace{0.2pt} p_1;q\hspace{0.2pt}}{\hspace{0.2pt}p_2;q\hspace{0.2pt}}$ distributing essentially all of the subsequent code to both branches.
Applying (\textbf{1a}) to 
the latter,
we consequently obtain $\cjump{\stnsp\negate{C}\hspace{-0.5pt}}{\stdsp i\hspace{0.5pt}}{p_1;q}$ with $\ret{\stsp i} = p_2;q$,
whose evaluations do not resort to implicit jumps. 
\subsubsection{(2a) expressing while-statements by jumps}
Attempting to represent the usual syntactic construction $\whileS{\hspace{1.9pt} C\hspace{-0.5pt}}{p}$ using jumps, we simply would not know where to jump when
$C$ does not hold: 
we need access to the entire syntax tree, \eg \hspace{1.5pt}$\whileS{\hspace{0.2pt}C\hspace{-1.7pt}}{\hspace{-1pt}p\hspace{-0.9pt}};q$, to acquire this information.
In response to that the syntax of while-statements gets extended to $\while{\hspace{-0.5pt}C\hspace{-1.5pt}}{\hspace{-1pt}p_1\hspace{-1.5pt}}{\hspace{-1pt}p_2\hspace{-1pt}}$ 
which by contrast can equivalently be represented by $\cjump{\stdsp\negate{C}\hspace{0.2pt}}{\hspace{2.1pt} j\hspace{0.5pt}}{p_1;\jump\hspace{2.5pt} i\stsp}$
with two fresh different labels $i$ and $j$, defining additionally a retrieving function
\[
\ret{\stsp i}\; \defeq\; \cjump{\negate{\mv{C}}}{\stsp j\stsp}{p_1;\jump{\stdsp i\stsp}}
\]
and $\ret{\hspace{0.1pt}j} \hspace{-0.7pt}\;\defeq\; p_2$
to resolve these.

It is moreover worth noting that the extra unconditional jump instruction $\jump{\hspace{0.2pt} i}$ following $p_1$
opens possibilities for interleaving which are not present with $\while{\hspace{-0.2pt}\mv{C}\hspace{-1.2pt}}{\hspace{-0.5pt}p_1\hspace{-0.5pt}}{p_2}$
and the according evaluation rule in Section~\ref{Sb:sem}
will take accounts of that by means of an exclusively for this purpose used $\skipp$.
\subsubsection{(2b) sequential normalisation of extended while-statements}
Provided by the extended syntax of while-statements we proceed as in (\textbf{1b}):
any (sub)program of the form
$\while{\hspace{0.2pt}\mv{C}\stnsp}{\hspace{-0.2pt}\mv{p}_1\hspace{-0.2pt}}{\stnsp\mv{p}_2};q$ gets replaced by
the equivalent $\while{\hspace{0.2pt}\mv{C}}{\mv{p}_1}{p_2;q}$. 
\subsubsection{(3) halting labels for sequential components}
Let $p_1, \ldots, p_n$ be sequential programs that shall run in parallel which is denoted by $\parallel\hspace{-4.9pt} p_1, \ldots, p_n$.
Without any significant change to the overall behaviour we can extend each component $p_i$ by the final jump to a fresh `halting' label $h_i$, \ie
$p_i ; \jump{h_i}$ defining
moreover $\rho(h_i) \defeq \skipp$.

By successive application of the rules (\textbf{1b}) and (\textbf{2b}) and then (\textbf{1a}) and (\textbf{2a}) to each $p_i ; \jump{h_i}$ we eventually obtain
$\parallel\hspace{-2.9pt} p\pr_1, \ldots, p\pr_n$ together with a retrieving function $\rho$ that resolves all labels.
This outcome can directly be interpreted as $n$ `flat' lists of labelled instructions
with essentially the same computational effect as $\parallel\hspace{-2.9pt} p_1, \ldots, p_n$.
\subsubsection{Summary}
  To sum up, this section outlined questions that arise when viewing evaluation of structured programs with conditional and while-statements replaced by jumps
  from the perspective of an assembly interpreter and how these questions are addressed by the framework.
  All that will reappear in the following chapters which by contrast are kept formal.
  \chapter{A Generic Concurrent Imperative Language}\label{S:lang}
  The notions of shallow and deep embeddings arise quite naturally when a language needs to be modelled
  in a formal system such as a proof assistant (\cf \cite{DBLP:conf/tphol/Harrison09a}). 
  The technique, notably applied by Nieto~\cite{DBLP:phd/dnb/Nieto02} and Schirmer~\cite{Schirmer:PhD} aiming to combine the merits of both,
  will be used in this setting. 
  More precisely, the syntax of the framework's language is basically deeply embedded
  by means of a free construction. 
  However, all
  indivisible state transformations as well as control flow conditions
  (such as the variable $C$ in the conditional and while-statements that appeared in the preceding chapter) are embedded shallowly
  or, in other words, the language is generated over all appropriately typed logical terms
  so that neither one-step state updates nor decisions whether
  \eg a conditional statement $\ite{\hspace{-0.2pt}\mv{C}\hspace{-1pt}}{\hspace{-0.5pt}p\hspace{-0.5pt}}{\hspace{-0.5pt}q}$ branches to $p$ or to $q$ 
  require any additional evaluations of language terms.
\section{The syntax}\label{Sb:lang}        
This is modelled by means of the parameterised
type $\langA{\alpha}$, where $\alpha$ is a type parameter representing underlying states and $\mathcal{L}$ -- a type constructor
sending any actual type $\tau$ supplied for $\alpha$ to the new type $\langA{\tau}$.
For example, $\langA{\ty{int} \times \ty{int}}$ would be an instance with
the states carrying two integer variables.
Note that such instantiations will be relevant for specific modellings only, 
\ie \stsp states remain abstract at the level of uniform program logical reasoning.

The terms of type $\langA{\alpha}$ are constructed by the following grammar\index{framework language!grammar}:
\[
\begin{array}{l c l}
  \langA{\alpha} & \mbox{::=} & \skipp \\
                 &  | & \basic\hspace{2.5pt} f \\ 
  &      |       & \cjump{\stsp C}{\stsp i\stsp}{\langA{\alpha}}\index{jump!conditional}\index{framework language!conditional jump} \\
  &      |      &  \langA{\alpha} ; \langA{\alpha}\index{framework language!sequential composition} \\
                     &      |      & \Parallel{\!\mathcal{L}^*_\alpha}\index{framework language!parallel composition} \\
                   &        |    & \while{\stsp C}{\langA{\alpha}}{\langA{\alpha}}\index{framework language!while-statement} \\
  &       |     & \ite{\stsp C}{\langA{\alpha}}{\langA{\alpha}}\index{framework language!conditional statement} \\
                   &      |      & \await{\stsp C}{\langA{\alpha}\index{framework language!await-statement}}
\end{array}
\]
where
\begin{itemize}
\item[-] $i\stsp$ is called a \emph{label}\index{framework language!label} ranging over the natural numbers (\ie is of type $\ty{nat}$, but any infinitely countable set of identifiers, such as strings,
may in principle be employed here as well)
\item[-] $f$ is called a \emph{state transformer}\index{framework language!state transformer} ranging over the values of type $\alpha \Rightarrow \alpha$
\item[-] $C$ is called a \emph{state predicate}\index{framework language!state predicate} ranging accordingly over the values of type $\alpha \Rightarrow \ty{bool}$.
\end{itemize}
Furthermore,  $\mathcal{L}^*_\alpha$ as usually stands for a finite sequence of elements of type $\langA{\alpha}$.
So $\Parallel{\hspace{-2.9pt}p_1, p_2, p_3}$ constructs a new value of type $\langA{\alpha}$
out of the three values of this type. In general, the term $\Parallel{\hspace{-1.9pt}p_1, \ldots, p_m}$ is 
called the \emph{parallel composition} of $p_i$ for $1\hspace{-1pt} \le i \le m$.
A more convenient and common notation $p_1 \hspace{-2.5pt}\parallel\hspace{-1pt} p_2$ is used
in place of $\Parallel{\hspace{-2.5pt}p_1, p_2}$. The reason for having $\Parallel{\hspace{-2.9pt}p_1, \ldots, p_n}$ as
a primitive, instead of using the binary operator with nesting,
will be given later in Section~\ref{Sb:corr_pcs}.
Regarding sequential composition, $p_1;p_2;p_3$ will be used as a shorthand for $p_1;(p_2;p_3)$. 
Moreover,  $\while{\hspace{2.5pt} C}{ p\stsp}{\skipp}$ will be abbreviated by $\whileS{\hspace{1.5pt} C}{p}$ whereas
 $\await{\stdsp C}{\skipp}$ -- by \stsp$\wait{\hspace{0.7pt}C}$\index{framework language!await-statement!$\wait$}.

A term of type $\langA{\alpha}$ is called \emph{jump-free}\index{jump-free!term} if it does not make use of $\com{cjump}$,
and it is called \emph{locally sequential}\index{locally sequential term} if it does not make use of the parallel operator.
\section{The underlying computational model}\label{Sb:sem}
Program steps are constructed following 
the principles of `Structural Operational Semantics' by Plotkin~\cite{PLOTKIN20043}.
Carrying on the state abstraction built into the language, 
let $\alpha$ be a fixed arbitrary type throughout this section.
\subsection*{Program steps}\label{Sb:psteps}
In the context of a \emph{(code) retrieving function}\index{retrieving function} $\rho\stsp$ of type $\ty{nat}\hspace{-0.5pt} \Rightarrow\hspace{-0.5pt} \langA{\alpha}$,
the effects of one program step\index{computation!step!program} performed by a term $p$ of type $\langA{\alpha}$ in a state $\sigma$ of type $\alpha$
are captured by means of the relation $\pstep\hspace{-2.9pt}$ of type
$\langA{\alpha}\hspace{-2.5pt} \times\hspace{-0.9pt} \alpha \Rightarrow \langA{\alpha} \hspace{-2.5pt}\times\hspace{-0.9pt} \alpha \Rightarrow \ty{bool}$,
connecting thus two configurations: a \emph{configuration}\index{configuration} $\zeta = (p,\stsp \sigma)$ comprises
\begin{enumerate}
\item[--] the \emph{program part}\index{configuration!program part} $p$ 
  denoted by $\progOf{\zeta}$,
\item[--] the \emph{state part}\index{configuration!state part} $\sigma$ 
  denoted by $\stateOf{\zeta}$. 
\end{enumerate}
A $\skipp$-\emph{configuration}\index{configuration!skip} has thus the form $(\skipp,\stsp \sigma)$ where $\sigma$ is some state.

The notation $\rpstep{\rho\hspace{0.2pt}}{(p,\stsp \sigma)\stdsp}{\hspace{-0.5pt}(q, \stsp\sigma\pr)}$ is used from now on to indicate that
the configuration $(p,\stsp \sigma)$
can be transformed via $\stsp\pstep\hspace{-2.9pt}$
to $(q,\stsp \sigma\pr)$ in the context of $\rho$.

The set of possible transformation steps is defined inductively by the rules listed in Figure~\ref{fig:pstep}, 
explained in more detail below.
\begin{figure}
\[
\begin{array}{l l}
\mbox{\it\sl Basic} & \Rule{}{\rpstep{\rho\stsp}{(\basic\hspace{2pt} f,\stdsp \sigma)\stsp}{\stnsp(\skipp,\stdsp f\hspace{2pt}\sigma)}} \\
& \\
\mbox{\it\sl CJump-True} & \Rule{\sigma \in\hspace{-0.2pt} C}{\rpstep{\rho\stsp}{(\cjump{\stsp C}{\stsp i\hspace{0.5pt}}{\hspace{0.5pt}p},\stdsp \sigma)\stsp}{\stnsp(\ret{\stsp i},\stdsp \sigma)}} \\
& \\
\mbox{\it\sl CJump-False} & \Rule{\sigma \notin\hspace{-0.2pt} C}{\rpstep{\rho\stsp}{(\cjump{\stsp C}{\stsp i\hspace{0.5pt}}{\hspace{0.5pt}p}, \stdsp\sigma)\stsp}{\stnsp(p,\stdsp \sigma)}} \\
& \\
 \mbox{\it\sl Sequential} & \Rule{\rpstep{\rho\stsp}{(p_1,\stdsp \sigma)\stsp}{\stnsp(p\pr_1,\stdsp \sigma\pr)}}
              {\rpstep{\rho\stsp}{(p_1 ; p_2,\stdsp \sigma)\stsp}{\stnsp(p\pr_1 ; p_2,\stdsp \sigma\pr)}} \\
 & \\
\mbox{\it\sl Sequential-skip} & \Rule{}{\rpstep{\rho\stsp}{(\skipp ; p,\stdsp \sigma)\stsp}{\stnsp(p,\stdsp \sigma)}} \\ 
 &  \\
\mbox{\it\sl Parallel} & \Rule{\rpstep{\rho\stsp}{(p_i,\stsp \sigma)\stsp}{\stnsp(p\pr_i, \stsp\sigma\pr)} \quad\quad 1 \le i \le m \quad\quad m > 0}
{\rpstep{\rho\stsp}{(\Parallel{\hspace{-1.2pt}p_1,\ldots p_i \ldots, p_m},\stdsp \sigma)\stsp}{\stnsp(\Parallel{\hspace{-1.2pt}p_1,\ldots p\pr_i \ldots, p_m},\stdsp \sigma\pr)}} \\
 & \\
\mbox{\it\sl Parallel-skip} & \Rule{}
     {\rpstep{\rho\stsp}{(\Parallel{\!\skipp,\ldots, \skipp},\stdsp \sigma)\stsp}{\stnsp(\skipp,\stdsp \sigma)}} \\
      & \\
\mathrm{\it\sl Await} & \Rule{\sigma \in\hspace{-0.2pt} C \quad\quad \rpsteps{\rho\stsp}{(p,\stdsp \sigma)\stsp}{\stnsp(\skipp,\stdsp \sigma\pr)}}
       {\rpstep{\rho\stsp}{(\await{\stsp C}{p},\stdsp \sigma)\stsp}{\stnsp(\skipp, \stdsp\sigma\pr)}} \\
 & \\
\mbox{\it\sl While-True} & \Rule{\sigma \in\hspace{-0.2pt} C \quad\quad x = \while{\stdsp C\stnsp}{p_1}{p_2}}
                          {\rpstep{\rho\stsp}{(x,\stdsp \sigma)\stdsp}{\stnsp(p_1 ; \skipp ; x,\stdsp \sigma)}} \\
 & \\
\mbox{\it\sl While-False} & \Rule{\sigma \notin\hspace{-0.2pt} C \quad\quad x = \while{\stdsp C\stnsp}{p_1}{p_2}}{\rpstep{\rho\stsp}{(x,\stdsp \sigma)\stdsp}{\stnsp(p_2,\stdsp \sigma)}} \\
                               & \\
\mbox{\it\sl Conditional-True} & \Rule{\sigma \in\hspace{-0.2pt} C}{\rpstep{\rho\stsp}{(\ite{\stsp C}{p_1}{p_2}, \stdsp\sigma)\stsp}{\stnsp(p_1,\stdsp \sigma)}} \\ 
 & \\
\mbox{\it\sl Conditional-False} & \Rule{\sigma \notin\hspace{-0.2pt} C}{\rpstep{\rho\stsp}{(\ite{\stsp C}{p_1}{p_2}, \stdsp\sigma)\stsp}{\stnsp(p_2,\stdsp \sigma)}} 
\end{array}
\]
\caption{Inductive rules for program steps.}
\label{fig:pstep}
\end{figure}
\paragraph{The rule} \hspace{-5pt}`{\it\sl Basic}'
declares how an \emph{atomic}\index{atomic!step}, \ie indivisible, computation step can be performed by a state transformer $f$. 
For example, a configuration of the form $(\basic\stsp(g\hspace{0.5pt} \circ\hspace{0.5pt} h), \stsp\sigma)$ can thus atomically be transformed 
to $(\skipp,\stsp g(h\: \sigma))$.
As opposed to that, $(\basic\hspace{2.5pt}h ; \basic\hspace{2.9pt} g, \stsp\sigma)$ is at
a finer level of granularity 
since there is no indivisible step to the same $\skipp$-configuration, now reachable by three program steps.
\paragraph{The rules}\hspace{-10pt} {\it\sl`CJump-True'} \textbf{and} {\it\sl`CJump-False'}
declare how a conditional jump is handled in dependence on $C$.
There and throughout this work, $\sigma\hspace{-0.5pt} \in\hspace{-0.5pt} C$ shall abbreviate $C\hspace{1.7pt}\sigma\hspace{-1pt} =\hspace{-1.5pt} \op{True}$ 
($\sigma\hspace{-0.5pt} \notin\hspace{-0.9pt} C$ accordingly). 
      Furthermore, from now on let $\top_\tau$ denote the universal predicate on
      the type $\tau$, \ie the one which evaluates to $\op{True}$ for any element of type $\tau$, \eg $\top_{\ty{nat}}$ is then the set $\naturals$.
For the sake of brevity we will just write $\top$ whenever $\tau$ is clear from the context.
The \emph{unconditional jump}\index{jump!unconditional} instruction can thus be defined by
\[
\jump{\stsp i} \hspace{1.5pt}\defeq\hspace{1.5pt}\cjump{\top}{\stsp i\stsp}{\skipp}
\]
such that 
$\rpstep{\rho}{(\jump{\hspace{1.2pt} i},\stsp \sigma)\stsp}{\hspace{-0.2pt}(\ret{\hspace{1.7pt} i}, \stsp\sigma)}$ follows for any $i\stsp$ and $\sigma$.
Moreover, $(\ret{\hspace{2pt} i},\stsp \sigma)$ is the only
successor of $(\jump{\hspace{1pt}i},\stsp \sigma)$, \ie no branching as one would expect.

It is worth to be once more highlighted that the framework does not impose any constraints on the branching condition $C$ except its type $\alpha \Rightarrow \ty{bool}$,
whereas according to the rules its evaluation takes only one step whatsoever, \ie proceeds atomically.
Although such liberty can be handy in abstract modelling, $C$ shall be instantiated by conditions
that can be handled by a single assembly instruction,
\eg `\emph{does certain register contain} $0$?', 
when a low-level representation needs to be eventually reached.
\paragraph{The rule}\hspace{-9pt} {\it\sl`Parallel'} essentially states that
a configuration $(x,\stsp \sigma\pr)$
is a successor of $(\Parallel{\hspace{-3.1pt}p_1,\ldots p_i \ldots, p_m},\stsp \sigma)$
iff
there is some $p\pr_i$ such that
$(p\pr_i,\stsp \sigma\pr)$ is a successor configuration of $(p_i,\stsp\sigma)$
and $x\hspace{-1.5pt} =\hspace{2pt} \Parallel{\hspace{-1pt}p_1,\ldots p\pr_i \ldots, p_m}$
where $1 \le i \le m$. For example $(\basic\hspace{2.1pt}f_1 \hspace{-1pt}\parallel\hspace{0.2pt} \basic\hspace{2pt}f_2, \stsp \sigma)$
has thus exactly two successors $(\skipp \parallel \basic\hspace{2.1pt}f_2, \stdsp f_1\stsp\sigma)$ and
$(\basic\hspace{2.1pt}f_1 \hspace{-1.2pt}\parallel \skipp, \stdsp f_2\hspace{2pt}\sigma)$ 
whereas each of these has in turn only one successor: $(\skipp\hspace{-0.1pt} \parallel\hspace{-0.1pt} \skipp, \stdsp f_2(f_1\hspace{1.2pt} \sigma))$ and
$(\skipp\hspace{-0.1pt} \parallel\hspace{-0.1pt} \skipp, \stdsp f_1(f_2\hspace{1.9pt} \sigma))$, respectively.
\paragraph{The rule} \hspace{-5pt}{\it\sl`Parallel-skip'}
complements the preceding rule: in cases when all of the parallel components have terminated, $\Parallel{\hspace{-1pt}\skipp,\ldots, \skipp}$
may terminate in one step as well.
Thus, $(\basic\hspace{2.1pt}f_1 \hspace{-2.1pt}\parallel\hspace{-1pt} \basic\hspace{2.1pt}f_2, \stsp \sigma)$ reaches either
$(\skipp,\stdsp f_2(f_1\hspace{1.2pt} \sigma))$
or $(\skipp,\stdsp f_1(f_2\hspace{2.1pt} \sigma))$ in altogether three program steps.
\paragraph{The rule} \hspace{-5.5pt}{\it\sl`Await'} states that
if $\sigma\hspace{-1.5pt} \in\hspace{-1.7pt} C$ holds and we have an evaluation of $p$ that starts with $\sigma$ and terminates
by means of $n\hspace{0.2pt} \in \naturals$ consecutive program steps in some state
  then $\await{\hspace{1.2pt} C\hspace{-0.7pt}}{\hspace{-0.4pt}p\hspace{-0.4pt}}$ also
  terminates in the same state, but by means of a single indivisible step.
Since $\hspace{-0.1pt}\pstep\hspace{-1.5pt}$ is the smallest relation closed under the rules in Figure~\ref{fig:pstep},
a configuration $(\await{\hspace{1pt}C\hspace{-0.7pt}}{\hspace{-1pt}p\hspace{-0.2pt}},\stsp \sigma)$ has no successors when
$\sigma\hspace{0.2pt} \notin \hspace{0.2pt}C$,
\ie\stdsp the evaluation is then blocked. 
Consequently, $\await{\hspace{0.1pt}\top\hspace{-0.7pt}}{\hspace{-0.5pt}p\hspace{-0.1pt}}$ cannot block an evaluation 
but
creates an \emph{atomic} \emph{section}\index{atomic!section}, abbreviated by $\langle p \rangle$ in the sequel,
which allows us to model behaviours where 
one parallel component may force \emph{all} others to instantly halt their computations until an atomic section terminates.
Taking up the example for the rule `{\it\sl Basic}',
an atomic section such as $\langle\basic\hspace{2.5pt}h ;\stsp \basic\hspace{2.9pt} g\rangle$ is computationally in principle the same as
$\basic\stsp(g\hspace{1.5pt} \circ\hspace{1.5pt} h)$
since for any state $\sigma$ the configuration $(\skipp, \stsp g(h\hspace{2.5pt} \sigma))$ is the unique successor
for both, $(\langle\basic\hspace{2.5pt}h ; \stsp\basic\hspace{2.9pt} g\rangle, \stsp\sigma)$ and
$(\basic\stsp(g\hspace{1pt} \circ\hspace{1pt} h), \stsp\sigma)$.

\vspace{-4pt}
\paragraph{The rule}\hspace{-10pt} {\it\sl`While-True'\stsp}
largely has the standard form, except that it additionally takes into account
the transformation of while-statements to
representation using jumps, sketched in Chapter~\ref{S:concept},
by means of the extra $\skipp$ in $(p_1 ; \skipp ; x,\stdsp \sigma)$ that
shall act as
a sort of placeholder for a jump in the proof of Proposition~\ref{thm:while-norm2}.

This deviation from the usual evaluation approaches is justified as follows.
Three layers can be identified in the framework's language: 
\begin{enumerate}
\item[(a)] $\skipp$, $\basic$, $\mathtt{cjump}$ and the sequential composition form the `inner' layer: the computational core,
\item[(b)] the parallel composition and $\mathtt{await}$ additionally
  allow for modelling of software systems with controlled interleaving and $\com{wait}$-synchronisation,
  forming the `middle' layer,
\item[(c)] conditional and while-statements form the `outer' layer and
           additionally enable structured programs serving thereby primarily modelling and verification purposes.
\end{enumerate}
Hence, in order to validate the computational model one has in principle to scrutinise the rules in the layers (a) and (b). 
The presence of while-statements in the language is on the other hand quintessential for usable program logics only:
for computation they shall be translated to the layer (a), the translation being a subject to validation too.
From this perspective,
the behaviour of a while-statement, determined by the rules {\it\sl`While-True'} and {\it\sl`While-False'},
has in first place to be provably equivalent to the behaviour of its layer (a) representation. 
This, in turn, would not be the case without the extra $\skipp$ in presence of interleaving. 
\subsection*{Environment steps}\label{Sb:esteps}
Apart from the program steps, the second integral component of the computational model is the \emph{environment}
which 
can also perform indivisible steps via the $\estep\!\stnsp$ relation\index{computation!step!environment} following the only one rule that
$(p,\stsp \sigma) \stsp\estep \stnsp (p\pr,\stsp \sigma\pr)$ implies $p = p\pr$ for any $p, p\pr, \sigma$ and $\sigma\pr$.
That is, an environment step can randomly modify 
the state of a configuration leaving its program part unchanged.
\section{Programs}
In line with the computational model, a \emph{program}\index{program} is constituted by a pair $(\rho,\stsp p)$ where $\rho$ is a code retrieving function and $p$ is a term of type $\langA{\alpha}$.
It should be clear that with $p$ being jump-free, a program does not actually need a retrieving function to perform any of its
computation steps or, more precisely, performs the same steps independently of how $\rho$ has been defined.
Therefore, retrieving functions will for the sake of brevity be mostly omitted in these cases 
such that a jump-free term of type $\langA{\alpha}$ may particularly be called a (jump-free) program\index{jump-free!program}\index{program!jump-free} as well.
By contrast, if \eg $\hspace{0.4pt}\jump{\hspace{0.7pt} i}$ occurs in $p$ then
there can be two 
program steps $\rpstep{\rho\hspace{0.2pt}}{\hspace{-1.5pt}(p,\stsp \sigma)\stsp}{\hspace{-0.5pt}(p_1, \sigma)}$ and
$\rpstep{\rho\pr\stnsp}{(p,\stsp \sigma)\stsp}{\stnsp(p_2, \sigma)}$ 
with different 
$p_1$ and $p_2$ if $\ret{\hspace{1.5pt} i}\stsp \neq\stsp \rho\pr\hspace{2.5pt}i$ holds.

Note that with certain pairs $(\rho, p)$ one could
achieve something similar to programs having infinite source code: if $p = \jump{\stdsp 0}$ and
$\ret{\stdsp i} = \jump{\stnsp(i + 1)}$ holds for all $i\hspace{0.7pt} \in\hspace{-0.2pt}\naturals$ then $(\rho, p)$ in principle corresponds to the `program', written below as a flat, yet \emph{infinite} list of labelled instructions
\vspace{-5pt}
\[
\begin{array}{l l}
 & \jump{0}\\
0: & \jump{1}\\
1: & \jump{2}\\
\ldots & \\[-7pt]
\end{array}\]
(that is, each line $i:\jump{\!(i+1)}$ forms a separate block labelled by $i$)
doing nothing except perpetual jumping forward through the code. Such creations can be regarded as a byproduct,
and in the sequel we will only focus on pairs $(\rho, p)$ that are well-formed. Intuitively, $(\rho, p)$ shall be considered well-formed
if all labels that are involved in the process of its evaluation can be statically computed. Formally, for any $p$ one can first define
the finite set $\ljumps\hspace{2.9pt}p$ of labels that occur in $p$.
Further, for any set of labels $L$ one can define the set of its \emph{successor labels} $\Next_\rho(L)$ 
by $\{i\hspace{-0.5pt} \in\hspace{-1.7pt} \ljumps(\ret{\hspace{0.5pt} j}) \hspace{3.2pt}|\hspace{3.7pt} j\hspace{-0.5pt}\in\hspace{-1.5pt} L \}$
and consider the closure 
$
\jumps(\rho, p) \stdsp\defeq\stdsp \bigcup_{n\in\naturals}\Next^n_\rho(\ljumps\hspace{2.9pt}p).
$
A program $(\rho, p)$ is then called \emph{well-formed}\index{program!well-formed} if the set $\jumps(\rho, p)$ is finite.

Now, if there is some $\rho\pr$ that coincides with $\rho$ on the set $\jumps(\rho, p)$ then 
$\rpstep{\rho\hspace{-0.5pt}}{\hspace{-0.7pt}(p,\stsp \sigma)\stsp}{\stnsp(p\pr, \stsp\sigma\pr)}$ entails
$\rpstep{\rho\pr\hspace{-0.7pt}}{\hspace{-0.9pt}(p, \stsp\sigma)\stsp}{\stnsp(p\pr,\stsp \sigma\pr)}$ for
      any $\sigma$, $\sigma\pr$ and $p\pr$.
That is, $\rho$ can be arbitrarily re-defined on any label $i\hspace{0.5pt} \notin\hspace{-0.7pt} \jumps(\rho, p)$ without affecting
the evaluation.
In particular, if $p$ is jump-free 
then we have $\jumps(\rho,p)\stnsp =\hspace{-1.5pt} \bot$ for any $\rho$ by the above definition,
which once more highlights that $p$ is thus a well-formed program and the choice of a retrieving function does not matter for its evaluation.

Lastly, taking up the question of sequentiality, a program $(\rho, p)$ is called \emph{sequential}\index{program!sequential} if
$p$ is locally sequential and $\rho\hspace{3.2pt} i\stdsp$ is moreover locally sequential 
for any label $i\hspace{0.2pt} \in\hspace{-1.2pt} \jumps(\rho, p)$.
From the above considerations immediately follows that a jump-free $p$ is sequential iff it is locally sequential,
regardless which $\rho$ we take.
Note that 
program steps retain sequentiality, \ie
$\rpstep{\rho}{(p,\stsp \sigma)\stsp}{\stnsp(p\pr,\stsp \sigma\pr)}$ entails
that $(\rho, p\pr)$ is sequential if $(\rho, p)$ is. 
\section{Potential computations of a program}\label{sub:pcs}
Having defined all possible one-step transformations on configurations,
an \emph{infinite potential computation}\index{computation!potential!infinite} $\sq$ of a program $(\rho, p)$  
is a sequence of configurations
$
  \sq = (p_0,\stsp \sigma_0), (p_1,\stsp \sigma_1),\ldots 
$
  where $p_0 = p$ and either
  \begin{enumerate}
    \item[(i)] $\rpstep{\rho\hspace{0.2pt}}{(p_{i},\stsp \sigma_{i})\stsp}{\hspace{-0.5pt}(p_{i+1},\stsp \sigma_{i+1})}$ or
    \item[(ii)] $(p_{i},\stsp \sigma_{i})\stsp \estep\hspace{-0.5pt} (p_{i+1},\stsp \sigma_{i+1})$
  \end{enumerate}
      holds for all $i\hspace{0.4pt} \in\hspace{-0.5pt} \naturals$.
      Note that an infinite sequence of configurations is actually
      a function of type $\ty{nat}\hspace{-0.4pt} \Rightarrow\hspace{-0.4pt} \langA{\alpha}\hspace{-1.59pt} \times\hspace{-0.59pt} \alpha$ and 
      any program has at least one such computation since each type has at least one inhabitant so that some environment step can always be performed.

Furthermore, let $\prefix{n}{\sq}$\index{computation!prefix} and $\suffix{n}{\sq}$\index{computation!suffix} respectively denote
the sequences obtained by taking the first $n\hspace{0.4pt} \in\hspace{0.2pt}\naturals$ configurations from $\sq$ (a \emph{prefix} of $\sq$)
and by 
dropping the first $n$ configurations from $\sq$ (a \emph{suffix} of $\sq$).
If $\sq$ is an infinite computation 
and $n\hspace{-0.9pt} >\hspace{-0.7pt} 0$ then the prefix $\stdsp\prefix{n}{\sq}\stsp$ forms a \emph{finite potential computation}\index{computation!potential!finite}  
$
  \sq\pr\hspace{-1.2pt} =\stnsp (p_0,\stsp \sigma_0), \ldots, (p_{n-1},\stsp \sigma_{n-1})
$
  being a list of length $n$, denoted by $|\sq\pr|$.
  For a finite $\sq$ the side condition $i\hspace{-0.7pt}<\hspace{-0.7pt}|\sq|$ must be provided when accessing its $i$-th configuration, for convenience denoted by $\sq_i$.
Moreover, the prefix $\prefix{n}{\sq}$ is also a finite potential computation whenever $0<n\le|\sq|$ holds. 

Similarly, the suffix $\suffix{m}{\sq}$ of a finite potential computation $\sq$ is also a finite potential computation
if $m\hspace{-1.7pt}<\hspace{-1pt}|\sq|$ is provided for otherwise $\suffix{m}{\sq}$ is an empty list. 
In case $\sq$ is an infinite sequence,  $\suffix{m}{\sq}$ is by contrast one of its infinite suffixes.
The equality $\suffix{m}{\sq_i}\hspace{-1.7pt} =\hspace{-0.7pt} \sq_{i + m}$, \ie the $i$-th configuration of the suffix is
the $(i+m)$-th configuration of the original sequence, holds for any $\sq, m$ and $i$ with
the side condition $i + m < |\sq|$ for finite $\sq$. 

When the last configuration of a finite computation $\sq$ coincides with the first configuration of $\sq\pr$
then we can compose the computation $\sq \scomp\hspace{-0.5pt} \sq\pr$\index{computation!composition} as follows.
We take the suffix $\suffix{1}{\sq\pr}$ and 
attach it to $\sq$ using the first transition of $\sq\pr$.
Thus, for any $i$ we have $(\sq\hspace{0.2pt} \scomp\hspace{-0.1pt} \sq\pr)_i = \sq_i$ if $i < |\sq|$ and $(\sq \scomp \sq\pr)_i\hspace{-0.9pt} = \sq\mystrut\pr_{i - |\sq| + 1}$ otherwise.
Note that $i < |\sq| + |\sq\pr| - 1$ (\ie $\stdsp i < |\sq\hspace{-0.5pt} \scomp\hspace{-1pt} \sq\pr|$) must additionally be provided in the latter case if $\sq\pr$ is finite.

From now on, $\rpcsi{p}{\rho}$ and $\rpcs{p}{\rho}$ will respectively denote the sets of all infinite and finite potential computations of $(\rho, p)$. Since computation steps of a jump-free $p$ do not depend on the choice of $\rho$,
we will just write $\pcsi{p}$ and $\pcs{p}$ in such cases.

Next chapter is devoted to program correspondences: a structured generic approach to semantic relations between programs.
\chapter{Stepwise Correspondences between Programs}\label{S:pcorr}
Potential computations of a program can be viewed as a transition graph having configurations as vertices and
atomic actions as edges with two different labels: one for the program steps and one for the environment steps.
Correspondences will allow us to systematically compare the transition graphs generated by programs, \ie their behaviour.
This in particular covers the question whether two programs behave in essentially the same way so that they
also exhibit the same relevant properties and can thus be
regarded as equivalent. 
For example the transformations replacing while-statements by jumps 
are supposed to result in an equivalent program.

In presence of interleaved actions even seemingly plain modifications to a model may produce significant behaviour deviations
so that reasoning demands very fine-grained equivalences. 
In the context of concurrency theory, the most commonly used equivalence relation on various sorts of labelled transition systems is therefore
\emph{bisimulation}\index{bisimulation} due to Park~\cite{Park81} which is for instance also central to the \emph{CCS} process algebra by Milner~\cite{Milner80}.
Configurations of a pushdown automaton can also be viewed as processes generating (usually infinite) transition graphs,
 and Stirling~\cite{STIRLING1998113} showed that bisimulation equivalence 
is decidable for a special class of pushdown automata. 
Generally, however, establishing bisimilarity (this is remarkably the same as finding a winning strategy in certain two-player games,
\cf\stdsp Burkart \emph{et al.}~\cite{DBLP:books/el/01/BurkartCMS01}) on infinite structures requires co-inductive reasoning tailored to the particular instance.

Deviating slightly from the bisimulation approach, programs in this framework will be regarded as equivalent 
when they are mutually corresponding  (the exact definitions are given in Section~\ref{Sb:corr_defs} below)
which is basically down to mutual (\ie bi-directional) similarity on the transition graphs.
This equivalence relation is weaker than bisimulation but turns out to be sufficient  for all purposes relevant in what follows. 
To outline the difference very abstractly, the vertices $s_1$ and $t_1$ in the transition graphs
\begin{diagram}[height=0.7cm, textflow]
  s_2 & \rTo^{b} & s_4 & \quad &   &   &  t_3 \\
  \uTo^{a} &     &     &       &   &   &  \uTo_{b}  \\
  s_1     &      &     &       &  t_1 & \rTo^{a}  & t_2 \\
  \dTo^{a} &     &      &      &   &   &  \dTo_{c} \\
  s_3     & \rTo^{c} & s_5 &   &  &   & t_4 \\
\end{diagram}
are not mutually similar
since neither $s_2$ nor $s_3$ can simulate the behaviours of $t_2$ so that $s_1$ and $t_1$ cannot be bisimilar either.
With the following modification 
\begin{diagram}[height=0.7cm, textflow]
  s_2 & \rTo^{b} & s_4 & \quad &   &   &  t_3 \\
  \uTo^{a} &     & \uTo_{b} &       &   &   &  \uTo_{b}   \\
  s_1     &  \rTo^{a}   & s_6  &       &  t_1 & \rTo^{a}  & t_2  \\
  \dTo^{a} &    & \dTo_{c} &      &   &   &  \dTo_{c} \\
  s_3     & \rTo^{c} & s_5&   &  &   & t_4\\
\end{diagram}
they become mutually similar yet not bisimilar: $s_2$ and $s_3$ still cannot simulate the behaviours of $t_2$.
\section{Definitions}\label{Sb:corr_defs}
We first define when a set of pairs of $\langA{}$-terms forms a simulation
and based on that -- when two programs are called corresponding.  
\begin{definition}\label{def:corr}
Let $X$ be a relation of type $\langA{\alpha}\hspace{-1.5pt} \times\hspace{-0.7pt} \langA{\beta} \Rightarrow \ty{bool}$,
where 
$\alpha$ and $\beta$ are state type abstractions.
Further, let $r$ be of type $\alpha\hspace{-0.2pt} \times\hspace{-0.7pt} \beta \Rightarrow \ty{bool}$, \ie a relation
on the underlying states, and let $\rho, \rho\pr$ be two code retrieving functions of type
$\ty{nat}\hspace{-0.1pt} \Rightarrow\hspace{-0.4pt} \langA{\alpha}$ and
$\ty{nat}\hspace{-0.1pt} \Rightarrow\hspace{-0.4pt} \langA{\beta}$, respectively.
Then $X$ will be called a \emph{simulation}\index{simulation} w.r.t. $\!\rho, \rho\pr$ and $r$ if the following conditions hold:
\begin{enumerate}
\item[(1)]
if $(p, q)\hspace{-0.4pt} \in\hspace{-1.7pt} X$ and $(\sigma_1, \sigma_2)\hspace{-0.2pt} \in\hspace{-0.2pt} r$
then for any $\rpstep{\rho\pr\hspace{-0.5pt}}{(q, \stsp\sigma_2)\stsp}{\stnsp(q\pr,\stsp \sigma\pr_2)}$ 
there is a program step $\rpstep{\rho\stsp}{(p,\stsp \sigma_1)\stsp}{\stnsp(p\pr,\stsp \sigma\pr_1)}$ 
with $(p\pr, q\pr) \in\hspace{-1.5pt} X$ and $(\sigma\pr_1, \sigma\pr_2) \in\hspace{-0.2pt} r$,
\item[(2)]
 if $(\skipp,\stsp q) \in\hspace{-1pt} X$ then $q = \skipp$,
\item[(3)]
 if $(p,\stsp \skipp) \in\hspace{-1pt} X$ then $p = \skipp$.
\end{enumerate}
\end{definition}
\begin{definition}\label{def:pcorr}
Two programs $(\rho, p)$ and $(\rho\pr, q)$ \emph{correspond}\index{program correspondence} w.r.t. \!$r$ 
if there exists some $X$ that contains the pair $(p, q)$ and is a simulation w.r.t. \hspace{-4pt}$\rho, \rho\pr, r$.
This will be denoted by $\pcorr{\hspace{-0.2pt}p\hspace{0.2pt}}{r}{\hspace{-0.7pt}q}$.
Furthermore, $(\rho, p)$ and $(\rho\pr,q)$ are called \emph{mutually corresponding}\index{program correspondence!mutual} 
w.r.t. \hspace{-0.9pt}$r$
if $\stsp\pcorrG{\rho\pr}{\rho}{\hspace{-0.2pt} q\hspace{0.3pt}}{\rconv{r}}{\hspace{-2pt}p}\stsp$ additionally holds,
where $\rconv{r}$ is the converse of $r$.
This in turn will be denoted by $\peqv{p}{r}{\stnsp q}$.
\end{definition}
The notations $\pcorrS{p}{r}{\hspace{-1.5pt}q}$ and $\peqvSG{\rho}{p\hspace{-0.2pt}}{r}{\hspace{-1.7pt} q}$ will be used whenever
$\rho$ and $\rho\pr$ are the same. With jump-free $p$ and $q$ we just write $\pcorrC{p\hspace{-0.2pt}}{r}{\hspace{-1.7pt} q}$ and $\peqvC{p\hspace{-0.2pt}}{r}{\hspace{-1.7pt} q}$.
Moreover, the index $r$ will be omitted when $r$ is the identity on the underlying states such that
if $p$ and $q$ are jump-free then we may simply write $\pcorrC{p}{}{q}$ and $\peqvC{p}{}{\hspace{-0.4pt}q}$. 
\subsection{Highlighting the fixed point structure}\label{S:gfp}
The definitions are presented above in a form often favourable for reasoning upon correspondences between programs.
On the other hand, this form might have slightly concealed the structure elaborated below 
focusing on jump-free programs 
so that retrieving functions can be omitted for the sake of clarity. 

Let from now on $\rcomp{r\hspace{1.1pt}}{\hspace{1.1pt}s}$ denote the composition\index{relational!composition} of relations $r$ and $s$ defined in the standard way by
$\{(a, b) \hspace{2.9pt}|\hspace{2.9pt} \exists c. \: (a, c)\hspace{0.2pt} \in r \wedge (c, b)\hspace{0.2pt} \in s\}$.
Thus, given a relation $r$ of type $\alpha\hspace{0.4pt} \times\hspace{0.2pt} \beta\hspace{0.2pt} \Rightarrow\hspace{-0.1pt} \ty{bool}$,
the transformer $\mathcal{T}_r$ sends each $X$ of type $\langA{\alpha}\hspace{-0.7pt} \times \langA{\beta}\hspace{0.1pt} \Rightarrow \hspace{0.1pt} \ty{bool}$
to the relation of the same type defined by
\[
\bigcup\{X\pr \;|\;\; \rcomp{(X\pr \otimes\stdsp r)\hspace{1.5pt}}{\hspace{1.5pt}\pstep}\hspace{-0.5pt} \subseteq\hspace{2pt} \rcomp{\pstep\hspace{-1.7pt}}{\hspace{3.9pt}(X \otimes\stdsp r) \hspace{2pt}\wedge\hspace{1.5pt} \op{C}\stdsp X\pr}\}
\]
where
\begin{enumerate}
\item[--] $s\stsp \otimes\stsp r$ denotes a relational product containing $((p, \sigma), (q, \sigma\pr))$ iff $(p, q) \in s$ and $(\sigma, \sigma\pr) \in r$, 
\item[--] $\op{C}\stdsp X\pr$ stands for the condition:
  for all $(p, q)\hspace{-0.1pt}\in\hspace{-1.2pt} X\pr$ we have that $p$ and $q$ are jump-free and also $p = \skipp$ iff $q = \skipp$.
\end{enumerate}
By this it should be clear that a relation $X$ is a simulation w.r.t.\! $r$ by Definition~\ref{def:corr} (restricted to jump-free terms) if and only if
$X$ is a postfixed point of $\mathcal{T}_r$ \ie $X\hspace{-1pt} \subseteq\hspace{-1pt} \mathcal{T}_r\stdsp X$.
Hence,  as the union of the respective simulations, $\sqsupseteq_r$ is the same as $\bigcup \{X \:|\: X\hspace{-1.2pt} \subseteq\hspace{-1.5pt} \mathcal{T}_r\stdsp X\}$ 
which is by definition the greatest fixed point of $\mathcal{T}_r$ \ie formally $\sqsupseteq_r\hspace{2.5pt} = \stnsp\nu\mathcal{T}_r$.
Furthermore, by the monotonicity of the relational composition and $\otimes$-operator we have $\mathcal{T}_r(\nu\mathcal{T}_r) = \nu\mathcal{T}_r$
and can particularly conclude that
$\sqsupseteq_r$ is itself a postfixed point of $\mathcal{T}_r$ and hence a simulation by Definition~\ref{def:corr}.

It is worth to be noted that by assuming $\approx\hspace{2.9pt} \subseteq\hspace{-1.5pt} \mathcal{T}_{\op{id}}(\approx)$, 
\ie that the mutual correspondence with $r\hspace{-1pt} =\hspace{-1pt} \op{id}$ is likewise a postfixed point of $\mathcal{T}_{\op{id}}$,
we could infer that the respective transition graphs of any jump-free $p$ and $q$
with $\peqvC{p}{}{q}$ are bisimilar
which, as pointed out before, is not true in this generality.
A particular consequence is 
that $\approx$ is not a simulation in the sense of Definition~\ref{def:corr}
although it clearly satisfies the conditions (2) and (3).

Lifting lastly the restriction to jump-free terms, the relation transformer additionally gets two retrieving functions as parameters 
and the above argumentation accordingly yields that $\pcorrG{\rho}{\rho\pr}{p\hspace{0.5pt}}{r}{\hspace{-1.2pt}q}$ holds iff
$(p,q) \in \nu\mathcal{T}_{\rho, \rho\pr, r}$ does.

\section{Properties of correspondences}\label{Sb:corr_props}
From the above considerations follows that the relation $\{(p, q)\;|\; \pcorr{p\stsp}{r}{\stnsp q}\}$ is
the largest simulation w.r.t. \hspace{-1.5pt}$\rho, \rho\pr, r$
and the conditions (1)--(3) of Definition~\ref{def:corr} consequently hold with $\{(p, q)\hspace{3pt}|\hspace{3.5pt} \pcorr{p}{r}{\hspace{-1pt} q}\}$ in place of $X$
(note once more that one cannot proceed likewise with the mutual correspondence for any $r$).

First, the condition (1) of Definition~\ref{def:corr} with $\{(p, q)\;|\; \pcorr{p}{r}{\hspace{-1.5pt} q}\}$ in place of $X$
can by induction be generalised to any number of steps:
\begin{lemma}\label{thm:pcorr_steps}
Assume $\pcorr{p}{r}{\stnsp q}$ and  
$\rpstepsN{\rho\pr\stnsp}{(q,\stsp \sigma_2)}{n}{\stnsp(q\pr,\stsp \sigma\pr_2)}$ where $n\hspace{0.1pt} \in\hspace{-0.5pt} \naturals$. Furthermore, suppose $(\sigma_1, \sigma_2) \in\hspace{-0.1pt} r$.
Then there exist $p\pr$ and $\sigma\pr_1$ such that
$\rpstepsN{\rho}{(p, \stsp\sigma_1)}{n}{(p\pr, \stsp\sigma\pr_1)}$ holds with $\pcorr{p\pr}{r}{\stnsp q\pr}$ and $(\sigma\pr_1, \sigma\pr_2) \in\hspace{0.1pt} r$.
\end{lemma}

The general transitivity rule for program correspondences has the following form:
\begin{lemma}\label{thm:ctrans}
  If $\stdsp\pcorrG{\rho_1}{\rho_2\hspace{-0.9pt}}{\hspace{-0.5pt}p_1\hspace{-1.5pt}}{r}{\hspace{-1.2pt} p_2}$ and
     $\pcorrG{\rho_2}{\rho_3\hspace{-0.5pt}}{\hspace{-0.9pt}p_2\hspace{-1.5pt}}{s}{\hspace{-1.2pt}p_3}$
then we also have $\pcorrG{\rho_1}{\rho_3}{p_1\hspace{-0.5pt}}{\rcomp{r\hspace{1.4pt}}{\hspace{1.2pt} s}}{\hspace{0.1pt} p_3}$. 
\end{lemma}
\begin{proof}
The assumptions give us two simulations: $X_1$ w.r.t. \!$\rho_1, \rho_2, r$ containing $(p_1, p_2)$, 
and $X_2$ w.r.t. \!$\rho_2, \rho_3, s$ containing $(p_2, p_3)$.
The relational composition $\rcomp{X_1\hspace{-0.5pt}}{\hspace{-0.2pt} X_2}$ contains the pair $(p_1, p_3)$ and is a simulation w.r.t.
$\!\rho_1, \rho_3,\stdsp \rcomp{r\hspace{1pt}}{\hspace{1pt} s}$
which follows straight from the properties of $X_1$ and $X_2$.
\end{proof}
With identities in place of $r$ and $s$, the preceding statement simplifies to 
\begin{corollary}\label{thm:ctrans2}
If $\pcorrG{\rho_1}{\rho_2}{p_1\stnsp}{}{p_2}$ and $\pcorrG{\rho_2}{\rho_3}{p_2}{}{p_3}$
then we also have $\pcorrG{\rho_1}{\rho_3}{p_1\stnsp}{}{p_3}$.
Furthermore, for jump-free $p_1,p_2,p_3$ we accordingly just have that $\pcorrC{p_1\hspace{-1pt}}{}{p_2}$ and $\pcorrC{p_2\hspace{-1pt}}{}{p_3}$
imply $\pcorrC{p_1\hspace{-1pt}}{}{p_3}$.
\end{corollary}

That $\sqsupseteq$ is moreover reflexive and hence a preorder is shown next.
Note that $\sqsupseteq_r$ does not have to be a preorder for any $r$ though.
\begin{lemma}\label{thm:crefl}
$\pcorrS{p\hspace{0.5pt}}{}{p}$. 
\end{lemma}
\begin{proof}
  The identity on $\langA{\alpha}$ gives us a simulation w.r.t. 
  the identity on $\alpha$.
\end{proof}
The central consequence of Corollary~\ref{thm:ctrans2} and Proposition~\ref{thm:crefl} is that
for any jump-free $p$ 
the set $\{(\rho, x) \hspace{3pt}|\hspace{3.9pt} \peqvS{p\hspace{-0.1pt}}{}{\hspace{-0.9pt}x}\}$ forms its equivalence class
whose elements indeed have a lot in common: 
one of the results in Section~\ref{Sb:pcorr-rule} is that programs $(\rho, p)$ and $(\rho, q)$ exhibit the same rely/guarantee properties whenever
$\peqvS{\hspace{-1pt}p\hspace{-0.2pt}}{}{\hspace{-0.7pt}q}$ holds.
Moreover, under the same assumption using Proposition~\ref{thm:ctrans} one can show that
for any $x$ and $r$ we have
$\pcorrS{p\hspace{0.5pt}}{r}{\hspace{-1pt}x}$ iff $\pcorrS{q\hspace{0.5pt}}{r}{\hspace{-0.7pt}x}$,
\ie $(\rho, p)$ and $(\rho, q)$ correspond to exactly the same programs.

These considerations essentially justify the choice of $\approx$ as \emph{the} program equivalence\index{program!equivalence} for this framework
and so we can step up to investigating on associativity and commutativity properties.

The sequential composition operator will be shown to be associative by Proposition~\ref{thm:seq-assoc}.
Regarding associativity of the parallel composition, 
consider the programs $\op{p} \defeq (\basic\hspace{1.79pt} f \hspace{-1pt}\parallel\hspace{-0.7pt} \skipp) \parallel \skipp$ and
$\op{q} \defeq \basic\hspace{1.7pt} f \hspace{-0.9pt}\parallel (\skipp \parallel \skipp)$
with any $f$ such that there exists a state $\overline{\sigma}\hspace{-0.5pt} \neq\hspace{-1pt} f\hspace{2.5pt}\overline{\sigma}$.  
If we assume $\models\op{p}\hspace{-0.2pt} \sqsupseteq\hspace{-0.1pt} \op{q}$
then by Proposition~\ref{thm:pcorr_steps} to the step $(\op{q},\stsp \overline{\sigma})\stdsp \pstep\stnsp (\basic\hspace{2pt} f \stnsp\parallel\hspace{-0.5pt} \skipp,\stsp \overline{\sigma})$
there must be a corresponding step $(\op{p},\stsp \overline{\sigma})\stdsp \pstep\stnsp (x,\stsp \overline{\sigma})$ with some $x$.
This is however impossible since  $(\op{p},\stsp \overline{\sigma}) \stdsp\pstep\stnsp ( (\skipp\hspace{-1pt} \parallel\hspace{-1pt} \skipp) \parallel\hspace{-1pt} \skipp,\stsp f\hspace{2.5pt}\overline{\sigma})$
is the only one program step from $(\op{p},\stsp \overline{\sigma})$, meaning that $\models\hspace{-1pt} \op{p}\hspace{-0.4pt} \not\approx\stnsp \op{q}$ holds.
Although the counterexample might be perceived as technical, it nonetheless leads to the conclusion
that the binary parallel operator is not associative at the level of small-step evaluation.
In other words, we
take the edge off by allowing `flat' lists of parallel components (\eg $\stsp\parallel\hspace{-2.7pt} \basic\hspace{2pt} f, \skipp, \skipp$)
whose reordering shall not give rise to semantic differences though and Corollary~\ref{thm:pcomm} makes that point explicit.
\begin{lemma}\label{thm:seq-assoc}
$\peqvS{p_1 ; p_2 ; p_3}{}{\stnsp (p_1 ; p_2) ; p_3}$. 
\end{lemma}
\begin{proof}
We show $\pcorrS{p_1 ; p_2 ; p_3}{}{(p_1 ; p_2) ; p_3}$ whereas $\pcorrS{(p_1 ; p_2) ; p_3}{}{p_1 ; p_2 ; p_3}$ follows similarly. 
The set
$X \defeq \{(u ; p_2 ; p_3,(u ; p_2) ; p_3) \;|\; u\hspace{0.4pt} \in\stnsp \top\}\stsp \cup\stsp \op{id}$
contains the pair $p_1 ; p_2 ; p_3$ and $(p_1 ; p_2) ; p_3$.
Further, suppose $\rpstep{\rho}{((u ; p_2) ; p_3,\stsp \sigma)\stsp}{\stnsp(x, \stsp\sigma\pr)}$ for some $u$.
If $u = \skipp$ then $x = p_2 ; p_3$ and $\sigma\pr\hspace{-1.2pt} =\hspace{0.5pt} \sigma$ which can be matched by the step
$\rpstep{\rho}{(\skipp ; p_2 ; p_3, \stsp\sigma)}{(p_2 ; p_3, \stsp\sigma)}$.
If  $u \neq \skipp$ then we have a program step $\rpstep{\rho\hspace{0.1pt}}{\hspace{-1pt}(u,\stsp \sigma)\stsp}{\hspace{-1.2pt}(u\pr,\stsp \sigma\pr)}$ with some $u\pr$ such that
$x\hspace{-0.5pt} =\hspace{-0.5pt} (u\pr ; p_2) ; p_3$. Thus, we further have 
$\rpstep{\rho\hspace{0.2pt}}{(u ; p_2 ; p_3,\stsp \sigma)\hspace{1pt}}{\hspace{-1pt}(u\pr ; p_2 ; p_3, \stsp\sigma\pr)}$ which is a matching step since
$(u\pr ; p_2 ; p_3, (u\pr ; p_2) ; p_3)\hspace{0.7pt} \in\hspace{-0.7pt} X$.
\end{proof}
\begin{lemma}\label{thm:pcomm1}
Let $I = \{1,\ldots, m\}$ with $m > 0$ and assume $q_i = p_{\pi(i)}\stsp$ for all $\stsp i\hspace{1pt} \in I$ where $\pi$ is a permutation on $I$.
Then $\stsp\pcorrS{\:\Parallel{\!p_1, \ldots, p_m\hspace{-1pt}}}{}{\:\Parallel{\hspace{-1.5pt} q_1, \ldots, q_m}}$. 
\end{lemma}
\begin{proof}
The following set contains the pair $(\Parallel{\!p_1, \ldots, p_m}, \Parallel{\!q_1, \ldots, q_m})$: 
\[
X \defeq  \{(\Parallel{\hspace{-2pt}u_{1}, \ldots, u_{m}}, \Parallel{\hspace{-2pt}v_1, \ldots, v_m}) \;|\; \forall i\stsp \in I.\:v_i\stnsp =\stnsp u_{\pi(i)} \} \cup \{(\skipp, \skipp)\}
\]
and is a simulation.
To establish that, suppose $\rpstep{\rho}{(\Parallel{\hspace{-2pt}v_1, \ldots, v_m},\stsp \sigma)\hspace{0.5pt}}{\hspace{-1pt}(x,\stsp\sigma\pr)}$ and
$(\Parallel{\hspace{-2pt} u_1, \ldots, u_m}, \Parallel{\hspace{-2pt}v_1, \ldots, v_m}) \in \stnsp X$ with some $x, \sigma, \sigma\pr$, $u_1,\ldots,u_m$ and $v_1,\ldots,v_m$. 
If $x$ is $\skipp$ so are $u_i, v_i$ for all $i\hspace{0.9pt} \in\hspace{-0.2pt} I$, \ie we are done. 
Otherwise there is a step $\rpstep{\rho\stsp}{(v_i, \sigma)\stsp}{(v_i\pr, \sigma\pr)}$ with some $i\hspace{0.7pt} \in\hspace{-0.5pt} I$
such that $x =\: \Parallel{\!v_1, \ldots v_i\pr \ldots,  v_m}$ holds.
We then also have
$\rpstep{\rho\stnsp}{\stnsp(\Parallel{\hspace{-2.7pt}u_1, \ldots u_{\pi(i)} \ldots, u_m}, \sigma)\stsp}{\stnsp(\Parallel{\hspace{-2.7pt}u_1, \ldots u\pr_{\pi(i)} \ldots, u_m}, \sigma\pr)}$
where $u\pr_{\pi(i)}$ is set to $v\pr_i$. The program step is sound because $u_{\pi(i)} = v_i$ is assumed.
\end{proof}
\begin{corollary}\label{thm:pcomm}
Let $I = \{1,\ldots, m\}$ with $m > 0$ and assume $q_i = p_{\pi(i)}\stsp$ for all $\stsp i\hspace{0.9pt} \in\hspace{-0.2pt} I$ where $\pi$ is a permutation on $I$.
Then $\stsp\peqvS{\:\Parallel{\hspace{-1.2pt}p_1, \ldots, p_m\hspace{-0.5pt}}}{}{\:\Parallel{\hspace{-1.2pt} q_1, \ldots, q_m}}$.   
\end{corollary}
\begin{proof}
The direction $\pcorrS{\:\Parallel{\hspace{-1.7pt}p_1, \ldots, p_m}\hspace{-0.5pt}}{}{\hspace{1.9pt}\Parallel{\hspace{-1.7pt}q_1, \ldots, q_m}}$ is just Proposition~\ref{thm:pcomm1}.
To show $\pcorrS{\:\Parallel{\hspace{-2pt}q_1, \ldots, q_m}\hspace{-0.5pt}}{}{\:\Parallel{\hspace{-2pt}p_1, \ldots, p_m}}$
we once more apply Proposition~\ref{thm:pcomm1}, but now with $\pi^{-1}$ for $\pi$. This is sound since
$\pi^{-1}$ is also a permutation on $I$ and we have
$p_i = p_{\pi(\pi^{-1}(i))} = q_{\pi^{-1}(i)}$ for all $i\hspace{1pt} \in\hspace{0.2pt} I$.
\end{proof}
\subsection{Closure properties}\label{Sb:corr-rules}
The rules presented in this section particularly enable syntax-driven derivations of
program correspondences, \ie derivations following syntactic patterns.
\begin{lemma}\label{thm:parallel-corr}
  Let $I = \{1,\ldots, m\}$ with $m > 0$ and assume
  \begin{enumerate}
  \item[\emph{(1)}] $\pcorr{p_i\hspace{-0.2pt}}{r}{\hspace{-0.5pt} q_i}\;$ for all $\stsp i\hspace{0.9pt} \in\hspace{-0.2pt} I$.
    \end{enumerate}
Then $\pcorr{\:\Parallel{\!p_1,\ldots, p_m}\hspace{-0.5pt}}{r}{\:\Parallel{\!q_1, \ldots,q_m}}$.
\end{lemma}
\begin{proof}
The assumption  (1) provides
simulations $X_1, \ldots, X_m$ w.r.t. \!\!$\rho, \rho\pr, r$ such that
$(p_i, q_i) \in\hspace{-0.9pt} X_i$ holds for all $i\hspace{1pt} \in I$.
Then let
\[
\hspace{-10pt}X\; \defeq\;  \{(\Parallel{\!u_1, \ldots, u_m},\stsp \Parallel{\!v_1, \ldots, v_m}) \hspace{3pt}|\hspace{3.9pt} \forall i\hspace{0.5pt} \in\hspace{0.1pt} I. \: (u_i, v_i)\hspace{0.2pt} \in\hspace{-1pt} X_i \} \cup
          \{(\skipp, \skipp)\}
\]
containing the pair $(\Parallel{\!p_1, \ldots, p_m},\stsp \Parallel{\!q_1, \ldots, q_m})$. 
To show that $X$ is a simulation w.r.t. $\!\rho, \rho\pr, r$
suppose we have a program step
$\rpstep{\rho\pr\hspace{-0.2pt}}{(\Parallel{\! v_1, \ldots, v_m},\stsp \sigma_2)\hspace{0.5pt}}{\hspace{-1pt}(x, \stsp\sigma\pr_2)}$
with $(\sigma_1, \sigma_2) \in r$ and
$(\Parallel{\! u_1, \ldots, u_m},\stsp \Parallel{\!v_1, \ldots, v_m}) \in\hspace{-1.2pt} X$.

If $x$ is $\skipp$ then we have 
$v_i\hspace{-0.5pt} = \skipp$ for all $i\stsp \in\hspace{-0.2pt} I$ and hence also $u_i\hspace{-0.5pt} = \skipp$ for all $i\hspace{-0.9pt} \in\hspace{-2pt} I$
due to $(u_i, v_i)\hspace{-1.5pt} \in\hspace{-3.5pt} X_i$.
If $x\hspace{-1.5pt}\neq\hspace{-1.2pt} \skipp$ then there is a program step
$\rpstep{\rho\pr}{(v_i, \stsp\sigma_2)\hspace{0.7pt}}{\hspace{-1pt}(v\pr_i,\stsp \sigma\pr_2)}$ with some $v\pr_i$ and $i\stsp \in\hspace{-0.2pt} I$
such that $x = \:\Parallel{\!\stnsp v_1, \ldots v\pr_i \ldots, v_m}$. Since we have $(u_i, v_i)\hspace{-0.55pt} \in\hspace{-1.79pt} X_i$ and $(\sigma_1, \sigma_2)\hspace{-0.9pt} \in\hspace{-0.9pt} r$,
there is a matching program step  $\rpstep{\rho\hspace{-0.2pt}}{\!(u_i, \sigma_1)\stsp}{\stnsp(u\pr_i, \sigma\pr_1)}$ with some $u\pr_i$ and $\sigma\pr_1$
such that $(u\pr_i, v\pr_i) \hspace{-1pt}\in\hspace{-2.7pt} X_i$ and also $(\sigma\pr_1, \sigma\pr_2)\hspace{0.29pt} \in\hspace{0.29pt} r$ hold.
This in turn enables the transition
\[
\rpstep{\rho\hspace{0.5pt}}{\stnsp(\Parallel{\!u_1, \ldots u_i \ldots, u_m}, \stsp\sigma_1)\stsp}{\stnsp(\Parallel{\!u_1, \ldots u\pr_i \ldots, u_m}, \stsp\sigma\pr_1)}
\]
with $(\Parallel{\!u_1, \ldots u\pr_i \ldots, u_m}, \Parallel{\!v_1, \ldots v\pr_i \ldots, v_m})\hspace{0.2pt} \in\stnsp X$.
\end{proof}
\begin{lemma}\label{thm:seq-corr}
  Assume
\begin{enumerate}
\item[\emph{(1)}] $\pcorr{p_1}{r}{q_1}$,
\item[\emph{(2)}] $\pcorr{p_2}{r}{q_2}$.
\end{enumerate}
  Then $\stsp\pcorr{p_1 ; p_2}{r}{q_1 ; q_2}$.
\end{lemma}
\begin{proof}
From (1) and (2) we respectively obtain two simulations $X_1, X_2$ w.r.t. $\rho, \rho\pr, r$
with $(p_1, q_1) \in\hspace{-1.5pt} X_1$ and $(p_2, q_2) \in\hspace{-1.5pt} X_2$.
Further, let
\[
\hspace{-10pt}X\; \defeq\; \{(u ; p_2,\stsp v ; q_2) \hspace{3pt}|\hspace{3.9pt} (u, v)\hspace{0.2pt} \in\hspace{-1pt} X_1\}\stdsp \cup\stsp X_2
\]  
such that $(p_1 ; p_2,\stdsp q_1 ; q_2) \in\hspace{-1.7pt} X$ holds in first place.
It remains thus to show that $X$ is a simulation w.r.t. \stnsp$\rho, \rho\pr, r$. 
To this end let $\rpstep{\rho\pr}{(v ; q_2,\stsp \sigma_2)\hspace{1pt}}{\hspace{-1pt}(x,\stsp \sigma\pr_2)}$ with $(\sigma_1, \sigma_2) \hspace{-0.5pt}\in\hspace{-0.7pt} r$ and
$(u, v)\hspace{-0.5pt} \in\hspace{-1.5pt} X_1$. If $v\hspace{-0.5pt} = \skipp$ then $u = \skipp,\stdsp x = q_2$ and $\sigma\pr_2\hspace{-0.7pt} = \sigma_2\;$ follow
so that we are done since $(p_2, q_2) \in\stnsp X_2$ and $X_2\stsp \subseteq\hspace{-0.2pt} X$.

If $v \neq \skipp$ then $x = v\pr ; q_2$ which is due to a step $\rpstep{\rho\pr}{(v,\stsp \sigma_2)\hspace{1pt}}{\hspace{-1pt}(v\pr, \stsp\sigma\pr_2)}$ with some $v\pr$.
By $(u, v)\hspace{0.5pt} \in\hspace{-1pt} X_1$ we then have $\rpstep{\rho\hspace{0.2pt}}{\hspace{-0.5pt}(u,\stsp \sigma_1)\hspace{1pt}}{\hspace{-1pt}(u\pr,\stsp \sigma\pr_1)}$ with
$(\sigma\pr_1, \sigma\pr_2)\hspace{0.2pt} \in\hspace{-0.2pt} r$ and $(u\pr, v\pr)\hspace{-0.1pt} \in\hspace{-1.5pt} X_1$. Thus, we derive the step 
$\rpstep{\rho\hspace{0.2pt}}{\stnsp(u ; p_2,\stsp \sigma_1)\hspace{1pt}}{\hspace{-1pt}(u\pr ; p_2,\stsp \sigma\pr_1)}$ with $(u\pr ; p_2,\stsp x)\hspace{0.2pt} \in\hspace{-0.7pt} X$.
\end{proof}

Next three propositions handle conditional, while- and await-statements taking additionally branching and blocking conditions into account.
\begin{lemma}\label{thm:cond-corr}
Assume
\begin{enumerate}
\item[\emph{(1)}] $\pcorr{p_1\hspace{-0.5pt}}{r}{\hspace{-0.5pt}p_2}$,
\item[\emph{(2)}] $\pcorr{q_1\hspace{-0.5pt}}{r}{\hspace{-0.5pt}q_2}$,
\item[\emph{(3)}]  for any $(\sigma_1, \sigma_2) \in r$ we have $\stsp\sigma_1 \in C_1$ iff $\stsp\sigma_2 \hspace{0.2pt}\in\hspace{0.2pt} C_2$.
\end{enumerate}
Then $\pcorr{\ite{\stsp C_1}{p_1}{q_1}\stsp}{r}{\ite{\stsp C_2}{p_2}{q_2}}$.
\end{lemma}
\begin{proof}
From (1) and (2) we respectively get simulations $X_p, X_q$ w.r.t. \hspace{-4pt}$\rho, \rho\pr, r$
with $(p_1, p_2)\hspace{0.2pt} \in\hspace{-1pt} X_p$ and $(q_1, q_2)\hspace{0.2pt} \in\hspace{-1pt} X_q$.
We define
\[\hspace{-10pt}X\; \defeq\;  X_p \cup X_q \cup \{(\ite{\stsp C_1}{p_1}{q_1},\stdsp\ite{\stsp C_2}{p_2}{q_2})\}\]
and show that $X$ is a simulation w.r.t. \hspace{-4pt}$\rho, \rho\pr, r$.
First, note that as the union of simulations, $X_p\hspace{-0.7pt} \cup\hspace{-1pt} X_q$ is a simulation as well.
Then for the remaining case, suppose we have a step $\rpstep{\rho\pr\hspace{-0.9pt}}{\hspace{-1pt}(\ite{\stsp C_2}{p_2}{q_2}, \stsp\sigma_2)\hspace{1pt}}{\hspace{-1pt}(x, \stsp\sigma_2)}$
with $(\sigma_1, \sigma_2) \in r$. If $\sigma_2 \in C_2$ then
$x = p_2$ and hence there is the matching transition
$\rpstep{\rho\hspace{0.5pt}}{(\ite{\stsp C_1}{p_1}{q_1}, \stsp\sigma_1)\hspace{1pt}}{\hspace{-0.9pt}(p_1,\stsp \sigma_1)}$ since $\sigma_1\hspace{-0.7pt} \in C_1$ by (3).
The case with $\sigma_2\hspace{0.5pt} \notin C_2$ is symmetric.
\end{proof}
\begin{lemma}\label{thm:while-corr}
  \stdsp Let $\stdsp p_{\mathit{while}}\hspace{-1.5pt} =\hspace{-1pt} \while{\stdsp C_1\hspace{-1.9pt}}{\hspace{-1.5pt}p_1\hspace{-1.9pt}}{\hspace{-1.5pt}q_1\hspace{-1.9pt}}$ and
  $q_{\mathit{while}}\stsp =\stsp \while{\stdsp C_2\stdsp}{p_2}{q_2}$,
and assume
\begin{enumerate}
\item[\emph{(1)}] $\pcorr{p_1\hspace{-0.5pt}}{r}{\hspace{-0.5pt}p_2}$,
\item[\emph{(2)}] $\pcorr{q_1\hspace{-0.5pt}}{r}{\hspace{-0.5pt}q_2}$,
\item[\emph{(3)}] for any $(\sigma_1, \sigma_2) \in r\stsp$ we have $\stsp\sigma_1\hspace{-0.5pt} \in C_1$ iff $\stsp\sigma_2 \hspace{0.2pt}\in\hspace{0.2pt} C_2$.
\end{enumerate}
Then $\pcorr{p_{\mathit{while}}}{r}{q_{\mathit{while}}}$.
\end{lemma}
\begin{proof}
From (1) and (2) we respectively get simulations $X_p, X_q$ w.r.t. \hspace{-4pt}$\rho, \rho\pr, r$
with $(p_1, p_2)\hspace{0.2pt} \in\hspace{-1pt} X_p$ and $(q_1, q_2)\hspace{0.2pt} \in\hspace{-1pt} X_q$.
The relevant set of pairs $X$ is defined by
\[
\begin{array}{l c l}
  \hspace{-10pt}X &\!\!\! \defeq &\!\! \{(p_{\mathit{while}},q_{\mathit{while}}), (\skipp;p_{\mathit{while}},\: \skipp;q_{\mathit{while}})\}\hspace{12pt} \cup\\
  &        &\!\! \{(u;\skipp;p_{\mathit{while}}, \:v;\skipp;q_{\mathit{while}}) \hspace{3pt}|\hspace{3.9pt} (u, v) \in\hspace{-1.5pt} X_p\} \cup X_q
\end{array}
\]
To establish $X$ as a simulation, suppose $(\sigma_1, \sigma_2) \in r$ and consider the following cases.

Regarding the first pair $(p_{\mathit{while}},q_{\mathit{while}})$, suppose $\rpstep{\rho\pr}{(q_{\mathit{while}}, \stsp\sigma_2)\hspace{1pt}}{\hspace{-1pt}(x,\stsp \sigma_2)}$.
If $\sigma_2\hspace{0.2pt} \in\hspace{0.1pt} C_2$ then $x = p_2;\skipp;q_{\mathit{while}}$ and, moreover, $\sigma_1\hspace{-0.1pt} \in\hspace{0.2pt} C_1$ by (3).
Hence there is the matching step $\rpstep{\rho\stsp}{(p_{\mathit{while}}, \stsp\sigma_1)\hspace{1pt}}{\hspace{-0.9pt}(p_1;\skipp;p_{\mathit{while}},\stsp \sigma_1)}$.

           The next case with $(\skipp;p_{\mathit{while}},\: \skipp;q_{\mathit{while}},)$ is clear because there is only one possible step
           $\rpstep{\rho\pr\hspace{-0.2pt}}{(\skipp;q_{\mathit{while}},\: \sigma_2)\stsp}{\stnsp(q_{\mathit{while}}, \sigma_2)}$ which is matched by the only possible
            $\rpstep{\rho\stsp}{(\skipp;p_{\mathit{while}},\: \sigma_1)\stsp}{\stnsp(p_{\mathit{while}}, \sigma_1)}$.
            
           Lastly, suppose $\rpstep{\rho\pr\hspace{-0.5pt}}{\hspace{-0.9pt}(v;\skipp;q_{\mathit{while}}, \stsp\sigma_2)\hspace{1pt}}{\hspace{-1pt}(x, \stsp\sigma\pr_2)}$
           with $(u, v)\hspace{-0.5pt}\in\hspace{-2.1pt} X_p$.
If we have $v = \skipp$ then $u = \skipp$ 
follows leading to the same situation as in the preceding case.
             Otherwise, with $v \neq \skipp$ we have $x = v\pr;\skipp;q_{\mathit{while}}$ due to a program step $\rpstep{\rho\pr}{\stnsp(v,\stsp \sigma_2)\hspace{1pt}}{\hspace{-1pt}(v\pr,\stsp \sigma\pr_2)}$.
             Since $X_p$ is a simulation, we obtain a transition
             $\rpstep{\rho}{\hspace{-1pt}(u, \stsp\sigma_1)\hspace{1pt}}{\hspace{-1pt}(u\pr,\stsp \sigma\pr_1)}$ with $(u\pr, v\pr)\hspace{-0.2pt} \in\hspace{-1.5pt} X_p$ and $(\sigma\pr_1, \sigma\pr_2)\hspace{-0.5pt} \in\hspace{-0.5pt} r$.
             The latter transition moreover entails $\rpstep{\rho\stsp}{(u;\skipp;p_{\mathit{while}},\stsp \sigma_1)\hspace{0.7pt}}{\hspace{0.7pt}(u\pr;\skipp;p_{\mathit{while}}, \stsp\sigma\pr_1)}$.
\end{proof}

\begin{lemma}\label{thm:await-corr}
Assume
\begin{enumerate}
\item[\emph{(1)}] $\pcorr{p_1\hspace{-0.5pt}}{r}{\hspace{-0.5pt}p_2}$,
\item[\emph{(2)}] $\sigma_2\hspace{0.1pt} \in\hspace{-0.3pt} C_2\stsp$ implies $\stsp\sigma_1\hspace{-0.4pt} \hspace{-0.4pt}\in C_1\stsp$ for any $(\sigma_1, \sigma_2)\hspace{0.1pt} \in r$.
\end{enumerate}
Then $\pcorr{\await{\stdsp C_1}{p_1\stnsp}\stdsp}{r}{\await{\stdsp C_2\stsp}{p_2}}$.
\end{lemma}
\begin{proof}
  The set $X \defeq\hspace{0.5pt} \{(\await{\stdsp C_1\hspace{-1.5pt}}{\hspace{-0.5pt}p_1\hspace{-2pt}},\stdsp \await{\hspace{1pt}C_2\hspace{-0.5pt}}{\hspace{-0.5pt}p_2\hspace{-1pt}}),\hspace{2pt} (\skipp,\stdsp \skipp)\}$
forms a simulation: $\rpstep{\rho\pr\hspace{-0.7pt}}{\hspace{-0.9pt}(\await{\stsp C_2\hspace{-0.2pt}}{\hspace{-0.5pt}p_2\hspace{-0.5pt}},\stsp \sigma_2)\hspace{1pt}}{\hspace{-1pt}(x,\stdsp \sigma\pr_2)}$ and $(\sigma_1, \sigma_2)\hspace{-1pt} \in\hspace{-1pt} r$ firstly entail
$\sigma_2\hspace{-0.7pt} \in\hspace{-1.1pt}  C_2$, $\rpsteps{\rho\pr\hspace{-0.9pt} }{\hspace{-0.9pt} (p_2,\stdsp \sigma_2)\hspace{-1pt}}{\hspace{-2pt}(\skipp,\stdsp \sigma\pr_2)}$ and $x\hspace{-0.9pt}  =\hspace{-1pt}  \skipp$.
By Proposition~\ref{thm:pcorr_steps} and (1) we can further obtain a state $\sigma\pr_1$ such that $(\sigma\pr_1,\stsp \sigma\pr_2)\hspace{0.2pt} \in\hspace{-0.7pt} r$ and $\rpsteps{\rho\stsp}{(p_1, \stsp\sigma_1)\hspace{0.5pt}}{\hspace{-0.7pt}(\skipp,\stsp \sigma\pr_1)}$ hold. 
From this we can in turn derive the matching step \stdsp
$\rpstep{\rho\stsp}{(\await{\stdsp C_1}{p_1},\stsp \sigma_1)\hspace{1pt}}{(\hspace{-1pt}\skipp,\stsp\sigma\pr_1)}$
 as $\sigma_1\hspace{-0.5pt} \in\hspace{0.4pt} C_1$ follows by (2). 
\end{proof}

The remaining two propositions
handle correspondences between the combinations of indivisible steps that will be relevant later on. 
\begin{lemma}\label{thm:basic-corr}
If $(f\hspace{2pt} \sigma_1,\stsp g\hspace{2.3pt}\sigma_2) \in\hspace{0.2pt} r\stsp$ holds for any pair $\stsp(\sigma_1, \sigma_2) \in r$
then we have $\pcorrC{\basic\hspace{2.5pt} f\hspace{0.5pt}}{r}{\hspace{-0.5pt}\basic\hspace{3.5pt} g}$.
\end{lemma}
\begin{proof}
The assumption provides that  
$\{(\basic\hspace{2.5pt} f, \stdsp\basic\hspace{3.5pt}g), \stdsp(\skipp,\stdsp \skipp)\}$
is a simulation w.r.t. \!$r$.
\end{proof}
\begin{lemma}\label{thm:await-corrL}
If for any pair $(\sigma_1, \sigma_2) \in\hspace{-0.2pt} r$ there exists some $\sigma\pr_1$
such that $\rpsteps{\rho\stdsp}{(p,\stsp \sigma_1)}{(\skipp, \stsp\sigma\pr_1)}$ and $(\sigma\pr_1,\stsp f\hspace{2.1pt}\sigma_2) \in r$ hold
then $\pcorrS{\stsp\langle p \rangle\stsp}{r}{\hspace{-0.5pt}\basic\hspace{2.7pt} f}$.
\end{lemma}
\begin{proof}
Let $X \defeq \{(\langle p \rangle, \stdsp\basic\hspace{2.5pt} f), \stdsp(\skipp,\stdsp \skipp)\}$.
In order to show that $X$ is a simulation w.r.t.\hspace{-2.1pt} $\rho$ and $r$, 
suppose $(\sigma_1,\stsp \sigma_2) \in\hspace{-0.5pt} r$. Thus, by the assumption there must be
some state $\sigma\pr_1$ with $(\sigma\pr_1, \stsp f\hspace{2pt}\sigma_2) \in r$ and $\rpsteps{\rho\stsp}{(p, \stsp\sigma_1)\hspace{0.5pt}}{\hspace{-1.5pt}(\skipp,\stsp \sigma\pr_1)}$,
such that $\rpstep{\rho\stsp}{(\langle p \rangle,\stsp \sigma_1)\hspace{1pt}}{\hspace{-0.5pt}(\skipp, \stsp\sigma\pr_1)}$ follows. 
\end{proof}
\section{Sequential normalisation}\label{Sb:seq-norm}
As pointed out in Chapter~\ref{S:concept}, certain code restructurings are needed prior to
a rational replacement of conditional
and while-statements by their representations using jumps.
Next two propositions show that no semantic differences arise by these restructurings 
whereas Section~\ref{Sb:jump-norm}
proceeds to the actual jump-transformations.
\begin{lemma}\label{thm:cond-norm1}
  Let
     \begin{enumerate}
      \item[\emph{(1)}] $L = \ite{\stdsp C}{p_1}{p_2}; q$,
      \item[\emph{(2)}] $R = \ite{\stsp C}{p_1 ; q}{p_2 ; q}$.
    \end{enumerate} 
     Then $\hspace{1.5pt}\peqvS{L}{}{\hspace{-1.5pt} R}$.
\end{lemma}
\begin{proof}
  We establish the correspondence $\pcorrS{\hspace{-0.5pt}L}{}{\hspace{-0.7pt}R}\stsp$
by showing that the set of pairs $\{(L,\stsp R)\}\cup\hspace{0.2pt} \op{id}$ is a simulation w.r.t. \!$\rho$.
Suppose $\rpstep{\rho\hspace{0.4pt}}{\hspace{-0.5pt}(R,\stsp \sigma)\hspace{1pt}}{\hspace{-1pt}(x,\stsp \sigma)}$ with some $x$ and $\sigma$.
If $\sigma\hspace{-0.7pt} \in\hspace{-0.9pt} C$
then $x\hspace{-1pt} =\hspace{-1pt} p_1;q$ and hence we have the matching step
$\rpstep{\rho}{\stnsp(L,\stsp \sigma)\hspace{1.2pt}}{\hspace{-0.5pt}(p_1;q,\stsp \sigma)}$.
The case with $\sigma\hspace{-0.5pt} \notin\hspace{-0.9pt} C$ is symmetric and the opposite direction
$\pcorrS{R\stsp}{}{\hspace{0.1pt} L}\stdsp$ can be concluded likewise.
\end{proof}
\begin{lemma}\label{thm:while-norm1}
  Let
     \begin{enumerate}
      \item[\emph{(1)}] $L = \while{\stsp C}{p_1}{p_2} ; q$,
      \item[\emph{(2)}] $R = \while{\stsp C}{p_1}{p_2 ; q}$.
    \end{enumerate}
Then $\hspace{1.5pt}\peqvS{L}{}{\! R}$.
\end{lemma}
\begin{proof}
We will establish the correspondence $\pcorrS{\hspace{-0.5pt}L\hspace{0.1pt}}{}{\hspace{-0.79pt}R}$
by showing that the set
$X\; \defeq \; \{(\skipp ; L, \:\skipp ; R),\: (L, R))\}  \cup
  \{(v ; \skipp ; L, \:v ; \skipp ; R)) \hspace{4pt}|\hspace{4.5pt} v\stsp \in \hspace{-1pt}\top \} \cup \op{id}$
  is a simulation w.r.t. \!$\rho$ whereas
  the opposite direction, \ie $\pcorrS{R\hspace{0.9pt}}{}{\hspace{0.2pt}L}$, follows in a symmetric way.

The case $(\skipp ; L, \stdsp \skipp ; R)$ is clear.
Next, assume $\rpstep{\rho\hspace{0.2pt}}{\hspace{-0.7pt}(R,\stsp\sigma)}{\hspace{0.2pt}(x,\stsp \sigma)}$ with some $x$ and $\sigma$.
If $\sigma\hspace{-0.59pt} \in\hspace{-0.9pt} C$
then $x \stnsp= p_1 ;\skipp;R$ and hence there is the matching step $\rpstep{\rho\stsp}{(L, \stsp\sigma)\stsp}{\hspace{-0.5pt}(p_1;\skipp;L,\stsp \sigma)}$
since $(p_1;\skipp;L,\stsp x)\hspace{0.29pt}\in\hspace{-0.9pt} X$.
If $\sigma\hspace{0.2pt} \notin\hspace{0.5pt} C$ then we have $x = p_2 ; q$
and hence we can take the step from $(L,\stsp \sigma)$ to the same $(x,\stsp \sigma)$.

Lastly, suppose $\rpstep{\rho}{(v ; \skipp ; R,\stsp \sigma)\hspace{1pt}}{\hspace{-1pt}(x,\stsp \sigma\pr)}$ with some $x, \sigma, \sigma\pr$.
If $v =\stnsp \skipp$ then $x = \stnsp\skipp ; R$ and $\sigma\pr\hspace{-1.7pt} = \sigma$. 
Hence $\rpstep{\rho}{\hspace{-0.5pt}(\skipp ; \skipp ; L,\stsp \sigma)\stsp}{\hspace{-1pt}(\skipp ; L,\stsp \sigma)}$ is the matching step.
If $v \neq \skipp$ then there is a step $\rpstep{\rho}{(v,\stsp \sigma)\stsp}{\stnsp(v\pr,\stsp \sigma\pr)}$
such that $x = v\pr ; \skipp ; R\hspace{1.5pt}$ holds. 
Thus, $\rpstep{\rho}{(v ; \skipp ; L,\stsp \sigma)\hspace{1pt}}{\hspace{-1pt}(v\pr ; \skipp ; L,\stsp \sigma\pr)}$ is the matching transition in that case.
\end{proof}
\section{Replacing conditional and while-statements by jumps}\label{Sb:jump-norm}
\begin{lemma}\label{thm:cond-norm2} 
The program equivalence
\[
\peqvS{\ite{\stsp C}{\stsp p\stsp}{\ret{\hspace{0.2pt} j}\stdsp}}{}{\cjump{\negate{C}}{j}{\stsp p\stsp}}
\]
holds for all $\rho, C, p$ and  $j$.
\end{lemma}
\begin{proof}
By taking the simulation w.r.t. \hspace{-2pt}$\rho$
\[\{(\ite{\stsp C}{\stsp p\stsp}{\ret{\hspace{0.2pt}j}\stdsp},\;\cjump{\negate{C}}{j}{\stsp p\stsp})\} \cup\stdsp \op{id}\] for the $\sqsupseteq$-direction and
\[\{(\cjump{\negate{C}}{j}{\stsp p\stsp},\; \ite{\stsp C}{\stsp p\stsp}{\ret{j}\stdsp})\} \cup\stdsp \op{id}\] for the opposite.
\end{proof}

The following proof utilises the extra $\skipp$ built into the rule `{\it\sl While-True}' (\cf\stsp Figure~\ref{fig:pstep}).
\begin{lemma}\label{thm:while-norm2}
Assume
\begin{enumerate}
\item[\emph{(1)}] $\ret{\hspace{1pt}i} = \cjump{\negate{C}}{\stsp j\stdsp}{\stdsp p;\jump{\hspace{1pt}i\stdsp}}$,
\item[\emph{(2)}] $\ret{\hspace{0.2pt}j} = q$.
\end{enumerate}
Then $\peqvS{\while{\stsp C}{p}{q}\stsp}{}{\ret{\hspace{1.1pt}i}}$.
\end{lemma}
\begin{proof}
To establish the $\sqsupseteq$-direction we define 
\[
\hspace{-0.5cm}X \hspace{4pt}\defeq\hspace{4pt} \{(L, \stsp\ret{\hspace{0.9pt}i}), \:(\skipp ; L, \:\jump{\hspace{1pt}i})\}\cup\stsp \{(u ; \skipp ; L, \: u ; \jump{\hspace{1pt}i}\hspace{4.5pt} |\hspace{4.5pt} u \stsp\in\hspace{-0.7pt} \top \} \cup\stsp \op{id}
\]
where $L$ stands for $\while{\stsp C}{p}{q}$, and proceed by showing that $X$ is a simulation w.r.t. \!$\rho$.

Regarding the pair $(L, \stsp\ret{\hspace{0.7pt}i})$, suppose there is a step $\rpstep{\rho\hspace{-0.7pt}}{\hspace{-1.2pt}(\ret{\hspace{0.5pt}i},\stsp \sigma)\hspace{1pt}}{\hspace{-1pt}(x,\stsp \sigma)}$ with some $x$ and $\sigma$.
If $\sigma\hspace{-0.2pt}\in\hspace{-0.2pt} C$ then $x = p;\jump{\hspace{1.2pt} i}$ which can be matched by the program step $\rpstep{\rho\hspace{0.2pt}}{(L,\stsp \sigma)\hspace{1pt}}{\hspace{-1pt}(p ; \skipp ; L,\stsp \sigma)}$ since $(p ; \skipp ; L,\: p; \jump{\hspace{0.5pt}i})\hspace{0.2pt}\in \hspace{-1.5pt} X$.
If $\sigma\hspace{0.2pt} \notin\hspace{0.2pt} C$ then $x = \ret{\hspace{0.7pt}j} = q\stsp$ follows and we can simply take the step from $(L,\stsp \sigma)$ to the same configuration $(q,\stsp \sigma)$ since $\op{id}\hspace{0.2pt} \subseteq\hspace{-1pt} X$.

Next, with $(\skipp ; L, \:\jump{\hspace{0.5pt} i})$ any transition $\rpstep{\rho}{\hspace{-0.5pt}(\jump{\hspace{0.5pt}i},\stsp \sigma)\hspace{1pt}}{\hspace{-1pt}(\ret{\hspace{0.5pt}i},\stsp \sigma)}$ 
can be matched straight by $\rpstep{\rho\hspace{0.4pt}}{(\skipp ; L, \stsp\sigma)\hspace{1pt}}{\hspace{-1.2pt}(L,\stsp \sigma)}$ since $(L,\stdsp \ret{\hspace{0.7pt}i})\hspace{0.2pt} \in \hspace{-1pt}X$.

Lastly, suppose we have $\rpstep{\rho\hspace{-0.1pt}}{\hspace{-1pt}(u; \jump{i},\stsp \sigma)\hspace{1pt}}{\hspace{-1pt}(x,\stsp \sigma\pr)}$ with some $x, \sigma, \sigma\pr$ and $u$. If $u =\stnsp \skipp$ then $x =\stnsp \jump{\hspace{1pt}i}\stsp$ and $\sigma\pr\hspace{-1pt} = \sigma$ follow which is straight matched by
the step $\rpstep{\rho\stsp}{(\skipp; \skipp ; L, \stsp\sigma)\stsp}{\stnsp(\skipp ; L, \stsp \sigma)}$.
If $u \neq \skipp$ then $x = u\pr ; \jump{\hspace{1pt}i}$ due to a step 
$\rpstep{\rho\hspace{0.2pt}}{\stnsp(u,\stsp \sigma)\hspace{1pt}}{\hspace{-1pt}(u\pr,\stsp \sigma\pr)}$ with some $u\pr$. This transition further entails
$\rpstep{\rho\stsp}{\stnsp(u ; \skipp ; L,\stsp \sigma)\hspace{1pt}}{\hspace{-1pt}(u\pr ; \skipp ; L,\stsp \sigma\pr)}$ such that $(u\pr ; \skipp ; L,\: u\pr ;\jump{\hspace{1pt}i}) \in\hspace{-1.5pt} X$.

The correspondence $\pcorrS{\ret{\hspace{0.9pt}i}\stsp}{}{\hspace{0.2pt} L}$ follows symmetrically. 
\end{proof}

Note that for the sake of simplicity the above proposition captured only the basic principle behind the transformation:
it does not account for the circumstance that the construction $p;\jump{\hspace{0.5pt}i}$ usually creates an additional subject to the sequential normalisation
described in Section~\ref{Sb:seq-norm}. For instance, if $p$ is a conditional statement then
$\jump{\hspace{1pt}i}$ needs to be `pushed' into both branches.
To address that,
(1) may be generalised to $\ret{\hspace{0.9pt}i} = \cjump{\negate{C}}{j}{p\pr}$
with some $p\pr$ satisfying
$\peqvS{p\pr\hspace{-0.5pt}}{}{\hspace{0.1pt}p;\jump{\hspace{1pt}i}}$.
\chapter{Potential Computations and Program Correspondences}\label{S:pcs-props}
This chapter takes up the topic of potential computations and elaborates on the question how potential computations of two
corresponding programs are related.
In line with this, it will also pay due attention to a characterisation of the program equivalence in terms of computations.

\section{Conditions on potential computations}
We start by defining the conditions that allow us to systematically access subsets of potential computations of a program.
\begin{definition}\label{def:progC}
  Let $\stsp G\stsp$ be of type $\alpha\stnsp \times\stnsp \alpha \Rightarrow \ty{bool}$, \ie a state relation.
  Then the \emph{program condition}\index{condition!program}  $\progC{\hspace{-0.5pt}G}$ comprises 
  all infinite potential computations $\sq$ where $(\stateOf{\sq_{i}},\stsp \stateOf{\sq_{i+1}}) \in\hspace{-0.5pt} G\stsp$ holds for each program step $\sq_{i} \stsp\pstep\hspace{-2.4pt} \sq_{i+1}$,
  as well as all finite non-empty prefixes of such $\sq$. 
\end{definition}
Program conditions can thus be used to encircle computations where the pre- and the poststate of each program step are related by some given $G$
and, as one might expect, environment steps will be handled accordingly next. 
\begin{definition}\label{def:envCi}
Let $R\stdsp$ be a state relation. Then the \emph{environment condition}\index{condition!environment}  $\envCi{\hspace{-2.2pt}R}\stsp$ comprises
all infinite potential computations $\sq$ where the condition $(\stateOf{\sq_{i}},\stsp \stateOf{\sq_{i+1}}) \in \hspace{-1pt} R\stsp$
holds for each environment step $\sq_{i}\stsp \estep\hspace{-2.7pt} \sq_{i+1}$.
\end{definition}
This definition addresses only infinite computations for the following reason.
As defined by Stirling~\cite{STIRLING1988347}, the set of \emph{actual computations}\index{computation!actual} of a program is
the set of its potential computations without any environment steps.
Thus, $\envCi{\hspace{-1.5pt}\bot}$
contains all infinite actual computations (here and in what follows, $\bot$ denotes the empty relation: the complement of the already introduced $\top$).
Prefixing $\envCi{\hspace{-2pt}\bot}$ does not give us the entire set of finite actual computations however:
actual computations of programs terminating on all inputs when run in isolation would for instance be omitted.
Therefore we explicitly recast Definition~\ref{def:envCi} as follows:
\begin{definition}\label{def:envC}
Let $R\stdsp$ be a state relation. Then the \emph{environment condition}\index{condition!environment}  $\envC{\hspace{-1.2pt}R}$ comprises
all finite potential computations $\sq$ where the condition $(\stateOf{\sq_{i}}, \stsp\stateOf{\sq_{i+1}}) \in \hspace{-1.2pt} R\stsp$
holds for any $i < |\sq| -1$ with $\sq_{i}\stsp \estep\hspace{-3pt} \sq_{i+1}$.
\end{definition}
The set of (finite or infinite) actual computations is then consequently the union of $\envCi{\hspace{-2.5pt}\bot}$ and $\envC{\hspace{-1pt}\bot}$.
However, only the notation $\envC{\hspace{-1.2pt} R}$ will be used in the sequel as it will be clear from the context which variant
of the environment condition is meant. 

The remaining two conditions apply to both, finite and infinite computations, in particular allowing us to address input/output properties of programs. 
\begin{definition}\label{def:inC-outC}
  Let $\stsp C$ be of type $\alpha \Rightarrow \ty{bool}$, \ie a state predicate.
  A potential computation $\sq$  satisfies the \emph{input condition}\index{condition!input} $\inC{\hspace{-0.5pt}C}$ if $\stsp\stateOf{\sq_0}\hspace{-0.5pt} \in\hspace{-0.7pt} C$.
  Furthermore, if there is an index $i$ with $\progOf{\sq_i}\hspace{-0.9pt} =\hspace{-0.5pt} \skipp$ and $\stateOf{\sq_\mu}\hspace{-0.7pt} \in\hspace{-1pt} C$ holds for the first such index $\mu$ 
  then $\sq$ satisfies the \emph{output condition}\index{condition!output} $\outC{C}$.
\end{definition}
Despite a different formulation,
the definition of $\outC{\hspace{-0.2pt} C}$ essentially coincides with its counterpart given by Stirling~\cite{STIRLING1988347}: one can show that
$\sq \in \outC{C}$ holds iff $\progOf{\sq_i}\stnsp = \skipp$  implies
that there exists some $j \le i\stsp$ such that $\progOf{\sq_j} = \skipp\stsp$ and $\stsp\stateOf{\sq_j}\hspace{-0.5pt} \in\hspace{-0.5pt} C\stsp$ hold 
for any $\stsp i\hspace{-0.5pt} <\hspace{-0.5pt} |\sq|\stsp$ in case $\sq$ is finite or just for any $i\in\hspace{-0.7pt} \naturals$ otherwise.
\section{Replaying finite potential computations}\label{Sb:corr_pcs}
The following proposition resorts to environment conditions for
showing how finite potential computations can be replayed
along program correspondences.
\begin{lemma}\label{thm:corr-sim}
Assume that $\pcorr{p}{r}{\stnsp q}$, $\sq^q\hspace{-0.5pt} \in \rpcs{q}{\rho\pr}$, $\sq^q\hspace{-0.5pt} \in \envC{\hspace{-1.2pt}R\pr}$ and 
$\stsp \rcomp{r\hspace{1.1pt}}{\hspace{0.1pt}R\pr} \hspace{-0.7pt}\subseteq\hspace{-1.2pt} \rcomp{R\hspace{1pt}}{\hspace{1pt} r}$ hold.
Further, let $\stsp\sigma$ be a state with $(\sigma,\stsp \stateOf{\sq^q_0})\stnsp \in\hspace{-0.5pt} r$.
Then there exists a computation $\sq^p\hspace{-0.79pt} \in\hspace{0.2pt} \rpcs{p}{\rho}$ with the same length as $\sq^q$
satisfying moreover the following conditions:
\begin{enumerate}
\item[\emph{(1)}] $\stateOf{\sq^p_0} = \sigma$,
\item[\emph{(2)}] $\sq^p\hspace{-0.5pt} \in\hspace{0.2pt} \envC{\hspace{-1.2pt}R}$,
\item[\emph{(3)}] $(\stateOf{\sq\mystrut^p_i},\stsp \stateOf{\sq\mystrut^q_i}) \in\hspace{-0.2pt} r\stdsp$ for any $\stsp i < |\sq^q|$,
\item[\emph{(4)}] $\pcorr{\progOf{\sq\mystrut^p_i}}{r}{\hspace{-1.2pt}\progOf{\sq\mystrut^q_i}}\stdsp$ for any $\stsp i < |\sq^q|$,
\item[\emph{(5)}] $\sq\mystrut^p_i\stdsp \estep\hspace{-2.5pt} \sq\mystrut^p_{i+1}\stsp$ iff $\stdsp\sq\mystrut^q_i\stsp \estep\hspace{-2.5pt} \sq\mystrut^q_{i+1}\stdsp$ for any $\stsp i < |\sq^q| - 1$. 
\end{enumerate}
\end{lemma} 
\begin{proof}
We proceed by induction on the length of $\sq^q$. The case $|\sq^q|\hspace{-1.5pt} =\hspace{-0.5pt} 0$ is trivial since 
$\sq^q\hspace{-0.9pt} \in \hspace{-0.5pt}\rpcs{q}{\rho\pr}$ implies that $\sq^q$ is not empty. Next, if $|\sq^q| = 1$ then the singleton sequence
$(p, \sigma) \in \rpcs{p}{\rho}$ clearly meets the conditions (1)--(5).

Further, suppose $|\sq|\hspace{-0.5pt} = n + 1$ with $n \ge 1$ and consider $\prefix{n}{\sq\mystrut^q}$, \ie the prefix $\sq\mystrut^q_0, \ldots, \sq\mystrut^q_{n-1}$ with the length $n$,
for which we have $\prefix{n}{\sq^q}\hspace{-0.4pt}\in\hspace{0.2pt} \rpcs{q}{\rho\pr}\hspace{-1pt} \cap\hspace{0.2pt}\envC{\hspace{-0.9pt}R\pr}$
and $(\sigma,\stsp \stateOf{\prefix{n}{\sq^q}_0})\stnsp \in\hspace{-0.9pt} r$.
Hence, by the induction hypothesis there is a computation $\sq^p \in\hspace{0.5pt} \rpcs{p}{\rho}$ with $|\sq^p|\stnsp = n\stsp$ satisfying moreover
the conditions (1)--(5) w.r.t. \!$\prefix{n}{\sq^q}$.
This particularly means:
\begin{enumerate}
\item[(a)] $(\stateOf{\sq\mystrut^p_{n-1}}, \stsp\stateOf{\sq\mystrut^q_{n-1}}) \in\hspace{0.2pt} r$ and
\item[(b)] $\pcorr{\progOf{\sq\mystrut^p_{n-1}}}{r}{\stnsp\progOf{\sq\mystrut^q_{n-1}}}$.
\end{enumerate}

If we have an environment step $\sq\mystrut^q_{n-1}\hspace{-1pt} \estep\hspace{-2.5pt} \sq\mystrut^q_{n}\:$ then $(\stateOf{\sq\mystrut^q_{n-1}},\stsp \stateOf{\sq\mystrut^q_{n}}) \in\hspace{-1.2pt} R\pr$
holds since $\sq^q\hspace{-0.25pt} \in\hspace{0.1pt} \envC{\hspace{-1pt}R\pr}$ is assumed. Hence, 
$(\stateOf{\sq\mystrut^p_{n-1}},\stsp \stateOf{\sq\mystrut^q_{n}}) \in\hspace{0.5pt} \rcomp{r\hspace{1pt}}{\hspace{0.2pt} R\pr}$ can be further inferred using (a).
Then the assumption $\rcomp{r\hspace{1.1pt}}{\hspace{0.35pt} R\pr}\hspace{-0.5pt} \subseteq\hspace{-1pt} \rcomp{R\hspace{1.7pt}}{\hspace{1.5pt} r}$ provides some $\sigma\pr$ such that
$(\stateOf{\sq\mystrut^p_{n-1}}, \stsp\sigma\pr) \in\hspace{-1pt} R$ and $(\sigma\pr,\stsp \stateOf{\sq\mystrut^q_{n}}) \in r$ hold.
We can thus extend $\sq^p$ by an environment step to 
$
\sq\mystrut^p_0, \ldots, \sq\mystrut^p_{n-1}\stsp \estep \stnsp(\progOf{\sq\mystrut^p_{n-1}},\stsp \sigma\pr)
$ 
which yields a computation in $\rpcs{p}{\rho}$ of length $n+1$ satisfying the conditions (1)--(5) w.r.t. \!$\sq^q$.

If we have a program step $\rpstep{\rho\pr\hspace{-0.7pt}}{\hspace{-0.5pt}\sq\mystrut^q_{n-1}}{\hspace{-2.5pt}\sq\mystrut^q_{n}}$
then with (a) and (b) we can further obtain
a step
$\rpstep{\rho}{\hspace{-0.9pt}\sq\mystrut^p_{n-1}}{\hspace{-0.7pt}(p\pr,\stsp \sigma\pr)}$ such that $\pcorr{p\pr}{r}{\hspace{-1.7pt}\progOf{\sq\mystrut^q_{n}}}$ and
$(\sigma\pr, \stsp\stateOf{\sq\mystrut^q_{n}})\hspace{-0.1pt}  \in\hspace{-0.2pt} r$ hold.
Extending thus $\sq\mystrut^p$ to 
$
\sq\mystrut^p_0, \ldots, \sq\mystrut^p_{n-1} \stsp\pstep\hspace{0.1pt} (p\pr, \sigma\pr)
$
we once more get 
a computation in $\rpcs{p}{\rho}$ of length $n + 1$ satisfying (1)--(5) w.r.t. \!$\sq^q$.
\end{proof}
In particular, with $r$ instantiated by $\op{id}$ whereas $R$ and $R\pr$ -- by $\top$, the above statement simplifies to:
\begin{corollary}\label{thm:corr-sim-id}
If $\pcorr{p}{}{q}$ and $\sq^q\stnsp \in \rpcs{q}{\rho\pr}$
then there is a computation $\sq^p\hspace{-0.5pt} \in \rpcs{p}{\rho}$ with $|\sq^p| = |\sq^q|$
satisfying moreover
\begin{enumerate}
\item[\emph{(1)}] $\stateOf{\sq^p_i} = \stateOf{\sq^q_i}\stsp$ for any $\stsp i < |\sq^q|$,
\item[\emph{(2)}] $\pcorr{\progOf{\sq^p_i}}{}{\progOf{\sq^q_i}}\stsp$ for any $\stsp i < |\sq^q|$,
\item[\emph{(3)}] $\sq^p_i\stsp \estep\hspace{-2.9pt} \sq^p_{i+1}$ iff $\stsp\sq^q_i \stsp\estep\hspace{-2.9pt} \sq^q_{i+1}\stsp$ for any $\stsp i < |\sq^q| - 1$. 
\end{enumerate}
\end{corollary}
\section{A characterisation of the program equivalence}\label{Sb:prg-equiv}
Apart from its primary application later in Section~\ref{Sb:pcorr-rule}, 
Proposition~\ref{thm:corr-sim} is also a step towards an interpretation of what a program equivalence
$\peqv{p\hspace{-0.2pt}}{}{\hspace{-0.5pt} q}$ actually means in terms of the potential computations of $(\rho, p)$ and $(\rho\pr, q)$:
this is expressed by the following proposition with a proof
utilising Corollary~\ref{thm:corr-sim-id}. 
\begin{lemma}\label{thm:corr-id-iff}
  \stdsp$\pcorr{p\hspace{-0.4pt}}{}{\hspace{-0.5pt}q}$ holds iff there exists a function $\Phi$ sending each $\sq^q \hspace{-0.5pt}\in\hspace{0.2pt} \rpcs{q}{\rho\pr}$
  to a non-empty set of computations $\Phi(\sq^q)$ such that
  \begin{enumerate}
  \item[\emph{(1)}] $\sq \in\hspace{0.1pt} \rpcs{p}{\rho}$ and $|\sq| = |\sq^q|\:$ for all $\sq \in \Phi(\sq^q)$,
  \item[\emph{(2)}] $\stateOf{\sq_i} = \stateOf{\sq^q_i}\:$ for any $\stsp i < |\sq^q|$ and $\sq \in \hspace{-0.2pt}\Phi(\sq^q)$,
  \item[\emph{(3)}] $\sq_i = \skipp\stsp$ iff $\stsp\sq^q_i = \skipp\;$ for any $\stsp i < |\sq^q|$ and $\sq \in \Phi(\sq^q)$,  
  \item[\emph{(4)}] $\sq_i \stsp\estep\hspace{-2.7pt} \sq_{i+1}$ iff $\stsp\sq^q_i\stsp \estep\hspace{-2.7pt} \sq^q_{i+1}\;$ for any $\stsp i < |\sq^q| - 1$ and $\sq \in \Phi(\sq^q)$,
  \item[\emph{(5)}]  $\Phi(\prefix{i+1\hspace{-0.5pt}}{\sq^q}) \subseteq\hspace{1pt} \prefix{i+1\hspace{-0.5pt}}{\Phi(\sq^q)}\;$ for any $\stsp i< |\sq^q|$ 
  \end{enumerate}
  where $\prefix{i+1}{\Phi(\sq^q)}$ denotes the set of prefixes of length $i+\hspace{-1pt}1$ taken from each computation in $\Phi(\sq^q)$.
\end{lemma}
\begin{proof}
  To show that the condition $\pcorr{\hspace{-0.5pt}p\hspace{-0.1pt}}{}{\hspace{-0.2pt}q}$ is necessary, suppose there is such $\Phi$ and let 
\[
X\hspace{3pt} \defeq \hspace{3pt} \{(\progOf{\sq_{|\sq^q|-1}},\stsp \progOf{\sq^q_{|\sq^q|-1}}) \hspace{3.9pt}|\hspace{3.9pt} \sq\mystrut^q\hspace{-0.2pt} \in \rpcs{q}{\rho\pr} \wedge \sq\hspace{0.5pt} \in \Phi(\sq^q)\}.
\]
Due to (1) and (2) we have $\Phi(q, \sigma) = \{(p, \sigma)\}$ for any state $\sigma$
(where $(q, \sigma)$ and $(p, \sigma)$ shall be understood as computations of length $1$), hence
$(p, q)\hspace{-0.1pt} \in\hspace{-1pt} X$ holds.
Next, we show that $X$ is a simulation w.r.t. $\hspace{-2.5pt}\rho$, $\rho\pr$ and $\op{id}$ (\cf\stdsp Definition~\ref{def:corr}).
To this end
we may assume $\sq^q\hspace{-0.9pt} \in\hspace{-0.1pt} \rpcs{q}{\rho\pr}$, $\sq\hspace{-0.2pt} \in\hspace{-0.5pt}  \Phi(\sq^q)$, $|\sq^q|\hspace{-1pt} = n +\hspace{-0.1pt} 1$ with $n\in\hspace{-0.2pt}\naturals$
and, moreover, a program step $\rpstep{\rho\pr\hspace{-0.7pt}}{\hspace{-1pt}(\progOf{\sq^q_n},\stsp \sigma)\hspace{1pt}}{\hspace{-1pt}(v, \stsp\sigma\pr)}$ 
for which we have to provide a matching step from the configuration $(\progOf{\sq_n},\stsp \sigma)$.
Let $\sq^{q_1}$ extend $\sq^q$ by two configurations as follows:
$
\sq^q\hspace{1pt} \estep\hspace{-1pt} (\progOf{\sq^q_n}, \stsp\sigma)\hspace{1pt} \pstep\hspace{-1pt} (v, \stsp\sigma\pr)
$.
Thus, we have $\sq^{q_1}\hspace{-1.7pt}\in \hspace{-0.5pt}\rpcs{q}{\rho\pr}$ and $\prefix{n+1}{\sq^{q_1}}\hspace{-1.7pt} =\hspace{-1pt} \sq^q$ so that by (5)
there exists a computation $\sq\pr\hspace{-1pt} \in\hspace{-0.7pt} \Phi(\sq^{q_1})$ with $\prefix{n+1}{\sq\pr}\hspace{-1pt} = \hspace{-1pt}\sq$.
Since $\sq\pr$ also meets the conditions (1)--(4), it must be of the form
$\sq\hspace{1pt} \estep\hspace{-1pt} (\progOf{\sq_n},\stsp \sigma) \hspace{1pt}\pstep\hspace{-1pt} (u,\stsp \sigma\pr)$
giving us the matching step $\rpstep{\rho\stsp}{\hspace{-0.4pt}(\progOf{\sq_n},\stsp \sigma)\stsp}{\hspace{0.2pt} (u,\stsp \sigma\pr)}$ 
because
$(u, v) \stnsp =\stnsp (\progOf{\sq\pr_{n+2}}, \progOf{\sq^{q_1}_{n+2}})$ \ie $(u, v) \in\hspace{-1pt} X$. 
For any $\sq^q\hspace{-1pt} \in\hspace{-0.5pt} \rpcs{q}{\rho\pr}$ and any $\sq\hspace{-0.5pt} \in \hspace{-1pt}\Phi(\sq^q)$ we can lastly conclude using (3)
that $\progOf{\sq_{|\sq^q|-1}} \hspace{-0.5pt}=\hspace{-0.5pt} \skipp\stdsp$ holds iff $\stdsp\progOf{\sq^q_{|\sq^q|-1}} \hspace{-0.5pt}=\hspace{-0.5pt} \skipp$ does.

In order to show that $\pcorr{\hspace{-0.5pt}p\hspace{-0.1pt}}{}{ q}$ is sufficient,
let $\Phi$ send each 
$\sq^q \in \rpcs{q}{\rho\pr}$ to the non-empty set
\[
\{\sq \in\hspace{0.2pt} \rpcs{p}{\rho}\hspace{3pt} |\hspace{3.5pt} \sq \mbox{ satisfies the conditions of Corollary~\ref{thm:corr-sim-id} w.r.t. }\! \sq^q \}
\]
such that (1)--(4) clearly hold.
Regarding (5), let $i < |\sq^q|$ and $\sq^{p_1}\hspace{-1pt} \in\hspace{-0.5pt} \Phi(\prefix{i+1}{\sq^q})$.
Since Corollary~\ref{thm:corr-sim-id} provides stepwise correspondence of $\sq^{p_1}$ and $\prefix{i+1}{\sq^q}$, 
we can particularly infer $\pcorr{\progOf{\sq^{p_1}_i}\hspace{-1pt}}{}{\hspace{-0.5pt}\progOf{\sq^q_i}}$.
Further, the suffix $\suffix{i}{\sq\mystrut^q}$ starts with $\sq\mystrut^q_i$ and is not empty so that
$\suffix{i}{\sq\mystrut^q}\hspace{-0.5pt} \in \rpcs{\progOf{\sq\mystrut^q_i}}{\rho\pr}$ holds. Resorting to Corollary~\ref{thm:corr-sim-id} once more,
we can replay this suffix by some $\sq^{p_2}\hspace{-1.2pt} \in\hspace{-0.5pt} \rpcs{\progOf{\sq^{p_1}_i}}{\rho}$ with $|\sq^{p_2}| \stnsp=\stnsp |\sq^q| -\hspace{0.2pt} i$.
Note that this construction ensures $\sq^{p_1}_i\hspace{-1.9pt} = \hspace{-0.5pt}\sq^{p_2}_0$,
\ie the last configuration of  $\sq^{p_1}$ is the first of $\sq^{p_2}$.
Thus, $\sq^{p_1}\hspace{-1.7pt} \scomp\hspace{0.1pt} \sq^{p_2} \in \rpcs{p}{\rho}$ is well-defined and
we have $\prefix{i+1}{(\sq^{p_1}\hspace{-1.7pt} \scomp\hspace{0.1pt} \sq^{p_2})} = \sq^{p_1}$
and, moreover,
$\sq^{p_1}\hspace{-1.9pt} \scomp\hspace{0.1pt} \sq^{p_2} \hspace{0.2pt}\in \Phi(\sq^q)$ since
$|\sq^{p_1}\hspace{-1.5pt} \scomp\hspace{0.1pt} \sq^{p_2}| = |\sq^{p_1}| + |\sq^{p_2}| - 1 = (i\hspace{0.7pt} + 1) + (|\sq^q| -\stdsp i) - 1 = |\sq^q|$. 
\end{proof}

Reflecting on this result, $\peqv{p}{\stnsp}{q}\stsp$ can be expressed by a second-order formula which
relates computations of $p$ and $q$ without self-references in contrast to the characterisation from Section~\ref{S:gfp} as an intersection of two greatest fixed points.
On the other hand, here we have the quite intricate condition (5) capturing that for any  $\sq^q \hspace{-1.2pt}\in\hspace{-1pt} \rpcs{q}{\rho\pr}$ we can find
some  $\sq^p\hspace{-1.5pt} \in\hspace{-0.9pt} \rpcs{p}{\rho}$ which not only corresponds to $\sq^q$ by means of (1)--(4) but is also
extendable in response to any possible extension of $\sq^q$.
Moreover, note that it is in principle possible to have some $\sq \in\hspace{-0.5pt} \Phi(\sq_1\hspace{-1.7pt} \scomp\hspace{-0.3pt} \sq_2)$
with $\prefix{n}{\sq}\hspace{0.2pt} \notin\hspace{-0.5pt} \Phi(\sq_1)$ where $n = |\sq_1|$ and $|\sq_2| > 1$ since the condition (5) is not an equality.

It is finally worth noting that projecting a computation $\sq_0, \ldots, \sq_n$ to the sequence of states $\stateOf{\sq_0}, \ldots, \stateOf{\sq_n}$ 
can be viewed as constructing a trace of the computation. Thus, by Proposition~\ref{thm:corr-id-iff}
from the assumption
$\peqv{\hspace{-0.1pt}p\hspace{-0.1pt}}{}{\hspace{-0.5pt}q}\stsp$
follows that the sets of such traces of $(\rho, p)$ and $(\rho\pr, q)$ are equal,
as one shall expect from equivalent programs.
Although such traces may be practical for various purposes, they omit important pieces of structural data,
particularly used to define the notion of fair computations in Section~\ref{Sb:faircomp}.
\chapter{A Hoare-style Rely/Guarantee Program Logic}\label{S:prog-log}
The following chapter applies most of the concepts presented so far and
is devoted to a Hoare-style program logic~\cite{Hoare}
for structured reasoning upon properties of programs with interleaved computations. 
\section{The extended Hoare triples}
Within the logic, program properties are expressed by means of the \emph{extended Hoare triples}\index{extended Hoare triple}, denoted by
$\rgvalidr{\rho}{R}{\hspace{-1pt}P}{\hspace{-0.1pt}p\hspace{0.2pt}}{Q}{\hspace{-0.7pt}G}$ where the component $R$ is called \emph{the rely}\index{condition!rely}, $P$ -- \emph{the precondition},
$Q$ -- \emph{the postcondition}, and $G$ -- \emph{the guarantee}\index{condition!guarantee}.
Following Stirling~\cite{STIRLING1988347}, the word \emph{extended} emphasises that such triples arise  
by a generalisation of the Hoare triples used in the method by Owicki and Gries~\cite{Owicki_Gries_76} 
where non-interference conditions are derived from annotated assertions. With the extended triples, such assertions are instead neatly captured
by the components that constitute the extension:
the rely and the guarantee namely.
\begin{definition}\label{def:rgval}
Let $(\rho, p)$ be a program, $R$ and $G$ -- state relations, 
whereas $P$ and $Q$ -- state predicates. Then the condition $\rgvalidr{\rho}{R}{\hspace{-1.7pt}P}{\hspace{-0.2pt}p\hspace{-0.2pt}}{Q}{\hspace{-1.5pt}G}$ holds iff
$\stdsp\envC{\hspace{-1pt}R}\hspace{0.4pt} \cap\hspace{0.4pt} \inC{\hspace{-0.9pt}P}\hspace{-0.9pt} \cap\hspace{0.2pt} \rpcs{p}{\rho}\stsp \subseteq\stsp \outC{\hspace{-0.3pt} Q}\hspace{0.2pt} \cap \progC{\hspace{-0.3pt} G}\stsp$ does.
\end{definition}
Next definition explicitly addresses properties of infinite computations too:
\begin{definition}\label{def:rgival}
The condition $\rgvalidri{\rho}{R}{\hspace{-1.5pt}P}{\hspace{-0.2pt}p\hspace{-0.2pt}}{Q}{\hspace{-1.5pt}G}$ holds iff
\begin{enumerate}
\item[(1)] $\envC{\hspace{-1.2pt}R}\hspace{0.4pt} \cap\hspace{0.4pt} \inC{\hspace{-1.2pt}P} \cap\hspace{0.5pt} \rpcsi{p}{\rho}\stdsp \subseteq\stsp \outC{\hspace{-0.3pt} Q} \cap \progC{\hspace{-0.5pt} G}\stsp $ and
\item[(2)] $\rgvalidr{\rho}{R}{\hspace{-1.5pt}P}{\hspace{-0.2pt}p\hspace{-0.2pt}}{Q}{\hspace{-1.5pt}G}$.
\end{enumerate}
\end{definition}
As shown next, these two definitions turn out to be logically equivalent so all rules in this chapter can refer to Definition~\ref{def:rgval} without any loss of generality. 
\begin{lemma}
$\rgvalidr{\rho}{R}{\hspace{-1.5pt}P}{p}{Q}{\hspace{-1.2pt}G}\stsp$ iff $\stsp\rgvalidri{\rho}{R}{\hspace{-1.5pt}P}{p}{Q}{\hspace{-1.2pt}G}$.
\end{lemma} 
\begin{proof}
  That the condition $\rgvalidr{\rho}{R}{\hspace{-1pt}P}{\hspace{0.4pt}p\hspace{0.4pt}}{Q}{\hspace{-0.9pt}G}\stsp$ is necessary
  follows immediately from Definition~\ref{def:rgival} whereas
  the opposite direction amounts to showing that 
  $\envC{\hspace{-1.5pt}R}\hspace{0.2pt} \cap\hspace{0.2pt} \inC{\hspace{-1.5pt}P}\hspace{-0.9pt} \cap \rpcsi{p}{\rho}\stsp \subseteq\stsp \outC{\hspace{-0.5pt} Q}\hspace{0.5pt}\hspace{-0.5pt} \cap\hspace{0.1pt} \progC{\hspace{-0.5pt} G}\stsp$ is implied by $\rgvalidr{\rho}{R}{\hspace{-1.5pt}P}{\hspace{-0.5pt}p\hspace{-0.5pt}}{Q}{\hspace{-1.5pt}G}$.
To this end let $\sq$ be an infinite computation in $\rpcsi{p}{\rho}$ with $\sq\hspace{0.2pt} \in\hspace{0.2pt} \envC{\hspace{-1pt} R}\hspace{0.5pt}  \cap\hspace{0.4pt} \inC{\hspace{-1pt} P}$.

To show $\sq\hspace{-0.2pt} \in\hspace{-0.5pt} \outC{\hspace{-0.3pt}Q}$ we can assume $\progOf{\sq_i} = \skipp$ with some $i\hspace{0.5pt} \in\stnsp \naturals$.
Then also the prefix $\prefix{i+1}{\sq}\hspace{-0.7pt} \in\hspace{-0.7pt} \rpcs{p}{\rho}$ satisfies $\envC{\hspace{-1.5pt} R}$ and $\inC{\hspace{-1.5pt} P}$ 
such that $\prefix{i+1}{\sq}\hspace{-0.5pt} \in\hspace{-1pt} \outC{\hspace{-0.4pt}Q}$ follows from $\rgvalidr{\rho}{R}{\hspace{-1pt}P}{\hspace{-0.5pt}p\hspace{-0.5pt}}{Q}{\hspace{-1pt}G}$. 
Since the last element of $\prefix{i+1}{\sq}$ is the $\skipp$-configuration $\sq_i$
we get some $j\stsp \le\stsp i$ with $\progOf{\sq_j} = \skipp$ and $\stateOf{\sq_j} \in\hspace{-0.4pt} Q$. 

Next, to show $\sq \in \hspace{-0.5pt}\progC{\hspace{-0.1pt}G}\stsp$ we proceed very similarly. 
Assuming a program step $\rpstep{\rho}{\sq_{i}\hspace{0.7pt}}{\hspace{-2.5pt}\sq_{i+1}}$ for some $i\hspace{0.7pt} \in\hspace{-0.2pt} \naturals$, we now consider
the prefix $\prefix{i+2}{\sq}\hspace{0.2pt} \in\hspace{-0.2pt} \rpcs{p}{\rho}$ which satisfies $\envC{\hspace{-1.2pt} R}$ and $\inC{\hspace{-1.2pt} P}$. 
Using $\rgvalidr{\rho}{R}{\hspace{-1.9pt}P}{\hspace{-0.9pt}p\hspace{-0.9pt}}{Q}{\hspace{-1.5pt}G}$ we can first infer $\prefix{i+2}{\sq}\hspace{0.2pt} \in \progC{\hspace{-0.2pt} G}$
to ultimately conclude $(\stateOf{\sq_{i}},\stsp \stateOf{\sq_{i+1}}) \in\hspace{-0.5pt} G$.
\end{proof}
\section{The program correspondence rule}\label{Sb:pcorr-rule}
In essence, this rule gives sufficient conditions for turning program correspondences
into implications on program logical properties. 
From now on, let $\rimg{R\hspace{1.2pt}}{\hspace{0.2pt}X}$ denote $\stsp\{b \hspace{3.7pt}|\hspace{3pt}\exists a\hspace{0.2pt}\in\hspace{-1pt} X. \hspace{2.5pt}(a, b)\hspace{0.2pt} \in\hspace{-0.9pt} R\}$,
\ie the image\index{relational!image} of a set $X$ under a relation $R$. 
\begin{lemma}\label{thm:pcorr-rule}
Assume $\pcorr{p\hspace{0.2pt}}{r}{\hspace{-0.5pt} q}$, $\rgvalidr{\rho\stsp}{R}{\hspace{-1.5pt}P}{\hspace{-0.9pt}p\hspace{-0.5pt}}{Q}{\hspace{-1pt}G}$ and
\begin{enumerate}
\item[\emph{(1)}] $P\pr \subseteq\stdsp \rimg{r\hspace{0.9pt}}{\hspace{0.7pt}P}$,
\item[\emph{(2)}] $\rcomp{r\hspace{1.5pt}}{\hspace{0.7pt}R\pr}\stsp \subseteq \stsp\rcomp{R\hspace{1.7pt}}{\hspace{1.7pt} r}$,
\item[\emph{(3)}] $\rimg{r\hspace{0.7pt}}{\hspace{1.1pt}Q} \stdsp\subseteq\stsp Q\pr$, 
 \item[\emph{(4)}] $\rcomp{\rcomp{\rconv{r}\hspace{-0.2pt}}{\hspace{1.7pt} G \hspace{1.2pt}}}{\hspace{1.7pt} r}\stdsp \subseteq\stsp G\pr$.
\end{enumerate}
Then $\rgvalidr{\rho\pr}{R\pr}{\hspace{-1.5pt}P\pr}{\hspace{-0.5pt}q\hspace{-0.5pt}}{Q\pr}{\hspace{-1.5pt}G\pr}$.
\end{lemma}
\begin{proof}
Let $\sq^q\hspace{-0.4pt} \in\hspace{0.1pt} \rpcs{q}{\rho\pr}\hspace{-0.2pt} \cap\stdsp \envC{\hspace{-0.2pt}R\pr} \cap\stdsp \inC{\hspace{-0.4pt}P\pr}$.
Then by (1) there is 
a state $\sigma \in\stnsp P$ such that $(\sigma,\stsp \stateOf{\sq^q_0})\hspace{-0.2pt} \in\hspace{-0.3pt} r$.
With (2), Proposition~\ref{thm:corr-sim} further provides the existence of 
a potential computation $\sq^p\hspace{-0.9pt} \in\hspace{-0.2pt} \rpcs{p}{\rho}$ with $|\sq^p| = |\sq^q|$ satisfying moreover
\begin{enumerate}
\item[(a)] $\sq^p \in\hspace{0.2pt} \envC{\hspace{-1.5pt}R}\hspace{0.5pt} \cap\hspace{0.2pt} \inC{\hspace{-1.1pt}P}$,
\item[(b)] $(\stateOf{\sq^p_i},\stsp \stateOf{\sq^q_i}) \in\hspace{-0.2pt} r\stsp$ for any $i < |\sq^q|$,
\item[(c)] $\pcorr{\progOf{\sq^p_i}\hspace{-0.2pt}}{r}{\hspace{-1.5pt}\progOf{\sq^q_i}}\stsp$ for any $i < |\sq^q|$,
\item[(d)] $\sq^p_i\hspace{1pt} \estep\hspace{-3pt} \sq^p_{i+1}\stsp$ iff $\stsp\sq^q_i\hspace{1pt} \estep\hspace{-3pt} \sq^q_{i+1}\stsp$ for any $i < |\sq^q| - 1$.   
\end{enumerate}
With (a) the assumed triple $\rgvalidr{\rho}{R}{\hspace{-1.5pt}P}{p}{Q}{\hspace{-1.5pt}G}$ gets applicable to $\sq^p$ yielding 
$\sq^p\hspace{-0.5pt} \in \outC{\hspace{-0.2pt}Q}\hspace{-0.1pt} \cap \progC{\hspace{-0.7pt}G}$ which is in turn used to establish
$\sq^q\hspace{-0.4pt} \in \outC{\hspace{-0.2pt}Q\pr}\hspace{-0.5pt} \cap \progC{\hspace{-0.2pt}G\pr}$.
 
First, suppose $\progOf{\sq^q_i}\hspace{-1.5pt} =\hspace{-0.5pt} \skipp$ with some $i\stnsp <\hspace{-1.2pt} |\sq^q|$.
Then $\progOf{\sq^p_i}\hspace{-1.5pt} =\hspace{-0.5pt} \skipp$ follows by (c).
Hence, $\sq^p \hspace{-1.2pt}\in\hspace{-0.7pt} \outC{\hspace{-0.2pt}Q}$ entails some $j \le i$ such that
$\sq^p_j$ is a $\skipp$-configuration and $\stateOf{\sq^p_j}\hspace{0.2pt} \in\hspace{-0.5pt} Q$. 
From (b) and (c) we further infer that $\sq^q_j$ must be also a $\skipp$-configuration and $(\stateOf{\sq^p_j},\stsp \stateOf{\sq^q_j})\hspace{-0.2pt} \in\hspace{-0.5pt} r$ holds.
That is, $\stateOf{\sq^q_j}$ is in the image of $Q$ under $r$ and hence $ \stateOf{\sq^q_j}\in\hspace{-0.2pt} Q\pr$ follows by (3).

Second, in order to show $\sq^q\hspace{-1pt} \in\hspace{-0.9pt} \progC{\hspace{-0.4pt}G\pr}\stsp$ suppose there is some  $i < |\sq^q| - 1$ with
a program step $\rpstep{\rho\pr}{\sq^q_i\hspace{1pt}}{\hspace{-2pt}\sq^q_{i+1}}$.
Hence from (d) follows that there must be also a program step at the same position on $\sq^p$,
\ie $\stsp\rpstep{\rho\hspace{0.1pt}}{\hspace{-0.5pt}\sq^p_i\hspace{1pt}}{\hspace{-2.7pt}\sq^p_{i+1}}$,
so that
$\sq^p\hspace{-1pt} \in\hspace{-0.7pt} \progC{\hspace{-0.1pt}G}$ entails $(\stateOf{\sq^p_i},\stsp \stateOf{\sq^p_{i+1}})\hspace{-0.5pt} \in\hspace{-0.7pt} G$.
Then $(\stateOf{\sq^q_i}, \stsp\stateOf{\sq^q_{i+1}})\hspace{-0.5pt} \in\hspace{-0.5pt} G\pr$ follows by (4) since both,
$(\stateOf{\sq^p_i},\stsp \stateOf{\sq^q_{i}})\hspace{-0.2pt} \in\hspace{-0.2pt} r$ and $(\stateOf{\sq^p_{i+1}},\stsp \stateOf{\sq^q_{i+1}})\hspace{-0.5pt} \in\hspace{-0.2pt} r$,
are provided by (b).
\end{proof}
This proposition can also be viewed as a generalisation of the principle, usually called `the rule of consequence':
from a valid triple $\rgvalidr{\rho\hspace{0.1pt}}{R}{\hspace{-0.7pt}P}{\hspace{-0.2pt}p\hspace{-0.2pt}}{Q}{\hspace{-0.7pt}G}$
one obtains valid triples replacing $R$ or $P$ by a stronger relation/condition whereas $G$ or $Q$ -- by a weaker.
Another key result of Proposition~\ref{thm:pcorr-rule} is: 
\begin{corollary}\label{thm:peqv-rg}
  \hspace{-5.7pt}  If $\peqvS{p\hspace{-0.5pt}}{}{\hspace{-0.7pt}q}$ then
  $\rgvalidr{\rho\hspace{-0.5pt}}{R}{\hspace{-0.4pt}P}{\hspace{0.2pt}p\hspace{0.3pt}}{Q}{\hspace{-0.2pt}G}\hspace{-1pt}$ iff
  $\rgvalidr{\rho\hspace{-0.5pt}}{R}{\hspace{-0.3pt}P}{\hspace{0.2pt}q\hspace{0.2pt}}{Q}{\hspace{-0.2pt}G}$.
\end{corollary}
\begin{proof}
With $\op{id}$ in place of $r$, the conditions (1)--(4) of Proposition~\ref{thm:pcorr-rule}
hold by the reflexivity of $\subseteq$. We apply this instance to $\pcorrS{p}{}{q}$ and to $\pcorrS{q}{}{p}$.
\end{proof}
In other words, equivalent programs exhibit the same rely/guarantee properties.
Taking for example the jump-transformation
\[
\peqvS{\ite{\stsp C}{\stsp p\stsp}{\ret{\hspace{-0.2pt} j}\stsp}}{}{\cjump{\negate{C}}{j}{\stsp p\stsp}}
\]
established by Proposition~\ref{thm:cond-norm2}, any property derived within the program logic for
$\ite{\stsp C}{\stsp p\stsp}{\ret{\hspace{-0.2pt} j}\hspace{0.4pt}}$ applies to the counterpart on the \emph{rhs} as well.
\section{Rules for $\skipp\stnsp$ and $\stsp\basic$}\label{Sb:skip-and-basic}
This section commences the syntax-driven part of the rely/guarantee program logic which comprises
the rules whose successive application allows us to turn triples into verification conditions
using backwards reasoning guided by the program syntax.
Note however that the syntax-driven approach
is hardly applicable to programs carrying only little syntactic structure:
programs that deploy while-statements are decisively more amenable to generation of verification conditions
than those which use only jumps.
This framework is therefore basically not aimed at verifying properties of any possible program
but only of those that are provably equivalent to some jump-free program for which the verification conditions are actually generated.
A particular consequence of this indirection is that 
a rule for $\com{cjump}$ as well as code retrieving functions are omitted in the syntax-driven part.
\begin{lemma}\label{thm:skip-rule}
$\rgvalid{R}{\hspace{-1.5pt}P}{\hspace{-0.5pt}\skipp\hspace{-0.5pt}}{P}{\hspace{-1pt}G}.$
\end{lemma}
\begin{proof}
A computation $\sq \in \pcs{\skipp}$ has no program steps and hence satisfies any program condition, 
whereas $\sq \in\hspace{-0.1pt} \outC{\hspace{-0.1pt}P}$ holds because $\sq \in\hspace{-0.2pt} \inC{\hspace{-0.2pt}P}$ is assumed.
\end{proof}
\begin{lemma}\label{thm:basic-rule}
Assume
\begin{enumerate}
\item[\emph{(1)}] $\rimg{R\stsp}{\hspace{0.2pt}P}\stsp \subseteq P$,
\item[\emph{(2)}] $P\subseteq\stsp \{\sigma \hspace{3pt}|\hspace{3.5pt} f\hspace{1.7pt}\sigma \hspace{0.2pt}\in Q\hspace{-0.2pt} \wedge (\sigma, f\hspace{1.7pt}\sigma) \in\hspace{-0.2pt} G \}$.
\end{enumerate}
Then $\;\rgvalid{R}{\hspace{-1.5pt}P}{\basic\hspace{2pt} f}{Q}{\hspace{-1.2pt}G}$.
\end{lemma}
\begin{proof}
Let $\sq\hspace{0.1pt} \in\hspace{0.1pt} \envC{\hspace{-0.4pt}R}\stdsp \cap\stdsp \inC{\hspace{-0.2pt}P} \cap\stdsp \pcs{\basic\hspace{2pt} f}$.
If $\sq$ does not perform any program steps then $\sq\hspace{-0.1pt} \in\hspace{-0.4pt} \outC{\hspace{-0.4pt}Q} \cap \progC{\hspace{-0.4pt}G}$ holds trivially.
Otherwise let $\sq_n\hspace{0.5pt}\pstep\hspace{-2.5pt} \sq_{n+1}$ be the first program step, \ie\stdsp
$\sq_n\hspace{-1pt} = (\basic\hspace{2.2pt} f,\stdsp \sigma_n)$ and $\sq_{n+1}\hspace{-1pt} = (\skipp,\stdsp f\hspace{1.7pt} \sigma_n)$ must hold with some $\sigma_n$
and, moreover, for any $m\hspace{-0.9pt} < \hspace{-0.2pt}n$ we have only environment steps $\stsp\sq_m\hspace{-0.4pt} \estep\hspace{-1.9pt} \sq_{m+1}$.
Thus, by (1) and $\stsp\stateOf{\sq_0}\hspace{-0.1pt} \in\hspace{-0.7pt} P$ we can further induce $\sigma_n \hspace{-0.9pt}\in\hspace{-0.9pt} P$   
so that $f\hspace{1.7pt}\sigma_n\hspace{-0.5pt} \in \hspace{-0.5pt}Q$ and $(\sigma_n,\stsp f\hspace{1.7pt}\sigma_n)\hspace{-0.5pt} \in \hspace{-0.5pt}G\hspace{0.2pt}$
follow by (2).
Using this, $\sq\hspace{-0.5pt} \in\hspace{-0.7pt} \progC{\hspace{-0.3pt}G}$ can be inferred because
$\sq$ does not have any other program steps from $n\hspace{-0.1pt}+\hspace{-0.2pt}1$ on, 
\ie $\suffix{n+1}{\sq}\hspace{-1pt} \in\hspace{-1.2pt} \pcs{\skipp}$, 
whereas  $\sq\hspace{-1.2pt} \in\hspace{-1.5pt} \outC{\hspace{-0.7pt}Q}\stsp$ follows since $\sq_{n+1}$ is the first $\skipp$-configuration on $\sq$.
\end{proof}
The assumption (1) in Proposition~\ref{thm:basic-rule} is usually called 
a \emph{stability}\index{stability} assumption: any environment step retains the precondition $P$, which is a general approach to reasoning about interference freedom.
It is also worth noting that $\skipp$ is the only language constructor that can get away without a stability assumption as Proposition~\ref{thm:skip-rule} has shown.
\section{A rule for the parallel composition operator}\label{Sb:parallel-rule}
The following rule and the proof are, in essence, extensions 
of the rule and the proof given by Stirling~\cite{STIRLING1988347} for the binary operator $p\stnsp \parallel q$. 
\begin{lemma}\label{thm:parallel-rule}
Let $I\hspace{-1pt} = \{1,\ldots,m\}$ with $m\hspace{-0.7pt} >\hspace{-0.5pt} 0$. Furthermore, let $G$ be a reflexive state relation and assume
\begin{enumerate}
\item[\emph{(1)}] $P\hspace{0.5pt} \subseteq\stsp \bigcap\nolimits_{k \in I}\hspace{-1.2pt} P_i $,
\item[\emph{(2)}] $R\hspace{0.5pt}  \subseteq\stsp \bigcap\nolimits_{k \in I}\hspace{-1.2pt} R_i$,
\item[\emph{(3)}] $\rgvalid{R_k}{\hspace{-2pt}P_k}{\hspace{-0.5pt}p_k\hspace{-1pt}}{Q_k}{\bigcap\nolimits_{l \in I\stsp \setminus \{k\}}\hspace{-2pt} R_l\hspace{0.2pt} \cap G}\stsp$ for any $k \in I$,
\item[\emph{(4)}]  $\rimg{R_k\hspace{-0.2pt}}{\hspace{1pt}Q_k} \hspace{-0.1pt}\subseteq\hspace{0.2pt}  Q_k \stsp$ for any $k\hspace{0.4pt} \in\hspace{-0.2pt} I$,
\item[\emph{(5)}] $\bigcap\nolimits_{k \in I}\hspace{-1pt} Q_k \hspace{0.5pt} \subseteq Q$.
\end{enumerate}
Then $\:\rgvalid{R}{\hspace{-1.5pt}P}{\hspace{-4pt}\Parallel{\hspace{-1.2pt} p_1,\ldots,p_m}}{Q}{\hspace{-1.2pt}G}$. 
\end{lemma}
\begin{proof}
  Let $\sq\hspace{-0.5pt} \in\hspace{-0.5pt} \envC{\hspace{-1.5pt}R} \cap\hspace{0.4pt} \inC{\hspace{-1.4pt}P}\hspace{-1pt} \cap\stsp \pcs{\hspace{1.2pt}\Parallel{\hspace{-1.7pt}p_1,\ldots,p_m}}$
and let $\sq\pr$ be the longest prefix of $\sq$ that does not contain a $\skipp$-configuration.
First, note that if $\sq$ does not contain such a configuration at all then $\sq\pr\hspace{-1.5pt} =\stnsp \sq$ holds.
Second, as $\sq\pr$ is not empty we have
$\sq\pr\hspace{-0.5pt} \in\hspace{0.4pt} \envC{\hspace{-1pt}R}\hspace{0.9pt} \cap\hspace{0.7pt} \inC{\hspace{-1pt}P}\hspace{-0.1pt} \cap\stsp \pcs{\hspace{1.2pt}\Parallel{\hspace{-1.2pt}p_1,\ldots,p_m}}$. 
Third, the program part of each configuration on $\sq\pr$ has the form $\Parallel{\hspace{-2.2pt}u_1,\ldots,u_m}$, \eg $\progOf{\sq\pr_0} =\stdsp \Parallel{\hspace{-1.5pt} p_1,\ldots,p_m}$.

For each $k\hspace{-0.5pt}\in\hspace{-1.2pt} I$ we can thus derive a computation $\sq^k\hspace{-1.9pt} \in\hspace{-0.5pt} \pcs{p_k}$ from $\sq\pr$ such that $|\sq^k| = |\sq\pr|$ for all
$k\hspace{0.5pt}\in\hspace{-0.2pt} I$
and moreover the following conditions are met:
\begin{enumerate}
\item[(a)] $\progOf{\sq\pr_i} =\; \Parallel{\!\progOf{\sq^{\scriptscriptstyle 1}_i}, \ldots, \progOf{\sq^m_i}}\stsp$ for any $i < |\sq\pr|$,
\item[(b)] $\stateOf{\sq\pr_i} = \stateOf{\sq^k_i}\stsp$ for any $i < |\sq\pr|$ and $k\hspace{0.1pt} \in\hspace{-0.2pt} I$,
\item[(c)] if $i < |\sq\pr| - 1$ then $\sq\pr_i \stsp\estep\hspace{-2.9pt} \sq\pr_{i+1}\stsp$ iff $\stsp\sq^k_i \stsp\estep\hspace{-2.9pt} \sq^k_{i+1}\stsp$
 holds for all $k\hspace{0.4pt} \in\hspace{-0.2pt} I$.
\end{enumerate}
Note that $\sq^k\hspace{-0.7pt} \in\hspace{0.5pt} \inC{\hspace{-1.1pt} P_k}\stdsp$ follows for any $k\hspace{0.5pt} \in\hspace{-0.1pt} I$ by (1) and (b).

Next, we show by contradiction that $\sq^k \hspace{-1.5pt}\in\hspace{-0.2pt} \envC{\hspace{-1.5pt}R_k}$ also holds for any $k \hspace{-0.2pt}\in\hspace{-0.9pt} I$.
Assuming the opposite, the set of indices $\bigcup_{k\in I}\stnsp M_k$ where
\[
M_k \hspace{2pt}\defeq\hspace{2pt} \{i \hspace{4pt}|\hspace{4pt} i < |\sq\pr| - 1 \stsp\wedge\stsp \sq^k_i \stsp\estep\hspace{-2.9pt} \sq^k_{i+1} \wedge\stsp (\stateOf{sq_i}, \stsp\stateOf{sq_{i+1}})\stsp \notin\hspace{-0.5pt} R_k \}
\]
is not empty. Then let $\mu$ be the least element of $\stsp\bigcup_{k\in I}\stnsp M_k$ and let $\stsp k_\mu\hspace{-2pt} \in\hspace{-1.7pt} I$
be chosen such that $\stsp \mu\hspace{0.4pt}  \in\hspace{-1pt}  M_{k_\mu}$.
Further, there must be some $n\hspace{0.4pt}  \in\hspace{-0.3pt}  I$ with a program step 
$\sq^n_\mu \hspace{1.5pt}\pstep\hspace{-2.5pt} \sq^n_{\mu+1}$:
would $\sq^n_\mu\hspace{1.5pt} \estep\hspace{-2.9pt}  \sq^n_{\mu+1}$ hold for all $n\hspace{-0.2pt}  \in\hspace{-0.5pt}I$
then we could first infer $\sq\pr_\mu \hspace{1.5pt} \estep\hspace{-2.5pt}  \sq\pr_{\mu+1}$ by (c),
then $(\stateOf{sq_\mu}, \stsp\stateOf{sq_{\mu+1}}) \hspace{-1.2pt} \in \hspace{-2.7pt} R$ since $\sq\hspace{-1.2pt} \in\hspace{-1.4pt} \envC{\hspace{-1.7pt} R}$,
and hence $(\stateOf{sq_\mu},\stsp \stateOf{sq_{\mu+1}})\hspace{-0.2pt} \in\hspace{-1.5pt} R_{k_\mu}$ by (2) in contradiction to $\mu\hspace{0.1pt} \in\hspace{-1.2pt}  M_{k_\mu}$.
Combined with the minimality of $\hspace{0.2pt}\mu\hspace{0.2pt}$ 
this means $\prefix{\mu+2}{\sq^n}\hspace{-0.9pt} \in\hspace{-0.15pt} \envC{\hspace{-0.55pt} R_n}\hspace{0.2pt}$.
With $\prefix{\mu+2}{\sq^n}\hspace{-0.2pt} \in\hspace{0.2pt} \inC{\hspace{-0.5pt} P_n}$
we obtain $\stsp\prefix{\mu+2}{\sq^n}\stnsp \in\stnsp \progC{\!(\bigcap\nolimits_{l \in I\stsp \setminus \{n\}} \hspace{-2pt} R_l\hspace{0.2pt}\cap\hspace{0.5pt} G)}$
by (3) and
since $k_\mu\hspace{-0.5pt} \hspace{-0.5pt}\in\hspace{-0.1pt} I\stsp\setminus\{n\}$ holds due to $\stsp\mu\hspace{0.7pt} \notin\hspace{-0.9pt} M_n\stsp$,
for the program step $\sq^n_\mu \hspace{1.5pt}\pstep\hspace{-2.5pt} \sq^n_{\mu+1}$ we infer $(\stateOf{sq_\mu},\stsp \stateOf{sq_{\mu+1}})\hspace{-0.2pt} \in\hspace{-1.2pt} R_{k_\mu}$
which once more contradicts $\mu\hspace{0.7pt} \in\hspace{-0.7pt} M_{k_\mu}$.

The intermediate result is that $\sq^k \hspace{-1.9pt}\in\hspace{-0.5pt} \pcs{p_k}\hspace{-0.5pt} \cap\hspace{-0.2pt} \inC{\hspace{-1.5pt} P_k}\hspace{-0.2pt} \cap\hspace{0.1pt} \envC{\hspace{-1.5pt} R_k}\stsp$
holds for each $k\hspace{0.2pt}\in\hspace{-0.4pt} I$.
Hence 
$\sq^k\hspace{-1.2pt} \in\hspace{-0.1pt} \outC{\hspace{-0.5pt}Q_k}\hspace{-0.7pt}\cap\hspace{-0.5pt} \progC{\hspace{-0.2pt}G}$
can be inferred for all
$k\hspace{0.2pt} \in\hspace{-0.3pt} I$ using (3) with subsequent weakening of the program condition.
This conclusion will be utilised below to establish $\sq \in \outC{\hspace{-0.5pt}Q}\hspace{0.1pt} \cap\hspace{-0.2pt} \progC{\hspace{-0.5pt}G}$.

First, to show $\sq\hspace{-0.7pt} \in\hspace{-1pt} \outC{\hspace{-0.5pt} Q}$ let $n \hspace{-0.2pt}=\hspace{-0.2pt} |\sq\pr|$ and 
suppose there is some $j < |\sq|$ with $\progOf{\sq_j}\hspace{-1pt} = \skipp$.
Then $j \ge n$ must hold because $\sq\pr$ does not contain any $\skipp$-configurations.  
Moreover, the step of $\sq$ at the position $n-1$ must be a program step to
the first $\skipp$-configuration. 
That is, we have
$\sq_{n-1} \hspace{0.2pt}\pstep\hspace{-2.5pt} \sq_n$ with $\sq_{n-1}\hspace{-1.7pt} = (\Parallel{\hspace{-2pt}\skipp, \ldots, \skipp},\stdsp \sigma)$ and
$\sq_n\hspace{-1.7pt} = (\skipp,\stsp \sigma)$ and some state $\sigma$ for which $\sigma\hspace{-1pt} \in\hspace{-1pt} Q\stsp$ has to be established.
Let $k\hspace{-0.2pt} \in\hspace{-1pt} I$ be fixed to this end. Using (a) we can infer $\sq^k_{n-1}\hspace{-1pt} = (\skipp, \stsp\sigma)$ such that from
$\sq^k\stnsp \in \outC{\hspace{0.1pt} Q_k}$ the existence of an index $i_k \le n\hspace{0.4pt} -\hspace{0.4pt} 1$ with $\sq^k_{i_k}\hspace{-1.5pt} = (\skipp,\stsp \sigma_{i_k})$
and $\sigma_{i_k}\hspace{-2.2pt} \in\hspace{-0.2pt} Q_k$ follows.
Furthermore, this property can be extended to $\stateOf{\sq^k_l}\hspace{-0.2pt} \in\hspace{-0.2pt} Q_k$ for all $l\stsp \ge\stsp i_k$
using (4) since there are only environment steps from $i_k$ on.
Altogether, we can conclude 
$\sigma\hspace{-0.9pt} \in\hspace{-0.5pt} \bigcap\nolimits_{k \in I}\hspace{-1.5pt} Q_k$
and hence $\sigma\hspace{-0.1pt} \in\hspace{-0.2pt} Q$ is a consequence of (5).

Second, to show $\sq\hspace{-1.2pt} \in\hspace{-2.1pt} \progC{\hspace{-0.9pt} G}\stdsp$ let $n\hspace{-0.5pt} = \hspace{-0.5pt}|\sq\pr|$ and
suppose $\sq_{i}\hspace{1pt} \pstep\hspace{-2.7pt} \sq_{i+1}$ with $i < |\sq| - 1$.
Note that $i < n$ holds: $i \ge n$ implies $\progOf{\sq_n} = \skipp$ and hence there would be only environment steps from $n$ on.
Further, if $i = n - 1$ then $\sq_{i} = (\Parallel{\!\skipp, \ldots, \skipp}, \stsp\sigma)$ and
$\sq_{i+1}\stnsp = (\skipp,\stsp \sigma)$ must hold with some state $\sigma$
and hence $(\sigma, \stsp\sigma) \in G$ follows by the reflexivity assumption.
Finally, if $i\stsp <\stsp n - 1$ then (c) provides the existence of a program step 
$\sq^k_{i}\stsp \pstep\hspace{-2.7pt} \sq^k_{i+1}$ with $k\hspace{-0.1pt} \in\hspace{-1pt}  I$ so that
$(\stateOf{\sq_i},\stsp \stateOf{\sq_{i+1}}) =\hspace{-0.1pt} (\stateOf{\sq^k_i}, \stsp\stateOf{\sq^k_{i+1}}) \in\hspace{-0.5pt} G$
follows from $\sq^k\hspace{-1.5pt} \in\hspace{-0.7pt} \progC{\hspace{-0.2pt}G}$.
\end{proof}

With $m = 2$, Proposition~\ref{thm:parallel-rule} can be simplified to a more convenient form:
\begin{corollary}\label{thm:parallel-rule2}
If $G$ is a reflexive state relation and
\begin{enumerate}
\item[\emph{(1)}] $P \subseteq\hspace{-0.5pt} P_1 \hspace{-0.5pt} \cap P_2 $,
\item[\emph{(2)}]  $R \hspace{0.5pt} \subseteq\hspace{-0.5pt}  R_1\hspace{-0.79pt} \cap R_2$,
\item[\emph{(3)}]  $\rgvalid{R_1}{\hspace{-1.5pt}P_1}{p_1\hspace{-1pt}}{Q_1}{\hspace{-1.5pt}R_2 \cap\hspace{0.2pt} G}$,
\item[\emph{(4)}]  $\rgvalid{R_2}{\hspace{-1.5pt}P_2}{p_2\hspace{-1pt}}{Q_2}{\hspace{-1.5pt}R_1 \hspace{-0.5pt}\cap\hspace{0.2pt} G}$, 
\item[\emph{(5)}]   $\rimg{R_1\hspace{0.1pt}}{\hspace{0.9pt}Q_1} \subseteq Q_1 $,
\item[\emph{(6)}]   $\rimg{R_2\hspace{0.1pt}}{\hspace{0.9pt}Q_2} \subseteq Q_2 $,
\item[\emph{(7)}]  $Q_1\hspace{-0.5pt} \cap\hspace{0.2pt} Q_2 \subseteq Q$
\end{enumerate}
then $\:\rgvalid{R}{\hspace{-1.5pt}P}{\hspace{-0.5pt}p_1 \hspace{-1.9pt}\parallel\hspace{-0.5pt} p_2\hspace{-0.5pt}}{Q}{\hspace{-1.5pt}G}$.
\end{corollary}
\subsubsection{An example}
The rules in this section enable structured reasoning upon properties of programs with interleaved computations, but are far from being complete
and the renown `parallel increment' example 
\[
\op{parallel\mbox{-}inc} \defeq \sv{x} := \sv{x} + 1 \parallel \sv{x} := \sv{x} + 1
\]
will be used below to substantiate the claim.
As a preparation to that, we first describe how actual programs working on actual states can be represented. 

In what follows, any state $\sigma$ on which a program operates comprises a mapping that sends each state variable\index{state variable} occurring in the program
(such as $\sv{x}$ in $\op{parallel\mbox{-}inc}$ above)
to a value of a fixed type assigned to the variable, \eg $\stsp\sigma\sv{x}$ is an integer in $\op{parallel\mbox{-}inc}$. 
Then an assignment\index{state variable!assignment} of a value, denoted by an appropriately typed logical term $t$, to a state variable $\sv{a}$ amounts to the indivisible step
$\basic(\lambda\sigma. \:\sigma_{[\sv{a}\stsp :=\stsp t]})$ with the updated state $\sigma_{[\sv{a} := t]}$ that sends
$\sv{a}$ to the value of $t$ and any $\sv{b} \neq \sv{a}\stsp$ to $\stsp\sigma\sv{b}$.
Note that if a state variable appears in $t$ then it actually stands for its image under $\sigma$. 
Thus, the assignment $\sv{x} := \sv{x} + 1$ is essentially a shorthand for $\basic(\lambda\sigma. \hspace{2pt}\sigma_{[\sv{x}\stsp :=\stsp \sigma\sv{x} + 1]})$.
Similar conventions apply to state predicates: for instance 
$\sv{a} = 0$ stands for $\{\sigma \hspace{2.5pt}|\hspace{2.9pt} \sigma\sv{a} = 0\}$.
Furthermore, primed variables will be used to capture state relations: \eg $\stdsp\sv{a} = \sv{a}\pr\hspace{-0.2pt} \wedge \sv{b} = \sv{a}\pr\stsp$
is the same as
$\stsp\{(\sigma,\stsp \sigma\pr) \hspace{3.5pt}|\hspace{3.5pt} \sigma\sv{a} = \sigma\pr\sv{a}\stsp \wedge\hspace{0.2pt} \sigma\sv{b} = \sigma\pr\sv{a}\}$.

Continuing with the `parallel increment' example, let $\sigma$ be a state assigning some integer value to the only state variable $\sv{x}$ and consider the triple:
\begin{equation}\label{eq:pinc}
\rgvalid{\bot}{\hspace{-1pt}\sv{x} = 0}{\op{parallel\mbox{-}inc}}{\sv{x} = 2}{\hspace{-1.5pt}\top}
\end{equation}
Since the rely condition $\bot$ confines the scope of reasoning to the actual computations of $\op{parallel\mbox{-}inc}$,
it is yet attainable 
to verify that (\ref{eq:pinc}) is indeed a property of $\op{parallel\mbox{-}inc}$ by plain unfolding Definition~\ref{def:rgval}:
this requires checking $\outC{\!(\sv{x} = 2)}$ only for seven computations that start with the state $\sigma_{[\sv{x} := 0]}$ and perform up to three program steps:
Figure~\ref{fig:p-inc} depicts the relevant transition subgraph where all environment steps are particularly omitted.
\begin{figure}
\begin{diagram}[height=2.5em,width=0.5em,notextflow]
                                                &         & (\op{parallel\mbox{-}inc},\stsp \sigma_{[\sv{x} := 0]}) &   &                                                   \\
                                                 & \ldTo  &                                                              & \rdTo  &                                              \\
\hspace{-2.5cm}(\skipp \parallel \sv{x} := \sv{x} + 1,\stsp \sigma_{[\sv{x} := 1]})  &        &                                                              &   & (\sv{x} := \sv{x} + 1 \parallel \skipp,\stsp \sigma_{[\sv{x} := 1]})    \\
                                                &  \rdTo  &                                                              & \ldTo  &                                              \\
&         & (\skipp \parallel \skipp, \stsp\sigma_{[\sv{x} := 2]}) &     &                                                  \\
&         & \dTo   & & \\
&         & (\skipp,\stsp \sigma_{[\sv{x} := 2]}) &     &                                                  \\
\end{diagram}
\caption{A transition subgraph of $\op{parallel\mbox{-}inc}$.}
\label{fig:p-inc}
\end{figure}

On the other hand, it is impossible to derive the same property immediately by the rule in Corollary~\ref{thm:parallel-rule2}:
assuming the opposite means that we have found some $R_i, P_i, Q_i$ with $i\stsp \in \{1, 2\}$ such that
\begin{enumerate}
\item[(a)] $(\sv{x} = 0)\hspace{0.2pt} \subseteq \hspace{0.2pt} P_1\hspace{-0.5pt} \cap P_2$,
\item[(b)] $\rgvalid{R_1}{\hspace{-1.5pt}P_1}{\sv{x} := \sv{x} + 1}{Q_1}{\hspace{-1.7pt}R_2}$,
\item[(c)] $\rgvalid{R_2}{\hspace{-1.5pt}P_2}{\sv{x} := \sv{x} + 1}{Q_2}{\hspace{-1.7pt}R_1}$,
\item[(d)] $Q_1\hspace{-0.7pt} \cap\hspace{0.2pt} Q_2\stsp \subseteq\stsp (\sv{x} = 2)$.
\end{enumerate}
Let $\sq \defeq (\sv{x}\hspace{-1pt} :=\hspace{-1pt} \sv{x}\hspace{0.2pt} +\hspace{-0.2pt}  1,\stsp \sigma_{[\sv{x} := 0]}) \hspace{1.5pt}\pstep\hspace{-1.5pt} (\skipp, \stsp\hspace{0.2pt}\sigma_{[\sv{x} := 1]})$ be a one-step computation
for which we have $\sq\hspace{-1.2pt} \in\hspace{-0.7pt} \envC{\hspace{-1.2pt}\bot}\hspace{-0.2pt} \cap\hspace{-0.2pt}  \pcs{\sv{x}\stnsp :=\stnsp \sv{x} \hspace{0.2pt} +\hspace{0.2pt}  1}$ by construction.
Furthermore, $\sq\hspace{0.1pt} \in\hspace{0.1pt} \inC{\hspace{-0.9pt}P_1} \hspace{-0.2pt}\cap \hspace{0.5pt} \inC{\hspace{-0.7pt}P_2}$ can be inferred from (a) allowing us to
apply (b) and (c) to $\sq$ in order to obtain $\sigma_{[\sv{x} := 1]} \in\hspace{-0.5pt} Q_1\hspace{-0.5pt}$ and $\sigma_{[\sv{x} := 1]} \in\hspace{-0.7pt} Q_2\hspace{0.5pt}$
such that $\sigma_{[\sv{x} := 1]} \in\hspace{-0.5pt} (\sv{x} = 2)$ follows using (d), which is evidently a contradiction.

This outcome shall however not lead to the hasty conclusion that (\ref{eq:pinc}) would not be derivable by the rules at all, for it is:
we can deploy two auxiliary variables (each local to one of the parallel components) keeping their sum equal to $\sv{x}$.
This allows us to establish a generalised property by an application of the parallel rule. Then deriving the triple (\ref{eq:pinc}) in principle amounts to `discarding' the auxiliaries
by means of the program correspondence rule in the same way as it will be accomplished in the case study Chapter~\ref{S:PM1},
where auxiliary variables are deployed for the same purpose: an application of the rule in Corollary~\ref{thm:parallel-rule2}. 
\section{A rule for the sequential composition operator}
The extra assumption (3) in the following statement will be justified below.
\begin{lemma}\label{thm:seq-rule}
Assume $G$ is a reflexive state relation and
\begin{enumerate}
\item[\emph{(1)}] $\rgvalid{R}{\hspace{-1.5pt}P}{\hspace{-0.5pt}p\hspace{-0.5pt}}{S}{\hspace{-1.5pt}G}$,
\item[\emph{(2)}] $\rgvalid{R}{\hspace{-1.5pt}S}{\hspace{-0.5pt}q\hspace{-0.5pt}}{Q}{\hspace{-1.5pt}G}$,
\item[\emph{(3)}] $q = \skipp$ implies $\rimg{R\hspace{1pt}}{\hspace{0.9pt}Q}\stsp \subseteq Q$.
\end{enumerate}
Then $\:\rgvalid{R}{\hspace{-1.5pt}P}{\hspace{-0.5pt}p;q\hspace{-0.5pt}}{Q}{\hspace{-1.5pt}G}$.
\end{lemma}
\begin{proof}
  Suppose $\sq\hspace{0.1pt} \in\hspace{0.1pt} \envC{\hspace{-1.2pt}R}\hspace{0.4pt} \cap\hspace{0.4pt} \inC{\hspace{-1pt}P}\hspace{-0.5pt} \cap\hspace{0.2pt}  \pcs{p;q}$.
  Further, let $\sq\pr$ be its longest prefix 
that does not reach a configuration of the form $(\skipp ; q,\stsp \sigma)$ and let $n = |\sq\pr|$.
Then for any $i < n$ there is some $u_i$ such that $\progOf{\sq_i}\hspace{-1pt} = \hspace{-0.5pt}u_i;q$ holds.

If $\sq$ does not reach a configuration of the form $(\skipp ; q,\stsp \sigma)$ then we have $\sq\pr\hspace{-0.5pt} = \sq$ and
$\progOf{\sq_i}\hspace{-0.2pt} = u_i;q$ for all $i < |\sq|$. This particularly means that there are no $\skipp$-configurations on $\sq$,
\ie $\sq\hspace{0.1pt} \in\hspace{-0.1pt} \outC{\hspace{0.2pt}Q}$ holds trivially.
We can further define $\sq\mystrut^p_i \hspace{-1pt}\defeq\hspace{-1pt} (u_i,\stsp \stateOf{\sq_i})$ for all $i < |\sq|$
inheriting the respective transitions of $\sq$.
This construction provides $\sq^p\hspace{-0.2pt} \in\hspace{0.1pt} \envC{\hspace{-1.2pt} R}\hspace{0.2pt}  \cap\hspace{0.5pt}  \inC{\hspace{-1.2pt} P}\hspace{-0.9pt} \cap \pcs{p}$ 
such that we can infer $\sq^p\hspace{-0.7pt}\in \hspace{-0.1pt}\outC{\hspace{-0.5pt} S}\hspace{-0.2pt} \cap\hspace{-0.2pt} \progC{\hspace{-0.5pt} G}$ using (1).
Then $\sq\hspace{0.1pt} \in\hspace{-0.5pt} \progC{\hspace{-0.2pt} G}$ follows from
$\sq^p\hspace{-0.9pt} \in\hspace{-0.55pt} \progC{\hspace{-0.2pt}G}$.

Otherwise we have $n < |\sq|$ and $\sq_n\hspace{-1pt} = (\skipp;q, \stsp\sigma_n)$.
In that case we define $\sq\mystrut^p_i \hspace{-0.5pt}\defeq (u_i,\stsp \stateOf{\sq_i})$ for all $i\stsp \le\stsp n$. 
In particular, $|\sq^p| = n + 1$ and $\sq^p_n$ is the first $\skipp$-configuration on $\sq^{p}$.
Since $\sq^p\hspace{-0.9pt} \in\hspace{-0.2pt} \envC{\hspace{-1.2pt} R}\hspace{0.2pt} \cap\hspace{0.2pt}  \inC{\hspace{-1.25pt} P}\hspace{-1pt} \cap\hspace{-0.2pt} \pcs{p}$ holds once more 
by construction, 
we also infer
$\sq^p\hspace{-1.25pt}\in \hspace{-0.2pt}\outC{\hspace{-0.5pt} S}\hspace{-0.2pt} \cap\hspace{-0.2pt} \progC{\hspace{-0.2pt} G}$ using (1).
Note however that in contrast to the preceding case, $\sigma_n \hspace{-0.9pt}\in\hspace{-0.4pt} S$
now follows from $\sq^{p}\hspace{-0.5pt} \in\hspace{0.2pt} \outC{\hspace{-0.5pt}S}$.

Further, if $\sq$ has only environment steps from $n$ on then $\sq\hspace{-0.1pt} \in\hspace{-0.5pt} \outC{\hspace{-0.2pt}Q}\stsp$  once more holds trivially,
whereas $\sq\hspace{0.1pt} \in\hspace{-0.5pt} \progC{\hspace{-0.4pt}G}\stsp$ is a consequence of $\sq^{p}\hspace{-0.9pt} \in\hspace{-0.5pt} \progC{\hspace{-0.2pt} G}$.

Otherwise, suppose $\sq$ makes a program step at $m$ with $n \le m < |\sq| -1$ and let $m_0$ be the least such index,
\ie $\sq_n \estep \hspace{-1.5pt}\ldots\hspace{2.5pt} \estep \hspace{-2.7pt}\sq_{m_0} \pstep \hspace{-2.7pt}\sq_{m_0 + 1}$
with $\sq_{m_0}\hspace{-1.9pt} = (\skipp;q,\stsp \sigma_{m_0})$ and $\sq_{m_0 + 1}\hspace{-1.2pt} = (q,\stsp \sigma_{m_0})$.
For each configuration $\sq_k$ with $n \le k \le m_0$ we thus have $\progOf{\sq_k} = \skipp;q$ and
can derive the computation $\sq^{q_0}\hspace{-1.5pt} = (q, \stsp\sigma_n)\stsp \estep\hspace{-1pt} \ldots\hspace{2pt} \estep\hspace{-1pt} (q, \stsp\sigma_{m_0})$ which is
composable with the suffix $\suffix{m_0+1}{\sq}$. Then let $\sq^q\hspace{-0.5pt} \defeq \sq^{q_0}\hspace{-1pt} \scomp\hspace{0.2pt} \suffix{m_0+1}{\sq}\stdsp$
such that $\sq^q\hspace{-0.7pt} \in \outC{\hspace{-0.1pt} Q} \cap \progC{\hspace{-0.2pt} G}$ follows from (2)
since $\sq^q\hspace{-0.5pt}  \in\hspace{0.1pt} \envC{\hspace{-1.2pt}R}\hspace{0.2pt} \cap \stsp\inC{\hspace{-0.9pt}S}\hspace{0.1pt} \cap \pcs{q}$ is provided by this construction.

Next we will show the output condition $\sq\hspace{-0.3pt} \in\hspace{-0.4pt} \outC{\hspace{-0.2pt} Q}$.
To this end let $j < |\sq|$ be an index with $\progOf{\sq_j}\hspace{-0.2pt} = \skipp$.
Then $j >\stnsp m_0$ must hold, whereas $j = m_0 + 1$ would entail $q = \skipp$.
In case $q = \skipp$ we have $\sq^q\hspace{-0.7pt} \in \pcs{\skipp}$. In combination with  $\sq^q\hspace{-0.5pt} \in \outC{Q}$
this yields $\sigma_n\hspace{-1.1pt} \in\hspace{-0.3pt} Q$. Furthermore, using (3) we can thus induce that $\stateOf{\sq^q_i}\hspace{-0.2pt} \in\stnsp Q$ holds for any $i < |\sq^q|$
and for $\stateOf{\sq^q_{m_0 - n}}$ in particular.
Since $\stateOf{\sq^q_{m_0 - n}} =\stsp \sigma_{m_0}\hspace{-0.7pt} = \stateOf{\sq_{m_0+1}}$ we can conclude this case.
With $q \neq \skipp$ as an additional assumption, $j\stsp >\hspace{0.2pt} m_0\hspace{0.4pt} +\hspace{0.2pt} 1$ follows by the above remarks.
Thus, we have $\progOf{\sq^q_{j - n - 1}} = \progOf{\sq_j} = \skipp\stsp$ such that $\sq\mystrut^q\hspace{-1pt} \in\hspace{-0.2pt} \outC{\hspace{-0.5pt}Q}$ entails
some $i \le j - n - 1$ with $\progOf{\sq^q_i} = \skipp\stsp$ and $\stsp\stateOf{\sq^q_i}\hspace{0.1pt} \in\hspace{-0.4pt} Q$.
If $i \le m_0\hspace{0.1pt} -\hspace{0.1pt} n$ then we get a contradiction in form of $q = \skipp$.
With $i > m_0 - n$ we can first infer $j \ge i + n + 1 > m_0 + 1$, 
and then $ \progOf{\sq_{i+n+1}} = \progOf{\sq^q_i} = \skipp\stsp$ and $\stsp \stateOf{\sq_{i+n+1}} =\stsp \stateOf{\sq^q_i} \in Q$.

Finally, the program condition $\sq\hspace{-1.7pt} \in\hspace{-2.1pt}  \progC{\hspace{-1pt}G}$ follows from $\sq^{p}\hspace{-2.4pt} \in\hspace{-2.2pt} \progC{\hspace{-1pt}G}$ and
$\sq^q\hspace{-0.1pt} \in\hspace{-0.1pt} \progC{\hspace{0.9pt}G}\stdsp$
if we additionally take into account that
the step $\stsp\sq_{m_0}\stnsp \pstep\stnsp \sq_{m_0 + 1}$, omitted by the composition $\sq^{q_0}\hspace{-1.5pt} \scomp\hspace{-0.4pt} \suffix{m_0+1}{\sq}\stsp$,
is covered by $G$ as well since $G$ is assumed to be reflexive.
\end{proof}

That the assumptions (1) and (2) alone are not generally sufficient to remain sound
in presence of the purely auxiliary
yet indispensable
$\skipp$-operator can be substantiated as follows. Suppose Proposition~\ref{thm:seq-rule} would be valid
without
(3) and let $\op{p} \defeq \sv{b} := \sv{a} + 1;\skipp$.
Hence, $\rgvalid{\sv{a} = \sv{a}\pr}{\hspace{-1pt}\sv{a}=0}{\hspace{-1.2pt}\op{p}\hspace{-0.7pt}}{\sv{b}=1}{\hspace{-1.9pt}\top}$ would be a consequence of
\begin{eqnarray*}
  \begin{aligned}
  &   \hspace{-20pt}\rgvalid{\sv{a}\stsp =\stsp \sv{a}\pr}{\hspace{-1pt}\sv{a}\stsp=\stsp 0}{\sv{b} := \sv{a} + 1}{\sv{b}\stsp=\stsp 1}{\hspace{-1.7pt}\top} \mbox{ and } \\
  &   \hspace{-20pt}\rgvalid{\sv{a}\stsp =\stsp \sv{a}\pr}{\hspace{-1pt}\sv{b}\stsp=\stsp 1}{\skipp}{\sv{b}\stsp=\stsp 1}{\hspace{-1.7pt}\top}
  \end{aligned}
\end{eqnarray*}  
where the former follows by
Proposition~\ref{thm:basic-rule} and the latter is just an instance of Proposition~\ref{thm:skip-rule}.
On the other hand, starting with some state $\sigma\hspace{-0.5pt} \in\hspace{-0.5pt} (\sv{a}\hspace{0.2pt} =\hspace{0.2pt} 0)$ 
the three-step computation
\[
\hspace{-0.7cm}(\op{p}, \stsp\sigma)\stsp \pstep (\skipp;\skipp,\stdsp \sigma_{[\sv{b} := 1]})\stsp \estep (\skipp;\skipp,\stdsp \sigma_{[\sv{b} := 0]})\stsp \pstep (\skipp,\stsp\sigma_{[\sv{b} := 0]})
\]
 clearly satisfies $\envC{\!(\sv{a}\stsp =\stsp \sv{a}\pr)} \cap \inC{\!(\sv{a}\stsp =\stsp 0)} \cap \pcs{\op{p}}$
but not $\outC{\!(\sv{b}\stsp=\stsp 1)}$.

The statement of Proposition~\ref{thm:seq-rule} is consequently not applicable to this example: the respective instance of assumption (3) cannot be satisfied since
the postcondition $\sv{b}\hspace{0.2pt} =1$ is not stable under the rely $\sv{a}\stsp =\stsp \sv{a}\pr$. 
\section{A rule for while-statements}
\begin{lemma}\label{thm:while-rule}
Assume $G$ is a reflexive state relation and
\begin{enumerate}
\item[\emph{(1)}] $\rimg{R\hspace{1pt}}{\hspace{0.5pt} P}\stsp \subseteq P$,
\item[\emph{(2)}] $\rgvalid{R}{\hspace{-1.2pt}P\hspace{-0.4pt} \cap\hspace{0.3pt} C}{\hspace{-0.5pt}p\hspace{-0.5pt}}{P}{\hspace{-0.5pt}G}$,
\item[\emph{(3)}] $\rgvalid{R}{\hspace{-1.2pt}P\hspace{-0.5pt} \cap \stnsp\negate{C}}{\hspace{-0.5pt}q\hspace{-0.5pt}}{Q}{\hspace{-1.5pt}G}$.
\end{enumerate}
Then $\:\rgvalid{R}{\hspace{-1.2pt}P}{\while{\stsp C\hspace{-0.2pt}}{p\hspace{0.4pt}}{q}}{Q}{\hspace{-1pt}G}$.
\end{lemma}
\begin{proof}  
  Let $p_{\mv{while}}$ stand for $\while{\stsp C\stnsp}{\hspace{-0.2pt}p}{q}$.
  We have to show
  $\sq\hspace{-0.4pt} \in\hspace{-0.5pt} \outC{\hspace{-0.4pt} Q}\hspace{-0.3pt}\cap\hspace{-0.3pt}\progC{\hspace{-0.5pt} G}\stsp$ 
  for any $\sq\hspace{-0.4pt} \in \hspace{-0.1pt}\pcs{p_{\mv{while}}}\hspace{-0.5pt} \cap \hspace{0.2pt}\envC{\hspace{-1.25pt}R}\hspace{0.3pt} \cap\hspace{0.1pt} \inC{\hspace{-1.25pt} P}$
  and proceed to this end by induction on the length of $\sq$,
which yields the induction hypothesis:
\begin{enumerate}
\item[($\mathcal{H}$)]\hspace{-5pt} \it{$\sq\pr\hspace{-0.3pt} \in\hspace{-0.1pt} \outC{\hspace{-0.2pt} Q}\hspace{-0.1pt} \cap\hspace{-0.3pt} \progC{\hspace{-0.1pt}G}\stsp$
  holds for any $\stsp\sq\pr\hspace{-0.3pt} \in\hspace{0.1pt} \pcs{p_{\mv{while}}}\hspace{-0.4pt} \cap\hspace{0.5pt} \envC{\hspace{-1.4pt}R} \hspace{0.2pt}\cap\hspace{0.3pt} \inC{\hspace{-1.1pt} P}$ with the length less than the length of $\sq$, {\em\ie} \!\!$|\sq\pr| < |\sq|$}.
\end{enumerate}
However, only the following consequence of ($\mathcal{H}$) will be employed below: 
\begin{enumerate}
\item[(4)] \it{if $\sq\pr$\hspace{-1pt} is a proper non-empty suffix of $\sq$ such that $\sq\pr_0\hspace{-0.7pt} =\hspace{-0.2pt} (p_{\mv{while}},\stsp \sigma)$ holds with some $\sigma \in\hspace{-0.7pt} P$ then
  $\sq\pr\hspace{-0.2pt} \in\hspace{-0.1pt}\outC{\hspace{-0.2pt}Q} \cap\hspace{-0.2pt} \progC{\hspace{-0.3pt}G}$}.
\end{enumerate}
If $\sq$ has no program steps then we are done. Otherwise we may assume that the first transition of $\sq$ is a program step, because
would we have $\sq_0\hspace{1pt} \estep\hspace{-2.5pt} \sq_1$ then $\suffix{1}{\sq}$ would be a proper non-empty suffix of $\sq$ such that
$\suffix{1}{\sq}_0\hspace{-0.7pt} = \sq_1\hspace{-0.2pt} = (p_{\mv{while}},\stsp \sigma)$ with some $\sigma \in\hspace{-0.7pt} P$ due to (1)
and hence $\suffix{1}{\sq}\hspace{0.2pt} \in\hspace{0.2pt}\outC{\hspace{-0.2pt}Q} \cap\hspace{-0.2pt} \progC{\hspace{-0.3pt}G}$ follows from (4).

That is, for the remainder of the proof we can assume $\sq_{0}\hspace{1.1pt} \pstep\hspace{-2.2pt} \sq_{1}$ such that $\sq_{0}\hspace{-1.5pt} = (p_{\mv{while}},\stsp \sigma_{0})$ and
$\sq_{1}\hspace{-1pt} = (x,\stsp \sigma_{0})$ hold with some $x$ and $\sigma_{0}\hspace{0.2pt} \in\hspace{-0.75pt} P$.

If $\sigma_{0}\hspace{-0.1pt} \notin\hspace{-0.2pt}  C$ then we have $x\hspace{-0.2pt} = q$ and
$\suffix{1}{\sq}\hspace{-0.1pt} \in\stnsp  \pcs{q}\hspace{-0.1pt} \cap\hspace{0.2pt} \envC{\hspace{-1.5pt}R}\hspace{0.4pt} \cap\hspace{0.2pt} \inC{\hspace{-2.5pt}(\stnsp P\hspace{-0.5pt} \cap\hspace{-1pt} \negate{C})}$.
Hence (3) entails $\suffix{1}{\sq} \in\hspace{-0.2pt} \outC{\hspace{-0.3pt} Q} \cap\hspace{-0.5pt} \progC{\hspace{-0.5pt} G}\stsp$ 
and we can conclude $\sq\hspace{-0.2pt} \in\hspace{-0.7pt} \outC{\hspace{-0.3pt} Q}\hspace{-0.2pt} \cap\hspace{-0.7pt} \progC{\hspace{-0.5pt} G}\stsp$ by the reflexivity of $G$.

If $\sigma_0\hspace{-0.4pt} \in\hspace{-0.4pt} C$ then $\suffix{1}{\sq}\hspace{-0.1pt} \in \pcs{p;\skipp;p_{\mv{while}}}\hspace{-0.5pt} \cap\hspace{0.2pt} \envC{\hspace{-1.1pt}R}\hspace{0.3pt} \cap\hspace{0.1pt}\inC{\hspace{-2pt}(\stnsp P\hspace{-0.75pt} \cap\hspace{0.1pt} C)}$.
Further, let $\sq\pr$ be the longest prefix of the suffix $\suffix{1}{\sq}$ which does not reach 
a configuration of the form $(\skipp ; \skipp; p_{\mv{while}}, \sigma)$ with some state $\sigma$,
and let $n_1 = |\sq\pr|$.
Then for any $i < n_1$ we have some $u_i \neq \skipp$ with $\progOf{\sq_{i +1}} = u_i;\skipp;p_{\mv{while}}$.

If $\sq\pr\hspace{-0.5pt} =\hspace{-0.1pt} \suffix{1}{\sq}$ then
$\sq$ has no $\skipp$-configurations at all and hence $\sq\hspace{-0.2pt} \in\hspace{-0.5pt} \outC{\hspace{-0.4pt}Q}$ holds trivially.
We can further define the computation 
$\sq\mystrut^p_i \defeq (u_i, \stsp\stateOf{\sq_{i + 1}})$ for all $i < n_1$ inheriting the respective transitions of $\sq$.
Then 
$\sq\hspace{-0.7pt} \in\hspace{-1.2pt} \progC{\hspace{-0.59pt}G}$ follows from the reflexivity of $G$ and
$\sq^p\hspace{-1.4pt} \in\hspace{-1.2pt} \progC{\hspace{-0.59pt}G}$ which in turn follows from (2)
since we have $\sq^p\hspace{-0.45pt} \hspace{-1pt}\in\hspace{-0.2pt} \pcs{p}\hspace{-0.5pt}\cap\hspace{0.2pt} \envC{\hspace{-1.25pt} R}\hspace{0.3pt} \cap\hspace{0.2pt} \inC{\!(\stnsp P\hspace{-0.7pt}\cap\hspace{0.3pt} C)}$ by construction.

Otherwise we can assume $n_1\hspace{-1pt} < |\sq| - 1$ and define $\sq\mystrut^p_i \defeq (u_i, \stdsp\stateOf{\sq_{i + 1}})$ now for all $i \le n_1$.
Then $\sq^p\hspace{-1.25pt} \in\hspace{-0.9pt} \outC{\hspace{-1.2pt} P}\hspace{-1.2pt} \cap\hspace{-0.5pt} \progC{\hspace{-0.5pt} G}\stsp$ follows once more from (2)
because $\sq^p \hspace{-1.7pt}\in\hspace{-0.9pt} \pcs{p}\hspace{-0.9pt}\cap\hspace{0.2pt}\envC{\hspace{-1.7pt} R}\hspace{-0.25pt}\cap\hspace{-0.29pt}\inC{\!(\stnsp P\hspace{-1pt}\cap
  \hspace{-0.1pt} C)}$ still holds by construction.
Furthermore, note that $\sq^p_{n_1}$ is the first $\skipp$-configuration on $\sq^p$. Translating these results to $\sq$ we obtain
$\sq_{n_1 + 1}\hspace{-1pt} =\hspace{-1pt} (\skipp;\skipp; p_{\mv{while}},\stsp \sigma_{n_1+ 1})$ with $\sigma_{n_1 + 1}\hspace{-0.2pt} \in\hspace{-1.5pt} P$
and conclude $\stsp\suffix{n_1+ 1}{\sq} \hspace{0.2pt}\in
\pcs{\skipp;\skipp; p_{\mv{while}}}\hspace{0.2pt} \cap\hspace{0.7pt} \envC{\hspace{-1.2pt} R}\hspace{0.4pt} \cap\hspace{0.5pt} \inC{\hspace{-1.2pt} P}$.

In case there are not at least two program steps: from $\skipp;\skipp; p_{\mv{while}}$ to $\skipp; p_{\mv{while}}$ and then from $\skipp; p_{\mv{while}}$ to  $p_{\mv{while}}$, 
the output condition
$\sq\in\hspace{-0.5pt}\outC{\hspace{1.55pt} Q}$ once more holds trivially, whereas 
$\sq\hspace{-0.2pt}\in\hspace{-0.7pt}\progC{\hspace{0.5pt} G}$ follows from $\sq^{p}\hspace{-1.4pt} \in\hspace{-0.7pt} \progC{\hspace{0.1pt} G}$
and the reflexivity of $G$.

Otherwise there must be some index $n_2$ such that $n_1 \hspace{-1.5pt}+ 2 < n_2 < |\sq|$  and $\sq_{n_2}\! = (p_{\mv{while}}, \sigma_{n_2})$ hold.
Using $\sigma_{n_1+ 1}\hspace{-0.2pt} \in\hspace{-1.29pt} P$ and the stability assumption (1) we can further induce $\sigma_{n_2}\hspace{-1.9pt} \in\hspace{-1.5pt} P$, such that
the hypothesis (4) is applicable to $\stsp\suffix{n_2}{\sq}$
yielding $\suffix{n_2}{\sq}\hspace{0.2pt} \in\hspace{-0.1pt} \outC{\hspace{-0.5pt}Q} \cap\hspace{-0.1pt} \progC{\hspace{-0.5pt}G}$.

To establish the output condition $\sq \in\hspace{-0.3pt} \outC{\hspace{0.7pt}Q}\stsp$ in this last branch,
let $j < |\sq|$ be an index with $\progOf{\sq_j} = \skipp$. 
By the above 
follows that $n_2 < j$ must hold,
\ie $\progOf{\suffix{n_2}{\sq}_{j - n_2}}\hspace{-0.5pt} = \skipp$.
Then $\suffix{n_2}{\sq}\hspace{0.1pt} \in\hspace{-0.2pt} \outC{\hspace{0.2pt} Q}$ entails some $i\stsp \le\stsp j\stsp -\stsp  n_2$ such that
$\suffix{n_2}{\sq}_{i} = \sq_{i + n_2} = (\skipp,\stsp \sigma_{i+n_2})$ with some $\sigma_{i+n_2} \hspace{-0.7pt}\in\hspace{-0.5pt} Q$.

Lastly, $\sq\hspace{-0.1pt} \in\hspace{-0.5pt} \progC{\hspace{-0.2pt}G}\stsp$ is a consequence of 
$\sq^{p} \hspace{-1pt} \in\hspace{-0.5pt} \progC{\hspace{-0.1pt}G}$, $\suffix{n_2}{\sq}\hspace{0.1pt} \in\hspace{-1pt} \progC{\hspace{-0.3pt}G}\stsp$ and the reflexivity of $G$. 
\end{proof}
\section{A rule for await-statements}
The main difficulty with triples of the form $\rgvalid{R}{\hspace{-0.7pt}P}{\await{\hspace{1.5pt} C\hspace{-0.5pt}}{p}}{Q}{\hspace{-0.7pt}G}$
is to capture that any state $\sigma$ before entering $p$
is related by $G$ to any state where $p$ might terminate when run in isolation starting with $\sigma$. 
As opposed to the one-step state transformers (\cf Proposition~\ref{thm:basic-rule}), such runs of $p$ can
comprise series of program steps and it would be rather very limiting
to require from each of them to comply with a transitive $G$.   
Therefore in the assumption (2) below, 
$\sigma$ represents a state before entering $p$ and is also invoked
in the postcondition to link it to
terminal states of $p$ by means of the guarantee $G$.
The same technique can be encountered by Nieto~\cite{DBLP:phd/dnb/Nieto02, 10.1007/3-540-36575-3_24}. 
It is also worth noting that deploying the assumption $(P\hspace{-0.7pt}  \cap\hspace{0.2pt} C)\hspace{-0.7pt} \times\hspace{-1pt} Q \subseteq G$
instead is in principle possible but would result in a considerably
weaker rule since $(P\hspace{-0.9pt} \cap \hspace{-0.1pt}C)\hspace{-0.9pt} \times\hspace{-0.9pt} Q$
generally gives only a rough over-approximation
of the actual input/output behaviour of $p$ restricted to the domain $P\hspace{-0.4pt} \cap\hspace{0.25pt} C$.  
\begin{lemma}\label{thm:await-rule}
Assume
\begin{enumerate}
\item[\emph{(1)}] $\rimg{R\hspace{1pt}}{\hspace{0.5pt}P}\hspace{0.5pt} \subseteq\hspace{-0.2pt} P$,
\item[\emph{(2)}] $\rgvalid{\bot}{\hspace{-2pt}P\hspace{-0.4pt} \cap\hspace{0.5pt} C\hspace{-0.3pt} \cap\hspace{-0.1pt} \{\sigma\}}{p}{\{\sigma\pr \;|\; \sigma\pr\hspace{-0.5pt} \in\hspace{-0.2pt} Q \wedge (\sigma,\stsp \sigma\pr) \in G\}}{\hspace{-1.9pt}\top}\:$ for any state $\sigma$.
\end{enumerate}
Then $\:\rgvalid{R}{\hspace{-1.5pt}P}{\await{\hspace{0.5pt}C\hspace{-0.5pt}}{p}}{Q}{\hspace{-1.5pt}G}$.
\end{lemma}
\begin{proof}
Let $\sq\hspace{-0.5pt} \in\hspace{-0.2pt} \envC{\hspace{-1.29pt}R}\hspace{0.1pt} \cap\hspace{0.1pt} \inC{\hspace{-1.29pt} P}\hspace{-0.7pt} \cap\hspace{0.2pt} \pcs{\await{\hspace{1pt} C\hspace{-1pt}}{p}}$.
If $\sq$ does not perform any program steps then we are done.
Next, suppose $\sq_i \hspace{0.5pt}\pstep\hspace{-1.2pt}\sq_{i+1}$ with $i\stsp <\stsp |\sq| - 1$ is the first such step. Then
$\sq_i\hspace{-1pt} = (\await{\hspace{0.9pt}C\hspace{-0.5pt}}{p},\stsp \sigma_i)$ and $\sq_{i+1}\hspace{-0.5pt} = (\skipp,\stsp \sigma_{i+1})$ 
such that $\sigma_i\hspace{-1pt} \in\hspace{-1.5pt} C$ and $(p,\stsp \sigma_{i})\hspace{-0.5pt} \pstepsN{n}\hspace{-1.7pt} (\skipp,\stsp \sigma_{i + 1})$ with some $n\hspace{-0.1pt}  \in\hspace{-1.1pt}  \naturals$ hold.
Further, noting that $\sigma_i\hspace{-1.2pt}  \in\hspace{-1.9pt}  P\stsp$ is induced by $\stateOf{\sq_0}\hspace{-0.7pt} \in\hspace{-2pt}  P$ and the stability assumption (1),
we can derive $\sq^p\hspace{-0.5pt}  \in\hspace{0.1pt}  \envC{\hspace{-1pt} \bot}\hspace{-0.1pt} \cap\hspace{0.4pt} \inC{\!(\hspace{-1pt}P\hspace{-0.4pt}  \cap\hspace{0.2pt} C)} \cap\pcs{p}$
with $|\sq^p| = n+1$, $\stateOf{\sq\mystrut^p_0} = \sigma_{i}$ as the initial and $\stateOf{\sq\mystrut^p_n} = \sigma_{i+1}$ as the final state.

Thus, we infer
$\sq^p \hspace{-0.7pt}\in\hspace{-0.5pt} \outC{\!(\{\sigma\pr \hspace{2.5pt}|\hspace{3.5pt} \sigma\pr\hspace{-0.7pt} \in\hspace{-0.4pt} Q\hspace{0.4pt} \wedge\hspace{0.1pt} (\sigma_{i},\stsp \sigma\pr)\hspace{-0.5pt} \in\hspace{-0.4pt} G\})}$
instantiating $\sigma$ by $\sigma_i$ in (2),
which yields $\sigma_{i+1}\hspace{-0.5pt} \in\hspace{-0.4pt} Q\stsp$ and $(\sigma_{i},\stsp \sigma_{i+1}) \in\hspace{-0.7pt} G$.
Then $\sq\hspace{-0.1pt} \in\hspace{-0.4pt} \outC{\hspace{-0.2pt} Q}$ follows 
as $\sq_{i+1}$ is the first $\skipp$-configuration on $\sq$, and
$\sq \in\hspace{-0.9pt} \progC{\hspace{0.7pt} G}$ -- 
since $\suffix{i+1}{\sq}\hspace{0.2pt} \in \pcs{\skipp}$. 
\end{proof}
\section{A rule for conditional statements}
\begin{lemma}\label{thm:cond-rule}
Assume $G$ is a reflexive state relation and
\begin{enumerate}
\item[\emph{(1)}] $\rimg{R\hspace{1.2pt}}{\hspace{0.5pt}P}\hspace{0.5pt} \subseteq\hspace{-0.5pt} P$,
\item[\emph{(2)}] $\rgvalid{R}{\hspace{-1.5pt}P\hspace{-0.4pt} \cap\hspace{0.4pt} C}{\hspace{-0.5pt}p\hspace{-0.5pt}}{Q}{\hspace{-1pt}G}$,
\item[\emph{(3)}] $\rgvalid{R}{\hspace{-1.5pt}P\hspace{-0.4pt}  \cap\hspace{-0.7pt} \negate{C}}{\hspace{-0.5pt}q\hspace{-0.5pt}}{Q}{\hspace{-1.2pt}G}$.
\end{enumerate}
Then $\:\rgvalid{R}{\hspace{-1.5pt}P}{\ite{\hspace{0.7pt}C\hspace{-0.2pt}}{p}{q}}{Q}{\hspace{-1pt}G}$.
\end{lemma}
\begin{proof}
  Let $\hspace{0.1pt}\sq\hspace{-0.5pt} \in\hspace{-0.2pt} \envC{\hspace{-1.4pt}R}\hspace{0.1pt} \cap\hspace{-0.1pt} \inC{\hspace{-1.2pt} P}\hspace{-1pt}\cap\hspace{0.2pt} \pcs{\ite{\hspace{0.5pt}C\hspace{-1pt}}{p}{q\hspace{-0.5pt}}}$.
  If $\sq$ has no program steps then we are done.
   Otherwise let $\sq_i\stsp \pstep\hspace{-1.5pt} \sq_{i+1}$ with $i < |\sq| - 1$ be the first such step. Hence we have
$\sq_i = (\ite{\hspace{0.9pt}C\hspace{-0.5pt}}{p}{q}, \stsp\sigma)$ and $\sq_{i+1}\hspace{-1pt} = (x,\stsp \sigma)$ with some $x$ and $\sigma$.
Further, note that the stability (1) and $\sq\hspace{-0.7pt} \in\hspace{-0.9pt} \inC{\hspace{-1.4pt} P}$ induce $\sigma\hspace{-0.2pt} \in\hspace{-0.7pt} P$.
If $\sigma\hspace{-0.3pt} \in\hspace{-0.1pt} C$ then $x=p\stdsp$ and
$\stdsp\suffix{i+1}{\sq}\hspace{0.1pt}\in\hspace{-0.1pt} \envC{\hspace{-1.1pt}R}\hspace{0.4pt} \cap\hspace{0.4pt} \inC{\hspace{-1.7pt}(\hspace{-0.5pt} P\hspace{-0.5pt} \cap\hspace{0.4pt} C)} \cap \pcs{p}\stsp$ follow.
Thus, with (2) we can first infer $\stsp\suffix{i+1}{\sq}\in\hspace{-0.1pt} \outC{\hspace{0.9pt}Q}\hspace{1pt} \cap \hspace{0.7pt}\progC{\hspace{0.7pt}G}$
and then  $\sq\hspace{-0.4pt}\in\hspace{-0.2pt} \outC{\hspace{0.75pt}Q} \hspace{1.1pt}\cap\hspace{0.9pt} \progC{\hspace{0.9pt}G}$.
The case $\sigma\hspace{0.1pt} \notin\hspace{-0.1pt} C$ is entirely symmetric.
\end{proof}
\section{A brief summary}
Starting with the definition of the extended Hoare triples, this section
was dedicated to the rely/guarantee program logic 
presenting
the program correspondence rule as well as a rule for each of the language constructors except $\com{cjump}$.
Although a rule for processing conditional jumps could simply be derived via the program equivalence in Proposition~\ref{thm:cond-norm2} combined with Proposition~\ref{thm:cond-rule},
the actual purpose of the program logic is to enable syntax-driven generation of verification conditions to extended Hoare triples with structured (\ie jump-free) programs only,
as pointed out at the beginning of Section~\ref{Sb:skip-and-basic}. Indirect reasoning upon rely/guarantee properties of
the corresponding programs using jumps is complementary backed by Proposition~\ref{thm:cond-norm2} and Proposition~\ref{thm:while-norm2}.
\newcommand{\imp}{\longrightarrow}
 \setcounter{equation}{0}
\chapter{Case Study: Peterson's Mutual Exclusion Algorithm}\label{S:PM1}
This chapter embarks on the first part of the case study
consisting of three parts concerned with structured derivations of the relevant properties of $\op{thread}_0$ and $\op{thread}_1$
that are shown in Figure~\ref{fig:pm} and shall model 
the threads in the algorithm by Peterson~\cite{PETERSON1981115}.
\begin{figure}
\[
\begin{array}{l l l}
\op{thread}_0 \hspace{2.5pt} \mv{cs}_0 &\hspace{-3pt}  \defeq & \sv{flag}_0 :=\hspace{-1pt}  \op{True};\\
             &        & \sv{turn} :=\hspace{-0.5pt}  \op{True};\\
& & \whileS{\stsp\sv{flag}_1 \hspace{-1.7pt}\wedge\hspace{-0.5pt} \sv{turn}\stsp}{\skipp};\\
& & \mv{cs}_0; \\
& & \sv{flag}_0 := \op{False}\\
& & \\
\op{thread}_1 \hspace{2.5pt} \mv{cs}_1 &\hspace{-3pt} \defeq & \sv{flag}_1 :=\hspace{-1pt}  \op{True};\\
             &        & \sv{turn} := \op{False};\\
& & \whileS{\stsp\sv{flag}_0\hspace{-1.2pt} \wedge\hspace{-1.7pt} \neg\sv{turn}\stsp}{\skipp};\\
& & \mv{cs}_1;\\ 
& & \sv{flag}_1 := \op{False}\\
\end{array}
\]
\caption{Two threads in the mutual exclusion algorithm.} 
\label{fig:pm}
\end{figure}
The definitions are respectively parameterised by the logical variables $\mv{cs}_0$ and $\mv{cs}_1$
being thus 
the placeholders for whatever needs to be accomplished within the critical sections
since they can be instantiated by any term of type $\mathcal{L}_\tau$ for any particular type of states $\tau$.
The parallel composition of the threads 
\[
\op{mutex}\hspace{1.5pt} \mv{cs}_0 \hspace{2pt} \mv{cs}_1 \defeq \op{thread}_0 \hspace{1.5pt}  \mv{cs}_0 \parallel  \op{thread}_1 \hspace{1pt}  \mv{cs}_1
\]
consequently models the actual protocol parameterised by the contents 
of the critical sections.
Note that this modelling is not tailored
to transformations to some low-level representations
in particular due to the compound conditions $\sv{flag}_1\hspace{0.15pt} \wedge\hspace{1pt}\sv{turn}$ and
$\sv{flag}_0\hspace{0.5pt} \wedge\hspace{-0.1pt} \neg\sv{turn}$
as none of them could normally be evaluated by means of one indivisible step, \ie without interleaving.
Nonetheless, this form fits better to the 
the main goal of the case study:
to illustrate the applicability of the proposed methods 
without delving into lots of technical, yet eventually important details deferred to Section~\ref{Sb:arg}.

It should be intuitively clear that the behaviours of $\op{mutex}\hspace{1.9pt} \mv{cs}_0 \hspace{2pt} \mv{cs}_1$ and 
the plain uncontrolled composition $\mv{cs}_0 \hspace{-2pt}\parallel\hspace{-1.4pt} \mv{cs}_1$ may substantially differ whenever $\mv{cs}_0$ and $\mv{cs}_1$ access some shared resource for reading and writing.
More specifically, many input/output properties that can be derived for
$\op{mutex}\hspace{2pt} \mv{cs}_0 \stdsp \mv{cs}_1$ from properties of $\mv{cs}_0$ and $\mv{cs}_1$
just basically do not hold for $\mv{cs}_0 \hspace{-2.9pt}\parallel\hspace{-2.9pt} \mv{cs}_1$ precisely because $\op{mutex}\hspace{2pt} \mv{cs}_0 \hspace{2pt} \mv{cs}_1$
enables certain privileged conditions to $\mv{cs}_0$ and $\mv{cs}_1$ regarding interleaving.
The central goal of this section is to encapsulate by means of a single generic rule the program logical steps transforming 
properties of $\mv{cs}_0$ and $\mv{cs}_1$ to properties of $\op{mutex}\hspace{2pt} \mv{cs}_0 \stdsp \mv{cs}_1$. 

First of all, in order to pin down the mentioned privileged conditions, let $\sv{shared}$ be a state variable
representing the resource to which both threads 
may request access for reading and writing.
Further, let $P_0, Q_0$ and $P_1, Q_1$ be some given predicates 
on the values of $\stsp\sv{shared}\stsp$ so that \eg by $P_0 \hspace{2pt}\sv{shared}$
(as explained in Section~\ref{Sb:parallel-rule} this is actually a shorthand for the set of states $\{\sigma \hspace{3pt}|\hspace{3pt} P_0\hspace{2.9pt} \sigma\sv{shared} \}$)
we can address the precondition of the critical section $\mv{cs}_0$
and by $Q_0\hspace{2.9pt}\sv{shared}$ -- its postcondition.
Thus, the properties that the critical sections shall exhibit in presence of interleaving
can initially be sketched by means of the following extended Hoare triples:
\[
\begin{array}{l}
\hspace{-25pt}\rgvalid{\stsp\sv{shared} = \sv{shared}\pr\wedge \ldots}{P_0 \hspace{2.1pt}\sv{shared}}{\mv{cs}_0}{Q_0 \hspace{2.1pt}\sv{shared}}{\ldots}\\
\hspace{-25pt}\rgvalid{\stsp\sv{shared} = \sv{shared}\pr\wedge \ldots}{P_1 \hspace{2.1pt}\sv{shared}}{\mv{cs}_1}{Q_1 \hspace{2.1pt}\sv{shared}}{\ldots}
\end{array}
\]
Among the relies, the equation $\sv{shared}\hspace{-0.7pt}  = \hspace{-0.7pt} \sv{shared}\pr$ in principle makes out the half of the privileged conditions:
each critical section can locally assume that the opposite thread 
does not modify $\sv{shared}$, addressing thus the core effect of mutual exclusion. 
Another half of the privileged conditions will accordingly amount to \emph{the absence} of  
the same equation
among the yet entirely open guarantees:
its presence there would 
simply preclude any modifications of $\sv{shared}$ 
within $\mv{cs}_0$ and $\mv{cs}_1$,
confining the scope of reasoning to instances of $\mv{cs}_0$ and $\mv{cs}_1$
for which $\mv{cs}_0 \!\parallel\! \mv{cs}_1$ would
actually be a more efficient alternative to $\op{mutex}\hspace{1.5pt} \mv{cs}_0 \stdsp \mv{cs}_1$.

Examining $\op{mutex}$ in more detail reveals that the condition
$\sv{flag}_1 \hspace{-0.5pt}\imp\hspace{-0.5pt} \neg\sv{turn}$ (\ie $\neg\sv{flag}_1\hspace{-1.9pt} \vee\hspace{-1pt} \neg\sv{turn}$)
is sufficient for $\op{thread}_0$  to enter its critical section.
However, the condition is not stable: suppose an evaluation of $\op{mutex}$ is
in a state $\sigma$ with $\neg \sigma\sv{flag}_1\hspace{-0.7pt} \wedge\hspace{0.2pt} \sigma\sv{turn}$ because $\op{thread}_1$ did not set $\sv{flag}_1$ yet
such that by doing so the evaluation runs into a state $\sigma\pr$ with $\sigma\pr\sv{flag}_1\hspace{-1.7pt} \wedge\hspace{-0.7pt} \sigma\pr\sv{turn}$.
Without being relevant for the behaviour of $\op{mutex}$,
this instability poses
a considerable problem from the verification perspective since practically each of the rely/guarantee program logic rules asserts stability in some form.
Deploying auxiliary variables often leads to a solution in such situations.
In this case, a weaker condition
\begin{equation}\label{eq:pm0-draft}
  `\op{thread}_1 \mbox{ did not modify } \sv{turn} \mbox{ yet'} \vee\hspace{-0.5pt} \neg\sv{flag}_1\hspace{-1.2pt} \vee\hspace{-1pt} \neg\sv{turn}
\end{equation}
shall be captured by means of auxiliaries. 
Firstly, (\ref{eq:pm0-draft}) is stable: if $\op{thread}_1$ sets $\sv{flag}_1$ then this happens entirely before it modifies $\sv{turn}$ due to the relative
order of these two assignments.
Secondly, (\ref{eq:pm0-draft}) is sufficient as well: $\op{thread}_0$ may safely enter its critical section in particular when $\op{thread}_1$ did not modify $\sv{turn}$ yet,
no matter what the values of $\sv{flag}_1$ and $\sv{turn}$ are.

Reifying this plan, the programs $\op{thread}^{\mathit{aux}}_0$ and $\op{thread}^{\mathit{aux}}_1$
\begin{figure}
\[
\begin{array}{l l l}
\op{thread}^{\mathit{aux}}_0 \hspace{0.5pt} \mv{cs}_0 &\hspace{-3pt} \defeq & \sv{flag}_0 := \hspace{-1.5pt}\op{True};\\
             &        & \hspace{-5pt}\langle\hspace{0.5pt} \sv{turn} :=\stnsp \op{True};\\
             &        &      \sv{turn\_aux}_0 :=\stnsp \op{True}\rangle;\\
& & \whileS{\stsp\sv{flag}_1 \hspace{-2pt}\wedge\hspace{-0.7pt} \sv{turn}\stsp}{\skipp};\\
& & \mv{cs}_0; \\
& & \sv{flag}_0 := \op{False}\\
& & \\
\op{thread}^{\mathit{aux}}_1 \hspace{0.5pt} \mv{cs}_1 &\hspace{-3pt} \defeq & \sv{flag}_1 := \hspace{-1.5pt}\op{True};\\
             &        &\hspace{-5pt}\langle\hspace{0.5pt} \sv{turn} := \op{False};\\
             &        & \sv{turn\_aux}_1 := \stnsp\op{True}\rangle;\\
& & \whileS{\stsp\sv{flag}_0 \hspace{-1.2pt}\wedge\hspace{-1.7pt} \neg\sv{turn}\stsp}{\skipp};\\
& & \mv{cs}_1; \\
& & \sv{flag}_1 := \op{False}\\
\end{array}
\]
\caption{The threads augmented by the auxiliary variables.} 
\label{fig:pm-aux}
\end{figure}
in Figure~\ref{fig:pm-aux} respectively deploy the variables $\sv{turn\_aux}_0$ and $\sv{turn\_aux}_1$ within the two extra created atomic sections
ensuring that \eg $\sv{turn\_aux}_1$ is set by $\op{thread}_1$ without any interruptions following the modification of $\stsp\sv{turn}$.
The essential effect is that we get $\sigma\sv{turn\_aux}_1\hspace{-1pt} \wedge\hspace{-1pt} \neg \sigma\sv{turn}$ where $\sigma$ denotes any state that $\op{thread}^{\mathit{aux}}_1$ reaches
by processing its atomic section.
The auxiliary model of the mutual exclusion algorithm is then accordingly defined by
\[
\op{mutex}^{\mathit{aux}}\hspace{0.5pt} \mv{cs}_0 \hspace{2pt} \mv{cs}_1 \;\defeq\; \op{thread}^{\mathit{aux}}_0 \hspace{0.2pt} \mv{cs}_0\stsp \parallel\stsp  \op{thread}^{\mathit{aux}}_1 \hspace{0.2pt} \mv{cs}_1
\]
It is quite justified to conjecture what will be elaborated later on, namely that the auxiliaries do not affect the behaviour of the actual $\op{mutex}$ model 
whatsoever.
Therefore without imposing any constraints on $\op{mutex}$ itself 
we can also presume that $\neg\sv{turn\_aux}_1$ and $\neg\sv{turn\_aux}_2$ initially hold
so that (\ref{eq:pm0-draft}) becomes $\neg\sv{turn\_aux}_1\hspace{-1pt} \vee\hspace{-1pt} \neg\sv{flag}_1\hspace{-1.2pt} \vee\hspace{-1pt} \neg\sv{turn}$.

Summing up, the states on which $\op{mutex}$ and $\op{mutex}^{\mathit{aux}}$ operate are shaped as follows. 
First, the variables $\sv{flag}_0, \sv{flag}_1$ and $\sv{turn}$ carry the Boolean values 
used for communication between threads. Second, the Boolean values of $\sv{turn\_aux}_0$ and $\sv{turn\_aux}_1$
will serve the verification purposes only, as explained above.
Third, there are $\sv{shared}, \:\sv{local}_0,\: \sv{local}_1$ whose exact type can remain open yet.
The variables $\sv{local}_0$ and $\sv{local}_1$ are respectively laid out for
$\mv{cs}_0$ and $\mv{cs}_1$ to make a local `snapshot' of $\sv{shared}$  
before submitting a modified
value back to $\sv{shared}$.  

To render the subsequent argumentation in a more structured way, let
\[
\vspace{-9pt}
\begin{array}{l c l}
\hspace{-15pt}\sv{cond}_0 &\hspace{-5pt} \defeq\hspace{-3pt} & \sv{turn\_aux}_1\hspace{-0.5pt} \wedge\hspace{0.2pt} \sv{flag}_1 \stdsp\imp\stdsp  \neg\sv{turn}\\
\hspace{-15pt}\sv{cond}_1 & \hspace{-5pt}\defeq\hspace{-3pt} & \sv{turn\_aux}_0\hspace{-0.5pt} \wedge\hspace{0.2pt}  \sv{flag}_0 \stdsp\imp\stdsp  \sv{turn}
\end{array}
\vspace{-1pt}
\]
and consider the extended Hoare triple
\begin{equation}\label{eq:pm1} 
\hspace{-25pt}\rgvalid{\op{R}_0}{\hspace{-1.5pt}\neg\sv{turn\_aux}_0 \hspace{-0.5pt}\wedge\hspace{-0.5pt} P_0\hspace{2.5pt}\sv{shared}}{\op{thread}^{\mathit{aux}}_0 \hspace{0.1pt} \mv{cs}_0}
        {Q_0\hspace{2.5pt}  \sv{shared}}{\hspace{-2pt}\op{R}_1}
\end{equation}
where $\op{R}_0$ denotes the rely condition
\[
\begin{array}{l r}
 \sv{flag}_0\hspace{-1pt} =\sv{flag}\pr_0\hspace{0.2pt} \wedge\hspace{0.2pt} \sv{turn\_aux}_0\hspace{-0.5pt} = \sv{turn\_aux}\pr_0\hspace{0.2pt} \wedge\hspace{0.2pt} \sv{local}_0\hspace{-0.5pt} = \sv{local}\pr_0\; \wedge  & \hspace{1.41cm}\mathrm{(a)} \\
 (\sv{turn\_aux}_0 \wedge\hspace{-0.5pt} \sv{flag}_0 \wedge \hspace{-0.5pt}\sv{cond}_0 \stdsp\imp\stdsp \sv{shared} = \sv{shared}\pr)\; \wedge  & \mathrm{(b)}\\
 (\sv{cond}_0\stdsp \imp\stdsp \sv{cond}\pr_0)\; \wedge  & \mathrm{(c)}\\
          (P_0\hspace{2.9pt}\sv{shared}\stdsp \imp\stdsp P_0\hspace{2.9pt}  \sv{shared}\pr) \wedge (Q_0\hspace{2.9pt}\sv{shared}\stdsp \imp\stdsp Q_0\hspace{2.9pt}  \sv{shared}\pr)  & \mathrm{(d)}
\end{array}
\]
whose interpretation is given line-by-line below:
\begin{enumerate}
\item[(a)] no step of $\op{thread}^{\mathit{aux}}_1$
  ever modifies $\sv{flag}_0$, $\sv{turn\_aux}_0$ and $\sv{local}_0$;
\item[(b)] the value of $\sv{shared}$ is retained 
           by any program step of $\op{thread}^{\mathit{aux}}_1$ from any state $\sigma$ where $\op{thread}^{\mathit{aux}}_0$ is requesting the exclusive access to $\sv{shared}$
  (\ie the condition $\sigma\sv{turn\_aux}_0\hspace{-1.1pt} \wedge\hspace{-1pt} \sigma\sv{flag}_0$ holds)
           and the request may be regarded as granted (\ie $\sigma\sv{cond}_0\hspace{-1pt}$ holds as well);
           note that this in particular applies when $\op{thread}^{\mathit{aux}}_0$ is inside its critical section;  
\item[(c)] $\sv{cond}_0$ is stable;
\item[(d)] each step of $\op{thread}^{\mathit{aux}}_1$ retains the pre- and the postcondition of $\mv{cs}_0$;
  the extent, to which actions of $\mv{cs}_1$ on the shared resource
  must be constrained in order to protect what $\mv{cs}_0$ achieves, varies for each particular instance of $P_0$ and $Q_0$:
  the threads have to cooperate in this sense to eventually reach the intersection of $Q_0\hspace{2.7pt}  \sv{shared}$ and $Q_1\hspace{2pt} \sv{shared}$
  together. 
\end{enumerate}
As one might expect, the guarantee $\op{R}_1$ in (\ref{eq:pm1}) is entirely symmetric to $\op{R}_0$:
\[
\begin{array}{l}
\hspace{-29pt}  \sv{flag}_1\hspace{-1.2pt} =\sv{flag}\pr_1 \wedge\hspace{0.5pt} \sv{turn\_aux}_1\hspace{-1.1pt} = \sv{turn\_aux}\pr_1 \wedge\hspace{0.5pt} \sv{local}_1\hspace{-0.5pt} = \sv{local}\pr_1\; \wedge\\
\hspace{-29pt}  (\sv{turn\_aux}_1\hspace{-1pt} \wedge \sv{flag}_1\hspace{-1pt} \wedge \sv{cond}_1\stdsp \imp\stdsp \sv{shared} = \sv{shared}\pr)\; \wedge \\
\hspace{-29pt}  (\sv{cond}_1\stdsp \imp\stdsp \sv{cond}\pr_1)\; \wedge \\
\hspace{-29pt}  (P_1\hspace{2.7pt}  \sv{shared}\stdsp \imp\stdsp P_1\hspace{2.2pt}  \sv{shared}\pr) \wedge (Q_1\hspace{2.7pt}  \sv{shared}\stdsp \imp\stdsp Q_1\hspace{2.2pt}  \sv{shared}\pr) 
\end{array}
\]
so that $\op{R}_0$ and $\op{R}_1$ just get swapped in the triple for $\op{thread}^{\mathit{aux}}_1$:
\[
\hspace{-22.5pt}\rgvalid{\op{R}_1}{\hspace{-2pt}\neg\sv{turn\_aux}_1\hspace{-0.5pt} \wedge\hspace{-0.5pt} P_1\hspace{1.7pt} \sv{shared}\stsp}{\op{thread}^{\mathit{aux}}_1 \hspace{0.1pt} \mv{cs}_1}
        {Q_1\hspace{1.7pt}  \sv{shared}}{\hspace{-2pt}\op{R}_0}
        \]
Regarding the critical sections, the extended Hoare triple
\begin{equation}\label{eq:pm3} 
\hspace{-1.2pt}  \rgvalid{\hspace{-0.5pt}\sv{shared}\hspace{-1pt} = \hspace{-1pt}\sv{shared}\pr\hspace{1pt} \wedge\hspace{1.9pt} \sv{local}_0\hspace{-1.5pt} =\hspace{-1pt} \sv{local}\pr_0}{P_0 \hspace{2.4pt}\sv{shared}}{\hspace{1.5pt}\mv{cs}_0\hspace{1.5pt}}{Q_0 \hspace{2.4pt}\sv{shared}}
          {\op{G}_0}\hspace{-1pt}
\end{equation}
where $\op{G}_0$ denotes the guarantee 
\[
\begin{array}{l r}
\sv{flag}_0\hspace{-1pt} =\sv{flag}\pr_0 \wedge\stsp \sv{flag}_1\hspace{-1.5pt} =\sv{flag}\pr_1 \wedge\stsp \sv{turn} =\sv{turn}\pr \wedge\stsp \sv{local}_1 \hspace{-1pt}= \sv{local}\pr_1\stnsp \;\wedge  & \hspace{21pt}\mbox{(a)}\\
 \sv{turn\_aux}_0\hspace{-0.5pt} = \sv{turn\_aux}\pr_0 \stsp\wedge\stsp \sv{turn\_aux}_1\hspace{-1pt} = \sv{turn\_aux}\pr_1 \;\wedge & \mbox{(b)}\\ 
(P_1\hspace{2.2pt} \sv{shared}\stdsp \imp\stdsp P_1\hspace{2.2pt} \sv{shared}\pr) \wedge (Q_1\hspace{2pt} \sv{shared}\stdsp \imp\stdsp Q_1\hspace{2pt} \sv{shared}\pr) & \mbox{(c)}
\end{array}
\]
captures the properties, parameterised by $P_0$, $Q_0$, $P_1$ and $Q_1$, that $\mv{cs}_0$ shall exhibit.
It will thus be assumed when deriving $(\ref{eq:pm1})$ in the next section.  
Note once more the presence of $\sv{shared} = \sv{shared}\pr$ among the rely conditions and its absence in $\op{G}_0$
which nonetheless requires from the program steps of $\mv{cs}_0$
\begin{enumerate}
\item[(a)] not to modify $\sv{flag}_0, \sv{flag}_1, \sv{turn}$ as these are reserved for communication between the two threads
           as well as the variable $\sv{local}_1$ which is `owned' by $\op{thread}^{\mathit{aux}}_1$;
\item[(b)] not to modify the auxiliary variables;
\item[(c)] to retain the conditions $P_1\hspace{2pt} \sv{shared}$ and $Q_1\hspace{1.7pt} \sv{shared}$ in line with $\op{R}_1$.
\end{enumerate}
Regarding $\mv{cs}_1$ we accordingly assume
\[
\hspace{-25pt}  \rgvalid{\sv{shared}\hspace{-1pt} = \hspace{-1pt}\sv{shared}\pr \hspace{-1pt}\wedge \sv{local}_1\hspace{-1.7pt} = \hspace{-1.2pt}\sv{local}\pr_1}{P_1 \hspace{2pt}\sv{shared}}{\mv{cs}_1}{Q_1 \hspace{2pt}\sv{shared}}
          {\hspace{-2pt}\op{G}_1}
          \]
where $\op{G}_1$ denotes                    
\[
\begin{array}{l}
\hspace{-29pt}\sv{flag}_0\hspace{-1pt} =\sv{flag}\pr_0 \wedge \sv{flag}_1\hspace{-1.2pt} =\sv{flag}\pr_1 \wedge \sv{turn} =\sv{turn}\pr\hspace{-0.5pt} \wedge \sv{local}_0 = \sv{local}\pr_0 \;\wedge   \\
\hspace{-29pt}\sv{turn\_aux}_0 = \sv{turn\_aux}\pr_0\stsp \wedge\hspace{0.7pt} \sv{turn\_aux}_1 \stnsp= \sv{turn\_aux}\pr_1\hspace{2.1pt}\wedge \\ 
\hspace{-29pt}(P_0\hspace{2.5pt} \sv{shared}\stdsp \imp\stdsp P_0\hspace{2.5pt}  \sv{shared}\pr) \wedge (Q_0\hspace{2.2pt}  \sv{shared}\stdsp \imp\stdsp Q_0\hspace{2.2pt}  \sv{shared}\pr) 
\end{array}
\]
\section{Processing $\op{thread}^{\mathit{aux}}_0$}\label{S:PM1:process}
The next goal is to establish (\ref{eq:pm1}) applying the rules of the rely/guarantee program logic.
For the first part of $\op{thread}^{\mathit{aux}}_0$ up to the critical section, \ie
\[
\begin{array}{l l l}
\hspace{-19pt}\op{thread}^{\mv{pre}}_0 \hspace{-0.2cm}&\; \defeq\;  &\sv{flag}_0 :=\hspace{-2pt} \op{True};\\
             &        & \hspace{-5pt}\langle\hspace{0.5pt} \sv{turn} :=\hspace{-2pt}  \op{True};\\
             &        & \sv{turn\_aux}_0 := \hspace{-2pt} \op{True}\rangle;\\
& & \whileS{\stsp\sv{flag}_1 \hspace{-1pt} \wedge\hspace{-0.5pt}  \sv{turn}}{\skipp}
\end{array}
\]
we can derive the property
\[
\begin{array}{l}
\hspace{-29pt}\rgvalid{\op{R}_0}{\neg\sv{turn\_aux}_0 \wedge\hspace{-1.5pt}  P_0\hspace{2.5pt}  \sv{shared}}{\\ \hspace{-9pt}\op{thread}^{\mv{pre}}_0\\ \hspace{-20pt}}
        {\hspace{2pt}\sv{flag}_0 \wedge\hspace{-0.5pt} \sv{turn\_aux}_0 \wedge\hspace{-1.5pt} P_0\hspace{2.5pt}  \sv{shared}\hspace{0.2pt} \wedge\hspace{-0.5pt} (\sv{flag}_1 \imp \neg\sv{turn})}{\hspace{-2pt}\op{R}_1}
        \end{array}
        \]
using $\stsp\sv{flag}_0\hspace{-0.5pt} \wedge\hspace{-0.5pt} \neg\sv{turn\_aux}_0 \wedge\hspace{-0.2pt} P_0\hspace{2.5pt}  \sv{shared}\stsp$ and $\stsp\sv{flag}_0 \hspace{-0.1pt}\wedge \sv{turn\_aux}_0 \wedge\hspace{-0.5pt} P_0\hspace{2.5pt}  \sv{shared}\stsp$
as the intermediate conditions which are both stable under $R_0$. 
Thus, it suffices to derive the triple
\begin{equation}\label{eq:pm4}
 \begin{array}{l}
\hspace{-29pt} \rgvalid{\op{R}_0}{\sv{flag}_0 \hspace{-0.5pt}\wedge\hspace{-0.2pt} \sv{turn\_aux}_0 \hspace{-0.2pt}\wedge\hspace{-1pt}  P_0\hspace{3pt} \sv{shared} \wedge\hspace{-0.5pt} (\sv{flag}_1 \imp \neg\sv{turn})}{\\ \hspace{-9pt}\mv{cs}_0 \\ \hspace{-20pt} }
         {\hspace{2pt}Q_0\hspace{3pt} \sv{shared}}{\hspace{-2pt}\op{R}_1}
 \end{array}
\end{equation}
to conclude (\ref{eq:pm1}).
To this end we will appeal to the assumption (\ref{eq:pm3}) 
despite the seeming discord: 
(\ref{eq:pm3}) intentionally relied on the privileged condition $\sv{shared}\hspace{-2.1pt} =\hspace{-2pt} \sv{shared}\pr$
whereas $\op{R}_0$ provides the same condition only under the premise $\sv{turn\_aux}_0 \hspace{-1pt}\wedge\hspace{-0.5pt}\sv{flag}_0 \hspace{-0.9pt}\wedge\hspace{-0.2pt} \sv{cond}_0$.
In other words, the rely in the assumed triple (\ref{eq:pm3}) is stronger than the rely in the anticipated conclusion (\ref{eq:pm4}).
This however makes merely the plain weakening, \ie the program correspondence rule (Proposition~\ref{thm:pcorr-rule}) with $\op{id}$ in place of $r$, impossible
but we apply the rule instantiating $r$ by a state relation that answers to the mentioned premise, namely
\[
\begin{array}{l}
\hspace{-10pt} \op{r}_0 \;\defeq\;  \{\stsp(\sigma,\stsp \sigma\pr) \hspace{3.5pt}|\hspace{3.7pt} \sigma = \sigma\pr\hspace{-1pt} \wedge\hspace{-0.5pt} \sigma\sv{turn\_aux}_0 \hspace{-0.5pt}\wedge\hspace{-0.5pt} \sigma\sv{flag}_0 \hspace{-0.5pt}\wedge\hspace{-0.5pt} \sigma\sv{cond}_0 \stsp\}.
\end{array}
\]
The proof obligations arising by an application of this instance of the program correspondence rule to (\ref{eq:pm3}) and (\ref{eq:pm4}) are:
\begin{enumerate}
\item[-] $\pcorrC{\mv{cs}_0}{\op{r}_0}{\hspace{-1.75pt}\mv{cs}_0}$ which gets passed on to the $\op{mutex}$ rule in form of an assumption, 
  noting that
  it does not additionally constrain the set of admissible $\mv{cs}_0$ instances
  as it merely requires from each step of
  $\mv{cs}_0$ to retain the condition $\sv{turn\_aux}_0\hspace{-0.1pt} \wedge\hspace{0.1pt} \sv{flag}_0\hspace{0.1pt} \wedge\hspace{0.1pt} \sv{cond}_0$
  and this was anyway already captured by the guarantee $\op{G}_0$;\vspace{3pt}
\item[-] $\rcomp{\op{r}_0\hspace{1pt}}{\hspace{0.5pt}\op{R}_0}\hspace{0.5pt} \subseteq \hspace{0.1pt}\rcomp{(\sv{shared}\hspace{-0.7pt} =\hspace{-0.7pt} \sv{shared}\pr\hspace{-0.5pt} \wedge\hspace{-0.5pt} \sv{local}_0 \hspace{-1.5pt}= \hspace{-1pt}\sv{local}\pr_0)\hspace{0.9pt}}{\hspace{1.1pt}\op{r}_0}\stsp$ which
  holds 
  since with $\op{r}_0$ on the \emph{lhs} we can derive $\sv{shared}\hspace{-0.5pt} =\hspace{-0.2pt} \sv{shared}\pr$ from $\op{R}_0$
  and, moreover, 
  $\sv{cond}_0 \imp\stsp \sv{cond}\pr_0$ is provided by $\op{R}_0$ in order to establish $\op{r}_0$ on the \emph{rhs};\vspace{3pt}
\item[-] $\sv{flag}_0 \hspace{-0.2pt}\wedge\hspace{0.5pt} \sv{turn\_aux}_0 \hspace{0.5pt}\wedge\hspace{-0.5pt} P_0\hspace{2.7pt} \sv{shared}\hspace{0.9pt} \wedge (\sv{flag}_1\hspace{-0.7pt} \imp\hspace{-0.2pt} \neg\sv{turn})\hspace{0.7pt} \subseteq\hspace{1pt} \rimg{\op{r}_0\hspace{0.5pt}}{\hspace{1.5pt}(\hspace{-1pt}P_0\hspace{2.7pt}\sv{shared})}$
which holds since $\sv{flag}_1 \hspace{-1pt}\imp \neg\sv{turn}\stsp$ entails $\sv{cond}_0$; \vspace{3pt}
\item[-]  $\rimg{\op{r}_0\hspace{0.7pt}}{\hspace{1.2pt}(Q_0 \hspace{2.5pt}\sv{shared})} \hspace{0.1pt}\subseteq\hspace{0.5pt} Q_0 \hspace{2pt}\sv{shared}\stsp$ which follows
  straight from $\op{r}_0 \hspace{-0.7pt}\subset\hspace{-0.2pt} \op{id}$;\vspace{3pt}
\item[-]  $\rcomp{\rcomp{\op{r}_0\hspace{2.5pt}}{\hspace{2.45pt}\op{G}_0}\hspace{2.1pt}}{\hspace{2.7pt}\op{r}_0}\hspace{0.5pt} \subseteq\hspace{-0.1pt}\op{R}_1$ which holds by 
  the decisive property that the conditions $\sv{turn\_aux}_0\hspace{-0.1pt} \wedge\hspace{0.2pt} \sv{flag}_0\hspace{-0.2pt} \wedge \hspace{-0.2pt}\sv{cond}_0$ and
  $\sv{turn\_aux}_1\hspace{-0.9pt} \wedge\hspace{0.2pt} \sv{flag}_1\hspace{-0.2pt} \wedge\hspace{-0.2pt} \sv{cond}_1$ are mutually exclusive and hence
  in particular entail $\sv{shared} = \sv{shared}\pr$.
\end{enumerate} 
\section{The intermediate result}
What we achieved so far is to derive the extended Hoare triple
\[
\hspace{-21pt}\rgvalid{\op{R}_0}{\hspace{-2pt}\neg\sv{turn\_aux}_0 \wedge\hspace{-0.5pt} P_0\hspace{2.5pt} \sv{shared}\stsp}{\op{thread}^{\mathit{aux}}_0 \hspace{0.1pt} \mv{cs}_0}
        {\stsp Q_0\hspace{2.5pt} \sv{shared}}{\hspace{-2pt}\op{R}_1}
        \]
 by successive application of the program logical rules from Section~\ref{S:prog-log}  assuming  
\[
 \hspace{-21pt} \rgvalid{\sv{shared} = \sv{shared}\pr\hspace{-0.7pt} \wedge \sv{local}_0 = \sv{local}\pr_0}{\hspace{-1.5pt}P_0 \hspace{2.5pt} \sv{shared}\stsp}{\hspace{-1pt}\mv{cs}_0\hspace{-1pt}}{\stsp Q_0 \hspace{2.5pt} \sv{shared}}
          {\hspace{-1pt}\op{G}_0}
\]
and
$
 \pcorrC{\stnsp\mv{cs}_0}{\op{r}_0}{\hspace{-0.07cm}\mv{cs}_0}.
$
In an entirely symmetric way we further derive
\[
\hspace{-21pt}\rgvalid{\op{R}_1}{\hspace{-2pt}\neg\sv{turn\_aux}_1 \wedge\hspace{-1.5pt} P_1\hspace{2.1pt} \sv{shared}\stsp}{\op{thread}^{\mathit{aux}}_1 \hspace{1pt} \mv{cs}_1}
        {\stsp Q_1\hspace{2.1pt} \sv{shared}}{\hspace{-2pt}\op{R}_0}
\]
from the assumptions 
\[
\hspace{-21pt} \rgvalid{\sv{shared} = \sv{shared}\pr\hspace{-1.5pt} \wedge\hspace{-0.5pt} \sv{local}_1 \hspace{-1.5pt}= \sv{local}\pr_1}{\hspace{-2pt}P_1 \hspace{1.7pt}\sv{shared}\stsp}
       {\hspace{-1.5pt}\mv{cs}_1\hspace{-1.5pt}}{\stsp Q_1 \hspace{1.7pt}\sv{shared}}
          {\hspace{-2pt}\op{G}_1}
\]
and
$
 \pcorrC{\stnsp\mv{cs}_1\hspace{-1pt}}{\op{r}_1}{\hspace{-3pt}\mv{cs}_1}
$
 with $\op{r}_1 \hspace{-0.5pt}\defeq\hspace{-0.5pt} \{(\sigma, \sigma\pr) \hspace{2pt}|\hspace{2pt} \sigma\hspace{-1.5pt} =\stnsp \sigma\pr\hspace{-0.7pt}  \wedge \hspace{0.5pt} \sigma\sv{turn\_aux}_1 \hspace{-0.5pt}\wedge \hspace{0.1pt}\sigma\sv{flag}_1\hspace{-0.5pt} \wedge\hspace{0.1pt} \sigma\sv{cond}_1 \}$. 

 Regarding the auxiliary model $\op{mutex}^{\mathit{aux}}$, \ie $\op{thread}^{\mathit{aux}}_0  \mv{cs}_0 \hspace{-0.5pt}\parallel\hspace{-0.5pt}  \op{thread}^{\mathit{aux}}_1 \mv{cs}_1$,
 we can thus conclude using the parallel composition rule (Corollary~\ref{thm:parallel-rule2}):
 \vspace{4pt}
\begin{lemma}[The auxiliary $\op{mutex}$ rule] \label{thm:mutex-aux}
  Assume\\[7pt]
  \begin{tabular}{l l}
    \hspace{-10pt}  \emph{(1)} &\hspace{-12pt}  $\rgvalid{\stsp\sv{shared}\hspace{-1pt} = \sv{shared}\pr\hspace{-2pt} \wedge \hspace{-1pt} \sv{local}_0\hspace{-1pt} = \hspace{-0.5pt} \sv{local}\pr_0}{\hspace{-2.5pt}P_0 \hspace{2.2pt}\sv{shared}\stsp}{\hspace{-1pt} \mv{cs}_0\hspace{-1pt} }{Q_0 \hspace{2.2pt}\sv{shared}}{\hspace{-2pt}\op{G}_0}$\\[2pt]
    \hspace{-10pt}  \emph{(2)} &\hspace{-12pt}  $\rgvalid{\stsp\sv{shared}\hspace{-1pt} = \sv{shared}\pr\hspace{-2pt} \wedge \hspace{-1pt} \sv{local}_1\hspace{-1pt} = \hspace{-0.5pt} \sv{local}\pr_1}{\hspace{-2.5pt}P_1 \hspace{2pt}\sv{shared}\stsp}{\hspace{-1pt} \mv{cs}_1\hspace{-1pt} }{Q_1 \hspace{2pt}\sv{shared}}{\hspace{-2pt}\op{G}_1}$\\[2pt]
   \hspace{-10pt} \emph{(3)} &\hspace{-12pt}  $\pcorrC{\mv{cs}_0}{\op{r}_0}{\hspace{-2.2pt}\mv{cs}_0}$\\[2pt]
   \hspace{-10pt} \emph{(4)} &\hspace{-12pt}  $\pcorrC{\mv{cs}_1}{\op{r}_1}{\hspace{-2.2pt}\mv{cs}_1}$.
  \end{tabular}\\[7pt]  
Then
\vspace{-21pt}
\[ \hspace{-2.9pt}
\rgvalidap{\op{id}}{\hspace{-0.5pt}P_0\hspace{2.1pt} \sv{shared} \wedge\hspace{-1pt} P_1\hspace{2pt} \sv{shared} \wedge \neg\sv{turn\_aux}_0 \wedge \neg\sv{turn\_aux}_1}
{\quad\op{thread}^{\mathit{aux}}_0 \hspace{0.1pt} \mv{cs}_0 \hspace{1pt}\parallel\hspace{1pt}  \op{thread}^{\mathit{aux}}_1 \hspace{0.1pt} \mv{cs}_1}
{ Q_0\hspace{2.5pt} \sv{shared} \wedge\hspace{-0.5pt} Q_1\hspace{2pt} \sv{shared}}{\hspace{-1.5pt}\top} 
\]
\end{lemma}
\section{A rule for $\op{mutex}$} 
Utilising this result, 
we once more resort to the program correspondence rule,
now instantiating $r$ by the state relation 
\[
\begin{array}{l c l l}
 \hspace{-10pt}  \op{r}_{\op{eqv}} & \hspace{-4pt}\defeq & \hspace{-4pt} \{ (\sigma,\stsp \sigma\pr) \hspace{4pt}| \hspace{-7pt}&  \sigma\sv{turn} = \sigma\pr\sv{turn}\hspace{0.2pt} \wedge
  \hspace{0.2pt}\sigma\sv{shared} = \sigma\pr\sv{shared} \hspace{2.5pt}\wedge \\ 
             &        &                         &  (\forall i\hspace{0.1pt} \in\hspace{-1pt} \{0, 1\}. \hspace{2.5pt}\sigma\sv{flag}_i\hspace{-1pt} = \sigma\pr\sv{flag}_i \wedge \sigma\sv{local}_i\hspace{-0.5pt} = \sigma\pr\sv{local}_i ) \}\\
\end{array}
\]
\ie the equivalence on states up to the values of the auxiliary variables.
Thus, a correspondence $\pcorrC{p\hspace{0.5pt}}{\op{r}_{\op{eqv}}}{\hspace{-1pt}q}$ allows us to transfer from $p$ to $q$ those state properties
that do not depend on the auxiliaries. For instance,  $Q\hspace{2.5pt}\sv{shared}$ is clearly such a property for any $Q$.
Hence, $Q(\sigma\sv{shared})$ naturally entails $Q(\sigma\pr\sv{shared})$
provided $(\sigma,\stsp \sigma\pr) \in\hspace{-0.5pt} \op{r}_{\op{eqv}}$.
\begin{lemma}[The $\op{mutex}$ rule]\label{thm:mutex}
    Assume\\[7pt]
  \begin{tabular}{l l}
    \hspace{-10pt}  \emph{(1)} &\hspace{-12pt}  $\rgvalid{\stsp\sv{shared}\hspace{-1pt} = \sv{shared}\pr\hspace{-2pt} \wedge \hspace{-1pt} \sv{local}_0\hspace{-1pt} = \hspace{-0.5pt} \sv{local}\pr_0}{\hspace{-2.5pt}P_0 \hspace{2.2pt}\sv{shared}\stsp}{\hspace{-1pt} \mv{cs}_0\hspace{-1pt} }{Q_0 \hspace{2.2pt}\sv{shared}}{\hspace{-2pt}\op{G}_0}$\\[2pt]
    \hspace{-10pt}  \emph{(2)} &\hspace{-12pt}  $\rgvalid{\stsp\sv{shared}\hspace{-1pt} = \sv{shared}\pr\hspace{-2pt} \wedge \hspace{-1pt} \sv{local}_1\hspace{-1pt} = \hspace{-0.5pt} \sv{local}\pr_1}{\hspace{-2.5pt}P_1 \hspace{1.7pt}\sv{shared}\stsp}{\hspace{-1pt} \mv{cs}_1\hspace{-1pt} }{Q_1 \hspace{1.9pt}\sv{shared}}{\hspace{-2pt}\op{G}_1}$\\[2pt]
   \hspace{-10pt} \emph{(3)} &\hspace{-12pt}  $\pcorrC{\mv{cs}_0}{\op{r}_0}{\hspace{-2.2pt}\mv{cs}_0}$\\[2pt]
   \hspace{-10pt} \emph{(4)} &\hspace{-12pt}  $\pcorrC{\mv{cs}_1}{\op{r}_1}{\hspace{-2.2pt}\mv{cs}_1}$\\[2pt]
   \hspace{-10pt} \emph{(5)} &\hspace{-12pt} $\pcorrC{\mv{cs}_0}{\op{r}_{\op{eqv}}}{\!\mv{cs}_0}$\\[2pt]
     \hspace{-10pt} \emph{(6)} &\hspace{-12pt} $\pcorrC{\mv{cs}_1}{\op{r}_{\op{eqv}}}{\!\mv{cs}_1}$.
  \end{tabular}\\[9pt]  
  Then $\stsp\rgvalid{\op{id}}{\hspace{-1.2pt}P_0\hspace{1.9pt} \sv{shared}\hspace{0.55pt} \wedge\hspace{-0.2pt} P_1\hspace{1.2pt} \sv{shared}}
  {\op{mutex}\hspace{1.2pt}\mv{cs}_0 \hspace{1.2pt}\mv{cs}_1}
  { Q_0\hspace{1.7pt} \sv{shared}\hspace{0.5pt} \wedge\hspace{0.2pt} Q_1\hspace{0.7pt} \sv{shared}}{\hspace{-1.5pt}\top}$. 
\end{lemma}
\begin{proof}
Firstly, the triple
\[\hspace{-20pt}\rgvalida{\op{id}}{\hspace{-1.5pt} P_0\hspace{2.5pt} \sv{shared} \wedge\hspace{-1.5pt}  P_1\hspace{2.1pt} \sv{shared} \wedge \hspace{-1.5pt} \neg\sv{turn\_aux}_0 \wedge \hspace{-1.5pt} \neg\sv{turn\_aux}_1}
{\quad\op{thread}^{\mathit{aux}}_0  \mv{cs}_0 \parallel  \op{thread}^{\mathit{aux}}_1  \mv{cs}_1}
{ Q_0\hspace{2.5pt} \sv{shared} \wedge\hspace{-1pt} Q_1\hspace{2.1pt} \sv{shared}}{\hspace{-1pt}\top}
\]
follows by Proposition~\ref{thm:mutex-aux} using (1) -- (4).
Then we can apply Proposition~\ref{thm:pcorr-rule} based on the correspondence
\[
\hspace{-20pt}\pcorrC{\op{thread}^{\mathit{aux}}_0 \mv{cs}_0 \parallel  \op{thread}^{\mathit{aux}}_1 \mv{cs}_1\;}{\op{r}_{\op{eqv}}}{\op{thread}_0 \hspace{2pt}\mv{cs}_0 \parallel  \op{thread}_1 \hspace{1pt} \mv{cs}_1}
\]
which follows by (5) and (6) since 
\[
\hspace{-20pt}\pcorrC{\langle \sv{turn} :=\hspace{-2pt} \op{True};\hspace{1pt} \sv{turn\_aux}_0 :=\hspace{-2pt} \op{True}\rangle\;}{\op{r}_{\op{eqv}}}{\sv{turn} :=\hspace{-2pt} \op{True}}
\]
and
\[
\hspace{-20pt}\pcorrC{\langle \sv{turn} :=\hspace{-2pt} \op{False};\hspace{1pt} \sv{turn\_aux}_1 :=\hspace{-2pt} \op{True}\rangle\;}{\op{r}_{\op{eqv}}}{\sv{turn} := \hspace{-2pt}\op{False}}
\]
can be derived using Proposition~\ref{thm:await-corrL}.
\end{proof}

To sum up, the assumptions (1) and (2) in the above rule are essential, addressing in particular the input/output properties of the critical sections.
As opposed to that, (3)--(6) get routinely 
discarded when the rule is applied to instances of $\mv{cs}_0$ and $\mv{cs}_1$
that sensibly avoid any form of access to the protocol and the auxiliary variables.
\section{An application of the $\op{mutex}$ rule}\label{sub:mutex-inst}
As motivated at the beginning of the case study, 
by Proposition~\ref{thm:mutex} we can draw sound conclusions regarding the behaviour of the entire protocol for each appropriate instance of the critical sections
as well as of the pre- and postconditions.
In order to illustrate that, the critical sections 
will be instantiated by programs, each adding an element to a shared set.

Recall that the type of $\sv{shared}, \sv{local}_0$ and $\sv{local}_1$ was left open as it is basically immaterial
to the functioning of the mutual exclusion protocol.
Let this type be  $\ty{int} \Rightarrow \ty{bool}$, \ie\stsp set of integers, in this example
and the respective tasks of the critical setions -- to insert $0$ and $1$ to $\sv{shared}$ using $\sv{local}_0$ and $\sv{local}_1$.
Thus, we instantiate $\mv{cs}_0 \mapsto \op{update}_0$ and $\mv{cs}_1 \mapsto \op{update}_1$ where
\[
\begin{array}{l c l}
\op{update}_0 & \hspace{-5pt}\defeq &\hspace{-1pt} \sv{local}_0 \hspace{5.99pt} := \sv{shared}; \\
                      &        & \sv{local}_0 \hspace{5.3pt} := \{0\} \cup \sv{local}_0; \\
                      &        &  \sv{shared}  := \sv{local}_0\\
                      &        & \\
\op{update}_1 & \hspace{-5pt}\defeq &\hspace{-1pt} \sv{local}_1\hspace{5.99pt}  := \sv{shared}; \\
                      &        & \sv{local}_1\hspace{5.3pt}  := \{1\} \cup \sv{local}_1; \\
                      &        &  \sv{shared} := \sv{local}_1\\
\end{array}
\]
and show that upon termination of the mutex protocol, $\sv{shared}$ always contains both, $0$ and $1$, \ie
\[
\hspace{-21pt}\rgvalid{\op{id}}{\!\top}
{\op{mutex} \hspace{2pt}\op{update}_0 \hspace{2pt} \op{update}_1}
{ 0\hspace{0.5pt} \in\hspace{0.7pt}\sv{shared}\stsp \wedge\hspace{-0.5pt} 1\hspace{0.2pt}\in\hspace{0.7pt}\sv{shared}}{\!\top}.
\]
The goal matches the conclusion of Proposition~\ref{thm:mutex} if we set $P_0$ and $P_1$ to $\top$,
$Q_0$ to $0\hspace{0.5pt} \in\hspace{0.2pt} \sv{shared}$ and $Q_1$ to $1\hspace{0.2pt} \in \hspace{0.1pt}\sv{shared}$.
Hence, to match its premises we first need to establish 
\[  
\hspace{-21pt}\rgvalid{\sv{shared} = \sv{shared}\pr\hspace{-1pt} \wedge \sv{local}_0 = \sv{local}\pr_0}{\hspace{-1.5pt}\top}{\op{update}_0 }{0\hspace{0.5pt} \in \sv{shared}}{\hspace{-1.5pt}\op{G}\pr_0}
\]
where $\op{G}\pr_0$ is consequently the following instance of $\op{G}_0$:
\[
\begin{array}{l}
\hspace{-25pt}\sv{flag}_0\hspace{-0.7pt} =\sv{flag}\pr_0 \wedge \sv{flag}_1 \hspace{-0.5pt}=\sv{flag}\pr_1 \wedge \sv{turn} =\sv{turn}\pr\hspace{2pt} \wedge  \\
\hspace{-25pt}\sv{local}_1 \hspace{-0.9pt}= \sv{local}\pr_1 \wedge \sv{turn\_aux}_0\hspace{-0.5pt} = \sv{turn\_aux}\pr_0 \wedge \sv{turn\_aux}_1\hspace{-0.5pt} = \sv{turn\_aux}\pr_1\hspace{2pt} \wedge \\
\hspace{-25pt}(1\in \sv{shared}\stdsp \imp\stdsp 1 \in \sv{shared}\pr). 
\end{array}
\]
The most intricate part constitutes the condition $1\hspace{-0.2pt}  \in\hspace{0.1pt}  \sv{shared}\stsp \imp\stsp 1\hspace{-0.2pt}  \in  \sv{shared}\pr$
which $\op{update}_0$ can surely guarantee. 
As mentioned earlier, its purpose is to ensure that $\op{update}_0$ does not compromise what
$\op{update}_1$ is trying to accomplish, \ie\stsp does not remove $1$ from $\sv{shared}$ in this instance.

Similar observations can be made regarding the triple for $\op{update}_1$:
\[  
\hspace{-22pt}\rgvalid{\sv{shared} = \sv{shared}\pr\hspace{-1pt} \wedge \sv{local}_1 = \sv{local}\pr_1}{\stnsp\top}{\op{update}_1 }{1 \in \sv{shared}}{\hspace{-1pt}\op{G}\pr_1}
\]
where $\op{G}\pr_1$ is the respective instance of $\op{G}_1$.
The remaining proof obligations 
\[
\begin{array}{l}
 \hspace{-27pt}\pcorrC{\op{update}_0}{\op{r}_0}{\op{update}_0}\\
 \hspace{-27pt}\pcorrC{\op{update}_1}{\op{r}_1}{\op{update}_1}\\
 \hspace{-27pt}\pcorrC{\op{update}_0}{\op{r}_{\op{eqv}}}{\op{update}_0}\\
 \hspace{-27pt}\pcorrC{\op{update}_1}{\op{r}_{\op{eqv}}}{\op{update}_1}
\end{array}
\]
can routinely be discarded by applying the syntax-driven correspondence rules in Section~\ref{Sb:corr-rules} since both,
$\op{update}_0$ and $\op{update}_1$, avoid any form of access to the state variables  
$\sv{flag}_0$, $\sv{flag}_1$, $\sv{turn}$, $\sv{turn\_aux}_0$ and $\sv{turn\_aux}_1$.
\setcounter{equation}{0}
\chapter{State Relations as Pre- and Postconditions}\label{S:gener}
A closer look at Definition~\ref{def:rgval} and the conclusion of the example application of the $\op{mutex}$ rule in the preceding chapter, \ie 
\[
\hspace{-15pt}\rgvalid{\op{id}}{\hspace{-2pt}\top}
{\op{mutex} \hspace{2pt}\op{update}_0 \hspace{2pt} \op{update}_1}
{\hspace{0.2pt} 0\hspace{0.55pt} \in\hspace{0.5pt}  \sv{shared} \wedge\hspace{-0.5pt} 1\hspace{0.15pt}  \in\hspace{0.7pt}  \sv{shared}}{\hspace{-2pt}\top}
\]
reveals that 
the postcondition 
encircles the admissible values of $\sv{shared}$ in all states where $\op{mutex} \hspace{2pt}\op{update}_0 \hspace{2pt} \op{update}_1$ might terminate
without any accounts to its initial values.
So, for instance, $\sv{shared}$ could initially carry the set $\{2\}$, but
$\{0, 1\}$ and $\{0, 1, 3\}$ and the like would perfectly satisfy the postcondition although
each of the critical sections merely adds its entry to the shared resource.
To come up to that we however would need state relations rather than state predicates:  
\[\hspace{-10pt}\sv{shared}\stsp \subseteq\stsp \sv{shared}\pr\hspace{-0.5pt} \wedge\hspace{0.5pt} 0\hspace{0.5pt}\in\hspace{0.5pt} \sv{shared}\pr \wedge \hspace{-0.5pt}1 \hspace{-0.1pt}\in\hspace{0.5pt} \sv{shared}\pr\]
gives a more accurate postcondition which would
allow us to conclude that neither $\{0, 1\}$ nor $\{0, 1, 3\}$ can be an output with the input $\{2\}$.

The insight is basically not new, and the problem is usually circumvented by fixing initial values of relevant state variables by means of logical `link' variables in preconditions
in order to refer to these in postconditions. That is, in the above example one could fix the initial value of $\sv{shared}$ by adding $x = \sv{shared}$
to the preconditions
and refer to $x$ in the postconditions by adding $x \subseteq \sv{shared}$.
However, van Staden~\cite{DBLP:conf/mpc/Staden15} in particular
addresses this point by means of an elegant approach leading to a systematic integration of the link variables technique with the program logic.
\section{Generalising the extended Hoare triples}
Firstly, taking a single state $\sigma_1$ as a precondition allows us to use a state relation $Q$ as a postcondition
that depends on the input:
$\rgvalidr{\rho}{R}{\hspace{-0.5pt}\{\sigma_1\}}{\hspace{-1pt}p\hspace{-0.5pt}}{\rimg{Q\hspace{0.2pt}}{\hspace{0.2pt}\{\sigma_1\}}}{\hspace{-1pt}G}$.
Note that
resorting to a singleton set
is decisive to this end because the plain approach $\rgvalidr{\rho}{R}{\hspace{-0.7pt}P}{\hspace{-0.7pt}p\hspace{-0.5pt}}{\rimg{Q\hspace{0.2pt}}{\hspace{0.2pt}P}}{\hspace{-0.7pt}G}$
does not achieve the same effect with any $P$ for if $\sigma\pr$ is in the image of $P$ under $Q$,
\ie $\stdsp\sigma\pr\hspace{-1.2pt} \in\hspace{-0.7pt} \rimg{Q\hspace{0.1pt}}{\hspace{0.1pt}P}$,
then there \emph{exists some} $\sigma\hspace{-0.79pt} \in\hspace{-1.59pt} P$
such that $(\sigma, \stsp \sigma\pr)\hspace{-1pt} \in\hspace{-1pt} Q$
which \emph{per se} does not rule out $(\sigma_1, \stsp \sigma\pr)\hspace{-1pt} \notin\hspace{-1pt} Q$ where $\sigma_1\hspace{-0.59pt} \in\hspace{-0.9pt} P$ is the actual input state.

Developing that further, 
$\rgvalidr{\rho\hspace{-0.5pt}}{R}{\hspace{-0.9pt}\{\sigma_1\}}{\hspace{-1pt}p\hspace{-0.7pt}}{\rimg{Q}{\{\sigma_1\}}}{\hspace{-1.5pt}G}$ in conjunction with
$\rgvalidr{\rho\hspace{0.2pt}}{R}{\hspace{-1pt}\{\sigma_2\}}{\hspace{-0.7pt}p\hspace{-0.5pt}}{\rimg{Q\hspace{0.5pt}}{\hspace{0.5pt}\{\sigma_2\}}}{\hspace{-1pt}G}$
express that any
$\sq \in \rpcs{p}{\rho}\hspace{-0.9pt} \cap\stsp \envC{\hspace{-1.5pt}R}\stsp \cap\stsp \inC\!(\{\sigma_1, \sigma_2\})$
satisfies the output condition $\outC{\hspace{-2pt}(\rimg{Q\hspace{0.7pt}}{\hspace{1pt}\{\stateOf{\sq_0}\}})}$.
Expressing the same using the universal quantifier leads to 
$\forall \sigma\hspace{-0.2pt}\in\hspace{-0.5pt} \{\sigma_1, \sigma_2\}.\hspace{2.5pt} \rgvalidr{\rho}{R}{\hspace{-0.5pt}\{\sigma\}}{\hspace{-0.5pt}p\hspace{-0.5pt}}{\rimg{Q\hspace{0.7pt}}{\hspace{0.75pt}\{\sigma\}}}{\hspace{-1pt}G}$.
So with a general parameter $P$ in place of $\{\sigma_1, \sigma_2\}$ the latter formula
captures exactly that if $p$ terminates in a state $\sigma\pr$ for an input state $\sigma \in\hspace{-0.2pt} P$
then $(\sigma, \stsp\sigma\pr) \in Q$.

The only deficiency in specifying the behaviour of a program $(\rho, p)$ by a set of triples
$\rgvalidr{\rho}{R}{\hspace{-0.5pt}\{\sigma\}}{\hspace{-0.9pt}p\hspace{-0.7pt}}{\rimg{Q\hspace{0.7pt}}{\hspace{0.7pt}\{\sigma\}}}{\hspace{-1pt}G}$
generated by $\sigma\hspace{-0.2pt}\in\hspace{-0.79pt} P$
is the blend of a state relation $Q$ as a postcondition with a state predicate $P$ as a precondition,
which creates certain disbalance that poses particular challenges for reasoning about sequential compositions and motivates the following definition.
\begin{definition}\label{def:rgval2}
Let $(\rho, p)$ be a program and $R, P, Q, G$ -- state relations.
Then  $\rgvalidrext{\rho}{R}{\hspace{-1pt}P}{\hspace{-0.7pt}p\hspace{-0.3pt}}{Q}{\hspace{-1pt}G}$ iff
  $\rgvalidr{\rho}{R}{\{\sigma\}}{\hspace{-0.7pt}p\hspace{-0.3pt}}{\rimg{Q\hspace{0.7pt}}{\hspace{0.5pt}\{\sigma_0\}}}{\hspace{-1pt}G}$ holds for all $(\sigma_0, \stsp\sigma) \in\hspace{-1pt} P$.
\end{definition}
Next proposition essentially provides an alternative and slightly more succinct definition:
\begin{lemma}\label{thm:pcondF} 
$\rgvalidrext{\rho}{R}{\hspace{-1pt}P}{\hspace{-0.7pt}p\hspace{-0.7pt}}{Q}{\hspace{-1pt}G}$ iff
$\rgvalidr{\rho}{R}{\hspace{-1pt}\rimg{P\hspace{-0.3pt}}{\hspace{0.5pt}\{\sigma_0\}}}{\hspace{-0.9pt}p\hspace{-0.9pt}}{\rimg{Q\hspace{0.25pt}}{\hspace{0.25pt}\{\sigma_0\}}}{\hspace{-1.2pt}G}$ holds for all states $\sigma_0$.
\end{lemma}
\begin{proof}
  A straight consequence of Definition~\ref{def:rgval2} and Definition~\ref{def:rgval}.
\end{proof}
A particular result is that
any triple
$\rgvalidr{\rho}{R}{\hspace{-0.9pt}P}{\hspace{-0.5pt}p\hspace{-0.5pt}}{Q}{\hspace{-0.9pt}G}$ is equivalent to 
$\rgvalidrext{\rho}{R}{\hspace{-1.2pt}\overline{P}}{\hspace{-0.5pt}p\hspace{-0.5pt}}{\overline{Q}}{\hspace{-1.2pt}G}$
where $\overline{A} \defeq \{(\sigma, \stsp\sigma\pr)\;|\; \sigma\pr\hspace{-0.9pt} \in\hspace{-1.5pt} A\}$
is a mapping that
embeds each state predicate $A$ into state relations.
On the other hand, by far not each triple
$\rgvalidrext{\rho}{R}{\hspace{-1.2pt}P}{\hspace{-0.5pt}p\hspace{-0.5pt}}{Q}{\hspace{-1.2pt}G}$ has an equivalent counterpart with
state predicates as pre- and postconditions: consider for instance 
\[
\rgvalidrext{\rho}{\op{id}}{\hspace{-1pt}\op{id}}{\hspace{-0.5pt}p\hspace{-0.5pt}}{\stsp\sv{shared}\hspace{0.2pt} \subseteq\hspace{0.2pt} \sv{shared}\pr}{\hspace{-1pt}\top}
\]
which encompasses those programs $(\rho, \stsp p)$ that do ultimately retain all elements in the set assigned to $\sv{shared}$ in the given input state
if this leads to termination.

A notable feature of such generalised triples is that each of the program logic rules from Section~\ref{S:prog-log}
can seamlessly be lifted 
by a systematic replacement of relational images with relational compositions.
This will be illustrated below by means of the program correspondence rule and the rule for $\basic$.
\begin{lemma}\label{thm:pcorr-rule2}
Assume $\pcorr{p}{r}{\hspace{-0.5pt}q}$,\; $\rgvalidrext{\rho}{R}{\hspace{-1pt}P}{\hspace{-0.5pt}p\hspace{-0.5pt}}{Q}{\hspace{-1pt}G}$ and 
\begin{enumerate}
\item[\emph{(1)}] $\rcomp{r\hspace{1.5pt}}{\hspace{0.59pt}R\pr}\hspace{0.5pt} \subseteq \hspace{-0.1pt}\rcomp{R\hspace{1.55pt}}{\hspace{1.55pt}r}$,
\item[\emph{(2)}] $P\pr\hspace{0.5pt} \subseteq \rcomp{P\hspace{0.55pt}}{\hspace{1.5pt}r}$,
\item[\emph{(3)}] $\rcomp{Q\hspace{1.4pt}}{\hspace{1.7pt}r}\hspace{0.5pt} \subseteq \stsp Q\pr$,
 \item[\emph{(4)}] $\rcomp{\rcomp{\rconv{r}\hspace{-0.7pt}}{\hspace{1.47pt}G}\hspace{1.4pt}}{\hspace{2pt}r}\hspace{0.9pt} \subseteq\hspace{0.5pt} G\pr$.
\end{enumerate}
Then $\rgvalidrext{\rho\pr}{R\pr}{\hspace{-1pt}P\pr}{\hspace{-0.5pt}q\hspace{-0.5pt}}{Q\pr}{\hspace{-1pt}G\pr}$.
\end{lemma}
\begin{proof}
By Proposition~\ref{thm:pcondF} we show $\rgvalidr{\rho\pr}{R\pr}{\hspace{-1.5pt}\rimg{P\pr\hspace{-1pt}}{\hspace{0.75pt}\{\sigma_0\}}}{\hspace{-1pt}q\hspace{-1pt}}{\rimg{Q\pr\hspace{-0.25pt}}{\hspace{0.7pt}\{\sigma_0\}}}{\hspace{-1.5pt}G\pr}$
with a fixed $\sigma_0$. Instantiating Proposition~\ref{thm:pcorr-rule} by
$\rimg{P\pr\hspace{-1.1pt}}{\hspace{1pt}\{\sigma_0\}}$ for $P\pr$ and $\rimg{Q\pr\hspace{-0.1pt}}{\hspace{1pt}\{\sigma_0\}}$ for $Q\pr$ as well as
$\rimg{P\hspace{0.2pt}}{\hspace{1.5pt}\{\sigma_0\}}$ for $\hspace{-0.5pt}P$ and $\rimg{Q\hspace{1pt}}{\hspace{1pt}\{\sigma_0\}}$ for $\hspace{-0.5pt}Q$,
it remains thus to establish the following three proof obligations: 
\begin{enumerate}
\item[(i)] $\rgvalidr{\rho\hspace{0.5pt}}{R}{\hspace{-0.9pt}\rimg{P\hspace{0.7pt}}{\hspace{1.5pt}\{\sigma_0\}}}{p}{\rimg{Q\hspace{1.5pt}}{\hspace{1.5pt}\{\sigma_0\}}}{\hspace{-0.9pt}G}$
  is a plain consequence of the assumption $\rgvalidrext{\rho\hspace{0.5pt}}{R}{\hspace{-1pt}P}{\hspace{-0.7pt}p\hspace{-1pt}}{Q}{\hspace{-1pt}G}$ by Proposition~\ref{thm:pcondF};
\item[(ii)] $\rimg{P\pr\hspace{-0.45pt}}{\hspace{0.9pt}\{\sigma_0\}} \subseteq\hspace{0.7pt} \rimg{r\hspace{0.5pt}}{\hspace{0.5pt}(\rimg{P\hspace{0.3pt} }{\hspace{0.5pt}\{\sigma_0\}})}$
  which is equivalent to $\rimg{P\pr\hspace{-1pt}}{\hspace{0.7pt}\{\sigma_0\}} \subseteq \rimg{(\rcomp{P\hspace{0.45pt}}{\hspace{1.4pt}r})\hspace{0.5pt}}{\hspace{0.5pt}\{\sigma_0\}}$
  and hence follows from (2);
\item[(iii)] $\rimg{r\hspace{0.9pt}}{\hspace{1pt}(\rimg{Q\hspace{0.7pt}}{\hspace{0.75pt}\{\sigma_0\}})} \subseteq\hspace{0.5pt} \rimg{Q\pr\hspace{-0.2pt}}{\hspace{0.75pt}\{\sigma_0\}}$
  which is equivalent to $\rimg{(\rcomp{Q\hspace{1.55pt}}{\hspace{1.75pt}r})}{\hspace{1pt}\{\sigma_0\}} \subseteq\hspace{0.2pt} \rimg{Q\pr\hspace{-0.2pt}}{\hspace{1pt}\{\sigma_0\}}$
  and hence follows from (3).
 \end{enumerate}
\end{proof}
\begin{lemma}\label{thm:basic-rule2}
Assume
\begin{enumerate}
\item[\emph{(1)}]  $\rcomp{P\hspace{0.9pt}}{\hspace{0.7pt}R}\hspace{1pt} \subseteq\hspace{-0.1pt} P$,
\item[\emph{(2)}] $P\hspace{0.2pt} \subseteq\stsp \{(\sigma,\stsp \sigma\pr) \hspace{3.5pt}|\hspace{3.9pt} (\sigma,\hspace{0.5pt} f\hspace{1.7pt}\sigma\pr) \in\hspace{-0.25pt} Q \wedge (\sigma\pr,\hspace{0.5pt} f\hspace{1.7pt}\sigma\pr) \in\hspace{-0.2pt} G\stsp \}$.
\end{enumerate}
Then $\stsp\rgvalidext{R}{\hspace{-1pt} P}{\basic\hspace{2.2pt} f}{Q}{\hspace{-0.5pt} G}$.
\end{lemma}
\begin{proof}
We have to show $\rgvalid{R}{\rimg{\hspace{-1pt}P\hspace{-0.1pt}}{\hspace{0.7pt}\{\sigma_0\}}}{\basic\hspace{2.2pt} f}{\rimg{Q\hspace{0.4pt}}{\hspace{0.4pt}\{\sigma_0\}}}{\hspace{-1pt}G}$
for any $\sigma_0$ and apply the instance of Proposition~\ref{thm:basic-rule} with $\rimg{P\hspace{0.2pt}}{\hspace{1.2pt}\{\sigma_0\}}$ for $P$ and $\rimg{Q\hspace{0.7pt}}{\hspace{0.9pt}\{\sigma_0\}}$ for $Q\stsp$.
This is sound because 
$\rimg{R\hspace{1pt}}{\hspace{1.1pt}(\rimg{P\hspace{0.45pt}}{\hspace{1.2pt}\{\sigma_0\}})} \subseteq\hspace{-0.1pt} \rimg{P\hspace{0.4pt}}{\hspace{1.2pt}\{\sigma_0\}}$ is a consequence of (1) whereas
$\rimg{P\hspace{0.29pt}}{\hspace{1.1pt}\{\sigma_0\}} \subseteq \{\sigma \hspace{3.5pt}|\hspace{3.9pt} f\hspace{1.7pt}\sigma\hspace{0.5pt} \in\hspace{0.5pt} \rimg{Q\hspace{0.7pt}}{\hspace{0.7pt}\{\sigma_0\}} \wedge (\sigma,\stsp f\hspace{1.7pt}\sigma) \in G\stsp \}$ -- of (2).
\end{proof}\vspace{2pt}

Next section proceeds to the onward topic concerning more structured reasoning upon rely/guarantee properties
by means of abstraction and reuse.
\section{Abstracting code fragments}
The key advantage of the generalised program logic presented in the preceding section is
that for the variety of environments specified by a rely $R$,
the input/output behaviour of a code fragment $p$ can often be so accurately captured 
by means of the pre- and postcondition in a triple $\rgvalidext{R}{\hspace{-1pt}P}{p}{Q}{\hspace{-1pt}G}$
that it is worth to be proved once and then reused when showing rely/guarantee properties of programs where $p$ occurs 
as a subroutine
interleaved with an environment in the specified variety.

The principle is illustrated by a small 
example addressing most of the essential points
that arise with abstraction of code fragments in the context of interleaving.
The triple
\begin{equation}\label{eq:inc}
\hspace{-21pt}\rgvalidext{\sv{x} = \sv{x}\pr}{\sv{x} = \sv{x}\pr}{\op{inc}}{\sv{x}\pr = \sv{x} + 1}{\sv{x} \le \sv{x}\pr}
\end{equation}
displays a relation between the inputs and the outputs of $\op{inc}\defeq \sv{x} := \sv{x} + 1$ 
if the environment condition $\sv{x}\hspace{-0.4pt} = \hspace{-0.2pt}\sv{x}\pr$ is met, \ie the variable $\sv{x}$
can be viewed as local to $\op{inc}$ when it gets deployed by a thread.
With any environment that does not satisfy $\sv{x} = \sv{x}\pr$,
the pre- and postcondition have to be adjusted accordingly.
For example, by weakening the rely to  $\sv{x} \le \sv{x}\pr$,
the precondition would need to be altered to $\sv{x} \le \sv{x}\pr$ to remain stable so that
the most accurate postcondition regarding 
$\sv{x}$ would consequently be $\sv{x} < \sv{x}\pr$.
Further, the guarantee in (\ref{eq:inc}) 
could also be amplified to $\sv{x} + 1\hspace{-0.2pt} = \sv{x}\pr$, which
would be the strongest possible only if the underlying states have no other variables than $\sv{x}$.
That is, if there is another variable, say $\sv{v}$, then 
$\sv{v} = \sv{v}\pr$ shall be guaranteed by $\op{inc}$ as well. 
To some extent this applies to the postcondition too:
in the purely sequential setting one would without much ado augment
the pre- and postcondition in (\ref{eq:inc}) by $\sv{v} = \sv{v}\pr$, 
but in presence of interleaving this would in turn require the additional rely $\sv{v} = \sv{v}\pr$ to ensure stability of thus updated precondition.

Staying within the variety specified by the rely $\sv{x} = \sv{x}\pr$, the triple (\ref{eq:inc}) can for example be reused twice to show 
\[\hspace{-21pt}\rgvalidext{\sv{x} = \sv{x}\pr}{\sv{x} = \sv{x}\pr}{\op{inc};\op{inc}}{\sv{x}\pr = \sv{x} + 2}{\sv{x} \le \sv{x}\pr}\]
making the reasoning more structured and efficient compared to the plain processing of $\sv{x} := \sv{x} + 1; \sv{x} := \sv{x} + 1$, \ie with the definition of $\op{inc}$ unfolded. 
For the second invocation of $\op{inc}$ one is however prompted to show
\begin{equation}\label{eq:inc2}
\hspace{-21pt}\rgvalidext{\sv{x} = \sv{x}\pr}{\sv{x}\pr = \sv{x} + 1}{\op{inc}}{\sv{x}\pr = \sv{x} + 2}{\sv{x} \le \sv{x}\pr}
\end{equation}
which is not an immediate instance of (\ref{eq:inc}) and demands the following general consideration:
\begin{lemma}\label{thm:ldiv}
Assume $P\pr$ is reflexive, $R\stsp \subseteq\hspace{-0.4pt} R\pr$, $G\pr \subseteq\hspace{0.2pt} G$ and 
\begin{enumerate}
\item[\emph{(1)}] $\rgvalidext{R\pr}{\hspace{-1.4pt}P\pr}{\hspace{-0.5pt}p\hspace{-0.5pt}}{Q\pr}{\hspace{-1pt}G\pr}$,
\item[\emph{(2)}] $\rcomp{P\hspace{0.7pt}}{\hspace{1.45pt}Q\pr} \subseteq\hspace{0.2pt} Q$.
\end{enumerate}
Then $\rgvalidext{R}{\hspace{-1pt}P}{\hspace{-0.5pt}p\hspace{-0.5pt}}{Q}{\hspace{-1pt}G}$.
\end{lemma}
\begin{proof}
By Definition~\ref{def:rgval2} we have to show $\rgvalid{R}{\hspace{-0.5pt}\{\sigma\}}{\hspace{-0.4pt}p\hspace{-0.4pt}}{\rimg{Q\hspace{0.79pt}}{\hspace{1pt}\{\sigma_0\}}}{\hspace{-1.5pt}G}$
for some states $\sigma_0$, $\sigma$ with $(\sigma_0,\stsp \sigma)\hspace{-0.2pt} \in\hspace{-1.2pt} P$.
First note that $\rgvalid{\hspace{0.2pt} R\pr}{\hspace{-0.5pt}\{\sigma\}}{p}{\rimg{Q\pr\hspace{0.1pt}}{\hspace{1.2pt}\{\sigma\}}}{\hspace{-1pt} G\pr}$
follows from (1) by Definition~\ref{def:rgval2} since $(\sigma,\stsp \sigma)\hspace{-1pt} \in\hspace{-2pt} P\pr$. 
Thus, only $\rimg{Q\pr\hspace{-0.9pt}}{\hspace{0.7pt}\{\sigma\}}\hspace{-0.2pt} \subseteq\hspace{-0.2pt}  \rimg{Q\hspace{0.7pt}}{\hspace{0.9pt}\{\sigma_0\}}$
remains to be established.
Assuming $(\sigma,\stsp \sigma\pr)\hspace{-0.9pt} \in\hspace{-1.2pt} Q\pr$ with some $\sigma\pr$ to this end
we can first infer $(\sigma_0,\stsp \sigma\pr) \in\hspace{-0.2pt} \rcomp{P\hspace{0.55pt}}{\hspace{1.2pt} Q\pr}\stsp$
and then $(\sigma_0,\stsp \sigma\pr) \in\hspace{-0.5pt}  Q\stsp$ using (2).
\end{proof}

Since $\sv{x} = \sv{x}\pr$ is reflexive and 
$
\rcomp{(\sv{x}\pr\hspace{-0.2pt} = \sv{x} + 1)\hspace{1pt}}{\hspace{0.9pt}(\sv{x}\pr\hspace{-0.2pt} = \sv{x} + 1)}\hspace{0.5pt} \subseteq\hspace{0.7pt} \sv{x}\pr\hspace{-0.2pt} = \sv{x} + 2
$ holds,
the triple (\ref{eq:inc2}) is derivable by the above proposition and (\ref{eq:inc}). 
\setcounter{equation}{0}
\chapter{Case Study: Strengthening the Specification}\label{S:PM2}
The $\op{mutex}$ rule in Proposition~\ref{thm:mutex} can directly be lifted to the generalised program logic with state relations as pre- and postconditions
proceeding similarly to Propositions~\ref{thm:pcorr-rule2} and~\ref{thm:basic-rule2}.
\section{Lifting the $\op{mutex}$ rule}
The key difference to Chapter~\ref{S:PM1} 
is that the parameters $P_0, Q_0, P_1, Q_1$ can now depend on two values of $\sv{shared}$ so that
\eg\stsp $Q_0$ can capture a relation between initial values of $\sv{shared}$ and its values when $\mv{cs}_0$ terminates. 
Our goal is thus the generalised conclusion of Proposition~\ref{thm:mutex}: 
\begin{equation} \label{mutexext-concl}\hspace{-21.5pt}
\rgvalidapext{\op{id}}{\hspace{-1.5pt}P_0\hspace{2.5pt} \sv{shared}\hspace{2.9pt} \sv{shared}\pr\hspace{-0.5pt} \wedge\hspace{-1pt} P_1\hspace{1.7pt} \sv{shared}\hspace{3.7pt} \sv{shared}\pr}
{\hspace{5pt}\op{thread}_0 \hspace{2pt} \mv{cs}_0 \parallel \hspace{0.1pt} \op{thread}_1 \hspace{1.5pt} \mv{cs}_1}
{ Q_0\hspace{2.7pt} \sv{shared}\hspace{3.5pt}\sv{shared}\pr\hspace{-0.5pt} \wedge\hspace{-0.5pt} Q_1\hspace{2.1pt} \sv{shared}\hspace{3.5pt} \sv{shared}\pr}{\hspace{-1.5pt}\top} 
\end{equation}
Note once more that with \eg\stsp $Q_0\hspace{2.7pt} \sv{shared}\hspace{3pt} \sv{shared}\pr$ in the above postcondition we respectively
relate the value of $\sv{shared}$ in an initial state $\sigma_0$ and
its value in a state where $\op{mutex}$ may terminate starting from $\sigma_0$. This is in contrast to relies and guarantees
where $\sv{shared}$ and $\sv{shared}\pr$ always refer to the respective values in any two \emph{consecutive} states
on a computation. 

Hence (\ref{mutexext-concl}) follows by Proposition~\ref{thm:pcondF} if we show
\begin{equation} \label{mutexext-concl2}\hspace{-24pt}
\rgvalidap{\op{id}}{\hspace{-1.5pt}(P_0\hspace{2pt} \sigma\sv{shared})\hspace{2.7pt} \sv{shared}\hspace{-0.5pt} \wedge\hspace{-1pt} (P_1\hspace{1pt} \sigma\sv{shared})\hspace{3.5pt} \sv{shared}}
{\hspace{7pt}\op{thread}_0 \hspace{1.7pt} \mv{cs}_0 \parallel \hspace{0.4pt} \op{thread}_1 \hspace{1.2pt} \mv{cs}_1}
{ (Q_0\hspace{1.9pt} \sigma\sv{shared})\hspace{3pt}\sv{shared}\hspace{-0.5pt} \wedge\hspace{-0.5pt} (Q_1\hspace{1.2pt} \sigma\sv{shared})\hspace{3pt} \sv{shared}}{\hspace{-1.5pt}\top} 
\end{equation}
for any state $\sigma$. In this context,
$P_i\hspace{1.7pt} \sigma\sv{shared}$ and $Q_i\hspace{1.7pt} \sigma\sv{shared}$ become fixed state predicates for $i\hspace{-0.7pt} \in\hspace{-1.5pt} \{0,1\}$
pointing to how Proposition~\ref{thm:mutex} shall be
instantiated in order to match (\ref{mutexext-concl2}): $P_0$ by $P_0\hspace{2.1pt} \sigma\sv{shared}$, $Q_0$ by $Q_0\hspace{2.1pt} \sigma\sv{shared}$ and so on.
To match the instantiated first premise of Proposition~\ref{thm:mutex}, namely
\[
\hspace{-25pt}\rgvalidap{\sv{shared}\hspace{-0.2pt} = \sv{shared}\pr\hspace{-0.5pt} \wedge \hspace{-0.2pt} \sv{local}_0\hspace{-1pt} = \hspace{-0.5pt} \sv{local}\pr_0}
          {(P_0\hspace{2.1pt} \sigma\sv{shared})\hspace{2.7pt} \sv{shared}}{\hspace{7pt}\mv{cs}_0}
          {(Q_0\hspace{2pt} \sigma\sv{shared})\hspace{3pt}\sv{shared}}{\stnsp \op{G}\pr_0}
\]          
with $\op{G}\pr_0$ denoting the respective instance of $\op{G}_0$
\[
\begin{array}{l}
\hspace{-25pt}\sv{flag}_0\hspace{-0.7pt} =\sv{flag}\pr_0 \wedge \sv{flag}_1\hspace{-1.5pt} =\sv{flag}\pr_1 \wedge \sv{turn} =\sv{turn}\pr \wedge \sv{local}_1 \hspace{-1pt}= \sv{local}\pr_1 \;\wedge \\
\hspace{-25pt} \sv{turn\_aux}_0\hspace{-0.5pt} = \sv{turn\_aux}\pr_0 \stsp\wedge\stsp \sv{turn\_aux}_1\hspace{-0.9pt} = \sv{turn\_aux}\pr_1 \;\wedge \\ 
\hspace{-25pt} ((P_1\hspace{1.5pt} \sigma\sv{shared})\hspace{2.7pt} \sv{shared}\stsp \imp\stsp (P_1\hspace{1.5pt} \sigma\sv{shared})\hspace{2.7pt} \sv{shared}\pr)\; \wedge \\
\hspace{-25pt} ((Q_1\hspace{1.15pt} \sigma\sv{shared})\hspace{2.5pt} \sv{shared}\stsp \imp\stsp (Q_1\hspace{1pt} \sigma\sv{shared})\hspace{2.5pt} \sv{shared}\pr), 
\end{array}
\]
we consequently make the following assumption
\begin{equation} \label{mutexext-asm1}
\hspace{-25pt}\rgvalidapext{\sv{shared}\hspace{-0.5pt} = \sv{shared}\pr\hspace{-1pt} \wedge \hspace{-0.2pt} \sv{local}_0\hspace{-1pt} = \hspace{-0.5pt} \sv{local}\pr_0}
          {\hspace{-2pt}P_0\hspace{2.9pt} \sv{shared}\hspace{2.9pt} \sv{shared}\pr}{\hspace{7pt}\mv{cs}_0}
          {Q_0\hspace{2.7pt} \sv{shared}\hspace{3.2pt}\sv{shared}\pr}{\stnsp\op{G}\prr_0}
\end{equation}
where 
$\op{G}\prr_0$ shall denote
\[
\begin{array}{l}
\hspace{-25pt}\sv{flag}_0 =\sv{flag}\pr_0 \wedge \sv{flag}_1\hspace{-1.5pt} =\sv{flag}\pr_1 \wedge \sv{turn} =\sv{turn}\pr \wedge \sv{local}_1 \hspace{-1pt}= \sv{local}\pr_1 \;\wedge \\
\hspace{-25pt} \sv{turn\_aux}_0\hspace{-0.5pt} = \sv{turn\_aux}\pr_0 \stsp\wedge\stsp \sv{turn\_aux}_1\hspace{-0.5pt} = \sv{turn\_aux}\pr_1 \;\wedge \\ 
\hspace{-25pt} (\forall v.\hspace{2pt} (P_1\hspace{1.55pt} v)\hspace{3.1pt} \sv{shared}\stsp \imp\stsp (P_1\hspace{1.55pt} v)\hspace{3.1pt} \sv{shared}\pr)\; \wedge \\
\hspace{-25pt} (\forall v.\hspace{2pt} (Q_1\hspace{0.9pt} v)\hspace{2.5pt} \sv{shared}\stsp \imp\stsp (Q_1\hspace{0.9pt} v)\hspace{2.5pt} \sv{shared}\pr). 
\end{array}
\]
In response to the instantiated second premise of Proposition~\ref{thm:mutex} we moreover make the symmetric assumption
\begin{equation} \label{mutexext-asm2}
\hspace{-25pt}\rgvalidapext{\sv{shared}\hspace{-0.5pt} = \sv{shared}\pr\hspace{-2pt} \wedge \hspace{-0.5pt} \sv{local}_1\hspace{-1pt} = \hspace{-0.5pt} \sv{local}\pr_1}
          {\hspace{-2pt}P_1\hspace{1.7pt} \sv{shared}\hspace{3.2pt} \sv{shared}\pr}{\hspace{7pt}\mv{cs}_1}
          {Q_1\hspace{2.7pt} \sv{shared}\hspace{3.2pt}\sv{shared}\pr}{\stnsp \op{G}\prr_1}
\end{equation}
where 
$\op{G}\prr_1$ accordingly denotes
\[
\begin{array}{l}
\hspace{-25pt}\sv{flag}_0 =\sv{flag}\pr_0 \wedge \sv{flag}_1\hspace{-1.5pt} =\sv{flag}\pr_1 \wedge \sv{turn} =\sv{turn}\pr \wedge \sv{local}_0 \hspace{-1pt}= \sv{local}\pr_0 \;\wedge \\
\hspace{-25pt} \sv{turn\_aux}_0\hspace{-0.5pt} = \sv{turn\_aux}\pr_0 \stsp\wedge \sv{turn\_aux}_1\hspace{-0.5pt} = \sv{turn\_aux}\pr_1 \;\wedge \\ 
\hspace{-25pt} (\forall v.\hspace{2pt} (P_0\hspace{2.4pt} v)\hspace{3.1pt} \sv{shared} \imp (P_0\hspace{2.4pt} v)\hspace{3.1pt} \sv{shared}\pr)\; \wedge \\
\hspace{-25pt} (\forall v.\hspace{2pt} (Q_0\hspace{1.7pt} v)\hspace{2.5pt} \sv{shared} \imp (Q_0\hspace{1.7pt} v)\hspace{2.5pt} \sv{shared}\pr). 
\end{array}
\]
The overall result is thus the following lifted rule for $\op{mutex}$: 
\begin{lemma}\label{thm:mutex-ext}
  Assume \emph{(\ref{mutexext-asm1})} and \emph{(\ref{mutexext-asm2})} and
  \begin{enumerate}
  \item[]\hspace{-21.5pt} $\pcorrC{\mv{cs}_0}{\op{r}_0}{\hspace{-2.2pt}\mv{cs}_0}$,
  \item[]\hspace{-21.5pt} $\pcorrC{\mv{cs}_1}{\op{r}_1}{\hspace{-2.2pt}\mv{cs}_1}$,
  \item[]\hspace{-21.5pt} $\pcorrC{\mv{cs}_0}{\op{r}_{\op{eqv}}}{\!\mv{cs}_0}$,
  \item[]\hspace{-17pt}$\pcorrC{\mv{cs}_1}{\op{r}_{\op{eqv}}}{\!\mv{cs}_1}$.
  \end{enumerate}
  Then we have \emph{(\ref{mutexext-concl})}.
\end{lemma}  
\section{An application of the lifted $\op{mutex}$ rule}\label{sub:mutex2-inst}
By contrast to Section~\ref{sub:mutex-inst}, the rule in Proposition~\ref{thm:mutex-ext} allows us to derive
a significantly finer input/output property of $\op{mutex} \hspace{2pt}\op{update}_0 \hspace{2pt} \op{update}_1$.
To this end we instantiate the parameters $P_0, P_1, Q_0, Q_1$ so as to make the triples  
(\ref{mutexext-asm1}) and (\ref{mutexext-asm2}) specify the behaviour of $\op{update}_0$ and $\op{update}_1$ taking into account
  also the input value $\sv{shared}$ in presence of interleaving.

More precisely, due to the inherent uncertainty whether the access to $\sv{shared}$ will be granted to $\op{update}_0$ or to $\op{update}_1$ in first place,
the pre- and postcondition of the below triple account for both cases:
\begin{equation}\label{upd0-ext}
\hspace{-27pt}\begin{array}{l}
\rgvalidext{\hspace{1pt}\sv{shared} = \sv{shared}\pr\hspace{-0.5pt} \wedge \sv{local}_0 = \sv{local}\pr_0}{\\\hspace{20pt}\sv{shared} = \sv{shared}\pr\hspace{-0.5pt} \vee \sv{shared} \cup \{1\} = \sv{shared}\pr}{\\ \quad\quad\; \op{update}_0 \\ \quad\;}
           {\sv{shared} \cup \{0\} = \sv{shared}\pr\hspace{-0.5pt} \vee \sv{shared} \cup \{0, 1\} = \sv{shared}\pr}{\!\op{G}\pr_0}
           \end{array}
\end{equation}
where $\op{G}\pr_0$ becomes now accordingly
\[
\hspace{-25pt}\begin{array}{l}
\sv{flag}_0 =\sv{flag}\pr_0 \wedge \sv{flag}_1 =\sv{flag}\pr_1 \wedge \sv{turn} =\sv{turn}\pr \wedge \sv{local}_1 = \sv{local}\pr_1 \;\wedge \\
 \sv{turn\_aux}_0 = \sv{turn\_aux}\pr_0 \wedge \sv{turn\_aux}_1 = \sv{turn\_aux}\pr_1\; \wedge \\
(\sv{shared} = \sv{shared}\pr \hspace{-0.5pt}\vee \sv{shared} \cup \{0\} = \sv{shared}\pr).
\end{array}
\]
Note that $\op{G}\pr_0$ could have been used in this form in Section~\ref{sub:mutex-inst} as well, but
the weaker guarantee $1\hspace{-0.9pt} \in\hspace{-0.2pt} \sv{shared} \stsp\imp\stsp 1\hspace{-0.7pt} \in\hspace{-0.2pt} \sv{shared}\pr$
was more adequate there
than $\sv{shared} = \sv{shared}\pr\hspace{-0.5pt} \vee \sv{shared} \cup \{0\} = \sv{shared}\pr$.

Since the  assumption (\ref{mutexext-asm2}) is similarly addressed by means of the respective triple for 
$\op{update}_1$, using Proposition~\ref{thm:mutex-ext} we can conclude
\[
\hspace{-27pt}\begin{array}{l}
  \rgvalidext{\op{id}}{\hspace{-1pt}(\sv{shared} = \sv{shared}\pr\hspace{-0.5pt} \vee \sv{shared} \cup \{1\} = \sv{shared}\pr)\hspace{2pt} \wedge\\
   \hspace{33pt} (\sv{shared} = \sv{shared}\pr\hspace{-0.5pt} \vee \sv{shared} \cup \{0\} = \sv{shared}\pr)}
             {\\ \hspace{25pt}\op{thread}_0 \hspace{2pt} \op{update}_0 \hspace{-0.1pt}\parallel \hspace{-0.1pt} \op{thread}_1 \hspace{1.5pt} \op{update}_1 \\\hspace{12.5pt}}
    {(\sv{shared} \cup \{0\} = \sv{shared}\pr\hspace{-0.5pt} \vee \sv{shared} \cup \{0, 1\} = \sv{shared}\pr)\hspace{2pt} \wedge\\
     \hspace{19pt} (\sv{shared} \cup \{1\} = \sv{shared}\pr\hspace{-0.5pt} \vee \sv{shared} \cup \{0, 1\} = \sv{shared}\pr)}{\hspace{-1.5pt}\top}
\end{array}
\]
which can further be simplified to
\begin{equation} \label{mutexext-result}
\hspace{-21pt}\rgvalidext{\op{id}}{\hspace{-1.5pt}\op{id}}
{\op{mutex} \hspace{2pt}\op{update}_0 \hspace{2pt} \op{update}_1}
{\sv{shared} \cup \{0, 1\} = \sv{shared}\pr}{\hspace{-1.5pt}\top}
\end{equation}
because from $\sv{shared}\hspace{-0.4pt} \cup\hspace{-0.9pt} \{0\} \hspace{-0.5pt}= \sv{shared}\pr\hspace{-0.5pt} \wedge\hspace{0.5pt}
\sv{shared}\hspace{-0.2pt} \cup\hspace{-0.7pt} \{1\} \hspace{-0.5pt}= \sv{shared}\pr$
we can infer $\{0, 1\}\hspace{0.2pt} \subseteq\hspace{0.7pt} \sv{shared}$ and hence
also $\sv{shared}\hspace{-0.2pt} \cup\hspace{-0.2pt} \{0, 1\} = \sv{shared}\pr$.

Unfolding the definitions, the triple (\ref{mutexext-result}) yields the following statement:
$\hspace{-0.5pt}\sigma\sv{shared}\hspace{0.5pt} \cup\hspace{0.5pt} \{0, 1\}\hspace{-1pt} = \hspace{-0.7pt}\sigma\pr\sv{shared}$ 
holds for all 
$\sq\hspace{-0.7pt} \in\hspace{-0.7pt} \pcs{\hspace{-0.2pt}\op{mutex} \hspace{1.7pt}\op{update}_0 \hspace{1.7pt} \op{update}_1\hspace{-0.5pt}}\hspace{0.1pt} \cap\hspace{0.29pt}
\envC{\hspace{0.1pt}\op{id}}$
where $\sigma \hspace{-1.7pt}=\hspace{-1pt} \stateOf{\sq_0}$ and $\sigma\pr$ is the state of the first $\skipp$-configuration on $\sq$, if such exists.
That the existence is provided even though not for any but for any \emph{fair} computation
$\sq\hspace{-0.2pt} \in\hspace{-0.5pt} \pcsi{\op{mutex} \hspace{2.1pt}\op{update}_0 \hspace{2.1pt} \op{update}_1}\hspace{-0.79pt} \cap\hspace{0.7pt} \envC{\hspace{-0.1pt}\op{id}}$
will ultimately be shown in the last part of the case study in Chapter~\ref{S:PM3}.
The techniques applied in Chapter~\ref{S:PM3} are derived from the general approach elaborated next. 
\setcounter{equation}{0}
\chapter{Verifying Liveness Properties of Programs}\label{S:live}
Liveness covers a variety of questions that is sometimes rather casually summarised by
`will something good eventually happen?', answering which may demand very involved reasoning 
in many instances.
In the current context we will confine `something good' to
\emph{certain computations of a program starting in a state satisfying an input condition reach a state satisfying some output condition}.
Making this point clear, let $\op{p}_1$ be the program
\[
\while{\stsp\sv{a} > 0\stsp}{\skipp\stsp}{\stsp\sv{b} :=\! \op{True}}
\]
whereas $\op{p}_2$ shall comprise the assignment $\sv{a}\hspace{-0.1pt} :=\hspace{-0.1pt} 0$ only
(note that Park~\cite{Park81} uses a similar introductory example for a similar purpose but in a different context). 
Whether \emph{certain} (which \emph{per se} does not rule out \emph{all}) computations of $\op{p}_1\hspace{-1.5pt} \parallel\hspace{-0.5pt}\op{p}_2$ that
start with a state satisfying $\sv{a}\hspace{-0.2pt} = \hspace{-0.2pt}  1\hspace{-0.5pt} \wedge\hspace{-1.5pt} \neg\sv{b}$
eventually reach a state satisfying $\sv{b}$ is an eligible question regarding liveness of $\op{p}_1\hspace{-2.1pt} \parallel\hspace{-0.5pt}\op{p}_2$.
One such computation evidently exists: $\op{p}_2$ assigns $0$ to $\sv{a}$, then $\op{p}_1$ exits the $\com{while}\;\sv{a}\hspace{0.2pt} > 0\;\com{do}\stsp\ldots\stsp$ cycle
and can assign $\op{True}$ to $\sv{b}$. On the other hand, a computation where we have only environment steps which do not modify $\sv{b}$ cannot reach a state satisfying $\sv{b}$ whatsoever. 
Hence, \emph{certain} computations can be consistent but surely cannot mean \emph{any} potential computation, be it finite or infinite.
Intuitively, \emph{certain} shall therefore encircle
the largest possible subset of potential computations of $\op{p}_1\hspace{-1pt} \parallel\hspace{0.1pt}\op{p}_2$ for which the answer to the question is a
\emph{provable} `\emph{yes}'.
More precisely, we may consider the following candidates: 
\begin{enumerate}
\item[(1)] $\pcs{\op{p}_1\hspace{-1.5pt}\parallel\hspace{-0.7pt}  \op{p}_2}$ -- \ie\stsp all finite potential computations of $\op{p}_1\! \parallel\hspace{-0.5pt} \op{p}_2$,
\item[(2)] $\pcsi{\op{p}_1\hspace{-1.5pt} \parallel\hspace{-0.7pt}  \op{p}_2}\hspace{-1pt} \cap\hspace{0.5pt} \envC{\hspace{-1pt}\bot}$ -- \stsp\ie\stsp
  all infinite actual computations of $\op{p}_1\hspace{-1.5pt} \parallel\hspace{-0.2pt} \op{p}_2$,
\item[(3)] $\pcsi{\op{p}_1\hspace{-3.1pt} \parallel\hspace{-2pt} \op{p}_2}\hspace{-2pt} \cap\hspace{0.01pt} \envC{\hspace{-1pt}\op{id}}$ -- \ie\stsp all infinite potential computations of $\op{p}_1\hspace{-2.7pt}  \parallel\hspace{-1.4pt}  \op{p}_2$ in presence of a `stuttering' environment.
\end{enumerate}
With (1) the answer is `\emph{no}': apart from many computations that just do not reach a state satisfying $\sv{b}$,
$\pcs{\op{p}_1\hspace{-2pt} \parallel\hspace{-1pt} \op{p}_2}$ also contains all prefixes of computations that do reach $\sv{b}$ once. Most of these prefixes, however, do not need to reach $\sv{b}$ themselves:
the singleton computation $(\op{p}_1\hspace{-2.5pt} \parallel\stnsp \op{p}_2,\stsp \sigma_0)$ is the simplest example.
With (2) the answer is also `\emph{no}': by considering all infinite computations 
without any environment steps
we actually exclude the most relevant ones that reach a state satisfying $\sv{b}$ and terminate.
By contrast, the option (3) not only covers all relevant computations of $\op{p}_1\hspace{-0.05cm}  \parallel \op{p}_2$
that reach $\sv{b}$ and terminate yielding subsequently to the `stuttering' environment, but also those with
infinite runs of the environment perpetually holding back all possible program steps, 
as well as of the $\com{while}\;\sv{a}\hspace{-0.5pt} >\hspace{-0.5pt} 0\;\com{do}\stsp\ldots\stsp$ cycles arbitrarily interleaved with the environment, 
all having $\neg\sv{b}$ as an invariant.
Casting such runs out
would result in an adequately large subset of computations of $\op{p}_1\hspace{-0.05cm}  \parallel \op{p}_2$ 
which nonetheless enables argumentation 
that the answer to the above question is `\emph{yes}'.

These considerations ineluctably lead to the conclusion that some form of \emph{fairness}
ought to be generally asserted on computations in $\rpcsi{p}{\rho}$ 
in order to achieve rational reasoning upon liveness properties of $p$,
which is by no means a new insight and various forms have been studied 
in the contexts of process algebras and temporal logics
(\cf~\cite{VANGLABBEEK2019100480, 10.1007/3-540-61474-5_84, 10.5555/903616} to name only a few).
Taking advantage of the syntactic structure built into the program part of each configuration,
next section develops a concise notion of fair potential computations.
%
\section{Program positions}\label{Sb:progpos}
The starting point is a view of the inductive rules for program steps in Figure~\ref{fig:pstep}
from the rewriting system perspective.
For example the rule
\[
\Rule{}{\rpstep{\rho\stsp}{(\basic\hspace{2pt} f,\stdsp \sigma)\stsp}{\stnsp(\skipp,\stdsp f\hspace{2pt}\sigma)}}
\]
can in principle be regarded as the identity used to reduce 
$(\sv{x} :=\!1 \parallel \sv{x} := \!2,\stsp \sigma)$ either to $(\skipp\hspace{-0.5pt} \parallel\hspace{-0.5pt} \sv{x} := \!2,\stdsp \sigma_{[\sv{x} := 1]})$
or to $(\sv{x} := \!1\hspace{-0.5pt} \parallel\hspace{-0.5pt} \skipp,\stdsp \sigma_{[\sv{x} := 2]})$
by means of the rule {\it\sl`Parallel'}.
Attaching two different positions to the subterms $\sv{x}\hspace{-1.1pt} :=\!1$ and
$\sv{x}\hspace{-0.5pt} :=\!2$ in order to indicate where a reduction can take place
would make the rule {\it\sl`Parallel'} essentially superfluous and
resemble rewriting of terms in general as described \eg by Baader and Nipkow~\cite{baader_nipkow_1998}.  
Furthermore, the configuration $(\sv{x} :=\!1; \sv{x} := \!2,\stdsp \sigma)$ can similarly be reduced to $(\skipp; \sv{x} := \!2,\stdsp \sigma_{[\sv{x} := 1]})$
by the rule {\it\sl`Sequential'},
but by contrast not to $(\sv{x} := \!1; \skipp,\stdsp \sigma_{[\sv{x} := 2]})$ as there is no rule to this end. 
This deviation from 
general term rewriting, where each subterm gets identified with a position, suggests the following insight:
despite the same tree structure, the terms $\sv{x}\hspace{-1pt} :=\!1 \!\parallel\! \sv{x}\hspace{-1pt} := \!2$ and $\sv{x}\hspace{-1pt} :=\!1; \sv{x}\hspace{-1pt} := \!2$
exhibit different sets of positions 
entitled for a substitution according to the inductive rules in Figure~\ref{fig:pstep}.

In line with the representation of term positions by words over the natural numbers deployed in~\cite{baader_nipkow_1998}, 
the function $\pos$ in this setting maps each term of type $\langA{\alpha}$ to a set of
words over $\ty{nat}$, referred to as \emph{program positions}\index{program position} and defined by the following recursive equations using pattern matching:
\[
\begin{array}{l c l}
\pos\: \skipp & = & \bot \\
\pos(\skipp ;q) & = & \{ 0 \} \\
\pos(p ;q) & = & \{0w \:|\; w \in\hspace{-0.5pt} \pos\hspace{1.7pt} p \} \\
\pos(\Parallel{\!\skipp, \ldots, \skipp}) & = &  \{ 0 \} \\
\pos(\Parallel{\!p_1, \ldots, p_m}) & = &  \bigcup_{i=1}^{m}\{iw \:|\; w \in \hspace{-0.7pt}\pos\hspace{1.9pt} p_i \} \\
\pos\hspace{2pt} p & = & \{ 0 \} \\
\end{array}
\]
so that \eg $\stsp\pos(\sv{x} :=\!1 \hspace{-0.5pt}\parallel \sv{x} := \!2) = \{10, 20\}$ whereas $\pos(\sv{x} :=\!1; \sv{x} := \!2) = \{00\}$.

Immediate consequences of the definition are:
\begin{enumerate}
\item[-] $\pos\hspace{1.7pt}p$ is empty iff $p = \skipp$,
\item[-] $\epsilon\hspace{0.9pt} \notin\hspace{-0.5pt} \pos\hspace{2pt}p$,
\item[-] if $x \in\hspace{-0.5pt} \pos\hspace{2.1pt}p$ then $x = x\pr0$ (where $x\pr$ could also be the empty word $\epsilon$),
\item[-] $\pos\hspace{2pt}p\stsp$ contains more than one word only if the parallel operator, applied to
  at least two arguments different from $\skipp$, occurs in $p$.
  This particularly means that $\pos\hspace{1.9pt}p$ contains exactly one word for any locally sequential $p \neq \skipp$. 
\end{enumerate}

Further, the function satisfying the following recursive equations using pattern matching
\[
\begin{array}{l c l}
\plook{p}{0} & = & \op{inr}\hspace{2pt}p \\
\plook{(p\hspace{1.5pt} ;q)}{0x} & = & \plook{p}{x} \\
\plook{(\Parallel{\!p_1, \ldots, p_m})}{ix} & = & \plook{p_i}{x}\\
\plook{p}{x} & = & \op{inl}\hspace{1pt}\verb|()|
\end{array}
\]
returns a value of the co-product type $\ty{1}\hspace{0.4pt} +\hspace{0.4pt} \langA{\alpha}$ where \verb|()| denotes the value 
of the unit type $\ty{1}$.
Note that there is some $t$ of type $\langA{\alpha}$ such that $\plook{p}{x}\hspace{-1.9pt} =\hspace{-0.7pt} \op{inr}\hspace{2.9pt}t$ holds  
whenever $x\hspace{-1pt} \in\hspace{-1.5pt}  \pos\hspace{2pt} p$.
Hence by $\plook{p}{x}$ we will for the sake of brevity refer to the $t$ 
and say that $\plook{p}{x}$
retrieves the subterm determined by the program position $x\hspace{-0.1pt} \in\hspace{-0.4pt}  \pos\hspace{2pt} p$.

Although $\plook{p}{x}\hspace{-1.5pt} =\hspace{-0.5pt} \op{inr}\hspace{2.19pt}t$ is satisfiable also with $x\hspace{-0.4pt} \notin\hspace{-0.7pt}  \pos\hspace{2pt} p$
(consider \eg \stdsp$\plook{\skipp}{0}$),
this modelling 
enables the following, in this context particularly useful property.
\begin{lemma}\label{thm:eq-lookup-pos}
If \stsp$x \in\hspace{-0.5pt} \pos\hspace{2pt} p$ and \stsp$\plook{p}{x}\hspace{-1pt} =\hspace{0.55pt} \plook{q}{x}$ then \stsp$x\hspace{0.1pt} \in\hspace{-0.37pt} \pos\hspace{2.2pt} q$.
\end{lemma}
\begin{proof}
  By induction on the length of $x$. Let $x \in\hspace{-0.5pt} \pos\hspace{2pt} p$ and suppose $\plook{p}{x}\hspace{-1pt} =\hspace{0.5pt} \plook{q}{x}$.
  If $x$ has the length $0$ then we are done since $\epsilon\hspace{1.2pt} \notin\hspace{-0.2pt} \pos\hspace{2pt}p$.
  If $x$ has the length $1$ then $x = 0$ and hence $\plook{p}{0}\hspace{-1pt} =\hspace{0.5pt} \plook{q}{0}$ entails $p = q$.
  
  Next, suppose $x$ has the form $ix\pr$ with $x\pr\hspace{-0.75pt} \neq\hspace{0.5pt} \epsilon$.
  In case $i =\hspace{-0.4pt} 0$ we infer that there are some $p_1,p_2$ and $q_1,q_2$ such that
  $p\hspace{-1.5pt} =\hspace{-1pt} p_1;p_2$ and $q\hspace{-1.5pt} =\hspace{-1pt} q_1;q_2$. Thus,
$x\hspace{-0.5pt} \in\hspace{-1.5pt} \pos\hspace{2pt} p$ entails $x\pr\hspace{-1.5pt} \in\hspace{-1pt} \pos\hspace{2pt} p_1$ whereas
  $\plook{p}{x}\hspace{-1.2pt} =\hspace{0.7pt} \plook{q}{x}$ entails  $\plook{p_1}{x\pr}\hspace{-1pt} =\hspace{0.5pt} \plook{q_1}{x\pr}$.
  From this $x\hspace{-0.2pt} \in\hspace{-0.7pt} \pos\hspace{2.2pt} q$ follows since $x\pr \in\hspace{-0.2pt} \pos\hspace{2.2pt} q_1$ is provided by the induction hypothesis.

  Lastly, in case $i >\hspace{-0.2pt} 0$ we infer that there are some $p_1, \ldots, p_m$ and $q_1, \ldots, q_{m\pr}$ where $i \le m$ and $i \le m\pr$ such that
  $p =\hspace{1.5pt} \Parallel{\!p_1, \ldots, p_m}$ and $q = \hspace{1.5pt}\Parallel{\!q_1, \ldots, q_{m\pr}}$.
  Thus,
$x\hspace{-0.7pt} \in\hspace{-1.4pt} \pos\hspace{2pt} p$ entails $x\pr\hspace{-1.75pt}\in\hspace{-1.4pt} \pos\hspace{2pt} p_i$ whereas
  $\plook{p}{x}\hspace{-1.5pt} =\hspace{0.2pt} \plook{q}{x}$ entails  $\plook{p_i}{x\pr}\hspace{-1.5pt} =\hspace{0.2pt} \plook{q_i}{x\pr}$.
  From this $x\hspace{0.2pt} \in\hspace{-0.2pt} \pos\hspace{2.2pt} q$ follows since $x\pr\hspace{-0.2pt} \in\hspace{-0.4pt} \pos\hspace{2.2pt} q_i$ is provided by the induction hypothesis.
\end{proof}
Replacing $\plook{p}{x}$ by $p\pr$ at the position $x\hspace{-0.1pt} \in\hspace{-1pt}   \pos\hspace{2pt} p$
is (once more in line with~\cite{baader_nipkow_1998}) denoted by $\psubst{p}{p\pr}{x}$ and defined by the equations 
\[
\begin{array}{l c l}
\psubst{p}{p\pr}{0} & = & p\pr \\
\psubst{(p\hspace{1.5pt} ;q)}{p\pr}{0x} & = & (\psubst{p}{p\pr}{x});q \\
\psubst{(\Parallel{\!p_1, \ldots, p_m})}{p\pr}{ix} & = & \Parallel{\!p_1, \ldots \hspace{1.5pt}\psubst{p_i}{p\pr}{x}\hspace{1.5pt} \ldots, p_m}
\end{array}
\]

Regarding the three syntactic operations together,
we have the property
\begin{equation}\label{eq:plook_psubst}
\plook{(\psubst{p}{q}{x})}{x\pr} = \left\lbrace\begin{array}{l l} q & \mbox{ if }\stsp x\hspace{0.1pt} =\hspace{0.1pt} x\pr \\
\plook{p}{x\pr} & \mbox{ otherwise}\end{array}\right.
\end{equation}
for any $p, q$ and $x, x\pr \hspace{-0.2pt} \in\hspace{-0.4pt} \pos\hspace{1.7pt}p$. 

A valuable consequence of Proposition~\ref{thm:eq-lookup-pos} and (\ref{eq:plook_psubst}) is:
\begin{corollary}\label{thm:rpos-psubst-retain}
  If there are two different positions $x\hspace{-1.2pt} \in\hspace{-2pt}\pos\hspace{2.1pt} p$ and $y\hspace{-1pt} \in\hspace{-2pt}\pos\hspace{2.1pt} p$ 
  then $y\hspace{0.2pt} \in\hspace{-0.5pt} \pos(\psubst{p}{q}{x})$.
\end{corollary}

The main objective behind these operations 
is to describe computation steps by means of substitutions at program positions 
as the following proposition shows.
\begin{lemma}\label{thm:rpos-pstep}
  $\rpstep{\rho\hspace{-0.2pt}}{\hspace{-1pt}(p, \stsp \sigma)\stsp}{\hspace{-1pt}(q, \stsp\sigma\pr)}$ iff
  there exist some $x\hspace{-0.2pt}  \in \hspace{-1pt} \pos\hspace{2pt} p$ and $p\pr$
such that $\rpstep{\rho\hspace{0.5pt}}{(\plook{p}{x},\stsp \sigma)\stsp}{\hspace{-1pt}(p\pr,\stsp \sigma\pr)}$ and $q = \psubst{p}{p\pr}{x}$.
\end{lemma}
\begin{proof}
By structural induction on $p$ for the if-direction, and on the relation $\pstep$ for the opposite. 
\end{proof}

The above statement however does not provide \emph{the} position
$x\hspace{-0.5pt} \in\hspace{-1pt} \pos\hspace{2pt}p$ that has been reduced (or `fired') by $\rpstep{\rho}{(p,\stsp \sigma)\hspace{1pt}}{\hspace{-1pt}(q,\stsp \sigma\pr)}$
so that $q = \psubst{p}{p\pr}{x}$.
On the other hand, the existence of different 
$x\hspace{0.1pt}, x\pr\hspace{-0.2pt} \in\hspace{-0.4pt}  \pos\hspace{2pt}p$
with 
$\rpstep{\rho}{\hspace{-0.5pt}(\plook{p}{x},\stsp \sigma)\hspace{1pt}}{\hspace{-1pt}(u,\stsp \sigma\pr)}$,
$\rpstep{\rho}{\hspace{-0.5pt}(\plook{p}{x\pr},\stsp \sigma)\hspace{1pt}}{\hspace{-1pt}(v,\stsp \sigma\pr)}$
and $\psubst{p}{u}{x}\hspace{-1.2pt} = \psubst{p}{v}{x\pr}$ 
entails $\plook{p}{x}\hspace{-1.2pt} =\hspace{-0.5pt} u$ and $\plook{p}{x\pr}\hspace{-1.2pt} = v$,
\ie the existence of program steps which behave like environment steps by retaining the program part of a configuration.

From $\rpstep{\rho\hspace{0.2pt}}{\hspace{-1pt}(p,\stsp \sigma)\hspace{0.2pt}}{\hspace{0.2pt}(p,\stsp \sigma\pr)}$
one can basically infer that $p$ is not jump-free and there is a label
$i$ with $\ret{\hspace{1pt}i}\stsp$ having the form $\cjump{\hspace{1.2pt}C\hspace{-0.9pt}}{\hspace{0.7pt}i}{w}$.
The following example illustrates that retrieving functions with such a property
can be used to produce very peculiar effects:
let $\rho$ be a retrieving function satisfying $\ret{\hspace{0.7pt}1}\hspace{-0.1pt}=\hspace{-0.1pt} \jump{\hspace{0.2pt}1}$ and
$\ret{\hspace{0.5pt}2}\hspace{-0.1pt} =\hspace{-0.1pt} \jump{\hspace{0.7pt}2}$, and
consider the program $(\rho,\stsp \jump{\hspace{0.1pt}1}\hspace{-0.5pt} \parallel\hspace{-0.5pt}\jump{\hspace{0.7pt}2})$. 
Then any $\sq \in \rpcsi{\jump{1}\hspace{-0.2pt} \parallel\hspace{-0.2pt} \jump{\hspace{0.1pt}2}}{\rho}\hspace{0.2pt} \cap\stsp \envC{\hspace{-0.5pt}\bot}$,
\ie any infinite actual computation of $(\rho, \jump{1} \parallel\hspace{-0.2pt} \jump{\hspace{0.1pt}2}$), has the form
\[
(\jump{1} \parallel\stnsp \jump{2},\stdsp \stateOf{\sq_0}) \pstep (\jump{1} \parallel\stnsp \jump{2},\stdsp \stateOf{\sq_0}) \pstep \ldots
\]
in particular concealing which of the two jumps has been deployed  by $\sq$ at all.  

Ruling such singularities generally out, only such retrieving functions $\rho$ where  
we have $\rho\hspace{3pt} i \hspace{-0.5pt}\neq\stnsp \cjump{\hspace{0.5pt}C\hspace{-1.7pt}}{ i\hspace{-0.7pt}}{\hspace{-0.7pt}w\hspace{-1pt}}$ 
for all $i, C, w$
will be considered from now on.
Then to each transition $\rpstep{\rho\hspace{0.2pt}}{(p,\stsp \sigma)\stsp}{\hspace{-0.5pt}(q,\stsp  \sigma\pr)}$ there is the unique \emph{fired}\index{program position!fired}
position $x \in\hspace{-0.5pt} \pos\hspace{1.9pt} p$
with a step $\rpstep{\rho\hspace{0.2pt}}{\hspace{-1pt}(\plook{p}{x}, \sigma)\stsp}{\hspace{-0.7pt}(p\pr, \sigma\pr)}$ such that $q = \psubst{p}{p\pr}{x}$.
\section{Fair computations}\label{Sb:faircomp}
It is essential to underline that in presence of await-statements, $\pos\hspace{2pt} p$
may also contain program positions that are not always available for a substitution:
if $x\hspace{0.2pt}  \in \hspace{-0.5pt} \pos\hspace{1.7pt} p$ with
$\plook{p}{x}\hspace{-1pt} =\hspace{0.7pt} \await{\hspace{1pt} C\hspace{-0.2pt}}{\hspace{-0.2pt}q\hspace{-0.2pt}}\stsp$ then
\begin{enumerate}
\item[-] $x$ is never available when $C = \bot$,
\item[-] availability of $x$ depends only on termination of $q$ when $C =\stnsp \top$,
\item[-] availability of $x$ additionally depends on the evaluation of $C$ otherwise.
\end{enumerate}
In order to keep the notion of fairness for the beginning as concise as possible,
we will focus on such program positions whose availability is state independent:
\begin{definition}
  Let $p$ be a term of type $\langA{\alpha}$. A position $x\hspace{-0.5pt}  \in \hspace{-1.2pt} \pos\hspace{1.7pt} p\stsp$
  is called \emph{always available}\index{program position!always available} if 
$\stsp\plook{p}{x}\hspace{-1.2pt} = \hspace{0.2pt}\await{\hspace{1pt} C}{q}\stsp$ implies
\begin{enumerate}
\item[(a)] $C = \stnsp\top$ and
\item[(b)] for any $\rho$ and a state $\sigma$ there exists $\sigma\pr$ such that $\rpsteps{\rho}{(q,\stsp \sigma)\hspace{0.5pt}}{\hspace{-1pt}(\skipp,\stsp \sigma\pr)}$.
\end{enumerate}
\end{definition}
Consequently, if $x\hspace{-0.1pt}  \in \hspace{-0.7pt} \pos\hspace{2pt} p$ is always available
then for any $\sigma$  there exist some $p\pr$ and $\sigma\pr$ such that
$\rpstep{\rho}{(\plook{p}{x},\stsp \sigma)\hspace{0.5pt}}{\hspace{-1pt}(p\pr, \stsp\sigma\pr)}$ and 
$\rpstep{\rho}{(p,\stsp \sigma)\hspace{0.5pt}}{\hspace{-1pt}(\psubst{p}{p\pr}{x},\stsp \sigma\pr)}$ hold.

A fair computation in principle does not keep an always available position from making its move once:
\begin{definition}\label{def:fair}
  An infinite potential computation $\sq\hspace{-0.3pt}  \in \hspace{-0.5pt}\rpcsi{p}{\rho}$ is called \emph{fair}\index{computation!fair}
  if for any $i\hspace{0.1pt} \in\hspace{-0.75pt}\naturals$ and
  any always available $x\hspace{-0.55pt} \in \hspace{-1.1pt}\pos(\progOf{sq_i})$
  there exists $j \ge i$ with a program step $\rpstep{\rho}{\stnsp\sq_j\hspace{0.3pt}}{\hspace{-2.7pt}\sq_{j+1}}$
  having $x$ as the fired position and satisfying $\plook{\progOf{\sq_j}}{x}\hspace{-1.2pt} =\hspace{-0.7pt} \plook{\progOf{sq_i}}{x}$.
This will be denoted by $\sq\hspace{0.1pt} \in\hspace{0.1pt} \rpcsf{p}{\rho}$.
\end{definition}
Thus, the definition leaves for instance open whether a program step will ever be performed by some 
$\sq\hspace{-0.7pt} \in\hspace{-0.4pt} \pcsf{\await{\stsp C\hspace{-1pt}}{\skipp}}$
with $C \hspace{-1pt}\subset\hspace{-2pt} \top$ because the only position $0$ is not always available.
Indeed, a plain generalisation of Definition~\ref{def:fair} from \emph{always available} to \emph{any} position
would make it inadequately restrictive as
we would be able to infer that any $\sq\hspace{-0.5pt} \in\hspace{-0.5pt} \pcsf{\await{\hspace{0.5pt} C\hspace{-0.7pt}}{\skipp}}$
reaches a state satisfying $C$ at least once. 
A way to generalise Definition~\ref{def:fair} to additionally enable reasoning upon liveness
properties of programs deploying await-statements beyond atomic sections will be discussed in Section~\ref{Sb:fair-await}.

For a program $(\rho, p)$ that is sequential, the property
$\sq\hspace{-1pt} \in\hspace{-1.2pt} \rpcsf{p}{\rho}$
simplifies to: $\sq\hspace{-0.4pt} \in\hspace{-0.5pt} \rpcsi{p}{\rho}$ and
for any $i\hspace{0.2pt} \in\hspace{-0.5pt} \naturals$ and any always available $x\hspace{-0.5pt} \in\hspace{-1pt} \pos(\progOf{sq_i})$
there exists $j\stsp \ge\stsp i$ with a program step
$\rpstep{\rho\stsp}{\hspace{-0.1pt}\sq_j\hspace{0.2pt}}{\hspace{-1.4pt}\sq_{j+1}}$. In other words, fairness for sequential programs means that
there are no infinite runs of environment steps as long as an always available position is present.

It shall lastly be outlined how this notion of fair computations allows us to substantiate the answer `\emph{yes}' in the introductory example
where we have $\pos(\op{p}_1\hspace{-1.9pt} \parallel\hspace{-0.7pt}\op{p}_2) = \{10, 20\}$ and both positions are moreover always available.
Hence, any computation $\sq \in \pcsf{\op{p}_1 \hspace{-1.9pt}\parallel\stnsp \op{p}_2}\hspace{-1.2pt} \cap\hspace{0.4pt} \envC{\hspace{-0.7pt}\op{id}}$ 
eventually reduces $\op{p_2}$ to $\skipp$ at the position $20$ assigning thereby $0$ to $\sv{a}$.
Note that the first parallel component must still be in the $\com{while}\;\sv{a}\hspace{-0.2pt} >\hspace{-0.7pt} 0\;\com{do}\stsp\ldots\stsp$ cycle as $\sv{a}$ had constantly the value $1$ prior to that. 
Now however only one available position remains and the states of all subsequent configurations will moreover satisfy $\sv{a} = 0$.
The first component will therefore eventually exit the $\com{while}\;\sv{a} > 0\;\com{do}\stsp\ldots\stsp$ cycle making in turn the subterm $\sv{b} := \!\op{True}\stsp$ available at the position $10$.
Finally, the environment, prior to keep on `stuttering' \emph{ad infinitum}, will first let $\sq$ reduce this position to $\skipp$ assigning thereby $\!\op{True}$ to $\sv{b}$, and then reduce $\skipp\! \parallel\! \skipp$ to $\skipp$. 
\section{A refutational approach to proving liveness properties}\label{S:refute}
The above paragraph can at best be considered as a sketch of a sketch of an argumentation that
$\sq\hspace{-1.2pt} \in\hspace{-1pt} \pcsf{\op{p}_1 \hspace{-2.9pt}\parallel\hspace{-1.5pt} \op{p}_2}\hspace{-2.1pt} \cap\hspace{-0.2pt}\inC{\!(\sv{a}\hspace{-0.5pt} =\hspace{-0.7pt} 1\hspace{-1pt} \wedge\hspace{-1.5pt} \neg\sv{b})}\hspace{-0.2pt} \cap\hspace{0.1pt} \envC{\hspace{-0.5pt}\op{id}}$ 
reaches a state satisfying $\sv{b}$. Nonetheless it pointed to how one can follow the program structure
in order to draw relevant conclusions, \ie the very basic approach is to some extent not unlike the syntax-driven rely/guarantee reasoning.

Generally, to show that any
fair computation of a program $p$ that starts in a state satisfying some given condition $P$ and has only `stuttering' environment steps eventually
reaches a state satisfying some given $Q$, one can attempt to refute the opposite, \ie that
there exists
$\sq\hspace{-0.2pt} \in\hspace{-0.1pt}  \pcsf{p}\hspace{-2pt} \cap\hspace{0.5pt}\inC{\hspace{-1.1pt} P}\hspace{-0.5pt} \cap\hspace{0.4pt} \envC{\hspace{-0.9pt}\op{id}}$ having $\neg Q$ as an invariant.
This motivates the following three definitions.
\begin{definition}\label{def:fcomp-from}
  Let $\sq$ be an infinite sequence of configurations, $p$ -- a jump-free program and $P, Q$ -- state predicates.
  Then the condition
  $\lcond{\sq\hspace{-0.5pt}}{\hspace{-0.7pt} P\hspace{-0.5pt}}{\hspace{0.9pt} p\hspace{0.9pt}}{\hspace{0.7pt}Q}$ holds
  iff $\sq\hspace{-0.2pt} \in\hspace{-0.2pt}  \pcsf{p}\hspace{-1.5pt} \cap\hspace{-0.1pt} \inC{\hspace{-1pt}P}$ and 
  $\stateOf{\sq_i} \notin\hspace{-0.5pt}  Q$ holds for all $i\hspace{0.7pt} \in\hspace{-0.1pt}\naturals$.
\end{definition}  
\begin{definition}\label{def:fcomp-fromN}
  The condition $\lcondN{\sq\hspace{-0.25pt}}{\! P\hspace{-0.5pt}}{\hspace{0.5pt} p\hspace{0.25pt}}{\hspace{0.7pt}Q}\stsp$ holds iff
  $\stsp\lcond{\sq\hspace{-0.25pt}}{\! P\hspace{-0.7pt}}{\hspace{0.5pt}p\hspace{0.4pt}}{\hspace{0.7pt}Q}$ 
  and 
  $\progOf{\sq_i} \hspace{-0.5pt}\neq \skipp$ holds for all $i\hspace{0.9pt} \in\hspace{-0.1pt}\naturals$.
\end{definition}
\begin{definition}\label{def:fcomp-fromT}
  Let additionally $n \in\hspace{-0.5pt} \naturals$.
  Then $\lcondT{\sq\hspace{-0.25pt}}{\hspace{-1.9pt} P\hspace{-0.5pt}}{\hspace{0.5pt}p\hspace{0.7pt}}{Q}{\hspace{0.2pt}n}\stsp$ holds iff
  $\sq\hspace{0.2pt}  \in\hspace{0.1pt} \pcsi{p} \hspace{-1pt} \cap\hspace{0.5pt} \inC{\hspace{-0.75pt}P}$ and $\stateOf{\sq_i}\hspace{0.1pt} \notin\hspace{-0.4pt}  Q\stsp$
  holds for all $i\hspace{1.1pt} \in\hspace{-0.1pt} \naturals$, and moreover
  $n$ is the least number such that $\progOf{\sq_n} = \skipp$.
\end{definition}
Thus, $\lcond{\sq\hspace{-0.2pt}}{\hspace{-1.5pt}P\hspace{-0.9pt}}{\hspace{0.5pt} p\hspace{0.2pt}}{\hspace{0.7pt}Q}$ entails
$\lcondN{\sq\hspace{-0.4pt}}{\!P\hspace{-0.5pt}}{\hspace{0.4pt}p\hspace{0.2pt}}{\hspace{0.2pt}Q}\hspace{-0.4pt} \vee\hspace{-0.5pt} (\exists n.\:\lcondT{\sq\hspace{-0.4pt}}{\hspace{-2pt}P\hspace{-0.5pt}}{\hspace{0.4pt} p\hspace{0.2pt}}{\hspace{0.2pt}Q}{n})$
and the disjunction is moreover exclusive:
the condition $\lcondN{\sq\stnsp}{\!P\stnsp}{p}{Q}$ covers the cases 
with a non-terminating $\sq$
and $\lcondT{\sq\hspace{-0.5pt}}{\hspace{-2.5pt}P\hspace{-0.9pt}}{\hspace{0.2pt}p\hspace{0.2pt}}{Q}{n}$ -- the cases with $\sq$ reaching the first $\skipp$-configuration by $n$ steps.
In particular this means that it is sound to refute the two cases in order to conclude
that there exists some $i$ with $\stateOf{\sq_i} \in\hspace{-0.5pt}  Q$ if $\sq\hspace{-0.1pt} \in\hspace{-0.4pt}  \pcsf{p}\hspace{-2pt} \cap\hspace{1pt} \inC{\hspace{-1.2pt}P}$.

Conversely, $\stsp\lcond{\sq\hspace{-0.7pt}}{\!P\hspace{-1.1pt}}{\hspace{0.2pt}p\hspace{0.3pt}}{\hspace{0.5pt}Q}\stsp$ 
follows straight from
$\lcondN{\sq\stnsp}{\hspace{-2pt}P\hspace{-0.9pt}}{\hspace{0.2pt}p\hspace{0.3pt}}{\hspace{0.5pt}Q}$,
but the same conclusion is not immediate with $\lcondT{\sq\hspace{-0.5pt}}{\hspace{-2.2pt}P\hspace{-1pt}}{\hspace{0.5pt}p\hspace{0.5pt}}{\hspace{0.5pt}Q\hspace{-0.2pt}}{\hspace{0.2pt}n}\stsp$ as Definition~\ref{def:fcomp-fromT}
does not explicitly require from $\sq$ to be fair. This point is however settled by the following general property. 
\begin{lemma}\label{thm:tfair}
  If $\sq\hspace{-0.1pt} \in\hspace{-0.1pt} \rpcsi{p}{\rho}$ and $\progOf{\sq_n} = \skipp\stsp$ 
  then $\stsp\sq \in \rpcsf{p}{\rho}$.
\end{lemma}
\begin{proof}
  Unfolding Definition~\ref{def:fair}, let $i\hspace{0.3pt} \in\hspace{-0.9pt}\naturals$ and $x\hspace{-0.7pt} \in \hspace{-1.5pt}\pos(\progOf{sq_i})$ be an always available position.
  Firstly, note that $i < n$ must hold since otherwise we would have $x\hspace{-0.4pt} \in\hspace{-1.4pt} \pos\hspace{2.2pt}\skipp$ which is impossible.
  Next, we proceed by showing that the additional assumption
  \begin{enumerate}
  \item[(a)]\hspace{-4pt} \emph{each program step $\rpstep{\rho\hspace{-0.2pt}}{\hspace{-0.2pt}\sq_j\stsp}{\hspace{-2pt}\sq_{j+1}}$ with $j\hspace{-0.1pt} \ge\hspace{0.1pt} i$
    and $\stsp\plook{\progOf{sq_j}}{x}\hspace{-1.4pt} =\hspace{-1pt} \plook{\progOf{sq_i}}{x}$ has the fired position $x\pr\hspace{-0.5pt} \neq\hspace{0.5pt} x$}
   \end{enumerate}
  leads to a contradiction. Indeed, with (a) it is possible to derive
    \begin{enumerate}
    \item[(b)] \emph{$j \ge i\stsp$ implies $\plook{\progOf{sq_j}}{x}\hspace{-1pt} = \plook{\progOf{sq_i}}{x}$}
    \end{enumerate}
    for any $j$, so that taking $n$ for $j$ in (b) we get $\plook{\skipp}{x}\hspace{-1pt} =\hspace{-1pt}\plook{\progOf{sq_i}}{x}$
    which once more yields  $x\hspace{-1.5pt} \in\hspace{-2pt} \pos\hspace{2.4pt}\skipp$ by Proposition~\ref{thm:eq-lookup-pos},
    meaning that the negation of (a) must be valid, \ie $\stsp\sq\hspace{-0.1pt} \in\hspace{0.1pt} \rpcsf{p}{\rho}$.
    Hence, it remains to show (b) 
    which will be accomplished by induction on $j$.

    The base case is clear since $j \hspace{-0.7pt}=\hspace{-0.5pt} 0$ implies $i\hspace{-0.2pt} =\hspace{-0.2pt} 0$.
    As the induction step we have to show $\plook{\progOf{sq_{j+1}}}{x}\hspace{-0.9pt} =\hspace{-0.2pt} \plook{\progOf{sq_i}}{x}$
    additionally assuming (b) and $j\hspace{-0.1pt} +\hspace{-0.2pt} 1\hspace{-0.5pt} \ge i$. Since with $j +\hspace{-0.5pt} 1\hspace{-0.7pt} = i$ the goal follows straight,
    with $j\stsp \ge\stsp i$ we infer $\plook{\progOf{sq_{j}}}{x}\hspace{-1pt} = \plook{\progOf{sq_i}}{x}$ from (b).
    Then 
    only $\plook{\progOf{sq_{j+1}}}{x}\hspace{-1.7pt} =\hspace{-1.2pt} \plook{\progOf{sq_j}}{x}$ remains to be established,
 noting that $x\hspace{0.1pt}  \in\hspace{-0.4pt}  \pos(\progOf{sq_{j}})$ moreover holds by Proposition~\ref{thm:eq-lookup-pos}.
    
 Firstly, if we have a program step
 $\rpstep{\rho\hspace{0.25pt}}{\hspace{-0.5pt}\sq_j\hspace{0.4pt}}{\hspace{-2.2pt}\sq_{j+1}}\hspace{-0.5pt}$
 then (a) entails $x\pr\hspace{-0.9pt} \neq x$ where $x\pr$ is its fired position. 
Thus,
$\plook{\progOf{sq_{j+1}}}{x}\hspace{-1pt} = \hspace{-1pt}\plook{(\psubst{\progOf{sq_j}}{p\pr}{x\pr})}{x}\hspace{-1pt}=\hspace{-1pt} \plook{\progOf{sq_j}}{x}$ follows using (\ref{eq:plook_psubst}).
    Secondly, with $\sq_j \hspace{0.2pt}\estep\hspace{-0.7pt} \sq_{j+1}$ we have $\progOf{\sq_{j+1}}\hspace{-0.5pt} = \hspace{-0.5pt}\progOf{\sq_{j}}$.
\end{proof}

To sum up the result:
\begin{corollary}\label{thm:split-eq}
 $\hspace{-2pt}\lcond{\sq\hspace{-0.5pt}}{\hspace{-1.5pt}P\hspace{-0.5pt}}{\hspace{0.9pt}p\hspace{0.7pt}}{\hspace{0.9pt}Q}$ iff  
  \hspace{0.2pt}$\lcondN{\sq}{\hspace{-1.75pt}P\hspace{-0.5pt}}{\hspace{0.7pt}p\hspace{0.9pt}}{\hspace{1pt}Q}\hspace{0.1pt} \vee (\exists n.\hspace{2.5pt} \lcondT{\sq}{\hspace{-1.9pt}P\hspace{-0.5pt}}{\hspace{0.7pt}p\hspace{0.9pt}}{\hspace{1pt}Q\hspace{0.2pt}}{\hspace{0.5pt}n})$
\end{corollary}
Note that the two disjuncts are also useful on their own:
one can for instance attempt to refute statements of the form
$\exists\sq\hspace{-0.9pt} \in\hspace{-0.9pt}  \envC{\hspace{-1pt}\op{id}}. \hspace{2.5pt}\lcondN{\sq\hspace{-0.7pt}}{\hspace{-2.5pt}P\hspace{-1pt}}{\hspace{0.3pt}p\hspace{0.2pt}}{\hspace{-0.7pt}\bot}$
which is the same as to establish 
that any $\sq\hspace{-1.9pt} \in\hspace{-1.9pt} \pcsf{p}\hspace{-2.1pt} \cap\hspace{0.2pt} \envC{\hspace{-1.2pt}\op{id}}\hspace{-0.7pt} \cap\hspace{-0.2pt} \inC{\hspace{-1.5pt}P}$
reaches a $\skipp$-configuration once,
\ie $\stsp p\stsp$ terminates on any input in $P$.

From the definitions follows that different conclusions can in general be drawn
depending on whether we have $\hspace{-0.5pt}\lcondN{\sq\hspace{-0.7pt}}{\hspace{-2pt}P\hspace{-0.7pt}}{\hspace{0.2pt}p\hspace{0.4pt}}{\hspace{0.59pt}Q}\stsp$
or $\lcondT{\sq\hspace{-0.7pt}}{\hspace{-2.1pt}P\hspace{-0.79pt}}{\hspace{0.25pt}p\hspace{0.2pt}}{\hspace{0.59pt}Q\hspace{-0.2pt}}{\hspace{0.2pt}n}$.
The syntax-driven rules presented in the remainder of this section handle therefore the two cases 
for each language constructor (with particular restrictions on await-statements for the non-terminating case) in place of $p$.

Starting with the simplest, an assumption $\lcondN{\sq}{\hspace{-1.5pt} P\hspace{-0.4pt}}{\stsp\skipp\stsp}{\hspace{0.7pt}Q\stsp}$ immediately leads to a contradiction, 
whereas $\lcondT{\sq}{\!P\hspace{-0.5pt}}{\hspace{0.7pt}\skipp\hspace{0.5pt}}{\hspace{0.7pt}Q}{\hspace{0.5pt}n}$ entails $n = \hspace{0.2pt}0$. 

Similarly to $\skipp$, $\lcondN{\sq\hspace{-1pt}}{\hspace{-2.5pt}P\hspace{-1pt}}{\hspace{0.2pt}\basic\hspace{2.1pt} f\hspace{-0.7pt}}{\hspace{0.4pt}Q}\stsp$ leads to a contradiction too
but, in contrast to $\skipp$, fairness of $\sq$ 
is decisive for this conclusion as it rules out infinite runs of environment steps.
Further, with $\lcondT{\sq}{\hspace{-2.2pt} P\hspace{-1pt}}{\hspace{0.5pt}\basic\hspace{2.1pt} f\hspace{-0.7pt}}{\hspace{0.4pt}Q}{n}$
in addition to $\stateOf{\sq_0}\hspace{-0.2pt} \in\hspace{-1.1pt} P$ we can infer
\begin{enumerate}
\item[(a)] $n > 0$, 
\item[(b)] $\sq_{n-1}\hspace{-1pt} = (\basic\hspace{2.1pt} f,\hspace{0.5pt} \sigma)$ and $\sq_n\hspace{-1pt} = (\skipp, \hspace{0.9pt}f\hspace{2pt}\sigma)$,
\item[(c)] $f \hspace{2pt} \sigma\hspace{0.2pt} \notin\hspace{0.2pt} Q$
\end{enumerate}
for some state $\sigma$, acquiring thus the information that 
the atomic action associated to $\stnsp f$ has been triggered at the point $n-1$ on $\sq$.
As will become evident later on, 
several such points can moreover indicate a relative order of events: 
as opposed to the parallel composition $\basic\hspace{2.1pt}f\hspace{-0.9pt}\parallel\hspace{-0.5pt}\basic\hspace{2.7pt}g$ where we can merely acquire that
$f$ and $g$ have been triggered at some respective points $n$ and $m$ on $\sq$, with
the sequential $\basic\hspace{2.1pt}f;\basic\hspace{2.7pt}g$ we would additionally have
$n < m$. 

Next two propositions handle the conditional and the await-statements.
\begin{lemma}\label{thm:await-live}
The following implications hold: 
\begin{enumerate}
\item[\emph{(1)}] if $\lcondN{\sq\hspace{-1.5pt}}{\hspace{-3.1pt}P\stnsp}{\langle p \rangle\hspace{-0.5pt}}{Q}\stsp$ then there exists a state $\sigma$ such that
  for all $\hspace{1pt}\sigma\pr\hspace{-0.5pt}$ the configuration $(\skipp,\stsp \sigma\pr)$ is not reachable from $(p,\stsp \sigma)$ by any sequence of program steps,
  \emph{\ie} $\hspace{-2pt}(p,\stsp \sigma) \not\psteps\hspace{-1pt} (\skipp,\stsp \sigma\pr)$;
\item[\emph{(2)}] if $\lcondT{\sq}{\hspace{-1.7pt}P\hspace{-0.7pt}}{\hspace{0.5pt}\await{\hspace{0.7pt}C\hspace{-0.5pt}}{\hspace{-0.2pt}p\hspace{-0.2pt}}\hspace{0.5pt}}{\hspace{0.9pt}Q\hspace{0.2pt}}{\hspace{0.7pt}n}\stsp$ then $\stsp n > 0$, $\stsp\stateOf{\sq_{n-1}} \in\hspace{-0.2pt} C$ and
  \[(p,\stsp \stateOf{\sq_{n-1}})\hspace{0.2pt} \psteps\hspace{-1pt} (\skipp,\stsp \stateOf{\sq_{n}}).\]
\end{enumerate}
\end{lemma}
\begin{proof}
  To show (1) we appeal to fairness as follows. Suppose for all states $\sigma$ there exists some $\sigma\pr$ such that
  $(p,\stsp \sigma)\hspace{0.2pt} \psteps \hspace{-0.79pt}(\skipp, \stsp\sigma\pr)$.  
Then the position $0$ (\ie $\langle p \rangle$) in $\progOf{\sq_0}$ is always available, and since $\sq$ is fair
there must be a program step
$\sq_j\hspace{0.5pt} \pstep\hspace{-2.1pt} \sq_{j+1}$ with $\progOf{\sq_{j}} = \langle p \rangle$ 
and
$\progOf{\sq_{j+1}} = \skipp$ which is a contradiction to $\lcondN{\sq}{\hspace{-1.4pt}P\hspace{-0.4pt}}{\stsp\langle p \rangle\hspace{0.2pt}}{\stsp Q}$.

To show (2), 
suppose $\lcondT{\sq\hspace{-0.55pt}}{\hspace{-2.4pt}P\stnsp}{\hspace{0.2pt}\await{\hspace{0.5pt}C\hspace{-1pt}}{\hspace{-0.7pt}p\hspace{-0.7pt}}\hspace{0.4pt}}{\hspace{0.4pt}Q}{n}$
and note that by contrast to (1)
not only atomic sections but 
all await-statements are thereby covered.
Since $\sq_n$ is the first $\skipp$-configuration on $\sq$ we can infer  
$n > 0$ and $(\await{\hspace{0.5pt}C\hspace{-1pt}}{\hspace{-1pt}p\hspace{-0.5pt}},\stdsp \stateOf{\sq_{n-1}})\stdsp \pstep \stnsp(\skipp,\stdsp \stateOf{\sq_{n}})$
from which
$\stateOf{\sq_{n-1}} \in\hspace{-0.7pt} C$ and $(p,\stdsp \stateOf{\sq_{n-1}}) \psteps\stnsp (\skipp, \stdsp\stateOf{\sq_{n}})$ follow.
\end{proof}
\begin{lemma}\label{thm:cond-live}
The following implications hold:
\begin{enumerate}
\item[\emph{(1)}] if $\lcondN{\sq\hspace{-0.5pt}}{\hspace{-1.9pt}P\hspace{-0.7pt}}{\hspace{0.4pt}\ite{\hspace{0.7pt} C\hspace{-1pt}}{\hspace{-0.5pt}p\hspace{-0.2pt}}{q}\hspace{-0.1pt}}{\hspace{0.5pt}Q}$ then there exists $n > \hspace{-0.5pt}0$
  such that either \:$\lcondN{\suffix{n}{\sq}\hspace{0.7pt}}{\hspace{-1pt}C\hspace{-0.2pt}}{\hspace{0.3pt}p\hspace{0.2pt}}{\hspace{0.5pt}Q}$ or
  $\:\lcondN{\suffix{n}{\sq}\hspace{0.5pt}}{\!\neg C}{\hspace{0.75pt}q\hspace{0.2pt}}{\hspace{0.5pt}Q}$; 
\item[\emph{(2)}] if $\lcondT{\sq\hspace{-0.5pt}}{\hspace{-1.9pt} P\hspace{-0.7pt}}{\hspace{0.4pt}\ite{\stsp C \hspace{-0.7pt}}{p}{q}}{\hspace{0.7pt}Q\hspace{-0.1pt}}{\hspace{0.2pt}n}$ then
  there exists some $m\hspace{-0.1pt}\in\hspace{-0.5pt}\naturals$ with \stsp$0\hspace{-0.4pt} <\hspace{-0.2pt} m$ and \stsp$m\hspace{-0.4pt} \le n$
    such that either \stsp$\lcondT{\suffix{m}{\sq}\hspace{-0.5pt}}{\hspace{-1.7pt}C\hspace{-1.1pt}}{\hspace{0.3pt}p\hspace{0.2pt}}{\hspace{0.4pt}Q\hspace{-0.25pt}}{\hspace{0.5pt}n-m}\stsp$ or $\lcondT{\suffix{m}{\sq}}{\!\neg C\hspace{-0.2pt}}{\hspace{0.59pt}q\hspace{0.2pt}}{\hspace{0.9pt} Q\hspace{0.25pt}}{\hspace{0.9pt} n-m}\:$.
\end{enumerate}
\end{lemma}
\begin{proof}
In (1) we particularly have $\sq\in\pcsf{\ite{\hspace{1pt} C\hspace{-0.5pt}}{p}{q}}$ which entails that a program step reducing the available position $0$ in $\progOf{\sq_0}$ will eventually be performed.
That is, there must be some $n > 0$ with $\sq_{n-1} \hspace{0.5pt}\pstep\hspace{-1.5pt} \sq_n$ such that
$\progOf{\sq_{n-1}}\hspace{-0.5pt} =\hspace{-0.5pt} \ite{\hspace{0.5pt} C\hspace{-1.2pt}}{\hspace{-1pt}p\hspace{-1pt}}{\hspace{-1pt}q\hspace{-1pt}}$ and
$\stateOf{sq\mystrut_{n}}\hspace{-1pt} =\hspace{-0.5pt} \stateOf{sq_{n-1}} $.
If $\stateOf{sq_{n-1}}\hspace{-1.2pt} \in\hspace{-1.9pt} C$ then $\suffix{n}{\sq}\hspace{-0.9pt} \in\hspace{-0.5pt} \pcsi{p}$ starts in a state satisfying $C$. 
Furthermore, as a suffix of $\sq$, the computation $\suffix{n}{\sq}\hspace{-1.5pt} \in\hspace{-1.2pt} \pcsi{p}$ neither reaches a $\skipp$-configuration nor a state satisfying $Q$ and is moreover fair.
Altogether 
$\lcondN{\suffix{n}{\sq}}{\hspace{-0.9pt}C\hspace{-0.5pt}}{\hspace{0.5pt}p\hspace{0.5pt}}{\hspace{0.75pt}Q}$ follows. 
With $\stateOf{sq_{n-1}} \hspace{0.5pt}\notin \hspace{-0.2pt}C$ we similarly conclude
$\lcondN{\suffix{n}{\sq}\hspace{0.2pt}}{\hspace{-1.5pt}\neg C}{\hspace{0.7pt}q\hspace{0.7pt}}{\hspace{0.75pt}Q}$.

In (2), $\lcondT{\sq\hspace{-0.5pt}}{\hspace{-2.55pt}P\hspace{-0.7pt}}{\hspace{0.25pt}\ite{\hspace{0.7pt} C\hspace{-1.2pt}}{p}{q\hspace{-0.5pt}}\hspace{0.2pt}}{\hspace{0.9pt} Q\hspace{-0.2pt}}{\hspace{0.5pt}n}\stsp$ entails some $i < n$ with $\sq_i\stsp \pstep\hspace{-2pt} \sq_{i+1}$
and $\progOf{\sq_i} = \ite{\stsp C\stnsp}{p}{q}$ because otherwise we would have $\progOf{\sq_n} \!\neq\! \skipp$. 
Hence, 
the suffix $\suffix{i+1}{\sq}$ starts either with $(p,\stsp \stateOf{\sq_{i+1}})$ where $\stateOf{\sq_{i+1}} \in \hspace{-0.2pt} C$ or
with $(q, \stsp\stateOf{\sq_{i+1}})$ where $\stateOf{\sq_{i+1}} \notin\hspace{-0.2pt} C$, and in both cases 
reaches the first $\skipp$-configuration in $n - i - 1$ steps.
\end{proof}

Let $\steqv{\sq\hspace{-1.5pt}}{\hspace{-1.5pt}\sq\pr}$ from now on denote that these infinite sequences of configurations 
run through the same states, \ie $\stdsp\stateOf{\sq_i}\hspace{-1pt} =\hspace{-1pt} \stateOf{\sq\pr_i}$ for all $i\hspace{-0.7pt}\in\hspace{-1.5pt}\naturals$.
This equivalence allows us to relate computations by their state traces and is of importance to all of the remaining rules starting with the sequential composition.
\begin{lemma}\label{thm:seq-live}
The following implications hold:
\begin{enumerate}
\item[\emph{(1)}] if $\lcondN{\sq\hspace{-0.2pt}}{\hspace{-1.7pt}P\hspace{-0.79pt}}{\hspace{0.5pt}p\hspace{0.9pt};\stnsp q\hspace{0.4pt}}{\hspace{0.5pt}Q}$ then there exists a computation $\steqv{\sq\pr\hspace{-1.2pt}}{\hspace{-0.9pt}\sq}$ and 
  either $\lcondN{\sq\pr\hspace{-1.4pt}}{\hspace{-2.9pt}P\hspace{-1.25pt}}{\hspace{-0.25pt}p\hspace{-0.4pt}}{\hspace{0.25pt}Q}$ or there exist $m, n\hspace{-1.2pt} \in\hspace{-1.9pt} \naturals$
  with $m\hspace{-1.9pt} <\hspace{-1.2pt} n$ such that we have $\lcondT{\sq\pr\hspace{0.3pt}}{\hspace{-1pt}P\hspace{-0.4pt}}{\hspace{0.5pt}p\hspace{0.7pt}}{\hspace{0.75pt}Q\hspace{-0.2pt}}{\hspace{0.7pt}m}$ and $\stsp\lcondN{\suffix{n}{\sq}\hspace{0.2pt}}{\hspace{-2pt}\top\hspace{-0.4pt}}{\hspace{0.7pt}q\hspace{0.25pt}}{\hspace{0.55pt} Q}$;  
\item[\emph{(2)}] if $\lcondT{\sq}{\hspace{-1.7pt}P\hspace{-0.5pt}}{\hspace{0.9pt}p\hspace{0.9pt};\stnsp q\hspace{0.4pt}}{\hspace{1pt}Q\hspace{0.5pt}}{\hspace{0.7pt}n}$ then there exist $\steqv{\sq\pr\hspace{-1.1pt}}{\hspace{-0.7pt}\sq}$ and some $m, l\stsp \in\hspace{-0.5pt} \naturals$ with $m < l \le n$ such that
  $\lcondT{\sq\pr}{\!P\hspace{-0.2pt}}{\hspace{0.7pt} p\hspace{0.5pt}}{\hspace{0.5pt}Q\hspace{0.2pt}}{\hspace{0.2pt}m}$ and $\stsp\lcondT{\suffix{l}{\sq}}{\hspace{-2pt}\top\hspace{-0.75pt}}{\hspace{0.92pt}q\hspace{0.25pt}}{\hspace{0.7pt}Q\hspace{0.1pt}}{\hspace{1.2pt}n-\hspace{0.79pt}l}$.  
\end{enumerate}
\end{lemma}
\begin{proof}
  In (1) we assume $\sq \in \pcsf{p;\stnsp q}$ that does not reach a $\skipp$-configuration.
  First, suppose $\progOf{\sq_i}\hspace{-1pt} \neq \skipp;q\stsp$ holds for all \stsp$i\hspace{0.7pt} \in\hspace{-0.2pt}\naturals$.
In this case we can infer 
that  for any $i\hspace{1pt}\in\hspace{0.1pt}\naturals\hspace{0.5pt}$ there is some $u_i$ with $\progOf{\sq_i}\hspace{-0.5pt} = u_i;q$.
Then let $\sq\pr_i\stsp \defeq\stsp (u_i,\stsp \stateOf{\sq_i})$.
Since $\sq\pr\hspace{-0.79pt}\in\hspace{-0.2pt} \pcsi{p}$ and $\steqv{\sq\pr\hspace{-0.7pt}}{\hspace{-0.1pt}\sq}$
are provided by this construction,
it only remains to be shown that $\sq\pr$ is fair.
Let $i\hspace{1pt}\in\naturals$ and let
$x\hspace{0.1pt} \in\hspace{-0.5pt} \pos\hspace{2pt}u_i$ be an always available position to this end.
Then $0x$ is an always available position of $u_i;q$ and hence
there is some $j \ge i$ with $(u_j;q,\stsp \sigma)\hspace{1pt}\hspace{0.5pt} \pstep \hspace{-0.5pt}(u_{j+1};q,\stsp \sigma\pr)$,
$\sigma\hspace{-0.9pt} = \stateOf{\sq_j}$, $\sigma\pr\hspace{-1.2pt} = \stateOf{\sq_{j+1}}$,
the fired position $0x$
and, moreover, $\plook{(u_j;q)}{0x}\hspace{-1pt} = \plook{(u_i;q)}{0x}$ \ie\stsp also $\stsp\plook{u_j}{x}\hspace{-1.1pt} = \plook{u_i}{x}$.
That $0x$ is fired yields a program step $(\plook{(u_j;q)}{0x},\stsp \sigma)\hspace{0.5pt} \pstep\hspace{-0.5pt} (p\pr, \stsp\sigma\pr)$ with some $p\pr$
such that $u_{j+1};q = \psubst{(u_j;q)}{p\pr}{0x} = (\psubst{u_j}{p\pr}{x});q$ holds.
Thus, $u_{j+1}\hspace{-1.5pt} = \psubst{u_j}{p\pr}{x}$ and $(\plook{u_j}{x}, \stsp\sigma)\hspace{1pt} \pstep\hspace{-0.7pt} (p\pr,\stsp \sigma\pr)$,
\ie $\stdsp x$ is the fired position of $(u_j, \stsp \sigma)\hspace{1pt} \pstep\hspace{-0.7pt} (u_{j+1},\stsp\sigma\pr)$.

Second, suppose $\mu$ is the smallest index with $\progOf{\sq_\mu} = \skipp;q$. We define $\sq\pr_i \defeq (u_i,\stsp \stateOf{\sq_i})$ for all $i\hspace{1.2pt} \le \mu$ and
$\sq\pr_i \defeq (\skipp,\stsp \stateOf{\sq_i})$ with $\sq\pr_{i-1}\hspace{0.5pt} \estep\hspace{-2.4pt} \sq\pr_{i}$ for all $i\hspace{0.5pt} > \mu$.
Thus, $\steqv{\sq\pr\hspace{-0.9pt}}{\hspace{-0.5pt}\sq}$ and
$\lcondT{\sq\pr\hspace{-0.5pt}}{\hspace{-1.9pt}P\hspace{-0.5pt}}{\hspace{0.5pt}p\hspace{0.7pt}}{\hspace{0.7pt}Q\hspace{0.3pt}}{\hspace{0.3pt}\mu}$ hold by construction. 
Since the position $0$ of $\progOf{\sq_\mu}$ is always available and $\sq$ is fair, we further have some $n > \mu$
with the program step $\sq_{n-1}\stsp \pstep\hspace{-2pt} \sq_n$ 
 such that $\progOf{\sq_{n-1}} = \skipp;q$ and, consequently, $\progOf{\sq_{n}}\hspace{-0.5pt} = q$.
 As nothing else than $\stateOf{\sq_n}\hspace{-0.7pt}  \notin\hspace{-0.9pt} Q$ may in general be asserted regarding the state $\stateOf{\sq_n}$,
 we conclude $\lcondN{\suffix{n}{\sq}}{\!\top\hspace{-0.5pt}}{\hspace{0.9pt}q\hspace{0.1pt}}{\hspace{0.7pt}Q}$.

 In (2) we have $\sq\hspace{-0.2pt} \in\hspace{-0.2pt}  \pcsi{p;q}$ and $\progOf{\sq_n} = \skipp$. 
 Hence, there must be the least index $\mu < n$ with $\progOf{\sq_\mu} = \skipp;q$.
 Then for any $i\hspace{1pt}\le\hspace{0.5pt} \mu$ there is some $u_i$ such that $\progOf{\sq_i} = u_i;q$ holds and
we once more define $\sq\pr_i \defeq (u_i, \stsp\stateOf{\sq_i})$ for all $i\hspace{1pt}\le\hspace{0.5pt} \mu$ and $\sq\pr_i \defeq (\skipp,\stsp \stateOf{\sq_i})$
with $\sq\pr_{i-1}\hspace{1pt} \estep\hspace{-1.9pt} \sq\pr_{i}$ for all $i\hspace{1pt}>\hspace{-0.2pt} \mu$. Thus, $\steqv{\sq\pr\hspace{-0.5pt}}{\hspace{-0.2pt}\sq}$ and
$\lcondT{\sq\pr\hspace{-0.9pt}}{\hspace{-2.5pt}P\hspace{-1pt}}{\hspace{0.2pt}p\hspace{0.1pt}}{\hspace{0.7pt}Q\hspace{-0.1pt}}{\hspace{0.3pt}\mu}$
hold by construction.
Further, there is some $l$ with $\mu\hspace{-0.1pt} <\hspace{0.1pt} l\hspace{0.1pt} \le\hspace{0.1pt} n$
such that $\progOf{\sq_{l}}\hspace{-1pt} = \hspace{-1pt}q$ 
for otherwise we would have $\progOf{\sq_{n}} = \skipp;q$. 
We therefore conclude $\stdsp\lcondT{\suffix{l}{\sq}\hspace{0.5pt}}{\hspace{-2pt}\top\hspace{-0.7pt}}{\hspace{0.95pt}q\hspace{0.4pt}}{\hspace{0.9pt}Q\hspace{0.25pt}}{\hspace{1pt}n-\hspace{0.9pt}l}$.
\end{proof}

Next proposition handles the cases with the parallel composition but, for the sake of simplicity, of two components only.
A very similar argumentation however leads to a proof of 
the general rule:
\emph{if $\lcondN{\sq\hspace{-0.7pt}}{\hspace{-1.55pt}P\hspace{1.7pt}}{\hspace{-0.9pt}\Parallel{\hspace{-1.5pt}p_1, \ldots, p_m}\hspace{-0.7pt}}{\hspace{0.4pt}Q}$ then there is a function sending each $k\hspace{-0.1pt} \in \hspace{-1pt}\{1, \ldots, m\}$ to a computation
$\sq^{p_k}$ such that}
\begin{enumerate}
\item[(i)] $\steqv{\sq^{p_k}\hspace{-2.1pt}}{\hspace{-0.4pt}\sq}$ \emph{for all} $k\hspace{0.2pt} \in \hspace{-0.4pt}\{1, \ldots, m\}$,
\item[(ii)] 
  $\lcond{\sq^{p_k}\hspace{-1.7pt}}{\hspace{-1pt}P\hspace{-0.5pt}}{\hspace{0.4pt}p_k\hspace{-0.9pt}}{\hspace{0.5pt}Q}\stsp$ \emph{for all} $k\hspace{0.1pt} \in \hspace{-0.4pt}\{1, \ldots, m\}$ \emph{and} 
\item[(iii)] \emph{there exists some} $k \in\hspace{-0.5pt} \{1, \ldots, m\}$ with
  $\lcondN{\sq^{p_k}\hspace{-1pt}}{\hspace{-1.2pt}P\hspace{-0.7pt}}{\hspace{0.5pt}p_k\hspace{-0.5pt}}{\hspace{0.7pt}Q}$.
\end{enumerate}
For $m\hspace{-1pt} = 2$, (ii) is the same as $\lcond{\sq^{p_1}\hspace{-2.1pt}}{\hspace{-1.5pt}P\hspace{-1pt}}{\hspace{0.2pt}p_1\hspace{-1.1pt}}{\hspace{0.2pt}Q} \wedge
\lcond{\sq^{p_2}\hspace{-1.9pt}}{\hspace{-1.5pt}P\hspace{-1pt}}{\hspace{0.2pt}p_2\hspace{-0.7pt}}{\hspace{0.2pt}Q}$
which by Corollary~\ref{thm:split-eq} is in turn equivalent to the formula
\[
\begin{array}{l}
(\lcondN{\sq^{p_1}}{\hspace{-1.75pt}P\hspace{-0.5pt}}{\hspace{0.7pt}p_1\hspace{0.1pt}}{\hspace{1pt}Q}\hspace{0.1pt} \vee
(\exists n.\hspace{2.5pt} \lcondT{\sq^{p_1}}{\hspace{-1.7pt}P\hspace{-0.5pt}}{\hspace{0.7pt}p_1\hspace{0.1pt}}{\hspace{1pt}Q\hspace{0.2pt}}{\hspace{0.9pt}n})) \; \wedge \\
(\lcondN{\sq^{p_2}}{\hspace{-1.75pt}P\hspace{-0.5pt}}{\hspace{0.7pt}p_2\hspace{0.1pt}}{\hspace{1pt}Q}\hspace{0.1pt} \vee
  (\exists n.\hspace{2.5pt} \lcondT{\sq^{p_2}}{\hspace{-1.7pt}P\hspace{-0.5pt}}{\hspace{0.7pt}p_2\hspace{0.1pt}}{\hspace{1pt}Q\hspace{0.2pt}}{\hspace{0.9pt}n})). 
\end{array}  
\]
Transforming it using the distributive law
results in the three disjuncts encountered as (a),(b) and (c) in the part (1) of Proposition~\ref{thm:parallel2-live} whereas the fourth
\[\lcondT{(\exists n.\hspace{2.5pt} \sq^{p_1}\hspace{-1pt}}{\hspace{-1.7pt}P\hspace{-0.9pt}}{\hspace{0.5pt}p_1\hspace{-1.1pt}}{\hspace{0.5pt}Q\hspace{-0.2pt}}{\hspace{0.5pt}n}) \hspace{0.9pt}\wedge\hspace{0.5pt}
\lcondT{(\exists n. \hspace{2.5pt} \sq^{p_2}\hspace{-0.7pt}}{\hspace{-1.5pt}P\hspace{-0.75pt}}{\hspace{0.5pt}p_2\hspace{-0.5pt}}{\hspace{0.5pt}Q}{\hspace{0.7pt}n})\]
contradicts (iii) and gets therefore discarded. The condition (iii) can thus disappear as well because it is implied by each of the disjuncts (a), (b) and (c).

Generally, if $\sq\hspace{-1pt} \in\hspace{-1pt} \pcsf{\hspace{1pt}\Parallel{\hspace{-2pt}p_1, \ldots, p_m}\hspace{0.7pt}}$
does not reach $\skipp$ then there must be
non-terminating components $p_{i_1},\ldots,p_{i_k}$ with $1 \le i_1 \hspace{-1.5pt}< \ldots < i_k\hspace{-1.5pt} \le m$ and $0\hspace{-0.5pt} < k$,
such that all $p_j$ with $j\hspace{-0.1pt} \notin\hspace{-0.1pt} \{i_1,\ldots,i_k\}$ can accordingly be assumed as terminating.
This amounts to $\sum_{k=1}^m {m \choose k}\hspace{-0.7pt} = 2^m\hspace{-1.7pt} - 1$ branches (\ie disjuncts) to be refuted
in order to refute $\lcondN{\sq\hspace{-0.5pt}}{\hspace{-2.1pt} P\hspace{1.5pt}}{\hspace{-1pt}\Parallel{\hspace{-1.25pt}p_1, \ldots, p_m}\hspace{-0.7pt}}{\hspace{1.2pt}Q}$
explaining also the three disjuncts for $m =\hspace{0.1pt} 2$.
As will become apparent in Chapter~\ref{S:PM3}, with $\op{mutex}$ we indeed will have to follow all of the three branches.
There is however a notable class of programs
with `independent' parallel components whose non-termination can be refuted independently of other components behaviour.
More precisely, in such cases 
it becomes possible to refute
$\bigvee_{i=1}^{m}\steqv{\sq^{p_i}\hspace{-2.1pt}}{\hspace{-0.5pt}\sq}\hspace{0.4pt} \wedge\hspace{-0.1pt} \lcondN{\sq^{p_i}\hspace{-1.7pt}}{\hspace{-1.7pt}P\hspace{-0.5pt}}{\hspace{0.5pt}p_i\hspace{-0.7pt}}{\hspace{0.79pt}Q}$, \ie $m$ branches.

The part (2) of Proposition~\ref{thm:parallel2-live} is extended to $m > 2$ components as follows: 
if $\lcondT{\sq\hspace{-1pt}}{\hspace{-2.1pt}P\hspace{1.4pt}}{\hspace{-1.5pt}\Parallel{\hspace{-2.1pt}p_1, \ldots, p_m}\hspace{-1.2pt}}{\hspace{0.57pt}Q\hspace{-0.32pt}}{\hspace{0.25pt}n}$
then for all $i\hspace{-0.2pt} \in\hspace{-0.9pt} \{1,\ldots,m\}$ 
there exist
$\steqv{\sq^{p_i}\hspace{-2.7pt}}{\hspace{-1pt}\sq}$ and $n_i <\hspace{0.5pt} n$
such that $\lcondT{\sq^{p_i}\hspace{-1.5pt}}{\hspace{-1.7pt} P\hspace{-0.7pt}}{\hspace{0.29pt}p_i\hspace{-0.7pt}}{\hspace{0.4pt}Q}{\hspace{0.4pt}n_i}$.
\begin{lemma}\label{thm:parallel2-live}
The following implications hold:
\begin{enumerate}
\item[\emph{(1)}] \hspace{-2.9pt}if $\lcondN{\sq\hspace{-0.5pt}}{\hspace{-2pt}P\hspace{0.29pt}}{\hspace{1.2pt}p\hspace{-0.94pt} \parallel\hspace{-0.9pt} q\hspace{1.1pt}}{\hspace{1.7pt}Q}$ then there exist $\steqv{\sq^p\hspace{-1.5pt}}{\stnsp\sq}$ and
  $\steqv{\sq^q\hspace{-1.5pt}}{\stnsp\sq}$ such that either
  \begin{enumerate}
  \item[\emph{(a)}] $\lcondN{\sq\mystrut^p\hspace{-1pt}}{\hspace{-1.7pt}P\hspace{-0.75pt}}{\hspace{0.2pt}p\hspace{0.1pt}}{\hspace{0.25pt}Q}$ and
                    $\lcondN{\sq\mystrut^q\hspace{-1pt}}{\hspace{-1.7pt}P\hspace{-0.71pt}}{\hspace{0.32pt}q\hspace{-0.15pt}}{\hspace{0.29pt}Q}$, or
  \item[\emph{(b)}] $\lcondN{\sq\mystrut^p\hspace{-1pt}}{\hspace{-1.7pt}P\hspace{-0.79pt}}{p}{\hspace{0.2pt}Q}$ and there is some $n$ with
                    $\lcondT{\sq\mystrut^q\hspace{-1pt}}{\hspace{-1.7pt}P\hspace{-0.75pt}}{\hspace{0.4pt}q\hspace{0.15pt}}{\hspace{0.2pt}Q\hspace{0.2pt}}{\hspace{0.4pt}n}$, or
  \item[\emph{(c)}] $\lcondN{\sq\mystrut^q\hspace{-1pt}}{\hspace{-1.7pt}P\hspace{-0.79pt}}{\hspace{0.4pt}q\hspace{0.05pt}}{\hspace{0.25pt}Q}$ and there is some $n$ with
                    $\lcondT{\sq\mystrut^p\hspace{-1.5pt}}{\hspace{-1.7pt}P\hspace{-0.9pt}}{\hspace{0.2pt}p}{\hspace{0.4pt}Q\hspace{0.2pt}}{\hspace{0.4pt}n}$;
   \end{enumerate}
\item[\emph{(2)}] \hspace{-1.5pt}if $\lcondT{\sq\hspace{-1.1pt}}{\hspace{-2.4pt}P\hspace{-1pt}}{\hspace{0.2pt} p\hspace{-1.75pt} \parallel\hspace{-1.5pt} q\hspace{-0.15pt}}{\hspace{0.39pt}Q\hspace{-0.3pt}}{\hspace{0.39pt}n}$ then there exist $\steqv{\sq^p\hspace{-2pt}}{\stnsp\sq}$, $\steqv{\sq^q\hspace{-2pt}}{\stnsp\sq}$, $n_1\hspace{-1.5pt} < n$ and $n_2 < n$ such that
  $\lcondT{\sq^p\hspace{-1pt}}{\!P\hspace{-0.7pt}}{\hspace{0.4pt}p\hspace{0.29pt}}{\hspace{0.55pt}Q\hspace{0.2pt}}{\hspace{0.7pt}n_1}$ and
  $\lcondT{\sq^q\hspace{-0.7pt}}{\!P\hspace{-0.7pt}}{\hspace{0.4pt}q\hspace{0.19pt}}{\hspace{0.55pt}Q\hspace{0.2pt}}{\hspace{0.7pt}n_2}$.
\end{enumerate}
\end{lemma}
\begin{proof}
In (1) we have a computation $\sq\hspace{-0.2pt} \in\hspace{-0.4pt} \pcsf{p\hspace{-1.2pt}\parallel\hspace{-1.2pt} q}$ that does not reach a $\skipp$-configuration.
Now, if there would be some $u, v, n_1, n_2$ with $\progOf{\sq_{n_1}} = \skipp \parallel v$ and
$\progOf{\sq_{n_2}} = u\hspace{-0.5pt} \parallel\hspace{-0.5pt} \skipp$
then we could infer $\progOf{\sq_m} = \skipp\hspace{-0.5pt} \parallel\hspace{-0.5pt} \skipp$ where $m$ is the maximum of $n_1$ and $n_2$,
and hence by fairness of $\sq$ we could further obtain some $n\hspace{-0.2pt} >\hspace{-1.2pt} m$
with $\progOf{\sq_n}\hspace{-0.2pt} =\hspace{-0.2pt} \skipp$ in contradiction to the assumption on $\sq$.
As a result, for all $i\hspace{0.7pt}\in\hspace{-0.9pt}\naturals$ we have some $u_i$ and $v_i$ such that $\progOf{\sq_i}\hspace{-0.5pt} = u_i\hspace{-1.9pt}\parallel\hspace{-1pt} v_i$ holds,
and define $\sq^p_i \hspace{-0.2pt}\defeq\hspace{-0.5pt} (u_i,\stsp \stateOf{\sq_i})$  
with $\sq\mystrut^p_i \hspace{0.5pt}\pstep\hspace{-2.4pt}\sq\mystrut^p_{i+1}$ if there is $\sq_i\hspace{0.7pt} \pstep\hspace{-2.4pt} \sq_{i+1}$
having the fired position of the form $1x$ where $x\hspace{0.4pt} \in\hspace{-0.79pt} \pos\hspace{2pt}u_i$, 
and $\sq\mystrut^p_i \stsp\estep\hspace{-2.5pt} \sq\mystrut^p_{i+1}$ otherwise.
Thus, $\sq^p\hspace{-1.7pt} \in\hspace{-0.9pt} \pcsi{p}$ and $\steqv{\sq^p\hspace{-3pt}}{\hspace{-1.9pt}\sq}$ are provided by this construction.
To establish that $\sq^p$ is fair,
let $i\hspace{0.2pt}\in\hspace{-0.5pt}\naturals$ and $x\hspace{-0.5pt} \in\hspace{-1pt} \pos\hspace{1.9pt}u_i$ be an always
available position.
Then $1x$ is an always available position of $u_i\hspace{-2.4pt} \parallel\hspace{-1.5pt} v_i$, and since $\sq$ is fair
there is some $j\hspace{0.5pt} \ge\hspace{0.5pt}  i$ with
$(u_j\hspace{-1.9pt} \parallel\hspace{-0.7pt}  v_j, \stsp\sigma)\hspace{0.7pt} \pstep\hspace{-0.5pt} (u_{j+1}\hspace{-1pt} \parallel\hspace{-0.5pt}  v_{j+1},\stsp \sigma\pr)$,
$\sigma\hspace{-0.5pt} = \stateOf{\sq_j}$, $\sigma\pr\hspace{-1.2pt}=\stateOf{\sq_{j+1}}$ such that
\begin{enumerate}
  \item[-] $\plook{u_j}{x} = \plook{(u_j\hspace{-1.2pt}  \parallel\hspace{-0.7pt}  v_j)}{1x} =  \plook{(u_i\hspace{-0.9pt}  \parallel\hspace{-0.5pt}  v_i)}{1x} = \plook{u_i}{x}$
  \item[-] $(\plook{(u_j \hspace{-1.2pt} \parallel\hspace{-0.5pt}  v_j)}{1x},\stsp \sigma) \hspace{0.5pt} \pstep\hspace{-0.5pt}  (p\pr,\stsp \sigma\pr)$
  \item[-] $u_{j+1} \hspace{-0.9pt} \parallel\hspace{-0.7pt}  v_{j+1} = \psubst{(u_j \hspace{-1.1pt} \parallel\hspace{-0.5pt}  v_j)}{p\pr}{1x} = \hspace{0.5pt}\psubst{u_j}{p\pr}{x}\hspace{-1.2pt}  \parallel\hspace{-0.7pt}  v_j$
\end{enumerate}
hold with some $p\pr$.
Then $u_{j+1}\hspace{-1.2pt} =\hspace{-0.7pt}\psubst{u_j}{p\pr}{x}$ and $(\plook{u_j}{x}, \stsp\sigma)\hspace{1pt} \pstep\hspace{-1pt} (p\pr,\stsp \sigma\pr)$,
\ie $\stdsp x$ is the fired position of $(u_j, \stsp \sigma)\stsp \pstep\hspace{0.2pt} (u_{j+1},\stsp\sigma\pr)$.
A computation $\steqv{\sq^q\hspace{-1pt}}{\hspace{-0.5pt}\sq}$ with $\sq^q\in \pcsf{q}$  is constructed likewise
-- the fired position is here of the form $2x$ where $x\hspace{0.2pt} \in\hspace{-0.5pt} \pos\hspace{2pt}v_j$. 

Thus, if $u_i\hspace{-1.2pt} \neq\hspace{-0.7pt} \skipp$ and $v_i\hspace{-1pt} \neq \skipp$ hold for all $i$ then
$\lcondN{\sq^p\hspace{-1.2pt}}{\hspace{-2.9pt} P\hspace{-0.9pt}}{\hspace{0.1pt}p\hspace{-0.15pt}}{\hspace{0.25pt} Q}$ and
$\lcondN{\sq^q\hspace{-0.7pt}}{\hspace{-1.9pt}P\hspace{-0.55pt}}{\hspace{0.79pt}  q\hspace{0.45pt}}{\hspace{0.79pt} Q}$ follow by the above constructions.
If we have $u_i\hspace{-1.1pt} \neq\hspace{-0.5pt} \skipp$ for all $i$ but there is some $j$ such that $v_j\hspace{-1pt} =\hspace{-0.1pt} \skipp$,
we conclude
$\lcondN{\sq^p\hspace{-1pt}}{\hspace{-1.9pt}P\hspace{-0.9pt}}{\hspace{0.1pt}p\hspace{-0.2pt}}{\hspace{0.4pt}Q}\stsp$ and
$\lcondT{\sq^q\hspace{-1pt}}{\hspace{-2.2pt}P\hspace{-1.1pt}}{\hspace{0.4pt}q\hspace{-0.2pt}}{\hspace{0.5pt}Q\hspace{-0.2pt}}{\hspace{0.5pt}n}\stsp$
where $n$ is the first index with $v_n\hspace{-1.7pt} =\hspace{-0.4pt}\skipp$.
The last case with $v_i \neq \skipp$ for all $i$ and some $j$ with $u_j\hspace{-1pt} = \skipp$ is handled accordingly.

 In (2) we assume $\sq\hspace{-0.2pt} \in\hspace{-0.2pt}  \pcsi{p\parallel q}$ and $\progOf{\sq_n}\hspace{-0.5pt} = \skipp$ where $n$ is the least such index.
 Then for all $i<n$ we have $\progOf{\sq_i}\hspace{-0.5pt} = u_i\hspace{-1.5pt} \parallel\hspace{-1pt} v_i$ with some $u_i, v_i$.
Furthermore, $\progOf{\sq_{n-1}} = \skipp\!\parallel \!\skipp$ and $\sq_{n-1}\stsp \pstep\hspace{-2.5pt} \sq_n$ must hold as well.
Consequently, there exist the least index $n_1\hspace{-1.5pt}<n$ with $u_{n_1}\hspace{-2.5pt} = \skipp$
and the least index $n_2\hspace{-1pt}<n$ with $v_{n_2}\! = \skipp$.
We can thus define the computation $\sq\mystrut^p$ by
$\sq\mystrut^p_i \defeq (u_i,\stsp \stateOf{\sq_i})$ for all $i\hspace{0.5pt}\le\hspace{-0.5pt} n_1$ and $\sq\mystrut^p_i \hspace{-0.9pt}\defeq\hspace{-0.5pt} (\skipp,\stdsp \stateOf{\sq_i})$
with $\sq^p_{i-1}\stsp \estep\hspace{-2.9pt} \sq^p_{i}$ for all $i>\hspace{-1pt} n_1$, such that $\steqv{\sq^p\hspace{-2pt}}{\stnsp\sq}$ and
$\lcondT{\sq^p\hspace{-1.2pt}}{\hspace{-1.95pt}P\hspace{-1pt}}{\hspace{0.4pt}p\hspace{0.1pt}}{\hspace{0.75pt}Q\hspace{0.05pt}}{\hspace{0.5pt}n_1}$
 follow by this construction.
 Finally,
 a computation $\steqv{\sq^q\stnsp}{\stnsp\sq}$ with $\lcondT{\sq^q\hspace{-0.1pt}}{\hspace{-1.5pt}P\hspace{-1pt}}{\hspace{0.57pt}q\hspace{0.25pt}}{\hspace{0.59pt}Q\hspace{0.25pt}}{\hspace{0.4pt}n_2}\stsp$ is defined likewise. 
\end{proof}

Shifting lastly the attention to while-statements, from an assumption of the form
$\lcondT{\sq\hspace{-0.4pt}}{\hspace{-1.7pt}P\hspace{-0.9pt}}{\hspace{0.79pt}\while{\hspace{1pt} C\stnsp}{p}{q}\stsp}{\hspace{0.79pt}Q\hspace{0.4pt}}{\hspace{0.7pt}n}$
we infer that
there is some $i\le n$ with $\lcondT{\suffix{i}{\sq}}{\!\neg C\hspace{-0.5pt}}{\hspace{0.59pt}q\hspace{0.29pt}}{\hspace{1.1pt}Q\hspace{0.49pt}}{\hspace{1pt}n - \hspace{0.5pt}i}$. That is, $\sq$
proceeds to $q$ in $i \le n$ steps reaching thereby a state in $\neg C$ from which it reaches a $\skipp$-configuration in $n -\stsp i$ steps.
Note that such $\sq$ in general does not have to visit a state in $C$.

With $\lcondN{\sq\hspace{-1pt}}{\hspace{-2.9pt}P\hspace{-1.2pt}}{\hspace{0.2pt}\while{\hspace{1pt} C\hspace{-1.2pt}}{\hspace{-1pt}p\hspace{-0.5pt}}{q\hspace{-1pt}}\hspace{0.2pt}}{\hspace{0.4pt}Q}$
we can by contrast infer that $\sq$ passes
through the set of states satisfying $C$ infinitely many times
whenever all fair computations of $p$ and $q$, starting respectively from $C$ and $\neg C$, reach a $\skipp$-configuration.
This is captured by the following statement
where strictly ascending infinite sequence of naturals in this setting means a function $\phi$ that maps each $i$ of type $\ty{nat}$ to
a value of the same type, for convenience denoted by $\phi_i$ instead of $\phi\hspace{2.7pt} i$, such that $\phi_n\hspace{-1pt} < \phi_m$ holds for all $n, m$ with $n < m$.
\begin{lemma}\label{thm:while-live}
If $\stsp\lcondN{\sq\hspace{-1.2pt}}{\hspace{-2.9pt}P\hspace{-1pt}}{\hspace{0.2pt}\while{\hspace{1pt}C\hspace{-1.5pt}}{\hspace{-1.2pt} p\hspace{-0.7pt}}{q\stnsp}\hspace{0.2pt}}{\hspace{0.2pt}Q}\stsp$ 
then either
\begin{enumerate}
\item[\hspace{-2cm}\emph{(1)}] 
  there exists a strictly ascending infinite sequence of natural numbers $\phi$ such that for any $n\hspace{0.2pt}\in\hspace{-0.5pt}\naturals$ there exist some 
  $\steqv{\sq\pr\hspace{-1.2pt}}{\hspace{-0.7pt}\suffix{\phi_n}{\sq}}\stsp$ and $\stsp m\hspace{-0.4pt}< \phi_{n+1} -\hspace{0.5pt} \phi_n$
  with $\lcondT{\sq\pr}{\stnsp C\hspace{-0.5pt}}{\hspace{0.4pt}p\hspace{0.5pt}}{\hspace{0.9pt}Q\hspace{0.1pt}}{\hspace{0.7pt}m}$, or
\item[\emph{(2)}] there exist some $n\hspace{0.2pt}>0\stsp$ and $\steqv{\sq\pr\stnsp}{\hspace{-0.5pt}\suffix{n}{\sq}}\stsp$ such that
  $\lcondN{\sq\pr\hspace{-0.5pt}}{\hspace{-1.2pt} C\hspace{-0.7pt}}{\hspace{0.2pt}p\hspace{0.2pt}}{\hspace{0.4pt}Q}$, or
\item[\emph{(3)}] there exists some $n > 0\stsp$ such that $\stsp\lcondN{\suffix{n}{\sq}}{\hspace{-2pt}\neg C\hspace{-0.5pt}}{\hspace{0.1pt}q}{\hspace{0.4pt}Q}$.
\end{enumerate}
\end{lemma}
\begin{proof}
  Let $p_{\com{while}}$ abbreviate $\while{\hspace{0.5pt} C\hspace{-2pt}}{\hspace{-0.9pt}p\hspace{-0.9pt}}{\hspace{-0.9pt}q\hspace{-0.9pt}}$ 
  and assume we have $\sq\hspace{-0.2pt} \in\hspace{-0.2pt} \pcsf{p_{\com{while}}}$ which
starts with a state in $\hspace{-1pt}P$, has $\neg Q$ as a state invariant and 
  does not reach a $\skipp$-configuration.
Moreover, we will assume the respective negations of (2) and (3) (note that the invariant $\neg Q$ is omitted because it is inherited from $\sq$ in both cases anyway): 
\begin{enumerate}
\item[(a)] \emph{for any $\sq\pr \hspace{-0.5pt}\in\hspace{-0.1pt} \pcsf{p}\hspace{-1pt} \cap\hspace{1.4pt} \inC{\hspace{0.55pt} C}$ and $n>\hspace{-0.5pt}0$ such that
$\stsp\steqv{\sq\pr\hspace{-1.2pt}}{\hspace{-1.2pt}\suffix{n}{\sq}}$ holds
           there exists $j\hspace{0.2pt}\in\hspace{-0.2pt}\naturals\stsp$ with $\progOf{\sq\pr_j} = \skipp$},
\item[(b)] \emph{for any $n\hspace{-0.5pt} >\hspace{-0.5pt} 0$ such that
$\stsp\suffix{n}{\sq}\hspace{-0.2pt} \in\hspace{-0.7pt} \pcsf{q} \hspace{-2.1pt}\cap\hspace{-0.05pt} \inC{\hspace{-2pt}(\neg C)}$ holds
there exists $m\hspace{-0.5pt}\in\hspace{-0.5pt}\naturals\stsp$ with $\progOf{\sq_{n+m}} = \skipp$}           
\end{enumerate}
in order to derive (1) constructively. To this end we first show that
for any $i\hspace{1.2pt}\in\hspace{-0.1pt}\naturals$ with a step $\sq_i\hspace{1pt} \pstep\hspace{-1.05pt} \sq_{i+1}\hspace{0.3pt}$ where
$\progOf{\sq_i}\hspace{-0.5pt} =\hspace{-0.5pt} p_{\com{while}}$ and $\stateOf{\sq_i}\hspace{-0.2pt} \in\hspace{-0.2pt} C$
there exists some $j\hspace{0.2pt} > i$ with a step $\sq_j\hspace{0.7pt} \pstep\hspace{-2.5pt} \sq_{j+1}\hspace{0.3pt}$ where
$\progOf{\sq_j}\hspace{-1pt} =\hspace{-1pt} p_{\com{while}}$ and $\stateOf{\sq_j}\hspace{-1pt} \in\hspace{-1.2pt} C$.
Notice that $\sq_{i+1} = (p;\skipp;p_{\com{while}},\stdsp\stateOf{\sq_i})$ in particular follows for such $i$. 

Further,
suppose
there is no $k > i$ with $\progOf{\sq_{k}} = \skipp;\skipp;p_{\com{while}}$. 
Then for all $k > i$ we have some $u_k \neq \skipp$ with $\progOf{\sq_{k}} = u_k;\skipp;p_{\com{while}}$, and 
define $\sq\pr_l \vspace{0.1pt}\defeq (u_{i+l+1},\stdsp \stateOf{\sq_{i+l+1}})$ for all $l$ 
inheriting the respective transitions from $\sq$.
Thus,  $\sq\pr\hspace{-0.5pt} \in\hspace{-0.2pt} \pcsf{p} \hspace{-1.7pt}\cap\hspace{0.7pt} \inC{\hspace{-0.25pt} C}$ and
$\steqv{\sq\pr\hspace{-1.2pt}}{\hspace{-1.2pt}\suffix{i+1}{\sq}}$ hold by construction.
Resorting to (a) we obtain some $j$ with $\progOf{\sq\pr_j} = u_{i+j+1} = \skipp\stdsp$, \ie a contradiction.

Then let $k > i$ be an index with $\progOf{\sq_{k}} = \skipp;\skipp;p_{\com{while}}$.
Since $\sq$ is a fair computation, we get some $l_1\stnsp > k$ with $\progOf{\sq_{l_1}} = \skipp;p_{\com{while}}$
           and further some $l_2 > l_1$ with $\progOf{\sq_{l_2}} = p_{\com{while}}$.
 Once more,
           since the position $0$ of $\progOf{\sq_{l_2}}$ is always available and $\sq$ is fair, there is also 
           some $j \ge\hspace{0.2pt} l_2$ with $\sq_j\hspace{1.2pt} \pstep\hspace{-2.9pt} \sq_{j+1}$ and $\progOf{\sq_j} = p_{\com{while}}$.
           In case $\stateOf{\sq_j} \hspace{-0.5pt}\notin\hspace{-0.2pt} C$ 
           we would have 
           $\suffix{j+1}{\sq} \in\hspace{-0.7pt} \pcsf{q}\! \cap\hspace{0.7pt} \inC{\hspace{-2pt}(\hspace{-0.9pt}\neg C)}$ 
           and using (b) obtain 
some $m$ with $\progOf{\sq_{j+m+1}} = \skipp$ in contradiction to the assumption that $\sq$ does not reach a $\skipp$-configuration.           

           As the intermediate result we obtain a function $\psi$ on the natural numbers 
           such that $\psi\hspace{2.7pt} n\hspace{0.4pt} >\hspace{-0.1pt} n$, $\sq_{\psi\hspace{1.5pt} n}\hspace{-0.7pt} \pstep\hspace{-2.5pt} \sq_{\psi\hspace{1.5pt} n \stsp + \stsp 1}$, $\progOf{\sq_{\psi\hspace{1.5pt}n}}\stnsp =\stnsp p_{\com{while}}$ and
           $\stateOf{\sq_{\psi\hspace{1.2pt}n}}\hspace{-0.9pt} \in\hspace{-0.7pt} C$
           hold for any $n$ 
           with $\sq_n \hspace{0.2pt}\pstep\hspace{-2.29pt} \sq_{n+1}$, $\progOf{\sq_n}\hspace{-0.5pt} =\hspace{-0.5pt} p_{\com{while}}$ and
           $\stateOf{\sq_n}\hspace{-0.5pt} \in\hspace{-0.5pt} C$.
           Moreover, note that an initial
           index $n_0$ with $\sq_{n_0}\hspace{-1.5pt} \pstep\hspace{-0.9pt} \sq_{n_0+1}$ and $\progOf{\sq_{n_0}}\hspace{-0.5pt} = p_{\com{while}}$
exists since $\sq$ is fair,
whereas $\stateOf{\sq_{n_0}}\hspace{-0.2pt} \in\hspace{-0.2pt} C$ must additionally hold for otherwise we would have
$\suffix{n_0+1}{\sq}\hspace{-0.7pt} \in\hspace{-0.7pt} \pcsf{q}\hspace{-2pt} \cap \hspace{0.1pt}\inC{\hspace{-2pt}(\hspace{-0.5pt}\neg C)}$
and hence
a contradiction using (b) as in the above paragraph.
A strictly ascending sequence $\phi$ is thus defined by $\phi_i \defeq 1 + \psi^i\hspace{1.2pt} n_0$.

Let $n$ be an arbitrary natural number for the remainder of the proof.
By the definition of $\phi$ we have $\progOf{\sq_{\phi_n}} = p;\skipp;p_{\com{while}}$, 
$\progOf{\sq_{\phi_{n+1} - 1}} = p_{\com{while}}$ and $\stateOf{\sq_{\phi_n}}\hspace{-1.9pt} \in\hspace{-1.7pt} C$.
This in particular means that 
there must be the least index $\mu$ such that $\phi_n \le \mu < \phi_{n+1}$ and
$\progOf{\sq_\mu} = \skipp;\skipp;p_{\com{while}}$ hold.
Then
for any $i$ with $\phi_n \le i \le \mu$ there is some $u_i$ with
$\progOf{\sq_i}\hspace{-0.5pt} =\hspace{-0.5pt} u_i;\skipp;p_{\com{while}}$.
Thus, we can define $\sq\pr_i \defeq (u_{\phi_n + i}, \stsp\stateOf{\sq_{\phi_n + i}})$
for all $i\le \mu - \phi_n\stsp$
and $\sq\pr_i \defeq (\skipp, \stsp\stateOf{\sq_{\phi_n + i}})$ for all $i>\hspace{-0.9pt} \mu - \phi_n$.
The computation $\sq\pr$ meets the conditions $\steqv{\sq\pr\hspace{-1.7pt}}{\hspace{-1pt}\suffix{\phi_n\stnsp}{\sq}}\stsp$ and
$\lcondT{\sq\pr\hspace{-0.5pt}}{\hspace{-0.7pt} C\hspace{-0.7pt}}{\hspace{0.4pt}p\hspace{0.4pt}}{\hspace{0.7pt}Q\hspace{-0.1pt}}{\hspace{0.5pt}\mu - \phi_n}$ by construction.
\end{proof}

Summing up, an assumption of the form
$\lcond{\sq\hspace{-0.9pt}}{\hspace{-1.55pt} P\hspace{-1pt}}{\hspace{-0.05pt} p\hspace{-0.1pt}}{\hspace{0.4pt}Q}$,
where 
$p$ does not contain any await-statements beyond atomic sections,
is first replaced by the two disjuncts in Corollary~\ref{thm:split-eq}
which in turn get processed by successive application of the presented syntax-driven rules down to
a series of branches to be closed, \ie refuted. In each of these we will have access to some additional piece of event/control flow information, 
which comprises a collection of
specific to the particular branch state properties exhibited by hypothetical computations running through 
the same states as certain suffixes of the initially fixed computation $\sq$.
This essentially allows us to project all this information back onto $\sq$
evincing thereby the relative order of events whenever it is present. 
For example, if a program starts with the assignment $\sv{a}\hspace{0.1pt} :=\hspace{-2pt} \op{True}$ followed by the conditional
$\ite{\stdsp\sv{b}\hspace{0.5pt} =\hspace{0.1pt} 1}{q}{r}$
then
one part of the branches will exhibit the properties 
$\sigma\sv{a}\stsp$ and $\sigma\pr\sv{b} \hspace{0.5pt}=\hspace{-0.5pt} 1$ with
$\sigma\hspace{-1pt} =\hspace{-0.5pt} \stateOf{\sq_1}$, $\sigma\pr\hspace{-1.2pt} = \hspace{-0.5pt}\stateOf{\sq_n}$ and $1 < n$,
whereas the other part -- the properties
$\sigma\sv{a}\stsp$ and $\sigma\pr\sv{b} \hspace{0.5pt}\neq\hspace{-0.5pt} 1$ with
$\sigma\hspace{-0.7pt} =\hspace{-0.2pt} \stateOf{\sq_1}$, $\sigma\pr\hspace{-1pt} = \hspace{-0.2pt}\stateOf{\sq_n}$ and $1\hspace{0.4pt} < n$.
Note that $1 < n$ reflects that $\op{True}$ is assigned to $\sv{a}$ first. 
Chapter~\ref{S:PM3} describes an application of
this technique in detail.
\section{Program correspondences and fair computations}\label{Sb:fair-corr}
The systematic replaying of computations along program correspondences, as implemented by the proof of Proposition~\ref{thm:corr-sim} handling finite cases,
similarly applies
to infinite computations yet without retaining fairness in general. 
Consider for instance the correspondence 
$\pcorrC{\!(\basic\hspace{1.9pt}f \hspace{-0.5pt}\parallel\hspace{-0.5pt} q)\hspace{-0.5pt}}{}{q}$ with a jump-free $q$.
Any $\sq^q\hspace{-0.5pt} \in \hspace{-0.2pt}\pcsf{q}$ is thus replayed by the corresponding 
$\sq\hspace{-0.1pt} \in\hspace{-0.1pt} \pcsi{\basic\hspace{1.7pt}f \hspace{-0.5pt}\parallel\hspace{0.5pt} q}$
that keeps 
the always available position $10$ from making its move once, 
\ie $\basic\hspace{1.9pt}f$ becomes so to say a `spare' component in this process.

This example underlines that retaining fairness demands additional restrictions on program correspondences.
We will focus on a \emph{componentwise} restriction particularly prohibiting such `spare' components.
To sketch the principle suppose $\pcorrC{p_1\hspace{-1.7pt}}{}{\hspace{-0.5pt}q_1}$ and $\pcorrC{\hspace{-0.5pt} p_2\hspace{-1pt}}{}{\hspace{-0.5pt}q_2}$.
Then we may require that any step from $q_1\hspace{-1.2pt} \parallel\hspace{-0.3pt} q_2$
having the fired position $ix$ where $x \in\hspace{-0.5pt} \pos\hspace{1.9pt} q_i$ and $i \stsp\in\hspace{-0.5pt}  \{1, 2\}$
is matched by a step from $p_1\hspace{-1.7pt} \parallel\hspace{-0.3pt} p_2$ with the fired position $iy$ where $y\hspace{-0.2pt} \in\hspace{-1.2pt} \pos\hspace{1.9pt}p_i$.
An exact definition of componentwise correspondence will be elaborated below.
\begin{definition}\label{def:non-block}
  A term $p$ of type $\langA{\alpha}$ will be called \emph{locally non-blocking}\index{locally non-blocking term} if for any $p$-subterm of the form
  $\await{\stsp C\stnsp}{\hspace{-0.5pt}p\pr\hspace{-1pt}}$
  we have $C \stnsp = \!\top$ and, moreover, 
  for any $\sigma$ there exists some $\sigma\pr$ such that $\rpsteps{\rho\hspace{0.2pt}}{\stnsp(p\pr,\stsp \sigma)}{\hspace{-1.5pt}(\skipp,\stsp \sigma\pr)}$.
Furthermore, a program $(\rho, p)$ is called \emph{non-blocking}\index{program!non-blocking} if
$p$ is locally non-blocking and $\rho\hspace{2.5pt} i\stsp$ is locally non-blocking
for any label $i\hspace{0.5pt} \in\hspace{-1pt} \jumps(\rho, p)$. 
\end{definition}
Note that the above definition addresses \emph{all} subterms of $p$
and not only those which are determined by the positions of $p$.
Consider for instance a jump-free $\op{p} = \basic\hspace{2pt}f ; \await{\stsp\bot}{\skipp}$.
Although the position $00$ of $\op{p}$, \ie $\basic\hspace{2pt}f$, is always available, the program is \emph{not} non-blocking
due to the `blocking' subterm $\await{\bot}{\skipp}$: any $\sq \in \pcsf{\op{p}}$ has a suffix $\suffix{n}{\sq}\hspace{0.2pt} \in \pcsf{\await{\bot}{\skipp}}$
without an available position.
By contrast, any $\progOf{\sq_i}$ is either $\skipp$ or has an always available position
when $\sq\hspace{-0.2pt} \in\hspace{-0.5pt} \rpcsf{p}{\rho}$ and $(\rho,p)$ is non-blocking.
This follows from the property that, as with the sequentiality, the non-blocking condition is retained by
the evaluation: 
for any step $\rpstep{\rho}{(p, \sigma)\stsp}{\hspace{0.2pt}(p\pr, \sigma\pr)}$
the program $(\rho, p\pr)$ is non-blocking if $(\rho, p)$ is.
\begin{definition}\label{def:componentwise}
Let $m \ge 1$ and $r$ be of type $\alpha \times \beta \Rightarrow \ty{bool}$.
Then two terms, $p$ of type $\langA{\alpha}$ and $q$ of type $\langA{\beta}$,
will be called \emph{corresponding componentwise}\index{program correspondence!componentwise} w.r.t. $\! m, r, \rho, \rho\pr$ if 
$p$ has the form $\Parallel{\hspace{-2pt} p_1,\ldots, p_m}$, $q$ has the form $\Parallel{\hspace{-1.7pt} q_1,\ldots, q_m}$ and, moreover, the conditions 
\begin{enumerate}
\item[(1)] $\pcorr{p_i\hspace{-0.9pt}}{r}{\hspace{-0.7pt}q_i}$
\item[(2)] $(\rho, p_i)$ is sequential
\item[(3)] $(\rho\pr, q_i)$ is non-blocking
\end{enumerate}
hold for all $i \in\stnsp \{1, \ldots, m\}$.
\end{definition}
Next proposition is essentially an extension of Proposition~\ref{thm:corr-sim} to infinite potential computation in presence of the `stuttering' environment.
It moreover presumes componentwise corresponding programs in order to replay fair computations as sketched above.
\begin{lemma}\label{thm:fair-corr-sim1}
Suppose $p$ and $\hspace{0.5pt}q$
  correspond componentwise w.r.t. $\hspace{-0.7pt}m, r, \rho, \rho\pr$
and let $\sq\mystrut^q\hspace{-1pt} \in\hspace{-0.5pt} \rpcsi{q}{\rho\pr}\hspace{-0.2pt} \cap\hspace{0.7pt} \envC{\hspace{-0.2pt}\op{id}}$. 
Furthermore, let $\stsp\sigma_0\stnsp$ be a state with $(\sigma_0,\stsp \stateOf{\sq^q_0}) \in\hspace{-0.2pt} r$.
Then there exists a computation
$\sq\mystrut^p \hspace{-1.9pt}\in\hspace{-0.9pt} \rpcsi{p}{\rho}\hspace{-0.5pt} \cap \envC{\hspace{-0.7pt}\op{id}}$
such that $\stateOf{\sq\mystrut^p_0}\hspace{-0.8pt} =\hspace{-0.5pt} \sigma_0$ and the following conditions hold for all $n \in\hspace{-0.5pt} \naturals$\emph{:}
\begin{enumerate}
\item[\emph{(1)}] either $\progOf{\sq^p_n}$ and $\progOf{\sq^q_n}$ correspond componentwise w.r.t. $\hspace{-2.5pt} m, r, \rho, \rho\pr$ or both are $\skipp$,
\item[\emph{(2)}] $(\stateOf{\sq^p_n}, \stdsp\stateOf{\sq^q_n}) \in\hspace{0.2pt} r$,
\item[\emph{(3)}] if $n\hspace{0.2pt}>0\stsp$ and $\stsp\rpstep{\rho\pr}{\hspace{-0.5pt}\sq\mystrut^q_{n-1}\hspace{0.2pt}}{\hspace{-1.5pt}\sq\mystrut^q_{n}}\stsp$
  has the fired position $kx$ 
  where $1\le k \le m$ then there exists some $x\pr$ such that
      $\rpstep{\rho\hspace{0.2pt}}{\hspace{-0.5pt}\sq\mystrut^p_{n-1}\hspace{0.1pt}}{\hspace{-2pt}\sq\mystrut^p_{n}}$ holds with the fired position $kx\pr$.
\end{enumerate}
\end{lemma}
\begin{proof}
We will show that to any finite computation $\sq\hspace{-0.2pt} \in\hspace{-0.4pt} \rpcs{p}{\rho}\hspace{0.4pt} \cap\hspace{1.2pt} \envC{\hspace{0.1pt}\op{id}}$ 
satisfying the conditions (1)--(3) for all $n < |\sq|$ there exists a configuration extending $\sq$ to
$\sq\pr\hspace{-0.7pt} \in\hspace{-0.2pt} \rpcs{p}{\rho}\hspace{0.4pt} \cap\hspace{1.1pt} \envC{\hspace{0.2pt}\op{id}}$ that in turn satisfies (1)--(3) for all $n \le |\sq| = |\sq\pr| - 1$.
Using this argument we can construct an admissible $\sq^p\hspace{-0.7pt} \in\hspace{-0.2pt} \rpcsi{p}{\rho}\hspace{0.2pt} \cap\hspace{0.7pt} \envC{\hspace{-0.5pt}\op{id}}$ recursively as follows.
Initially, $(p, \sigma_0)$ is assigned to $\sq^p_0$. 
Next, let $n\hspace{0.5pt} \in\hspace{-0.5pt} \naturals$ and assume that
the prefix $\sq\mystrut^p_0, \ldots, \sq^p_n\stsp$ being a computation in $\rpcs{p}{\rho}\hspace{-0.2pt} \cap\hspace{0.25pt} \envC{\hspace{-0.4pt}\op{id}}$ of length $n+1$
meets the conditions (1)--(3) for all $i \le n$. Then any configuration extending the prefix in the above sense  can be assigned to $\sq\mystrut^p_{n+1}$.

Now, let $\sq\hspace{-1.2pt} \in\hspace{-1.5pt} \rpcs{p}{\rho}\hspace{-0.9pt} \cap\hspace{-0.15pt} \envC{\hspace{-1.2pt}\op{id}}$, 
$\stsp n\hspace{-0.9pt} =\hspace{-1pt} |\sq|$ and suppose that
(1)--(3) hold for all $i\hspace{-0.5pt} <\hspace{0.2pt} n$. 
Firstly, if there is an environment step $\sq\mystrut^q_{n-1} \estep\hspace{-2.5pt} \sq\mystrut^q_{n}\stsp$ then $\sq\mystrut^q_{n}\hspace{-0.9pt} =\hspace{-0.9pt} \sq\mystrut^q_{n-1}$
holds because
$\sq\mystrut^q \hspace{-0.75pt}\in\hspace{-0.1pt} \envC{\hspace{-0.2pt}\op{id}}$.
By setting $\sq_n \hspace{-0.9pt}\defeq\hspace{-0.5pt} \sq_{n-1}$ we extend $\sq$ accordingly to  
the computation $\stsp\sq_0, \ldots, \sq_{n-1}\hspace{0.2pt} \estep\hspace{-2.2pt} \sq_{n}$
satisfying (1)--(3) for all $i \le\hspace{1pt} n$.

Secondly, with a program step $\rpstep{\rho\pr\hspace{-0.5pt}}{\stnsp\sq\mystrut^q_{n-1}\hspace{-0.1pt}}{\hspace{-2.1pt}\sq\mystrut^q_{n}}$ 
we particularly have that
$\progOf{\sq_{n-1}}$ and $\progOf{\sq^q_{n-1}}$ correspond componentwise w.r.t. $\! m, r, \rho, \rho\pr$ by the condition (1)
since $\progOf{\sq^q_{n-1}}\hspace{-0.7pt} \neq\hspace{-0.5pt} \skipp$. 
This in turn means $\progOf{\sq_{n-1}}\hspace{-1pt} =\hspace{1.7pt} \Parallel{\hspace{-2pt}u_1,\ldots, u_m}$ and
$\progOf{\sq^q_{n-1}}\stnsp = \hspace{2,5pt}\Parallel{\!v_1,\ldots, v_m}$
hold with some $u_1, \ldots, u_m$ and $v_1, \ldots, v_m\stsp$ satisfying $\stsp\pcorr{u_k\hspace{-0.5pt}}{r}{\hspace{-0.75pt}v_k}\stsp$ for all $1 \le k \le m$.

Further, if $\progOf{\sq\mystrut^q_{n}}\hspace{-0.5pt} \neq\hspace{-0.5pt} \skipp$ then
the fired position of $\rpstep{\rho\pr\hspace{-0.5pt}}{\sq^q_{n-1}\hspace{0.5pt}}{\hspace{-2.5pt}\sq\mystrut^q_{n}}$ is
of the form $kx$ where $1\le k \le m$ and
$x\hspace{0.2pt} \in \hspace{-0.7pt}\pos\hspace{1.9pt}v_k$
such that
\begin{enumerate}
\item[--] $\rpstep{\rho\pr}{(\plook{v_k}{x}, \stsp\stateOf{\sq\mystrut^q_{n-1}})\stsp}{(w,\stsp \stateOf{\sq\mystrut^q_{n}})}$ and 
\item[--] $\progOf{\sq\mystrut^q_n}= \:\Parallel{\!v_1,\ldots \psubst{v_k}{w}{x} \ldots, v_m}$
\end{enumerate}
hold with some $w$,
\ie we have $\rpstep{\rho\pr\hspace{-0.5pt}}{\hspace{-1pt}(v_k,\stsp \stateOf{\sq^q_{n-1}})\hspace{0.5pt}}{\hspace{-1pt}(\psubst{v_k}{w}{x},\stsp \stateOf{\sq\mystrut^q_{n}})}$.
Using $\pcorr{u_k\stnsp}{r}{\hspace{-1.2pt} v_k}$ we derive thus 
a matching step $\rpstep{\rho\hspace{0.5pt}}{\hspace{-0.7pt}(u_k,\stsp \stateOf{\sq_{n-1}})\stsp}{\stnsp(u\pr_k,\stsp \sigma\pr)}$
with some $u\pr_k$ and $\sigma\pr$ satisfying $\pcorr{u_k\pr\hspace{-0.2pt}}{r}{\hspace {-1.1pt}\psubst{v_k}{w}{x}}$ and
$(\sigma\pr,\stsp \stateOf{\sq\mystrut^q_n})\hspace{-0.1pt} \in\hspace{0.01pt} r$.
Let $x\pr$ denote the fired position of $\rpstep{\rho\hspace{0.2pt}}{\hspace{-0.9pt}(u_k,\stsp \stateOf{\sq_{n-1}})\stsp}{\stnsp(u\pr_k,\stsp \sigma\pr)}$,
\ie $\stdsp x\pr\hspace{-0.79pt} \in\hspace{-0.79pt} \pos\hspace{1.9pt}u_k$ and there is some $w\pr$ 
such that $\rpstep{\rho\hspace{-0.1pt}}{\hspace{-0.9pt}(\plook{u_k}{x\pr},\stsp \stateOf{\sq^p_{n-1}})\hspace{0.9pt}}{\hspace{-1pt}(w\pr, \stsp\sigma\pr)}$ and
$u\pr_k\hspace{-1.2pt} =\hspace{-0.7pt} \psubst{u_k}{w\pr}{x\pr}$ hold. 
Consequently, the program step
\[\rpstep{\rho\stsp}{(\Parallel{\!u_1,\ldots u_k \ldots,  u_m}, \stdsp\stateOf{\sq_{n-1}})\stsp}{(\Parallel{\!u_1,\ldots \psubst{u_k}{w\pr}{x\pr} \ldots, u_m},\stsp \sigma\pr)}\]
has the fired position $kx\pr$.
Thus, by setting $\sq_n\hspace{-2.5pt} \defeq\hspace{-0.7pt} (\Parallel{\!u_1,\ldots \psubst{u_k}{w\pr}{x\pr} \ldots, u_m},\stsp \sigma\pr)$
we extend $\sq$ to $\stsp\sq_0, \ldots, \sq_{n-1}\hspace{0.4pt} \pstep\hspace{-2pt} \sq_{n}\stsp$
satisfying (1)--(3) for all $i \le\hspace{1pt} n$.

Finally, in case $\progOf{\sq^q_{n}}\hspace{-1.2pt} =\stnsp \skipp$ we can infer $v_k\hspace{-1.2pt} =\stnsp \skipp$ 
and $u_k\hspace{-1.2pt} =\stnsp \skipp$ for all $k\hspace{0.35pt} \in\hspace{-0.5pt} \{1, \ldots, m\}$ and
extend $\sq\stsp$ to $\stsp\sq_0, \ldots, \sq_{n-1}\stsp \pstep\hspace{-1pt} (\skipp,\stdsp \stateOf{\sq_{n-1}})$.
\end{proof}

Next proposition complements this result
showing that the conditions (1)--(3) are sufficient to 
draw conclusions about fairness of replayed computations. 
\begin{lemma}\label{thm:fair-corr-sim2}
  Let $\sq^p \hspace{-1.2pt}\in\hspace{-0.2pt} \rpcsi{p}{\rho}$ and $\sq^q\hspace{-1pt} \in\hspace{-0.2pt} \rpcsi{q}{\rho\pr}$ be two computations
  satisfying
  the conditions \emph{(1)--(3)} of 
  the preceding proposition
  for all $n\hspace{-0.4pt} \in\hspace{-0.7pt} \naturals$.
Then $\sq^p$ is fair whenever $\sq^q$ is.
\end{lemma}
\begin{proof}
Let $y$ be an always available position of $\progOf{\sq^p_i}$ with some $i\hspace{0.79pt}\in\hspace{-0.4pt}\naturals$.
By (1), $\progOf{\sq^p_i}$ corresponds to $\progOf{\sq^q_i}$ componentwise w.r.t. $\hspace{-1.5pt}m, r, \rho, \rho\pr$. 
This means
that $\progOf{\sq^p_{i}}\stnsp =\hspace{2pt} \Parallel{\!u_1,\ldots, u_m}$ and $\progOf{\sq^q_{i}}\stnsp = \hspace{2pt}\Parallel{\!v_1,\ldots, v_m}$
hold with some $u_1, \ldots, u_m$ and $v_1, \ldots, v_m$ satisfying $\pcorr{u_k}{r}{\hspace{-1pt}v_k}$ for all $1 \le k \le m$ because $\progOf{\sq^p_i} \neq \skipp$.

If $y =\hspace{-0.5pt} 0\stsp$ then we consequently have $u_k \hspace{-1.5pt} = v_k \stnsp= \skipp$ for all $k\stnsp \in\hspace{-1.2pt} \{1, \ldots, m\}$.
Since $sq\mystrut^q$ is assumed to be fair, there must further be some $j\hspace{-1pt} \ge\hspace{-1pt} i$ such that
$\progOf{sq\mystrut^q_j}\hspace{-1pt} = \hspace{2.5pt}\Parallel{\!\skipp, \ldots, \skipp}$
and $\progOf{sq\mystrut^q_{j+1}}\hspace{-1.2pt} =\hspace{-1pt} \skipp$. By the condition (1) we can thus once more infer
$\progOf{sq\mystrut^p_j}\hspace{-1.5pt} = \hspace{2.5pt}\Parallel{\hspace{-2.5pt}\skipp, \ldots, \skipp}$ and $\progOf{sq\mystrut^p_{j+1}}\hspace{-1.2pt} =\hspace{-1pt} \skipp$,
\ie we have the program step $\stsp\rpstep{\rho\hspace{0.7pt}}{\sq\mystrut^p_j\hspace{1.5pt}}{\hspace{-2.1pt}\sq\mystrut^p_{j+1}}$ with the fired position $0$.

Next, 
suppose $y\hspace{-0.5pt} =\hspace{-0.7pt} kx$ with $k\hspace{-0.29pt} \in\hspace{-0.7pt} \{1, \ldots, m\}$ and $x\hspace{-0.29pt}  \in\hspace{-1pt} \pos\hspace{2pt}u_k$.
Hence $v_k\hspace{-1.5pt} \neq\hspace{-0.7pt} \skipp$ due to $\pcorr{u_k\hspace{-0.7pt}}{r}{\hspace{-1.7pt} v_k}$ and since
$(\rho\pr, v_k)$ is assumed to be non-blocking, there must be an always available 
$x\pr\hspace{-0.59pt} \in\hspace{-0.59pt} \pos\hspace{2.1pt}v_k$. 
Furthermore,
$kx\pr$ is an always available position of $\progOf{\sq^q_i}$ and
since $\sq^q$ is fair we get some $j \ge i$ such that the conditions 
\begin{enumerate}
\item[--] $\progOf{sq^q_j} = \:\Parallel{\hspace{-1.7pt} v\pr_1,\ldots, v\pr_m}$ 
\item[--] $x\pr\hspace{-0.4pt} \in\hspace{-0.55pt} \pos\hspace{2pt}v\pr_k$
\item[--] $\rpstep{\rho\pr\hspace{-0.2pt}}{\sq^q_j\stsp}{\hspace{-2.5pt}\sq^q_{j+1}}\stsp$ has the fired position $kx\pr$
\end{enumerate}
hold with some $v\pr_1,\ldots, v\pr_m$.
This allows us to draw the following conclusions regarding $\sq^p$. 
Firstly, by the condition (1) we have $\progOf{\sq^p_j} = \:\Parallel{\!u\pr_1,\ldots, u\pr_m}$ with some locally sequential $u\pr_1,\ldots, u\pr_m$. 
Secondly, the condition (3) yields a position $x\prr\hspace{-0.9pt} \in \hspace{-0.5pt}\pos\hspace{2.1pt}u\pr_k$
that has been fired by $\rpstep{\rho\hspace{0.4pt}}{\sq^p_j\hspace{1pt}}{\hspace{-2.5pt}\sq^p_{j+1}}$.

It remains to
show that the relevant position $kx$ in $\progOf{\sq^p_i}$ must have been fired by a step $\rpstep{\rho}{\hspace{-0.79pt}\sq^p_l\hspace{0.4pt}}{\hspace{-2.5pt}\sq^p_{l+1}}$
with $i \le l \le j$. Assuming the opposite 
we can induce $\plook{u_k}{x}\hspace{-1.2pt} = \plook{u\pr_k}{x}$ and therefore  $x \hspace{0.1pt}\in\hspace{-0.55pt} \pos\hspace{2.1pt}u\pr_k$ holds by Proposition~\ref{thm:eq-lookup-pos}.
On the other hand, from the componentwise correspondence follows that $u\pr_k$ is locally sequential and
has thus at most one program position, \ie\stdsp $x = x\prr$.
\end{proof}

Combining the results of this section, the proof of the following corollary
consequently resorts first to Proposition~\ref{thm:fair-corr-sim1} and then to Proposition~\ref{thm:fair-corr-sim2}
showing how 
termination properties can be carried along componentwise correspondences.
The statement's applicability will be substantiated
in the next chapter establishing 
termination of fair computations of $\op{mutex}$ indirectly via 
the auxiliary model $\op{mutex}^{\mathit{aux}}$.
\begin{corollary}\label{thm:fair-corr-sim3}
Let $p$ and $q$ correspond componentwise w.r.t. \!$m, r, \rho, \rho\pr$
and assume 
\begin{enumerate}
\item[\emph{(1)}] $P\pr\hspace{-0.1pt} \subseteq\hspace{0.9pt} \rimg{r\hspace{1pt}}{\hspace{0.5pt}P}$,
\item[\emph{(2)}] for any $\sq^p\hspace{-1pt} \in\hspace{-0.2pt} \rpcsf{p}{\rho}\hspace{-0.9pt} \cap\hspace{1.5pt} \inC{\hspace{0.39pt}P}\hspace{0.7pt} \cap\hspace{1.5pt} \envC{\hspace{0.59pt}\op{id}}$ there exists some 
  $i$ with
      $\progOf{\sq_i}\hspace{-0.9pt} =\hspace{-0.5pt} \skipp$.
\end{enumerate} 
Then  for any computation $\sq^q\hspace{-1.5pt} \in\hspace{-0.9pt} \rpcsf{q}{\rho\pr}\hspace{-2.5pt} \cap\hspace{-0.4pt} \inC{\hspace{-1.25pt}P\pr} \hspace{-1.7pt}\cap\hspace{0.2pt} \envC{\hspace{-1pt}\op{id}}$ there also exists some $i$ with
      $\progOf{\sq_i} = \skipp$.
\end{corollary}
\begin{proof}
  Let $\sq^q\hspace{-0.79pt} \in \hspace{-0.55pt}\rpcsf{q}{\rho\pr}\hspace{-2.1pt} \cap\hspace{-0.1pt} \inC{\hspace{-0.9pt} P\pr}\hspace{-1.4pt} \cap\hspace{0.2pt} \envC{\hspace{-0.5pt}\op{id}}$.
  Then from (1) we get a preimage $\sigma\hspace{-0.4pt} \in\hspace{-1pt} P$ of $\stateOf{\sq^q_0}$ under $r$, \ie $(\sigma,\stsp \stateOf{\sq^q_0})\hspace{-0.2pt} \in\hspace{-0.1pt} r$.
  As a result, Proposition~\ref{thm:fair-corr-sim1} provides
a corresponding computation $\stsp\sq^p \hspace{-1pt}\in\hspace{-0.4pt} \rpcsi{p}{\rho}\hspace{-0.2pt} \cap\hspace{0.55pt} \inC{\hspace{-0.55pt}P}\hspace{-0.29pt} \cap\hspace{0.9pt}\envC{\hspace{-0.2pt}\op{id}}$ which must be moreover fair by Proposition~\ref{thm:fair-corr-sim2}.
Thus, from (2) we obtain some $i$ with $\progOf{\sq^p_i} = \skipp$
and can conclude $\progOf{\sq^q_i} = \skipp\stsp$ since $\pcorr{\progOf{\sq^p_i}\hspace{0.2pt}}{r}{\hspace{-1.5pt}\progOf{\sq^q_i}}$.
\end{proof}
\section{A brief summary}
Starting with the 
question which computations of a particular program reach a state satisfying a predicate that one normally expects them to reach,
this chapter 
first focused on developing a concise notion of fair computations without full support for await-statements however
(the topic addressed in Section~\ref{Sb:fair-await}).
Based on this, a refutational approach to proving liveness properties has led to a collection of syntax-driven rules for processing programs which are jump-free and also
do not make use of the await-statements beyond atomic sections.
Similarly to the program logic in Chapter~\ref{S:prog-log} with the input/output and invariant properties,
processing statements by these rules can not only 
supply arguments why a program exhibits some liveness property, but also
hint to the reasons why it does not.
Lastly, correspondence conditions, stronger than in Proposition~\ref{thm:corr-sim},
ensuring that any fair computation of a program can be replayed by a fair computation of a corresponding program
have been elaborated.
Next chapter applies all of the techniques to establish termination of $\op{mutex}$.
\setcounter{equation}{0}
\chapter{Case Study: Proving Termination}\label{S:PM3}
Recall the model $\op{mutex}$ of the Peterson's mutual exclusion algorithm shown in Figure~\ref{fig:pm}.
The proofs 
in Chapter~\ref{S:PM1} 
implicitly utilised the property
that the threads cannot be simultaneously within their critical sections in order to derive
that a certain condition holds \emph{upon termination} of $\op{mutex}$,
\ie for all computations of $\op{mutex}$ that reach a $\skipp$-configuration.
In this sense it would have taken significantly less efforts to achieve  
exactly the same result with a slightly modified model
simply making $\op{thread}_0$ (or $\op{thread}_1$, or both) spin in perpetual waiting to enter $\mv{cs}_0$ (or $\mv{cs}_1$),
\ie with $\whileS{\op{True}}{\skipp}$ in place of $\whileS{\hspace{1pt}\sv{flag}_1 \hspace{-2pt}\wedge\hspace{-1pt} \sv{turn}}{\stnsp\skipp\stnsp}$.
All fairness considerations become basically irrelevant for (non-)termination of such a contrived model: no computation will evidently reach a $\skipp$-configuration anyway.
This is in contrast to the following, more subtle modification attempting to `optimise' the protocol by ignoring $\sv{turn}$, \ie
\[
\begin{array}{l}
 \sv{flag}_0 :=\stnsp \op{True};\\
\whileS{\stdsp\sv{flag}_1}{\skipp};\\
\mv{cs}_0; \\
\sv{flag}_0 := \stnsp\op{False}
\end{array}
\]
being the modified $\op{thread}_0$ whereas $\op{thread}_1$ is modified accordingly. Although intuitively clear, it is not straight forward anymore
(one can make $\sv{turn}$ an auxiliary variable proceeding similarly to Chapter~\ref{S:PM1})
to show that the threads still would not interfere on a shared resource. But the actual point is that a considerable part of the fair computations would indeed reach a $\skipp$-configuration:
one thread can pass to its critical section without any interruption by the other which can then freely move to its `busy waiting' phase.
As opposed to that, not only a part but \emph{all} fair computations of $\op{mutex}$ eventually reach a $\skipp$-configuration. 

To sum up, the ultimate goal of this part of the case study is a proof that any computation in
$\hspace{-0.5pt}\pcsf{\op{mutex} \hspace{1.5pt} \mv{cs}_0 \hspace{1.5pt} \mv{cs}_1} \hspace{-0.5pt}\cap\hspace{1.9pt} \envC{\hspace{0.79pt}\op{id}}$
reaches a $\skipp$-configuration provided $\mv{cs}_0$ and $\mv{cs}_1$ unconditionally terminate.  
Also for this task we will resort to certain properties of the auxiliary model $\op{mutex}^{\mathit{aux}}$, shown in Figure~\ref{fig:pm-aux},
which are not directly accessible by $\op{mutex}$,
\ie $\op{mutex}^{\mathit{aux}}$ is once more scrutinised first.

\section{Preparations}
Following the approach described in Section~\ref{S:refute},
we seek to refute the statement
\begin{equation}\label{eq:pm30}
\hspace{-20pt}\exists\sq\hspace{-0.25pt} \in\hspace{-0.5pt} \envC{\hspace{-0.7pt}\op{id}}. \hspace{2.5pt}\lcondN{\sq}{\!\neg\sv{turn\_aux}_0\hspace{-1pt} \wedge\hspace{-1.7pt} \neg\sv{turn\_aux}_1}
{\hspace{0.5pt}\op{mutex}^{\mathit{aux}}\hspace{1.7pt} \mv{cs}_0\hspace{1.7pt}\mv{cs}_1}{\hspace{-0.5pt}\bot} 
\end{equation}
assuming that any claim of non-termination of $\mv{cs}_0$ or $\mv{cs}_1$ can be refuted, \ie
\begin{equation}\label{eq:pm3-asm1}
  \centering
  \begin{split}
   \hspace{-1.1cm} & \forall\sq\hspace{2pt}P.\hspace{3.3pt}\NlcondN{\sq}{\hspace{-1.9pt}P\hspace{-0.5pt}}{\hspace{-0.2pt}\mv{cs}_0}{\hspace{-0.2pt}\bot}\\
    & \forall\sq\hspace{2pt}P.\hspace{3.1pt}\NlcondN{\sq}{\hspace{-1.9pt}P\hspace{-0.5pt}}{\hspace{-0.2pt}\mv{cs}_1}{\hspace{-0.2pt}\bot}
  \end{split}
  \end{equation}
Now, from the assumption (\ref{eq:pm30}) we derive a computation $\sqt\hspace{-1.2pt}\in\hspace{-1.2pt} \envC{\hspace{-1.2pt}\op{id}}$
having moreover the property
\begin{equation}\label{eq:pm32}
  \hspace{-5pt}\lcondN{\sqt\hspace{0.1pt}}{\!\neg\sv{turn\_aux}_0\hspace{-0.9pt} \wedge\hspace{-1.5pt} \neg\sv{turn\_aux}_1\hspace{-0.5pt}}
  {\hspace{0.9pt}\op{mutex}^{\mathit{aux}}\hspace{1pt} \mv{cs}_0\hspace{1.7pt}\mv{cs}_1}{\hspace{-0.2pt}\bot} 
\end{equation}
It is essential to
note that $\sqt$ will refer to this computation in the course of the entire refutation process,
\ie up to Proposition~\ref{thm:PM3-aux}.
Furthermore, the shorthand $\sqts_i$ will be used for the state $\stateOf{\sqt_i}$ of the $i$-th configuration on $\sqt$.

Sketching a general refutation plan,
an application of the rule for the parallel composition in Proposition~\ref{thm:parallel2-live} to (\ref{eq:pm32})
would produce the three branches:
\begin{enumerate}
\item[(a)] $\op{thread}^{\mathit{aux}}_0$ and $\op{thread}^{\mathit{aux}}_1$ do not terminate, or
\item[(b)] $\op{thread}^{\mathit{aux}}_0$ terminates but $\op{thread}^{\mathit{aux}}_1$ does not, or
\item[(c)] $\op{thread}^{\mathit{aux}}_1$ terminates but $\op{thread}^{\mathit{aux}}_0$ does not.
\end{enumerate}

In the branch (a) we could further infer that both threads got stuck in their `busy waiting' phases, \ie
the conditions $\sv{flag}_1 \hspace{-1.7pt}\wedge\hspace{-0.7pt} \sv{turn}$ and $\sv{flag}_0  \hspace{-1pt}\wedge \hspace{-1.5pt}\neg\sv{turn}$ must recur infinitely often on $\sqt$,
in particular 
meaning that the value of $\sv{turn}$ flips infinitely many times. 
On the other hand, we could also derive that starting from a certain point, the condition 
$\sv{turn\_aux}_0\hspace{-1.2pt} \wedge\hspace{-0.7pt} \sv{turn\_aux}_1$ holds perpetually on $\sqt$
and hence from this very point the value of $\sv{turn}$ must remain constant 
in contradiction to the previous conclusion.

In the branch (b) only  $\op{thread}^{\mathit{aux}}_1$ gets stuck in its `busy waiting' phase, \ie
only the condition $\sv{flag}_0 \hspace{-1pt}\wedge \hspace{-1.5pt}\neg\sv{turn}$ must recur infinitely often on $\sqt$.
On the other hand, regarding $\op{thread}^{\mathit{aux}}_0$ we could now derive that it eventually assigns $\op{False}$ to $\sv{flag}_0$ at some point on $\sqt$,
and hence $\neg\sv{flag}_0$ holds perpetually starting from this point.
Altogether we would get a state on $\sqt$ where both, $\neg\sv{flag}_0$ and $\sv{flag}_0$, hold. 

The branch (c) is clearly symmetric to (b): replacing $\neg\sv{turn}$ by $\sv{turn}$ as well as swapping $0$ and $1$ in the preceding paragraph yields
the sought contradiction. We therefore will not treat this branch explicitly in what follows.

Although rather superficial, these considerations nonetheless reveal that in order to close a branch
we appeal to certain invariants 
(\eg $\sv{flag}_0$ keeps its value once $\op{False}$ assigned) which hold on any $\sq \in \pcsi{\op{mutex}^{\mathit{aux}}}\hspace{-1pt} \cap\hspace{0.5pt} \envC{\hspace{-0.9pt}\op{id}}$. 
These invariants are accurately captured 
by the guarantees of the following extended Hoare triple
\begin{equation}\label{eq:pm33}
\hspace{-20pt}\rgvalidi{\op{id}}{\stnsp\neg\sv{turn\_aux}_0 \wedge\stnsp \neg\sv{turn\_aux}_1}{\op{mutex}^{\mathit{aux}}\hspace{0.5pt} \mv{cs}_0 \hspace{1.9pt}\mv{cs}_1 }{\top}{\hspace{-2pt}\op{G}_{\stnsp\mv{global}}}
\end{equation}
where $\op{G}_{\stnsp\mv{global}}$ accordingly stands for
\[
\begin{array}{l}
  \hspace{-25pt}(\sv{turn\_aux}_0 \stsp\imp\stdsp \sv{turn\_aux}\pr_0) \wedge   (\sv{turn\_aux}_1\stsp \imp\stdsp \sv{turn\_aux}\pr_1) \; \wedge  \\
\hspace{-25pt}(\sv{turn\_aux}_0 \wedge \sv{flag}\pr_0 \stsp\imp\stdsp \sv{flag}_0) \wedge  (\sv{turn\_aux}_1 \wedge \sv{flag}\pr_1 \stsp\imp\stdsp \sv{flag}_1) \; \wedge  \\
\hspace{-25pt}(\sv{turn\_aux}_0 \wedge \sv{turn\_aux}_1\stsp \imp\stdsp \sv{turn} = \sv{turn}\pr). 
\end{array}
\]
The triple (\ref{eq:pm33}) is derivable using the program logic rules from Chapter~\ref{S:prog-log}
under the assumptions
$\rgvalid{\top}{\hspace{-1.5pt}\top}{\mv{cs}_0 }{\top}{\hspace{-1.5pt}\op{G}_{\stnsp\mv{cs}}}$ and
$\rgvalid{\top}{\hspace{-1.5pt}\top}{\mv{cs}_1\hspace{-0.5pt}}{\top}{\hspace{-1.5pt}\op{G}_{\stnsp\mv{cs}}}$
with $\op{G}_{\stnsp\mv{cs}}$  
merely requiring 
from each program step of $cs_0$ and $cs_1$ 
to leave the auxiliary and the protocol variables unmodified:
\[
\begin{array}{l}
\hspace{-21pt}\sv{turn\_aux}_0\hspace{-0.7pt} =\sv{turn\_aux}\pr_0 \wedge \sv{turn\_aux}_1\hspace{-1pt} =\sv{turn\_aux}\pr_1 \; \wedge \\
\hspace{-21pt}\sv{flag}_0\hspace{-1.1pt} =\sv{flag}\pr_0 \wedge \sv{flag}_1\hspace{-1.5pt} =\sv{flag}\pr_1 \wedge \sv{turn} =\sv{turn}\pr.
\end{array}
\]
It is worth noting that the guarantees in $\op{G}_{\stnsp\mv{global}}$, 
immaterial from the perspective of potential interferences
between the two threads (\cf Proposition~\ref{thm:mutex-aux}), become decisive arguments in the liveness context.

Now, the primary purpose of the triple (\ref{eq:pm33}) is to conclude $\sqt\hspace{-0.2pt} \in\hspace{-0.9pt} \progC{\hspace{-0.5pt}\op{G}_{\stnsp\mv{global}}}$
using that $\sqt\hspace{-0.55pt} \in\hspace{-0.55pt} \pcsi{\op{mutex}^{\mathit{aux}}\hspace{0.5pt} \mv{cs}_0 \hspace{1.9pt}\mv{cs}_1}\hspace{-1.5pt} \cap\hspace{-0.1pt}
\inC{\hspace{-1.7pt}(\neg\sv{turn\_aux}_0\hspace{-0.9pt} \wedge\hspace{-2pt} \neg\sv{turn\_aux}_1)}\hspace{-0.3pt} \cap\hspace{0.2pt} \envC{\hspace{-0.7pt}\op{id}}$ is a straight consequence of (\ref{eq:pm32}).
This notably means that 
  \begin{enumerate}
    \item[(i)] $\sqts_i\sv{turn\_aux}_0 \stsp\imp\stdsp \sqts_{i+1}\sv{turn\_aux}_0$
    \item[(ii)] $\sqts_i\sv{turn\_aux}_1\stsp \imp\stdsp \sqts_{i+1}\sv{turn\_aux}_1$
    \item[(iii)] $\sqts_i\sv{turn\_aux}_0\hspace{-0.7pt} \wedge\hspace{-0.5pt} \sqts_{i+1}\sv{flag}_0\stsp \imp\stdsp \sqts_i\sv{flag}_0$
    \item[(iv)] $\sqts_i\sv{turn\_aux}_1\hspace{-0.7pt} \wedge\hspace{-0.5pt} \sqts_{i+1}\sv{flag}_1\stsp \imp\stdsp \sqts_i\sv{flag}_1$  
      \item[(v)] $\sqts_i\sv{turn\_aux}_0\hspace{-0.7pt} \wedge\hspace{-0.5pt} \sqts_i\sv{turn\_aux}_1 \stsp\imp\stdsp \sqts_i\sv{turn} = \sqts_{i+1}\sv{turn}$
  \end{enumerate}
  hold \emph{globally} on $\sqt$, \ie for all $i$ regardless 
  of the kind of the transition from $\sqt_i$ to $\sqt_{i+1}$:
  the program steps are covered by the above conclusion $\sqt\hspace{-0.1pt} \in\hspace{-0.79pt} \progC{\hspace{0.7pt}\op{G}_{\stnsp\mv{global}}}$ whereas the environment --
  by the property $\sqt\hspace{-0.4pt} \in\hspace{-0.4pt} \envC{\hspace{0.9pt}\op{id}}$ since (i)--(v) are especially valid when
  $\sqts_i\hspace{-0.7pt} = \sqts_{i+1}$.
  
Next two propositions utilise (i)--(v) in order to derive relevant properties for certain suffixes of $\sqt$.
\begin{lemma}\label{thm:PM3a}
  \!Let $\hspace{0.3pt}n, m\hspace{-0.1pt} \in\hspace{-0.4pt} \naturals$ with $\hspace{0.7pt}\sqts_n\sv{turn\_aux}_0$ and $\hspace{0.7pt}\sqts_m\sv{turn\_aux}_1$.
  \hspace{-2.5pt}Then there exists $d\hspace{0.7pt}\in\hspace{-0.5pt}\naturals$ such that  $\stsp\sqts_i\sv{turn} = \sqts_j\sv{turn}$
holds for all $\hspace{1pt}i\hspace{0.7pt} \ge d$ and $j\hspace{0.7pt} \ge d$.
\end{lemma}
\begin{proof}
  Let $d$ be the maximum of $n$ and $m$.
  With (i) and (ii) we can induce $\sqts_i\sv{turn\_aux}_0\hspace{0.2pt} \wedge\hspace{0.2pt} \sqts_i\sv{turn\_aux}_1$ for all $i\ge d$.
  Then $\stsp\sqts_i\sv{turn} = \sqts_{i+1}\sv{turn}$ follows for all $i\ge d$ by (v).
  Thus, $\stsp\sqts_i\sv{turn} = \sqts_{j}\sv{turn}$ holds for all $i \ge d$ and $j \ge d$.
\end{proof}
\begin{lemma}\label{thm:PM3b}
Assume $\sqts_m\sv{turn\_aux}_0$ and $\neg\sqts_n\sv{flag}_0$ with $m \le n$. Then
$\neg\sqts_i\sv{flag}_0$ holds for all $\hspace{1pt} i\hspace{0.7pt} \ge n$.
\end{lemma}
\begin{proof}
  With (i) we first induce $\sqts_i\sv{turn\_aux}_0$ for all $i\ge m$.
Then with $\neg\sqts_n\sv{flag}_0$ 
and (iii) we conclude $\neg\sqts_i\sv{flag}_0$ for all $\stsp i\hspace{0.7pt}\ge n$.
\end{proof}

\section{The actual refutation}\label{Sb:act-refute}
Following the general plan we start with processing (\ref{eq:pm32}) by successive application of the syntax-driven rules from
Section~\ref{S:refute}.
Resolving the parallel composition operator in the first place, an application of Proposition~\ref{thm:parallel2-live} to (\ref{eq:pm32}) yields two computations $\sq^0$ and $\sq^1$
with $\steqv{\sq^0\hspace{-2pt} }{\hspace{-0.7pt} \sqt}$
and $\steqv{\sq^1\hspace{-2pt} }{\hspace{-0.7pt} \sqt}$
such that either
\begin{eqnarray}
  \begin{aligned}
  \label{eq:pm34}
 & \lcondN{\sq^0\hspace{-0.2pt}}{\hspace{-1.2pt}\neg\sv{turn\_aux}_0\hspace{-0.5pt} \wedge\hspace{-1.2pt} \neg\sv{turn\_aux}_1}
{\stsp\op{thread}^{\mathit{aux}}_0\hspace{1pt} \mv{cs}_0}{\bot} \\ 
\label{eq:pm35}
 & \lcondN{\sq^1\hspace{-0.2pt}}{\hspace{-1.2pt}\neg\sv{turn\_aux}_0\hspace{-0.5pt} \wedge\hspace{-1.2pt}  \neg\sv{turn\_aux}_1}
       {\stsp\op{thread}^{\mathit{aux}}_1\hspace{1pt} \mv{cs}_1}{\bot}
       \end{aligned}
\end{eqnarray}
or
\begin{eqnarray}
        \begin{aligned}
          \label{eq:pm36}
          & \lcondT{\sq^0\hspace{-0.2pt}}{\hspace{-1.2pt}\neg\sv{turn\_aux}_0\hspace{-0.5pt} \wedge\hspace{-1.2pt} \neg\sv{turn\_aux}_1}
       {\stsp\op{thread}^{\mathit{aux}}_0\hspace{1pt} \mv{cs}_0}{\bot\hspace{0.2pt}}{\hspace{1pt}n} \\
\label{eq:pm37} 
 & \lcondN{\sq^1\hspace{-0.2pt}}{\hspace{-1.2pt}\neg\sv{turn\_aux}_0\hspace{-0.5pt} \wedge\hspace{-1.2pt}  \neg\sv{turn\_aux}_1\stnsp}
       {\stsp\op{thread}^{\mathit{aux}}_1\hspace{1pt} \mv{cs}_1}{\bot}
      \end{aligned}
\end{eqnarray}
hold with some $n\hspace{-0.05pt}\in\hspace{-0.35pt}\naturals$.
Note once more that the third disjunct is omitted in this presentation because it is the same as (\ref{eq:pm36}) with the indices $0$ and $1$ swapped.

Further processing (\ref{eq:pm34})  
using the rules from Section~\ref{S:refute} leads to several trivially refutable branches such as
\begin{eqnarray*}
  \begin{aligned}
& \lcondN{\sq\pr\stsp}{\stnsp\neg\sv{turn\_aux}_0 \wedge \stnsp\neg\sv{turn\_aux}_1}{\stdsp\sv{flag}_0 \hspace{-0.2pt}:=\! \op{True}\stdsp}{\bot}\: \mbox{, or} \\
    & \lcondN{\sq\prr\stnsp}{\stnsp\top\hspace{-0.7pt}}{\stdsp\langle \sv{turn} :=\! \op{True};\stdsp\sv{turn\_aux}_0 \hspace{-0.4pt}:= \!\op{True}\rangle\stdsp}{\bot}
\end{aligned}
\end{eqnarray*}
for some $\sq\pr$ and $\sq\prr$.
Moreover, note that all branches claiming termination of the `busy waiting' phase of $\op{thread}^{\mathit{aux}}_0$ or $\op{thread}^{\mathit{aux}}_1$
are simply refutable as well because all other parts of the program are unconditionally terminating due to (\ref{eq:pm3-asm1})
which is a contradiction to (\ref{eq:pm34}).
As a result, it remains only one non-trivial branch exhibiting the following properties:
\begin{enumerate}
\item[(a)] there are some $m_0, m_1\hspace{-0.5pt} \in\hspace{-0.2pt} \naturals$ with $\sqts_{m_0}\sv{turn\_aux}_0$ and $\sqts_{m_1}\sv{turn\_aux}_1$
\item[(b)] there are two strictly ascending sequences of natural numbers $\phi, \psi$ and also some $d_0\hspace{-1pt}>\hspace{-1pt}m_0$ and $d_1\hspace{-1pt}>\hspace{-1pt}m_1$ such that both
  \begin{enumerate}
    \item[(b$_0$)] $\sqts_{\phi_{i} + d_0}\sv{flag}_1\hspace{-0.7pt} \wedge\hspace{-0.2pt} \sqts_{\phi_{i} + d_0}\sv{turn}$
    \item[(b$_1$)] $\sqts_{\psi_{i} + d_1}\sv{flag}_0\hspace{-0.7pt} \wedge\hspace{-1.5pt} \neg\sqts_{\psi_{i} + d_1}\sv{turn}$
  \end{enumerate}
hold for all $i\hspace{0.5pt}\in \naturals$.
\end{enumerate}
Then by Proposition~\ref{thm:PM3a} and (a) we obtain some $d$ such that the equality $\sqts_i\sv{turn} = \sqts_j\sv{turn}$ holds for all $i, j\ge d$,
\ie we have $\sqts_{\phi_{d} + d_0}\sv{turn} = \sqts_{\psi_{d} + d_1}\sv{turn}$ 
in particular due to $\phi_{d}\hspace{-1pt} \ge\hspace{-1pt} d$ and $\psi_{d}\hspace{-1pt} \ge\hspace{-1pt} d$.
In contradiction to that,
$\sqts_{\phi_{d} + d_0}\sv{turn}$ and $\neg\sqts_{\psi_{d} + d_1}\sv{turn}$ follow respectively from (b$_0$) and (b$_1$).

Next we backtrack to the pending branch (\ref{eq:pm37}) that still needs to be processed. 
Note that successive application of the rules from Section~\ref{S:refute} to
\[
\hspace{-15pt}\lcondT{\sq^0}{\hspace{-2pt}\neg\sv{turn\_aux}_0 \stnsp\wedge \stnsp\neg\sv{turn\_aux}_1}
       {\stsp\op{thread}^{\mathit{aux}}_0\hspace{0.7pt} \mv{cs}_0\stsp}{\bot\hspace{-0.2pt}}{\hspace{0.9pt}n}
       \]
       goes ahead without any further branching, thereby 
       producing a series of conclusions about the shape of 
       $\sq^0$ and hence also of $\sqt$.         
In particular, we obtain $\neg\sqts_{n}\sv{flag}_0$, \ie $\neg\sv{flag}_0$ holds on the $n$-th state of $\sqt$ 
since the assignment $\sv{flag}_0\stnsp := \stnsp\op{False}$  
is carried out by the last program step on $\sq^0$.
Moreover, we get some $m\hspace{-0.5pt}<n$ with $\sigma_{m}\sv{turn\_aux}_0$,
and hence Proposition~\ref{thm:PM3b} provides a more advanced conclusion
that $\neg\sigma_i\sv{flag}_0$ holds for all $i\ge n$. 
       Regarding the opposite thread, 
       further processing of the assumption
\[
\hspace{-15pt}\lcondN{\sq^1\hspace{-0.5pt}}{\!\neg\sv{turn\_aux}_0 \hspace{-1pt}\wedge\hspace{-1.9pt} \neg\sv{turn\_aux}_1}{\stsp\op{thread}^{\mathit{aux}}_1 \mv{cs}_1}{\bot}
\]
as with (\ref{eq:pm34}) results in only one non-trivially refutable branch
where we have a strictly ascending sequence $\phi$ and some $d$ 
such that $\sigma_{\phi_{i} + d}\sv{flag}_0 \wedge\hspace{-1pt} \neg\sigma_{\phi_{i} + d}\sv{turn}$ holds for all $i\hspace{0.15pt} \in \hspace{-0.79pt}\naturals$
and hence for $n$ in particular, \ie $\sigma_{\phi_{n} + d}\sv{flag}_0$. On the other hand we also have $\neg\sigma_{\phi_{n} + d}\sv{flag}_0$ since $\phi_{n}\hspace{-1pt} \ge n$.

This concludes the refutation process of (\ref{eq:pm30}) so that
the following proposition can summarise the result: 
\begin{lemma}\label{thm:PM3-aux}
Assume
\begin{enumerate}
\item[\emph{(1)}] $\rgvalid{\top}{\hspace{-1.9pt}\top}{\mv{cs}_0 }{\top}{\!\op{G}_{\mv{cs}}}$,
\item[\emph{(2)}] $\rgvalid{\top}{\hspace{-1.9pt}\top}{\mv{cs}_1 }{\top}{\!\op{G}_{\mv{cs}}}$,
\item[\emph{(3)}] $\lcondN{\sq}{\!P\hspace{-0.7pt}}{\mv{cs}_0\hspace{-0.5pt}}{\hspace{-0.7pt}\bot}$ yields a contradiction for all $\sq$\hspace{-0.7pt} and $P$,
\item[\emph{(4)}] $\lcondN{\sq}{\!P\hspace{-0.9pt}}{\mv{cs}_1\hspace{-0.5pt}}{\hspace{-0.7pt}\bot}$ yields a contradiction for all $\sq$\hspace{-0.7pt} and $P$.
\end{enumerate}
Then any $\sq\hspace{-0.2pt} \in\hspace{-0.2pt} \pcsf{\op{mutex}^{\mathit{aux}}\hspace{1pt} \mv{cs}_0 \hspace{1.7pt}\mv{cs}_1}\hspace{-1.5pt} \cap\hspace{0.2pt} \inC{\hspace{-2.5pt}(\neg\sv{turn\_aux}_0 \hspace{-0.5pt}\wedge\hspace{-1.5pt} \neg\sv{turn\_aux}_1)} \cap\hspace{0.2pt} \envC{\hspace{-0.5pt}\op{id}}$
has some $n\hspace{0.55pt} \in\hspace{-0.5pt} \naturals$ such that $\progOf{\sq_n}\hspace{-0.5pt} = \skipp$.
\end{lemma}
\section{Termination of $\op{mutex}$}
To derive the corresponding result for $\op{mutex}$ we
proceed similarly to Section~\ref{S:PM1} using 
Corollary~\ref{thm:fair-corr-sim3} instead of Proposition~\ref{thm:corr-sim}.
The program correspondence,
\[
\hspace{-10pt}\pcorrC{\hspace{-1pt}\op{mutex}^{\mathit{aux}} \hspace{0.2pt}\mv{cs}_0\hspace{2pt}\mv{cs}_1\hspace{-0.7pt}}{\op{r}_{\op{eqv}}}{\!\op{mutex} \hspace{2pt} \mv{cs}_0\hspace{2pt}\mv{cs}_1}
\]
established in Section~\ref{S:PM1},
is however \emph{per} \emph{se} not sufficient because Corollary~\ref{thm:fair-corr-sim3} presumes 
$\pcorrC{\op{thread}^{\mathit{aux}}_0 \hspace{0.2pt} \mv{cs}_0}{\op{r}_{\op{eqv}}}{\hspace{-2.1pt}\op{thread}_0 \hspace{2.15pt} \mv{cs}_0}\stsp$ and 
$\stsp\pcorrC{\op{thread}^{\mathit{aux}}_1 \hspace{0.2pt} \mv{cs}_1}{\op{r}_{\op{eqv}}}{\hspace{-2.1pt}\op{thread}_1 \hspace{1.5pt} \mv{cs}_1}$
On the other hand, these correspondences have been implicitly derived from the assumptions
$\pcorrC{\stnsp\mv{cs}_0\hspace{-0.2pt}}{\op{r}_{\op{eqv}}}{\hspace{-2.5pt}\mv{cs}_0}$ and $\pcorrC{\stnsp\mv{cs}_1\hspace{-0.7pt}}{\op{r}_{\op{eqv}}}{\hspace{-2.5pt}\mv{cs}_1}$
in the proof of Proposition~\ref{thm:mutex} since it used the syntax-driven rules from Section~\ref{Sb:corr-rules} only.
Furthermore, note that $\op{thread}^{\mathit{aux}}_0$, $\op{thread}^{\mathit{aux}}_1$, $\op{thread}_0$ and $\op{thread}_1$
are by definition sequential and non-blocking, provided $\mv{cs}_0$ and $\mv{cs}_1$ are sequential and non-blocking.

This essentially means that $\op{mutex}^{\mathit{aux}} \hspace{1pt}\mv{cs}_0\hspace{2pt}\mv{cs}_1$ and
$\op{mutex} \hspace{2pt} \mv{cs}_0\hspace{2pt}\mv{cs}_1$ correspond componentwise w.r.t. \hspace{-1pt}$\op{r}_{\op{eqv}}$ by Definition~\ref{def:componentwise}
whenever $\mv{cs}_0$ and $\mv{cs}_1$ are sequential and non-blocking as well as
$\pcorrC{\stnsp\mv{cs}_0\hspace{-0.2pt}}{\op{r}_{\op{eqv}}}{\hspace{-2.5pt}\mv{cs}_0}$ and
$\pcorrC{\stnsp\mv{cs}_1\hspace{-0.5pt}}{\op{r}_{\op{eqv}}}{\hspace{-2.5pt}\mv{cs}_1}$ hold.

\begin{lemma}\label{thm:PM3}
Assume
\begin{enumerate}
\item[\emph{(1)}] $\rgvalid{\top}{\hspace{-1.7pt}\top}{\hspace{-0.5pt}\mv{cs}_0 \hspace{-0.5pt}}{\top}{\hspace{-1.7pt}\op{G}_{\mv{cs}}}$,
\item[\emph{(2)}] $\rgvalid{\top}{\hspace{-1.7pt}\top}{\hspace{-0.5pt}\mv{cs}_1 \hspace{-0.5pt}}{\top}{\hspace{-1.7pt}\op{G}_{\mv{cs}}}$,
\item[\emph{(3)}] $\lcondN{\sq}{\!P\hspace{-0.59pt}}{\hspace{0.1pt}\mv{cs}_0\hspace{-0.5pt}}{\hspace{-0.5pt}\bot}$ yields a contradiction for any $\sq$\hspace{-0.7pt} and $P$,
\item[\emph{(4)}] $\lcondN{\sq}{\!P\hspace{-0.59pt}}{\hspace{0.1pt}\mv{cs}_1\hspace{-0.9pt}}{\hspace{-0.5pt}\bot}$ yields a contradiction for any $\sq$\hspace{-0.7pt} and $P$,
\item[\emph{(5)}] $\pcorrC{\stnsp\mv{cs}_0\hspace{-0.2pt}}{\op{r}_{\op{eqv}}}{\hspace{-2.25pt}\mv{cs}_0}$,
\item[\emph{(6)}] $\pcorrC{\stnsp\mv{cs}_1\hspace{-0.5pt}}{\op{r}_{\op{eqv}}}{\hspace{-2.25pt}\mv{cs}_1}$,
\item[\emph{(7)}] $\mv{cs}_0$ and $\mv{cs}_1$ are sequential and non-blocking.
\end{enumerate}
Then any $\sq\hspace{-0.5pt} \in\hspace{-0.5pt} \pcsf{\op{thread}_0\hspace{1.7pt} \mv{cs}_0 \hspace{-1.7pt}\parallel\stnsp \op{thread}_1\hspace{1.5pt}\mv{cs}_1}\hspace{-2.2pt} \cap\envC{\hspace{-1.2pt}\op{id}}$
eventually reaches a $\skipp$-configuration.
\end{lemma}
\begin{proof}
Using Proposition~\ref{thm:PM3-aux} and the assumptions (1)--(4) we infer that any
$\sq\hspace{-0.2pt} \in\hspace{-0.2pt} \pcsf{\op{mutex}^{\mathit{aux}}\hspace{0.5pt} \mv{cs}_0 \hspace{2pt}\mv{cs}_1}\hspace{-1.9pt}  \cap\hspace{0.15pt} \inC{\!(\neg\sv{turn\_aux}_0\hspace{-0.9pt} \wedge\hspace{-2pt} \neg\sv{turn\_aux}_1)} \cap\hspace{-0.05pt} \envC{\hspace{-0.7pt}\op{id}}$
reaches  a $\skipp$-configuration. Moreover, $\op{mutex}^{\mathit{aux}}\hspace{0.5pt} \mv{cs}_0 \hspace{2pt}\mv{cs}_1$ and
$\op{mutex}\hspace{2pt} \mv{cs}_0 \hspace{2pt}\mv{cs}_1$ correspond componentwise w.r.t. $\hspace{-2.7pt}\op{r}_{\op{eqv}}$
due to the assumptions (5)--(7).
Hence, in order to apply Corollary~\ref{thm:fair-corr-sim3} only   
$\hspace{-0.5pt}\top\hspace{-1pt} \subseteq \hspace{0.4pt}\rimg{\op{r}_{\op{eqv}}\hspace{-0.1pt}}{\hspace{1.5pt}(\neg\sv{turn\_aux}_0\hspace{-0.5pt} \wedge \hspace{-1.5pt}\neg\sv{turn\_aux}_1)}$
remains to be shown, \ie\stsp for any $\sigma$ there is some $\sigma\pr$ with 
$\neg\sigma\pr\sv{turn\_aux}_0$, $\neg\sigma\pr\sv{turn\_aux}_1$ and $(\sigma\pr, \sigma)\hspace{0.5pt}\hspace{-0.5pt} \in\hspace{-0.5pt} \op{r}_{\op{eqv}}$.
Set $\sigma\pr \defeq \sigma_{[\sv{turn\_aux}_0\stsp :=\stsp \op{False},\hspace{2pt} \sv{turn\_aux}_1\stsp :=\stsp \op{False}]}$ to this end.
\end{proof}

Regarding lastly the instantiation of $\mv{cs}_0$ and $\mv{cs}_1$ by $\op{update}_0$ and $\op{update}_1$,
defined in Section~\ref{sub:mutex-inst}, 
any $\sq\hspace{-1.5pt} \in\hspace{-1.5pt} \pcsf{\op{mutex} \hspace{2.1pt}\op{update}_0 \hspace{2.5pt} \op{update}_1}\hspace{-2.5pt} \cap\hspace{-0.2pt} \envC{\hspace{-1.4pt}\op{id}}$
reaches a $\skipp$-configuration by Proposition~\ref{thm:PM3}.

The side condition `if such exists' in the final result of Section~\ref{sub:mutex2-inst} is now discarded as follows:
$\hspace{-0.5pt}\sigma\sv{shared}\hspace{-0.5pt} \cup\hspace{-0.5pt} \{0, 1\}\hspace{-1.1pt} = \hspace{-0.9pt}\sigma\pr\sv{shared}$ holds
for any computation $\sq\hspace{-0.2pt} \in\hspace{-0.5pt} \pcsf{\op{mutex} \hspace{2.5pt}\op{update}_0 \hspace{2.1pt} \op{update}_1}\hspace{-1.5pt} \cap\hspace{0.7pt} \envC{\hspace{-0.29pt}\op{id}}$
where $\sigma \hspace{-1.2pt}=\hspace{-1pt} \stateOf{\sq_0}$ and $\sigma\pr$ is the state of the first $\skipp$-configuration on $\sq$.

\setcounter{equation}{0}
\chapter{Conclusion}\label{S:concl}
This report gave a detailed presentation of a framework,
built as a conservative extension to the simply typed higher-order logic
and geared towards modelling, verification
and transformation of concurrent imperative programs.
The essential points were:
\begin{itemize}
\item[-] a concise computational model encompassing fine-grained interleaving, state abstraction and jump instructions;
\item[-] stepwise program correspondence relations
  allowing us in particular to reason upon properties of programs that use jumps in place of conditional and while-statements via a semantic equivalence to their structured counterparts; 
\item[-] a Hoare-style rely/guarantee program logic seamlessly tied in with the state abstraction
and featuring the program correspondence rule as a generalised rule of consequence;
\item[-] a light-weight extension of the logic enabling 
         state relations in place of pre/postconditions;
       \item[-] a refutational approach to proving liveness properties 
         driven by a concise 
         notion of fair computations.
\end{itemize}
The remainder of the chapter discusses further enhancements to the framework.
\section{Fair computations and await-statements}\label{Sb:fair-await}
As pointed out in Section~\ref{Sb:faircomp},
a plain generalisation of Definition~\ref{def:fair} from \emph{always available} to \emph{any} program position
would essentially require from in this sense fair computations of $\await{\hspace{1pt}C\hspace{-1.2pt}}{p}$ to visit a state satisfying $C$ at least once,
confining thus the scope of liveness reasoning to rather a handpicked (and sometimes even empty) subset of the potential computations of
programs that make use of await-statements beyond atomic sections.
Whereas the concept of an always available position can conveniently be detached from any notion of computation,
availability of a position referring to some 
$\wait{\hspace{0.2pt}C}$ cannot:
depending on the computational context such program positions may be regarded either as available or not. 
More precisely, such a position is certainly unavailable on a computation that does not visit a state in $C$ at all
but it might become available on a computation which does,
leading to the question: how often shall a computation pass through a state in $C$
so that one can sensibly require from it to reduce the subterm   
$\wait{\hspace{0.4pt} C}$ to $\skipp$ once? A finite number might be appropriate to tackle some very specific models,
but the general answer shall be: \emph{infinitely often}.
Always available positions fit in this scheme as a corner case: any infinite computation visits a state satisfying $\top$ infinitely often.

As a result of these considerations, part (1) of Proposition~\ref{thm:await-live} could be generalised 
as follows:  
\begin{enumerate}
\item[--] if $\lcondN{\sq\stnsp}{\hspace{-2pt}P\hspace{-0.5pt}}{\hspace{0.5pt}\await{\hspace{0.7pt}C}{p}\hspace{0.5pt}}{\hspace{0.5pt}Q}\stsp$ 
  then \emph{either}
  \begin{enumerate}
    \item[(a)] there exists a state $\sigma$ such that
      for all $\hspace{1pt}\sigma\pr\hspace{-0.5pt}$ the configuration $(\skipp, \sigma\pr)$ is not reachable from $(p, \sigma)$ by any sequence of program steps, \emph{\ie} $\hspace{-2pt}(p, \sigma) \not\psteps\stnsp (\skipp, \sigma\pr)$, \emph{or}
    \item[(b)] \emph{the condition $C\hspace{-1pt}$ recurs only finitely many times on} $\sq$.
      \end{enumerate}
\end{enumerate}
Once more, Proposition~\ref{thm:await-live} would be an instance of the above rule because $\top$ does trivially recur infinitely often on any infinite computation.
Processing thus programs that resort to await-statements for synchronisation purposes we would have to close the additional branch (b). 
The according variation of $\op{mutex}^{\mathit{aux}}$ 
\[
\begin{array}{l || l}
  \hspace{-21pt}  \sv{flag}_0\hspace{-0.7pt} :=\stnsp \op{True}; & \hspace{10pt}\sv{flag}_1\hspace{-1pt} :=\stnsp\op{True}; \\
\hspace{-25.5pt}   \langle \sv{turn} :=\stnsp \op{True}; & \hspace{5pt}\langle \sv{turn} :=\hspace{-0.2pt}\op{False}; \\
\hspace{-20pt}               \sv{turn\_aux}_0 :=\hspace{-1.5pt} \op{True}\rangle; & \hspace{10pt}\sv{turn\_aux}_1 :=\hspace{-1.5pt}\op{True}\rangle; \\
 \hspace{-21pt}  \wait{\neg\sv{flag}_1\hspace{-2.1pt} \vee\hspace{-1pt} \neg\sv{turn}};\hspace{5pt} & \hspace{10pt}\wait{\neg\sv{flag}_0\hspace{-1.7pt} \vee\hspace{-0.45pt}\sv{turn}};\\
 \hspace{-21pt}   \mv{cs}_0; & \hspace{10pt}\mv{cs}_1;\\
\hspace{-21pt}\sv{flag}_0\hspace{-0.5pt} := \hspace{-0.5pt}\op{False} & \hspace{10pt}\sv{flag}_1\hspace{-0.5pt} := \stnsp\op{False}
\end{array}
\]
can be used to illustrate this.
As with $\op{mutex}^{\mathit{aux}}$ in Chapter~\ref{S:PM3} we process this model by successive application of the syntax-driven rules from Section~\ref{S:refute}.
In contrast to Chapter~\ref{S:PM3}, with this variation we clearly do not need to apply the rule for while-statements, 
but encounter instead a branch carrying the assumption of the form
$\lcondN{\sq}{\hspace{-1.9pt}P\hspace{-0.7pt}}{\hspace{0.9pt}\wait{\hspace{0.5pt}\neg\sv{flag}_1\hspace{-2.1pt} \vee\hspace{-1.5pt} \neg\sv{turn}\hspace{0.5pt}}\hspace{0.2pt}}{\hspace{1.2pt}Q}$
where $\steqv{\sq\hspace{-0.5pt}}{\hspace{-0.5pt}\sqt}$.

Appealing to the above presented, generalised version of Proposition~\ref{thm:await-live} we could infer that
$\neg\sv{flag}_1\hspace{-2.25pt} \vee\hspace{-1.7pt} \neg\sv{turn}$ recurs only finitely many times on $\sq$ and hence also on $\sqt$ due to $\steqv{\sq}{\sqt}$.
Furthermore, non-termination assumption for the second thread would let us symmetrically derive that
$\neg\sv{flag}_0\hspace{-1.25pt} \vee\hspace{-0.15pt} \sv{turn}$ recurs only finitely many times on $\sqt$.
We could thus reason that there must be
\begin{enumerate}
\item[--] a suffix $\suffix{d_0}{\sqt}$ on which $\sv{flag}_1\hspace{-1.5pt} \wedge\hspace{-0.55pt} \sv{turn}$ holds globally, and
\item[--] a suffix $\suffix{d_1}{\sqt}$ on which $\sv{flag}_0\hspace{-1.2pt} \wedge\hspace{-1.9pt} \neg\sv{turn}$ holds globally.
\end{enumerate}
Since `globally' in particular means here \emph{ad infinitum},  
quite similarly to Chapter~\ref{S:PM3} we would get a contradiction in form of a state $\sigma$ on $\sqt$ such that $\sigma\sv{turn}$ and $\neg\sigma\sv{turn}$ hold
regardless whether we have $d_0\hspace{-1pt} \le d_1$ or $d_0\hspace{-1pt} \ge d_1$. 
\section{Enhancing the refutational approach}\label{Sb:enh-refute}
Each rule in Section~\ref{S:refute} essentially just carried around the parameter $Q$, set to capture an invariant condition on $\sq$
when drawing conclusions from assumptions of the form $\lcondN{\sq\hspace{-1pt}}{\hspace{-2.7pt}P\hspace{-1.1pt}}{p}{Q}$ or $\lcondT{\sq\hspace{-1pt}}{\hspace{-2.5pt}P\hspace{-0.9pt}}{p}{Q}{n}$.
The sole motivation behind that was to achieve a better perception of the underlying approach.
This parameter can actually be taken out
and the invariant asserted only once for the initial computation:
the refutational approach using the syntax-driven rules in Section~\ref{S:refute}
works the same
regardless if we start with the assumption
$\exists\sq\hspace{-0.7pt} \in\hspace{-0.7pt} \envC{\hspace{-0.39pt}\op{id}}. \hspace{2.5pt}\lcondN{\sq\hspace{-0.5pt}}{\hspace{-2pt}P\hspace{-0.9pt}}{\hspace{0.2pt}p\hspace{0.2pt}}{\hspace{0.2pt}Q}$
\stdsp or \stdsp
$\exists\sq\hspace{-0.79pt} \in\hspace{-0.7pt} \envC{\hspace{-0.39pt}\op{id}}. \hspace{2.5pt}\lcondN{\sq\hspace{-0.5pt}}{\hspace{-2pt}P\hspace{-0.9pt}}{\hspace{0.2pt}p\hspace{0.2pt}}{\hspace{-0.9pt}\bot} \hspace{-0.2pt}\wedge\hspace{-0.2pt}
\forall i. \hspace{2.9pt}\stateOf{\sq_i}\hspace{-0.5pt} \notin\hspace{-0.9pt} Q$.
In the latter case we however avoid propagation of the invariant
$\forall i. \hspace{2.9pt}\stateOf{\sq\pr_i}\hspace{-0.5pt} \notin\hspace{-0.9pt} Q$ 
explicitly down to each computation $\sq\pr$ that arises by successive application of the rules.

It is beneficial to keep track of environment conditions instead. To this end,
Definition~\ref{def:fcomp-fromN} and Definition~\ref{def:fcomp-fromT} shall respectively be updated to: 
\begin{definition}\label{def:fcomp-fromN-mod}
  $\hspace{-2pt}\lcondN{\sq\hspace{-0.1pt}}{\hspace{-1.7pt} R\hspace{0.9pt}}{\hspace{0.2pt}P\hspace{-0.2pt}}{\hspace{0.79pt}p}\stsp$ holds iff
  $\sq\hspace{-0.4pt} \in\hspace{-0.5pt} \pcsf{p}\hspace{-2pt}  \cap\hspace{0.5pt} \envC{\hspace{-0.9pt}R} \hspace{0.79pt}\cap\hspace{0.59pt} \inC{\hspace{-0.9pt}P}
  \hspace{-0.2pt}\cap\hspace{0.55pt} \outC{\hspace{-0.5pt}\bot}$
  (note that $\sq\hspace{-0.5pt} \in\hspace{-0.4pt} \outC{\hspace{-0.5pt}\bot}$ is by Definition~\ref{def:inC-outC} equivalent to $\progOf{\sq_i}\hspace{-1pt} \neq\hspace{-0.5pt} \skipp$ for all $i\hspace{0.7pt} \in\hspace{-0.2pt} \naturals$).
  \end{definition}
\begin{definition}\label{def:fcomp-fromT-mod}
  Let additionally $n\hspace{0.25pt} \in\hspace{-0.25pt} \naturals$.
  Then $\lcondT{\sq\hspace{-0.2pt}}{\hspace{-1.9pt} R\hspace{0.51pt}}{\hspace{-0.29pt}P\hspace{-0.5pt}}{\hspace{0.29pt}p\hspace{0.4pt}}{\hspace{0.9pt}n}\stsp$ holds iff
  $\sq\hspace{0.1pt}  \in\hspace{-0.2pt} \pcsi{p} \hspace{-1.7pt} \cap\hspace{0.25pt}\envC{\hspace{-1.27pt}R}\hspace{0.31pt} \cap\hspace{0.21pt} \inC{\hspace{-1.2pt}P}$
  and $n$ is the least number such that $\progOf{\sq_n}\hspace{-0.5pt} = \skipp$.
\end{definition}
Adjusting the syntax-driven rules in Section~\ref{S:refute} to this setting, 
only the rule for the parallel composition cannot simply pass the environment condition parameter to the emerging computations.
More precisely, the proof of Proposition~\ref{thm:parallel2-live} derives from $\sq\hspace{-1pt} \in\hspace{-1.2pt} \pcsf{p\! \parallel\! q}$
two computations $\sq^p\hspace{-1.2pt} \in\hspace{-1pt} \pcsf{p}$ and $\sq^q\hspace{-0.5pt} \in \pcsf{q}$ such that program steps of $\sq^p$
become supplementary environment steps on $\sq^q$ and \emph{vice versa}. For that reason we cannot draw any particular conclusions
regarding environment conditions on $\sq^p$ and $\sq^q$ from the assumption $\sq\hspace{-0.5pt} \in\hspace{-0.5pt} \envC{\hspace{-1.2pt}R}$ alone,
but assuming additionally the extended Hoare triples
\begin{eqnarray}\label{concl:guars1}
  \begin{aligned}
 &\hspace{-7pt} \rgvalid{\top}{\!\top}{\stnsp p\hspace{-0.5pt}}{\top}{\!G_{\hspace{-0.5pt}p}} \hspace{-4pt} & \mbox{with some } G_{\hspace{-0.5pt}p} \hspace{-0.7pt}\supseteq R\\
    &\hspace{-7pt}\rgvalid{\top}{\!\top}{\stnsp q\hspace{-0.5pt}}{\top}{\!G_q} & \mbox{with some } G_q\hspace{-0.7pt} \supseteq R
    \end{aligned}
\end{eqnarray}
would consequently provide
$\sq^p\hspace{-1.2pt} \in\hspace{-0.1pt} \envC{\hspace{1.5pt}G_q}$ and $\sq^q\hspace{-0.7pt} \in\hspace{-0.2pt} \envC{\hspace{1.79pt}G_p}$.
That is,
an adjusted rule for the parallel composition operator
shall derive
from the assumptions $\lcondN{\sq\hspace{-0.5pt}}{\hspace{-2.1pt}R\hspace{0.27pt}}{\hspace{-0.45pt}P\hspace{-0.9pt}}{\hspace{0.2pt}p\hspace{-1.5pt} \parallel\hspace{-1.2pt} q}$ and (\ref{concl:guars1})
that
there exist $\steqv{\sq^p\hspace{-1.5pt}}{\stnsp\sq}$ and
  $\steqv{\sq^q\hspace{-1.5pt}}{\stnsp\sq}$ such that either
  \begin{enumerate}
  \item[(a)] $\lcondN{\sq\mystrut^p\hspace{-1pt}}{\hspace{-1pt}G_q\hspace{-0.75pt}}{\hspace{-0.1pt}P\hspace{-0.5pt}}{\hspace{0.4pt}p}$ and
                    $\lcondN{\sq\mystrut^q\hspace{-1pt}}{\hspace{-1pt}G_{\hspace{-0.4pt}p}\hspace{-0.7pt}}{\hspace{-0.05pt}P\hspace{-0.5pt}}{\hspace{0.5pt}q}$, or
  \item[(b)] $\lcondN{\sq\mystrut^p\hspace{-1pt}}{\hspace{-1pt}G_q\hspace{-0.75pt}}{\hspace{-0.5pt}P\hspace{-0.59pt}}{\hspace{0.4pt}p}$ and there is some $n$ with
    $\lcondT{\sq\mystrut^q\hspace{-1pt}}{\hspace{-1pt}G_{\hspace{-0.5pt}p}\hspace{-0.9pt}}{\hspace{-0.25pt}P\hspace{-0.5pt}}{\hspace{0.37pt}q\hspace{0.1pt}}{\hspace{0.45pt}n}$, or
  \item[(c)] $\lcondN{\sq\mystrut^q\hspace{-1pt}}{\hspace{-1pt}G_{\hspace{-0.5pt}p}\hspace{-0.79pt}}{\hspace{-0.1pt}P\hspace{-0.7pt}}{\hspace{0.4pt}q}$ and there is some $n$ with
                    $\lcondT{\sq\mystrut^p\hspace{-1.2pt}}{\hspace{-1pt}G_q\hspace{-0.9pt}}{\hspace{-0.1pt}P\hspace{-0.9pt}}{\hspace{0.4pt}p\hspace{-0.05pt}}{\hspace{0.45pt}n}$.
   \end{enumerate}
Note that with $G_{\hspace{-0.5pt}p}\hspace{-1pt} = G_q\hspace{-1pt} =\! \top$ the assumptions (\ref{concl:guars1}) become valid for all $p$ and $q$ so that we in principle get back the part (1) of Proposition~\ref{thm:parallel2-live}.

Considering $\op{mutex}^{\mathit{aux}}$ once more but under a shifted angle, 
we may thus refute the assumption
$\lcondN{\sqt\hspace{-0.1pt}}{\hspace{-1pt}\op{id}\hspace{1pt}}{\hspace{1pt}\neg\sv{turn\_aux}_0 \stsp\wedge\hspace{-0.5pt} \neg\sv{turn\_aux}_1\stsp}{\stsp\op{mutex}^{\mathit{aux}}\hspace{0.7pt}\mv{cs}_0\hspace{1.9pt}\mv{cs}_1}$ for a fixed initial computation $\sqt$.
To apply the above parallel composition rule it is feasible to instantiate $G_{\hspace{-0.5pt}p}$ by $\op{G}_0$ and $G_q$ by $\op{G}_1$ where
\begin{eqnarray*}
  \begin{aligned}
  &\hspace{-5pt}  \op{G}_0 \hspace{-7pt}& \defeq\hspace{7pt} & \sv{turn\_aux}_1\hspace{-1pt} =\sv{turn\_aux}\pr_1 \hspace{-0.2pt}\wedge \sv{flag}_1\hspace{-1.5pt} =\sv{flag}\pr_1 \\
    &\hspace{-5pt}  \op{G}_1 \hspace{-7pt}& \defeq\hspace{7pt} & \sv{turn\_aux}_0\hspace{-1pt} =\sv{turn\_aux}\pr_0 \hspace{-0.2pt}\wedge \sv{flag}_0\hspace{-1.5pt} =\sv{flag}\pr_0
    \end{aligned}
\end{eqnarray*}
noting that both, $\op{G}_0$ and $\op{G}_1$, are reflexive state relations. As the result, we in particular would have to close the branch with 
two computations $\steqv{\sq^0\hspace{-1.5pt}}{\stnsp\sqt}$ and $\steqv{\sq^1\hspace{-1.5pt}}{\stnsp\sqt}\stsp$ such that
\begin{eqnarray*}\label{concl:triples}
  \begin{aligned}
 &\hspace{-5pt} \lcondN{\sq^0\hspace{-0.5pt}}{\hspace{-1.7pt}\op{G}_1\hspace{-0.79pt}}{\hspace{-0.2pt}\neg\sv{turn\_aux}_0\hspace{-1pt} \wedge\hspace{-1.5pt} \neg\sv{turn\_aux}_1}
       {\stsp\op{thread}^{\mathit{aux}}_0\hspace{1pt} \mv{cs}_0}\\
 &\hspace{-5pt}\lcondN{\sq^1\hspace{-0.5pt}}{\hspace{-1.7pt}\op{G}_0\hspace{-0.7pt}}{\hspace{-0.2pt}\neg\sv{turn\_aux}_0\hspace{-1pt} \wedge\hspace{-1.5pt} \neg\sv{turn\_aux}_1}
       {\stsp\op{thread}^{\mathit{aux}}_1\hspace{0.7pt} \mv{cs}_1}
   \end{aligned}
\end{eqnarray*}
hold and the set of hypothetical solutions to $\sq^0$ and $\sq^1$ gets strictly reduced
in comparison to Section~\ref{Sb:act-refute} due to the properties 
$\sq^0\hspace{-1.2pt} \in\hspace{-0.1pt} \envC{\hspace{-0.5pt}\op{G}_1}$ and $\sq^1\hspace{-1.4pt} \in\hspace{-0.1pt} \envC{\hspace{-0.5pt}\op{G}_0}$.
The merits of the refined technique will be highlighted in the next section.
\section{Abstraction and granularity}\label{Sb:arg}
As pointed out at the beginning of the case study in Chapter~\ref{S:PM1}, 
the conditions $\sv{flag}_1\hspace{-1.7pt} \wedge\hspace{-0.5pt} \sv{turn}$ and $\sv{flag}_0\hspace{-1pt} \wedge\hspace{-1.7pt} \neg\sv{turn}$,
deployed by $\op{mutex}$ to keep one of the threads spinning,
are actually examples of modelling that on the one hand makes the essentials significantly clearer, 
but on the other -- abstracts over interleaving to which an eventual implementation may be subjected    
because from the low-level code perspective an evaluation of \eg $\stdsp\sv{flag}_1\hspace{-2.1pt} \wedge\hspace{-0.9pt} \sv{turn}$ requires a series of atomic actions:
reading the respective Boolean values out of the memory, computing their conjunction and storing the result so that a conditional jump can correctly be performed.   
By then we in particular could not say whether the shared memory locations still contain the values that have been read out
at the start of such a `checkout' phase.

\newcommand{\chk}{\mv{checkout}_0}
\newcommand{\bcond}{\sv{mxcond}}
Expressing that in terms of $\langA{}$, we `unfold' $\op{thread}^{\mathit{aux}}_0$ to
the model shown in Figure~\ref{fig-PM-unfold} 
where the additional parameter $\chk$ shall be instantiated by
\begin{figure}
\[
\begin{array}{l}
\hspace{74pt}\sv{flag}_0 := \!\op{True};\\
\hspace{70pt}\langle \sv{turn} :=\! \op{True};\\
\hspace{75pt}  \sv{turn\_aux}_0 :=\hspace{-1.5pt} \op{True}\rangle; \\
\hspace{74pt}\chk;\\
\hspace{74pt}\whileS{\stdsp\bcond_0\stsp}{\chk};\\
\hspace{74pt}\mv{cs}_0; \\
\hspace{74pt}\sv{flag}_0\hspace{-1pt} := \op{False}
\end{array}
\]
\caption{$\op{thread}^{\mathit{aux}}_0$ `unfolded' .}
  \label{fig-PM-unfold}
\end{figure}
\begin{equation}
  \label{concl:inst}
\begin{array}{l l l l l}
 \hspace{-5pt}\chk &\hspace{-4pt} \mapsto &  \sv{local\_flag}_0 \hspace{7pt} &\hspace{-9pt} := &\hspace{-4pt} \sv{flag}_1;\\
 & & \sv{local\_turn}_0 \hspace{2.1pt} &\hspace{-9pt} := &\hspace{-4pt}  \sv{turn};\\
& &\bcond_0 \hspace{0.5pt} &\hspace{-9pt} := &\hspace{-4pt} \sv{local\_flag}_0\hspace{-0.7pt} \wedge\hspace{-0.7pt} \sv{local\_turn}_0
\end{array}
\end{equation}
with $\op{thread}^{\mathit{aux}}_1$ modified accordingly using $\mv{checkout}_1$.

Thus, the model resides at a finer level of granularity than $\op{mutex}^{\mathit{aux}}$ in Figure~\ref{fig:pm-aux},
and the conclusion that it retains the properties, 
presented in Proposition~\ref{thm:mutex-aux} and Proposition~\ref{thm:PM3-aux},
is not at all immediate:
a state satisfying both, $\sv{flag}_1\stnsp$ and $\sv{turn}$, is not explicitly required anymore
to keep  $\op{thread}^{\mathit{aux}}_0$ spinning in the `busy waiting' phase so that especially
the liveness reasoning in Chapter~\ref{S:PM3} demands a revision.
The crucial point to this end is that based on the enhancements to the refutational approach outlined in the preceding section,
it is
strikingly more expedient 
to leave $\chk$ as another parameter alongside $\mv{cs}_0$ specifying a variety of its admissible instances as follows. 

Let $\op{R}_0$ from now on denote the condition that the state variables $\sv{local\_flag}_0$, $\sv{local\_turn}_0$ and $\bcond_0$ remain unmodified,
whereas $\op{G}_0$ -- that all other variables remain unmodified.
Symmetrically we have $\op{R}_1$ and $\op{G}_1$ for the state variables $\sv{local\_flag}_1$, $\sv{local\_turn}_1, \bcond_1$ and their complement.
Then the assumptions \\[9pt]
\begin{tabular}{l l}
\hspace{-5pt}(1) &\hspace{-7pt} $\rgvalid{\top}{\!\top}{\chk}{\top}{\!\op{G}_0}$\\[2pt]
\hspace{-5pt}(2) &\hspace{-7pt} $\rgvalid{\op{R}_0}{\hspace{-1pt}\bcond_0}{\chk}{\top}{\bcond\pr_0 \imp\hspace{0.7pt} \bcond_0}$\\[2pt]
\hspace{-5pt}(3) &\hspace{-7pt} $\rgvalid{(\sv{turn} = \sv{turn}\pr)\hspace{-0.5pt} \wedge\hspace{-1.1pt} \op{R}_0}{\hspace{-1pt}\top}{\hspace{-0.5pt}\chk\hspace{-0.5pt}}{\stsp \bcond_0 \imp\hspace{0.7pt} \sv{turn}}{\hspace{-1.7pt}\top}$\\[2pt]
\hspace{-5pt}(4) &\hspace{-7pt} $\rgvalid{(\sv{flag}\pr_1\hspace{-1pt} \imp\hspace{0.4pt} \sv{flag}_1)\hspace{-0.5pt} \wedge\hspace{-1pt} \op{R}_0}{\hspace{-1.2pt} \neg\sv{flag}_1}{\hspace{-1pt}\chk\hspace{-1pt}}{\neg\bcond_0}{\hspace{-2pt}\top}$\\[2pt]  
\hspace{-5pt}(5) &\hspace{-7pt} $\forall\sq\:P\:R.\hspace{2.5pt}\NlcondN{\sq}{\hspace{-1.5pt}R\hspace{0.5pt}}{\hspace{0.2pt}P\hspace{-0.1pt}}{\hspace{-0.1pt}\chk}$.
\end{tabular}\\[9pt]
(and accordingly for $\mv{checkout}_1$) are 
sufficient to prove termination of the model parameterised by $\chk$ and $\mv{checkout}_1$,
being in particular resolvable by the instantiation (\ref{concl:inst}).
What is more, this abstraction makes the verification process significantly more structured and concise 
as will be sketched below
keeping to a refutation plan that in many aspects leans on the plan used in Chapter~\ref{S:PM3}.

Utilising the enhancement to the refutational approach 
drafted in the preceding section, we thus seek to refute the assumption
\begin{equation}\label{concl:init-asm}
\hspace{-9pt}\lcondN{\sqt\hspace{0.2pt}}{\hspace{-1.2pt}\op{id}\hspace{0.2pt}}{\hspace{0.4pt}\op{P}\hspace{-0.2pt}}{\stsp\op{mutex}^{\mathit{aux}}\hspace{1pt}\chk\hspace{3.1pt}\mv{cs}_0\hspace{2.7pt}\mv{checkout}_1\hspace{1.7pt}\mv{cs}_1}
\end{equation}
with a fixed initial computation $\sqt$ and $\op{P}$ denoting the state predicate
\[
 \hspace{-10pt}\neg\sv{turn\_aux}_0 \wedge\hspace{-1.5pt} \neg\sv{turn\_aux}_1\hspace{-0.5pt}\wedge \bcond_0 \wedge\hspace{-0.5pt} \bcond_1.\]
 Provided by the assumptions (1) and $\rgvalid{\top}{\!\!\top}{\!\mv{checkout}_1\!}{\top}{\!\!\op{G}_1}$,
 the global guarantees $\op{G}_{\stnsp\mv{global}}$ (\cf (\ref{eq:pm33}) in Chapter~\ref{S:PM3})
 \begin{equation*}\label{concl:pm33}
\hspace{-21pt}\rgvalidi{\op{id}}{\hspace{-1.5pt}\op{P}}{\op{mutex}^{\mathit{aux}}\hspace{1pt}\chk\hspace{3.1pt}\mv{cs}_0\hspace{2.7pt}\mv{checkout}_1\hspace{1.7pt}\mv{cs}_1 }{\top}{\hspace{-2pt}\op{G}_{\stnsp\mv{global}}}
 \end{equation*}
 remain derivable and apply thus to $\sqt$.
 Recall that $\op{G}_{\stnsp\mv{global}}$ in particular contains $\sv{turn\_aux}_0\hspace{-0.4pt} \wedge\hspace{0.3pt} \sv{turn\_aux}_1\hspace{-1pt} \imp\hspace{0.5pt} \sv{turn} = \sv{turn}\pr$ and
$\sv{turn\_aux}_0\hspace{-1pt} \imp\hspace{0.5pt} \sv{turn\_aux}\pr_0$ and $\sv{turn\_aux}_1\hspace{-1pt} \imp\hspace{0.5pt} \sv{turn\_aux}\pr_1$.
In the current setting we can moreover derive
 \begin{eqnarray}\label{concl:guar1}
   \begin{aligned}
   &\hspace{-25pt} \rgvalidi{\top}{\!\top}{\op{thread}^{\mathit{aux}}_0\hspace{1pt}\chk\hspace{3pt} \mv{cs}_0}{\top}{\!\op{R}_1} \\
     &\hspace{-25pt} \rgvalidi{\top}{\!\top}{\op{thread}^{\mathit{aux}}_1\hspace{1pt}\chk\hspace{3pt} \mv{cs}_1}{\top}{\!\op{R}_0}
   \end{aligned}
 \end{eqnarray}
 and
  \begin{eqnarray}\label{concl:guar2}
   \begin{aligned}
   &\hspace{-25pt} \rgvalidi{\op{R}_0}{\!\bcond_0}{\op{thread}^{\mathit{aux}}_0\hspace{0.5pt}\chk\hspace{2pt} \mv{cs}_0\hspace{-1pt}}{\top}{\!\bcond\pr_0\hspace{-1pt}\imp\hspace{0.2pt} \bcond_0} \\
     &\hspace{-25pt} \rgvalidi{\op{R}_1}{\!\bcond_1}{\op{thread}^{\mathit{aux}}_1\hspace{0.5pt}\chk\hspace{2pt} \mv{cs}_1\hspace{-1pt}}{\top}{\!\bcond\pr_1\hspace{-1pt} \imp\hspace{0.2pt} \bcond_1}
     \end{aligned}
   \end{eqnarray}
Using (\ref{concl:guar1}), 
an application of the parallel composition rule, sketched in the preceding section, 
to (\ref{concl:init-asm}) accordingly yields two computations $\sq^0$ and $\sq^1$ with $\steqv{\sq^0\hspace{-1.9pt} }{\hspace{-0.7pt} \sqt}$
and $\steqv{\sq^1\hspace{-1.9pt} }{\stnsp \sqt}$ such that either
\begin{enumerate}
  \item[(a)] $\lcondN{\sq^0\hspace{-1pt}}{\hspace{-1.5pt}\op{R}_0\hspace{-0.5pt}}{\hspace{0.2pt}\op{P}\hspace{-0.4pt}}
{\hspace{0.59pt}\op{thread}^{\mathit{aux}}_0\hspace{1pt}\chk\hspace{2.1pt} \mv{cs}_0} \mbox{ and }\\[3pt] 
              \lcondN{\sq^1\hspace{-1pt}}{\hspace{-1.5pt}\op{R}_1\hspace{-0.5pt}}{\hspace{0.2pt}\op{P}\hspace{-0.4pt}}
                     {\hspace{0.59pt}\op{thread}^{\mathit{aux}}_1\hspace{1pt}\mv{checkout}_1\hspace{2pt} \mv{cs}_1} \mbox{, or }$
  \item[(b)] $\lcondN{\sq^0\hspace{-1.1pt}}{\hspace{-1.7pt}\op{R}_0\hspace{-0.7pt}}{\hspace{0.2pt}\op{P}\hspace{-0.5pt}}
{\hspace{0.59pt}\op{thread}^{\mathit{aux}}_0\hspace{1pt}\chk\hspace{2.1pt} \mv{cs}_0} \mbox{ and there is some } n \mbox{ with }\\[3pt] 
              \lcondT{\sq^1\hspace{-0.7pt}}{\hspace{-1.45pt}\op{R}_1\hspace{-0.5pt}}{\hspace{0.15pt}\op{P}\hspace{-0.59pt}}
                     {\hspace{0.59pt}\op{thread}^{\mathit{aux}}_1\hspace{1pt}\mv{checkout}_1\hspace{2pt} \mv{cs}_1}{\hspace{0.7pt}n} \mbox{, or }$
   \item[(c)] $\lcondN{\sq^1\hspace{-1pt}}{\hspace{-1.92pt}\op{R}_1\hspace{-0.9pt}}{\hspace{0.2pt}\op{P}\hspace{-0.7pt}}
{\stsp\op{thread}^{\mathit{aux}}_1\hspace{1pt}\mv{checkout}_1\hspace{2pt} \mv{cs}_1} \mbox{ and there is some } n \mbox{ with }\\[3pt] 
              \lcondT{\sq^0\hspace{-1pt}}{\hspace{-1.7pt}\op{R}_0\hspace{-0.2pt}}{\hspace{0.05pt}\op{P}\hspace{-0.7pt}}
                     {\hspace{0.5pt}\op{thread}^{\mathit{aux}}_0\hspace{1pt}\mv{checkout}_0\hspace{2.1pt} \mv{cs}_0}{\hspace{0.7pt}n}$.
   \end{enumerate}

An argumentation closing the branch (a) can be briefly outlined as follows.
Processing both statements in (a) up to $\mv{checkout}_0$ and $\mv{checkout}_1$  using the syntax-driven rules 
yields 
a state on $\sq^0$ where $\sv{turn\_aux}_0$ holds and a state on $\sq^1$ where $\sv{turn\_aux}_1$ holds.
Projected back onto $\sqt$,
we once more (\cf Chapter~\ref{S:PM3})
get some $d$ such that
$\sv{turn\_aux}_0\hspace{-1.1pt} \wedge\hspace{-1.1pt}  \sv{turn\_aux}_1$ and hence also $\sv{turn} = \sv{turn}\pr$
hold globally on the suffix $\suffix{d}{\sqt}$.

Propagating this result the other way round along the state equivalence $\steqv{\sq^0\hspace{-1.9pt} }{\hspace{-0.7pt} \sqt}$\stdsp,
for any two states $\sigma$, $\sigma\pr$ on the suffix $\suffix{d}{\sq^0}$
we have  
$\sigma\sv{turn} = \sigma\pr\sv{turn}$.
We clearly get the same result for the suffix $\suffix{d}{\sq^1}$ due to $\steqv{\sq^1\hspace{-2.5pt} }{\hspace{-1pt} \sqt}$.
Next two paragraphs focus on drawing further conclusions for $\op{thread}^{\mathit{aux}}_0 \hspace{1pt}\mv{checkout}_0\hspace{2.1pt} \mv{cs}_0$.

Unconditional termination for all instances of $\mv{cs}_0$ 
has been assumed in Chapter~\ref{S:PM3} whereas the assumption (5) additionally
asserts the same for all instances of $\chk$.
Thus
$\lcondN{\sq^0\hspace{-0.9pt}}{\hspace{-1.5pt}\op{R}_0\hspace{0.25pt}}{\hspace{0.9pt}\op{P}\hspace{0.2pt}}
{\hspace{1.2pt}\op{thread}^{\mathit{aux}}_0\hspace{0.7pt}\chk\hspace{2.1pt} \mv{cs}_0}$
can hold only due to non-termination of the `busy waiting' phase $\whileS{\stsp\bcond_0}{\stnsp\chk\stnsp}$.
More precisely, processing $\lcondN{\sq^0\hspace{-1.2pt}}{\hspace{-1.7pt}\op{R}_0\hspace{-0.29pt}}{\hspace{0.39pt}\op{P}\hspace{-0.45pt}}
{\hspace{0.59pt}\op{thread}^{\mathit{aux}}_0\hspace{0.5pt}\chk\hspace{2.1pt} \mv{cs}_0}$
using syntax-driven rules
yields a strictly ascending infinite sequence $\phi$ such that 
for all $i\hspace{0.7pt}\in\hspace{-0.5pt}\naturals$
we have a computation $\steqv{\prescript{i}{}{\sq}\hspace{-0.7pt}}{\hspace{-0.7pt}\suffix{\phi_i}{\sq^0}}\stsp$
with $\lcondT{\prescript{i}{}{\sq}\hspace{-0.2pt}}{\hspace{-1.9pt}\op{R}_0\hspace{-0.25pt}}{\hspace{0.7pt}\bcond_0\hspace{-0.2pt}}{\hspace{0.1pt}\chk\hspace{-0.3pt}}{\hspace{0.5pt}\prescript{i}{}{n}}$
for some $\prescript{i}{}{n}\hspace{-0.5pt} <\hspace{-0.2pt} \phi_{i+1}\hspace{-0.5pt} - \phi_i$
(the notations $\prescript{i}{}{\sq}$ and $\prescript{i}{}{n}$ shall just emphasise
that the values are actually parameterised by $i$).
From this we can further infer
\begin{enumerate}
\item[(i)]\hspace{-2.9pt}$\lcondT{\prescript{d}{}{\sq}\hspace{-0.5pt}}{\hspace{-1.5pt}(\sv{turn}\hspace{-0.1pt} =\hspace{-0.4pt} \sv{turn}\pr)\hspace{-1.25pt} \wedge\hspace{-1.29pt} \op{R}_0\hspace{-0.1pt}}{\hspace{0.7pt}\bcond_0}{\hspace{0.2pt}\chk}{\hspace{0.7pt}\prescript{d}{}{n}}\stsp$ with  $\stsp\steqv{\prescript{d}{}{\sq}\hspace{-1pt}}{\hspace{-0.2pt}\suffix{\phi_d}{\sq^0}}$
\item[(ii)]\emph{for all states $\sigma$ on $\sq^0$ we have $\sigma\bcond_0 \hspace{-1pt}=\hspace{-2pt} \op{True}$}.
\end{enumerate}
To (i) note that $\sigma\sv{turn} = \sigma\pr\sv{turn}$ holds for any $\sigma$, $\sigma\pr$ on $\suffix{d}{\sq^0}$ and
for environment steps on $\prescript{d}{}{\sq}$ 
in particular, since $d \le \phi_d$.
Regarding (ii), if there would be a state $\sigma^0_h = \stateOf{\sq^0_h}$ with $\neg\sigma^0_h\bcond_0$
then using (\ref{concl:guar2}) we can derive $\neg\sigma^0_j\bcond_0$ for all $j \ge h$ %
and hence for $\phi_h\hspace{-1.5pt} \ge\hspace{-0.5pt} h$ in particular, which is a contradiction since the state
$\stateOf{\sq^0_{\phi_h}}$ 
satisfies $\bcond_0$.

Now with (i) we can take the decisive advantage of the $\chk$ abstraction:
the triple (3) ensures that for all $\sq\hspace{-1.2pt} \in\hspace{-1pt} \pcs{\chk}\hspace{-0.9pt} \cap\hspace{-0.1pt} \envC{\!((\sv{turn}\hspace{-0.9pt} =\hspace{-1.1pt} \sv{turn}\pr)\hspace{-1pt} \wedge\hspace{-1.9pt} \op{R}_0)}$ we have
$\sq\hspace{-0.2pt} \in\hspace{-0.2pt} \outC{\hspace{-1.5pt}(\bcond_0\hspace{-1pt} \imp\hspace{0.5pt} \sv{turn})}$ whereas 
$\prescript{d}{}{\sq}$ not only meets the relevant conditions but also reaches a $\skipp$-configuration in $\prescript{d}{}{n}$ steps by (i).
Hence we have a state $\sigma$ on $\suffix{\phi_d}{\sq^0}$ such that $\sigma\sv{turn}\hspace{-2pt} = \hspace{-4.1pt}\op{True}$ holds due to (ii).
Then a projection back onto $\sqt$
yields that $\sv{turn}$ must hold globally on $\suffix{d}{\sqt}$
because the value of $\sv{turn}$ is constant on $\suffix{d}{\sqt}$.

Recasting the argumentation in the last two paragraphs to $\op{thread}^{\mathit{aux}}_1$
accordingly yields that $\neg\sv{turn}$ also holds globally on $\suffix{d}{\sqt}$.
Altogether we particularly get $\stateOf{\sqt_d}\hspace{-0.2pt} \in\hspace{-1.2pt} \bot$
and hence the branch (a) is closed.

As in Chapter~\ref{S:PM1} and Chapter~\ref{S:PM3} we can establish a program correspondence
between the model in Figure~\ref{fig-PM-unfold} and
\[
\begin{array}{l}
\hspace{74pt}\sv{flag}_0 := \!\op{True};\\
\hspace{74pt}\sv{turn} :=\! \op{True};\\
\hspace{74pt}\chk;\\
\hspace{74pt}\whileS{\stdsp\bcond_0\stsp}{\chk};\\
\hspace{74pt}\mv{cs}_0; \\
\hspace{74pt}\sv{flag}_0\hspace{-1pt} := \op{False}
\end{array}
\]
translating the relevant protocol properties along the correspondence.
Applying subsequently the transformations from Section~\ref{Sb:seq-norm} and Section~\ref{Sb:jump-norm}
in order to replace the while-statements by jumps in both threads, we in principle arrive at
the componentwise equivalent representation 
\[
\hspace{-5pt}\begin{array}{l l || l l}
 &\hspace{-7pt} \sv{flag}_0 := \hspace{-1pt}\op{True}; &  & \hspace{-7pt}\sv{flag}_1 := \hspace{-1pt}\op{True};\\
 &\hspace{-7pt}  \sv{turn} := \hspace{-0.7pt}\op{True};   &  & \hspace{-7pt}\sv{turn} := \op{False}; \\
  &\hspace{-7pt} \chk;                     &  &\hspace{-7pt} \mv{checkout}_1; \\
  1: &\hspace{-7pt} \com{cjump}\:\neg\bcond_0\hspace{2.9pt}\com{to}\hspace{2.5pt}2     & \hspace{2pt} 3: &\hspace{-7pt} \com{cjump}\:\neg\bcond_1\;\com{to}\;4 \\ 
     &\hspace{-4.5pt} \com{otherwise}\;\chk;\jump{1}\;\com{end};  &  & \hspace{-4.5pt}\com{otherwise}\;\mv{checkout}_1;\jump{3}\;\com{end}; \\
2: &\hspace{-7pt} \mv{cs}_0;                                    & \hspace{2pt} 4: &\hspace{-7pt} \mv{cs}_1;\\
   &\hspace{-7pt} \sv{flag}_0 \hspace{-1pt}:= \op{False}        &     &\hspace{-7pt} \sv{flag}_1 \hspace{-1pt}:= \op{False}
\end{array}
\]
which is notably closer to a template for provably correct assembly level implementations
derived from the
Peterson's mutual exclusion algorithm.

\backmatter

\bibliographystyle{unsrt}
\bibliography{framework.bib}

\printindex
\end{document}